\documentclass[11pt]{article}
\usepackage{style}
\bibliography{refs}

\begin{document}
	
\title{Polynomial and analytic methods for classifying complexity of planar graph homomorphisms}
\author{
  Jin-Yi Cai\\
  \texttt{jyc@cs.wisc.edu}
  \and
  Ashwin Maran\\
  \texttt{amaran@wisc.edu}
}
\date{}
\maketitle

\begin{abstract}
    We introduce some polynomial and analytic methods in the  classification program
    for the complexity of planar graph homomorphisms. These methods allow us to
    handle infinitely many lattice conditions and isolate the new P-time tractable
    matrices represented by tensor products of matchgates. We use these methods
    to prove a complexity dichotomy for $4 \times 4$ matrices that says Valiant's
    holographic algorithm is universal for planar tractability in this setting.
\end{abstract}
\thispagestyle{empty}

\newpage

\tableofcontents

\thispagestyle{empty}
\clearpage
\pagenumbering{arabic}

\section{Introduction}\label{sec: introduction}

Given graphs $G$ and $H$, a mapping from $V(G)$ to $V(H)$ is called
a \textit{homomorphism} if edges of $G$ are mapped to edges of $H$.
This is put in a more general or quantitative setting by the notion of a \textit{partition function}.
Let  $M = (M_{ij})$ be a $q \times q$ symmetric matrix.
In this paper we consider non-negative arbitrary real entries
$M_{ij} \in \mathbb{R}_{\geq 0}$; 
if $M_{ij} \in \{0, 1\}$, then $M$ is 
the unweighted adjacency matrix of a graph
$H = H_M$. Given $M$,
the partition function  $Z_{M}(G)$ for any input undirected multi-graph $G = (V, E)$
is defined as
$$Z_{M}(G) = \sum_{\sigma: V \rightarrow [q]} \prod_{(u, v) \in E} M_{\sigma(u) \sigma(v)}.$$
Obviously isomorphic graphs $G \cong G'$ have the same value $Z_{M}(G) = Z_{M}(G')$, and thus every $M$
defines a graph property $Z_{M}(\cdot)$. For a 0-1 matrix $M$,
$Z_{M}(G)$ counts the number of homomorphisms from $G$ to $H$.
Graph homomorphism (GH) encompasses a great deal of graph properties~\cite{lovasz2012large}.

Each $M$ defines a
computational problem, denoted by $\EVAL(M)$: given an input graph $G$ compute $Z_{M}(G)$.
The  complexity of 
$\EVAL(M)$ has been a major focus of research.
A number of increasingly general complexity dichotomy theorems have been achieved.
Dyer and Greenhill \cite{dyer2000complexity} proved that for any 0-1 symmetric matrix $M$, computing $Z_{M}(G)$ is either in P-time or is $\#$P-complete. Bulatov and Grohe \cite{bulatov2005complexity}  found a complete classification for $\EVAL(M)$ for all nonnegative matrices $M$. 
Goldberg, Grohe, Jerrum, and Thurley \cite{goldberg2010complexity} then proved a dichotomy for all real-valued matrices $M$. Finally, Cai, Chen, and Lu \cite{cai2013graph} established a dichotomy for all complex valued matrices $M$.
We also note that graph homomorphism can be viewed as a special case of
counting CSP.
For counting CSP,
a series of results established a complexity dichotomy for any set of constraint functions
${\cal F}$, going from 
0-1 valued~\cite{bulatov2013complexity,dyer2010complexity,dyer2011decidability,dyer2013complexity} to nonnegative rational valued~\cite{bulatov2012complexity}, to
nonnegative real valued~\cite{cai2016nonnegative}, to all complex valued functions~\cite{cai2017complexity}.

Parallel to this development, Valiant~\cite{valiant2008holographic} introduced
\emph{holographic algorithms}. It is well known that counting the number of perfect
matchings (\#PM) is \#P-complete~\cite{valiant1979complexity}. On the other hand,
since the 60's, there has been a famous FKT algorithm~\cite{kasteleyn1961statistics,temperley1961dimer,kasteleyn1963dimer,kasteleyn1967graph} that can
compute \#PM on planar graphs in P-time. Valiant's holographic algorithms
greatly extended its reach, in fact so much so that a most intriguing question arises: Is this   a \emph{universal} algorithm that 
\emph{every} counting problem expressible as a sum-of-products
that \emph{can be solved} in P-time on planar graphs
(but \#P-hard in general) \emph{is solved} by this method alone?
Such a universality statement must appear to be extraordinarily, if not overly, ambitious.


After a series of work~\cite{cai2009holant,cai2016complete,cai2019holographic, backens2017new, backens2018complete, yang2022local, fu2019blockwise, fu2014holographic, cai2019holographic}
it was established that for every set of complex valued constraint functions ${\cal F}$
on the Boolean domain (i.e., $q=2$) there is a 3-way exact
classification for  
\#CSP(${\cal F}$): 
 (1) P-time solvable, 
 (2) P-time solvable over planar graphs but \#P-hard over general graphs, 
 (3) \#P-hard over planar graphs.
Moreover,
category (2) consists of precisely  those problems that can be solved by
Valiant's holographic algorithm using FKT.
Cai and Maran~\cite{cai2023complexity} showed that for $\EVAL$ the
same 3-way exact classification holds even on the domain $q=3$,
and category (2) again consists of precisely those problems that can be
solved by Valiant's holographic algorithm using FKT.
So far little is known 
for higher domain problems ($q >3$) on this universality question.

Let $\PlEVAL(M)$ denote the problem $\EVAL(M)$ 
when the input graphs are restricted to planar graphs.
Planar GH is also intimately related to 
quantum isomorphism of graphs, a relaxation of classical isomorphism~\cite{atserias2019quantum}.
It is known that graphs $H$ and $H'$ are quantum isomorphic iff there is a perfect
winning strategy in a two-player graph isomorphism game in which the players share and
 perform measurements on an entangled quantum state. This is also equivalent to
the existence of a quantum permutation matrix transforming $H$ to $H'$~\cite{lupini2020nonlocal}.
Let $M$ and $M'$ be the adjacency matrices of $H$ and $H'$. 
Man\v{c}inska and Roberson~\cite{manvcinska2020quantum} proved that $H$ and $H'$ are quantum isomorphic 
iff $Z_M(G) = Z_{M'}(G)$ for every planar graph $G$, i.e.,  $H$ and $H'$  define the same Planar GH 
problem. Furthermore, a fascinating consequence of this line of work is that
it is undecidable whether $\PlEVAL(M) = \PlEVAL(M')$~\cite{manvcinska2020quantum}, which hints
at the difficulty that we face in this paper.

%
Our goal is to classify the complexity of  $\PlEVAL(M)$, i.e., when 
the input $G$ is restricted to be planar for $Z_M(G)$.
(The underlying graph $H_M$ is not restricted to planar graphs.)
%
We want to classify the problems $\PlEVAL(M)$:
What is the computational complexity of $Z_M(G)$ from  planar input graphs $G$?
We present some strong polynomial and analytic techniques that will
help us approach this problem.
We demonstrate the power of these techniques by giving a complete classification of the complexity of
$\PlEVAL(M)$ for all non-negative real valued
full rank $4 \times 4$ matrices
$M$.
The full rank $4 \times 4$ case is particularly important as it is the first case where tensor
product of matchgates~\cite{valiant2008holographic,valiant2001quantum} defines new P-time tractable problems.
We prove that an exact classification according to the three categories above holds
for this class,
and a holographic reduction to FKT remains 
a \emph{universal} algorithm for category (2).
\section{Preliminaries}\label{sec: preliminaries}

\subsection{Model of Computation}\label{sec: modelComputation}

The Turing machine model is naturally suited to the study of computation over discrete structures such as integers or graphs. 
When  $M \in \mathbb{R}^{q \times q}$, for  $\PlEVAL(M)$ 
one usually restricts $M$ to be a matrix
with only algebraic numbers. This is strictly for the consideration
of the model of computation, even though allowing all
real-valued matrices would be more natural. 
 
 There is a formal (albeit nonconstructive) method to treat  $\PlEVAL(M)$ 
 for arbitrary real-valued matrices $M$ and yet stay strictly
 within the Turing machine model in terms of  bit-complexity.
In this paper, because our proof depends heavily on analytic
argument with continuous functions on the real line, this
logical formal view becomes necessary.

To begin with, we recall a theorem from field theory:
Every extension field ${\bf F}$ over  ${\mathbb Q}$
by a finite set of real numbers is  a finite algebraic extension ${\bf E}'$
of a certain  purely transcendental   extension field ${\bf E}$
over ${\mathbb Q}$,
which has the form ${\bf E} = {\mathbb Q}(X_1,
\ldots, X_m)$ where $m \ge 0$ and $X_1,
\ldots, X_m$ are algebraically independent~\cite{jacobson1985basic} (Theorem 8.35, p.~512).
${\bf F}$ is said to have
a finite transcendence degree $m$ over ${\mathbb Q}$.
It is known that $m$ is uniquely defined for ${\bf F}$.
Since 
$\rm{char}~{\mathbb Q} =0$,
the finite algebraic extension ${\bf E}'$ over ${\bf E}$
is actually simple, ${\bf E}' = {\bf E}(\beta)$ for some $\beta$,
and it is specified by a
minimal polynomial in ${\bf E}[X]$.
Now given a real matrix $M$, let ${\bf F} = {\mathbb Q}(M)$ 
be the extension field by adjoining the entries of $M$.
We  consider $M$ is fixed for the problem $\PlEVAL(M)$,
and thus we may assume (nonconstructively) that the form
${\bf F} = {\bf E}(\beta)$
 and ${\bf E} = {\mathbb Q}(X_1,
\ldots, X_m)$ are given. (This means, among other things,
that the minimal polynomial of $\beta$ over ${\bf E}$ is given,
and all arithmetic operations can be performed on ${\bf F}$.)

Now, the computational problem $\PlEVAL(M)$ is the following:
Given a planar 
$G$, compute $Z_{M}(G)$ as an element in ${\bf F}$
(which is expressed as a polynomial in $\beta$ with coefficients in ${\bf E}$).
More concretely, we can show that this is equivalent to the following problem $\COUNT(M)$:
The input  is a pair $(G,x)$,
  where $G=(V,E)$ is a planar graph and $x\in {\bf F}$.
The output is\vspace{-0.1cm}
$$
\text{\#}_{M}(G,x)= \Big|\big\{\sigma:V\rightarrow
  [q]\hspace{0.08cm}: \hspace{0.08cm} \prod_{(u, v) \in E} m_{\sigma(u), \sigma(v)}=x\big\}\Big|,\label{full_COUNTM}
  $$
a non-negative integer. Note that, in this definition,
we are basically combining terms with the same 
product value in the definition of $Z_{M}(G)$.

Let $n=|E|$.
Define  $X$ to be the  set of all possible product values
appearing in $Z_{M}(G)$:
\begin{equation}\label{full_definitionreuse}
X=\left\{\prod_{i,j\in [q]}m_{ij}^{k_{ij}}\hspace{0.08cm}\Big|\hspace{0.1cm}
  \text{integers $k_{ij}\ge 0$ and $\sum_{i,j\in [q]}k_{ij}=n$}
\right\}.
\end{equation}
There are $\binom{n+q^2-1}{q^2-1} = n^{O(1)}$ many integer sequences
$(k_{i,j})$ such that  $k_{i,j}\ge 0$ and $\sum_{i,j\in [q]}k_{i,j}=n$.
$X$
  is defined as a set, not a multi-set.
After removing repeated
elements the cardinality $|X|$ is also polynomial in $n$.
 For fixed and given ${\bf F}$
 the elements in $X$ can be enumerated in polynomial time in $n$.
 (It is important that ${\bf F}$ and $q$ are all treated as fixed constants.)
It then follows from the definition that
  $\text{\#}_{M}(G,x)= 0$ for any $x\notin X$.
This gives us the following relation:
\[Z_M(G)= \sum_{x\in X} x \cdot \text{\#}_{M}(G,x),\ \ \ \text{for any 
  graph $G$,}\]
and thus, $\PlEVAL(M)\le \COUNT(M).$

For the other direction,
  we construct, for any $p\in [|X|]$ (recall that $|X|$ is polynomial in $n$),
  a  planar graph $T_pG$ from $G$ by replacing
  every edge of $G$ with $p$ parallel edges.
Then,
$$
Z_M(T_pG)= \sum_{x\in X} x^p \cdot \text{\#}_{M}(G,x),\ \ \ \text{for any 
  graph $G$.}
$$
This is a Vandermonde system; it has full rank since
elements in $X$ are distinct by definition. So by
querying $\PlEVAL(M)$ for the
  values of $Z_{M}(T_pG)$,
  we can solve it in polynomial time
  and get $\text{\#}_{M}(G,x)$ for\vspace{0.0015cm} every non-zero $x\in X$.
To obtain $\text{\#}_{M}(G,0)$ (if $0\in X$), we note that
$$
\sum_{x\in X} \text{\#}_{M}(G,x) =q^{|V|}.
$$
This gives us a polynomial-time reduction and thus, $\COUNT(M)\le \PlEVAL(M)$.
We have proved
\begin{lemma}\label{lemma: full_count}
For any fixed  $M \in \mathbb{R}^{q \times q}$,  
   $\PlEVAL(M)\equiv \COUNT(M)$.
\end{lemma}
Thus, $\PlEVAL(M)$ can be identified with the 
problem of producing those polynomially many integer coefficients
in the canonical expression for $Z_M(G)$ as a sum
of (distinct) terms from $X$.

This  formalistic view has 
the advantage that we can treat the complexity 
of  $\PlEVAL(M)$ for  general $M$, and not restricted to algebraic numbers. 
 Thus, numbers such as $e$ or $\pi$
need not be excluded.
More importantly, in this paper this generality is essential, due to the proof technique that we employ.
Furthermore, once freed from this restriction we in fact explicitly use
transcendental numbers as a tool in our proof (see \cref{lemma: latticeHardnessDelta}).
In short, in this paper, treating the complexity of  $\PlEVAL(M)$ for  general real $M$
is not a \emph{bug} but a \emph{feature}.

However, we note that this treatment 
has the following subtlety.  
For the computational problem $\PlEVAL(M)$
the formalistic view demands that
${\bf F}$ be specified in the form ${\bf F} = {\bf E}(\beta)$.
Such a form exists, and its specification is of
constant size when  measured
in terms of the size  of the input graph $G$. 
However, 
in reality many basic questions for transcendental numbers
are unknown.
For example, it is still unknown whether $e + \pi$ or $e \pi$ are
rational, algebraic irrational or transcendental,
and it is open whether ${\mathbb Q}(e, \pi)$ has  transcendence degree
2 (or 1) over  ${\mathbb Q}$, i.e., whether $e$ and $\pi$ are algebraically
independent.
The formalistic view here non-constructively
assumes this information is given for ${\bf F}$.
A  polynomial time reduction $\Pi_1  \le \Pi_2$  from one problem
to another  in this setting merely
implies that the \emph{existence} of a polynomial time algorithm
for $\Pi_2$ logically implies the 
\emph{existence} of  a  polynomial time algorithm
for $\Pi_1$. We do not actually obtain such
an algorithm constructively. 

This logical detour not withstanding, if a reader
is  interested only in the complexity of $\PlEVAL(M)$ 
for integer matrices $M$, then the
complexity dichotomy proved in this paper holds
according to the standard definition of $\PlEVAL(M)$ 
for integral $M$ in terms of
the model of computation; the fact that this is proved 
in a broader setting for all real matrices $M$ 
is irrelevant. This  is akin to the situation
in analytic number theory, where one might be 
interested in a question strictly about the ordinary
integers, but the theorems are proved
in a broader setting of analysis.

\subsection{Definitions}
As we will refer to various different types of matrices
throughout this paper,
it will be helpful to establish some notation.
Given a positive integer $q \geq 1$,
we let $\text{Sym}_{q}(X)$
denote the set of $q \times q$ symmetric matrices
such that each entry of the matrix is from the set
$X \subset \mathbb{R}$.
For example, with $X = \mathbb{R}$
 (respectively, $\mathbb{R}_{\geq 0}$, or
$\mathbb{R}_{\neq 0}$)
$\text{Sym}_{q}(X)$ denotes $q \times q$ symmetric matrices with
 real  (respectively, non-negative, or non-zero) entries.
Similarly, we  let $\text{Sym}^{\tt{F}}_{q}(X)$ denote the
set of $q \times q$ full rank symmetric matrices
such that each entry of the matrix is from the set
$X \subset \mathbb{R}$, and let
$\text{Sym}_{q}^{\tt{pd}}(X)$ denote the set of
$q \times q$ positive definite symmetric matrices
with entries from $X$.
	
Now, consider some $M \in \text{Sym}_{q}(\mathbb{R})$
with entries $M_{ij} \in \mathbb{R}$ for $i, j \in [q]$.
Given a planar, undirected multi-graph $G = (V, E)$,
we can perform certain elementary operations (that preserve
planarity) on the graph $G$
to transform it into a new graph $G'$,
such that $Z_{M}(G') = Z_{M'}(G)$ for some matrix $M'$.
For most of this paper we will use two such operations,
thickening and stretching.

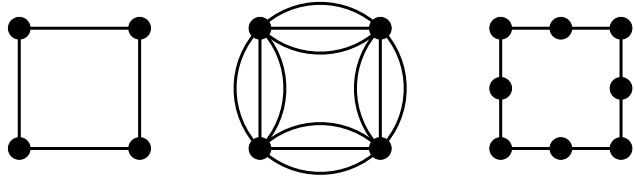
\begin{figure}[b]
	\centering
	\scalebox{0.8}{\begin{tikzpicture}[line join=miter, draw opacity=1]

\node[circle, draw=black, fill=black] (c) at (-1, -1){};
\node[circle, draw=black, fill=black] (c) at (-1, 1){};
\node[circle, draw=black, fill=black] (c) at (1, 1){};
\node[circle, draw=black, fill=black] (c) at (1, -1){};

\draw[line width=0.5mm, black] (-1, -1) -- (-1, 1);
\draw[line width=0.5mm, black] (-1, 1) -- (1, 1);
\draw[line width=0.5mm, black] (1, 1) -- (1, -1);
\draw[line width=0.5mm, black] (1, -1) -- (-1, -1);

\begin{scope}[xshift = 4cm]

\node[circle, draw=black, fill=black] (c) at (-1, -1){};
\node[circle, draw=black, fill=black] (c) at (-1, 1){};
\node[circle, draw=black, fill=black] (c) at (1, 1){};
\node[circle, draw=black, fill=black] (c) at (1, -1){};

\draw[line width=0.5mm, black] (-1, -1) -- (-1, 1);
\draw[line width=0.5mm, black] (-1, 1) -- (1, 1);
\draw[line width=0.5mm, black] (1, 1) -- (1, -1);
\draw[line width=0.5mm, black] (1, -1) -- (-1, -1);

\draw[line width=0.5mm, black] (-1, -1) to [in=225, out=135] (-1, 1);
\draw[line width=0.5mm, black] (-1, -1) to [in=-45, out=45] (-1, 1);

\draw[line width=0.5mm, black] (-1, 1) to [in=135, out=45] (1, 1);
\draw[line width=0.5mm, black] (-1, 1) to [in=225, out=-45] (1, 1);

\draw[line width=0.5mm, black] (1, 1) to [in=45, out=-45] (1, -1);
\draw[line width=0.5mm, black] (1, 1) to [in=135, out=225] (1, -1);

\draw[line width=0.5mm, black] (1, -1) to [in=45, out=135] (-1, -1);
\draw[line width=0.5mm, black] (1, -1) to [in=-45, out=225] (-1, -1);

\end{scope}

\begin{scope}[xshift = 8cm]

\node[circle, draw=black, fill=black] (c) at (-1, -1){};
\node[circle, draw=black, fill=black] (c) at (-1, 1){};
\node[circle, draw=black, fill=black] (c) at (1, 1){};
\node[circle, draw=black, fill=black] (c) at (1, -1){};

\draw[line width=0.5mm, black] (-1, -1) -- (-1, 1);
\draw[line width=0.5mm, black] (-1, 1) -- (1, 1);
\draw[line width=0.5mm, black] (1, 1) -- (1, -1);
\draw[line width=0.5mm, black] (1, -1) -- (-1, -1);

\node[circle, draw=black, fill=black] (c) at (-1, 0){};
\node[circle, draw=black, fill=black] (c) at (0, 1){};
\node[circle, draw=black, fill=black] (c) at (1, 0){};
\node[circle, draw=black, fill=black] (c) at (0, -1){};

\end{scope}

\end{tikzpicture}}
	\caption{A graph $G$, the thickened graph $T_{3}G$, and the stretched graph $S_{2}G$.}
	\label{fig:thickeningStretching}
\end{figure}

From any planar multi-graph $G = (V, E)$, and a positive integer $n$,
we can construct the planar multi-graph $T_{n}G$,
by replacing every edge in $G$ with $n$ parallel edges
between the same vertices.
This process is called \emph{thickening}.
Clearly $Z_{M}(T_{n}G) = Z_{T_{n}M}(G)$,
where $T_{n}M \in \text{Sym}_{q}(\mathbb{R})$
with entries $\big((M_{i, j})^{n}\big)$ for $i, j \in [q]$.
In particular, $\PlEVAL(T_{n}M) \leq \PlEVAL(M)$ for all $n \geq 1$.

Similarly, from any planar multi-graph $G = (V, E)$,
and a positive integer $n$, we can construct the
planar multi-graph $S_{n}G$ by replacing every edge
$e \in E$ with a path of length $n$.
This process is called \emph{stretching}.
It is also easily seen that  $Z_{M}(S_{n}G) = Z_{S_{n}M}(G)$,
where $S_{n}M = M^{n}$, the $n$-th power of $M$.
So, we also have  $\PlEVAL(S_{n}M) \leq \PlEVAL(M)$ for all $n \geq 1$.
\section{Reduction from the Potts Model}\label{sec: vertex_color}

We now consider the thickening operation more closely.
For $m, n \ge 1$, let
$$\mathcal{P}_{m}(n) = \left\{\mathbf{x} = (x_{i})_{i \in [m]}
\in \mathbb{Z}^{m} \hspace{0.08cm}\Big|\hspace{0.1cm}
(\forall\ i \in [m]) \ [x_{i} \geq 0]\ \mbox{ and }
\sum_{i \in [m]}x_{i} = n\right\}.$$
We note that
given any graph $G = (V, E)$,
\begin{equation}\label{equation: full_thickening}
    Z_{M}(T_{n}G) = Z_{T_{n}M}(G) = \sum_{x \in X(G)}x^{n} \cdot \text{\#}_{M}(G, x)
\end{equation}
where
\begin{equation}\label{equation: full_thickening-X(G)-set}
X(G) = \left\{\prod_{i,j\in [q]}M_{ij}^{k_{ij}}\hspace{0.08cm}\Big|\hspace{0.1cm}
\mathbf{k} = (k_{ij})_{i, j \in [q]} \in
\mathcal{P}_{q^{2}}(|E|)\right\},
\end{equation}
and
\begin{equation}\label{equation: full_thickening-number-of-mappings}
\text{\#}_{M}(G,x)= \Big|\big\{\sigma:V\rightarrow [q]\hspace{0.08cm}: \hspace{0.08cm} \prod_{(u, v) \in E} M_{\sigma(u), \sigma(v)}=x\big\}\Big|.
\end{equation}
Note that given any $x \in X(G)$,
$\text{\#}_{M}(G, x)$ does not depend on $n$,
but depends only on the entries of the matrix $M$.
We will deal with this dependence now.

\begin{definition}\label{definition: generatingSet}
    Let $\mathcal{A} \subseteq \mathbb{R}_{\neq 0}$
    be a set of non-zero real numbers.
    A finite set 
    $\{g_{t}\}_{t \in [d]} \in (\mathbb{R}_{> 1})^{d}$,
    for some integer $d \geq 0$, 
    is called a generating set of $\mathcal{A}$ if
    for every $a \in \mathcal{A}$, there exists a
    \emph{unique} $(e_0, e_{1}, \dots, e_{d}) \in \{0, 1\} \times {\mathbb{Z}}^{d}$ 
    such that $a = (-1)^{e_0} {g_{1}^{e_{1}} \cdots g_{d}^{e_{d}}}$.
\end{definition}

\begin{remark*}
The uniqueness of the exponents implies the following property of $\{g_{t}\}_{t \in [d]}$:
Whenever $\prod_{t\in [d]}g_t^{a_t} = \prod_{t\in [d]}g_t^{b_t}$ we have $a_t=b_t$ for all $t \in [d]$,
i.e., any such expression has unique exponents.
\end{remark*}

\begin{lemma}\label{lemma: generatingSet}
    Every finite set $\mathcal{A} \subset \mathbb{R}_{\neq 0}$
    of non-zero real numbers has
    a generating set.
\end{lemma}
\begin{proof}
    Consider the multiplicative group $\mathcal{G}$ 
    generated by 
    the positive real numbers 
    $\{|a| : a \in \mathcal{A}\}$.
    It is  a subgroup of the multiplicative group
    $(\mathbb{R}_{> 0}, \cdot)$.
    Since  $\mathcal{A}$ is finite, and 
    $(\mathbb{R}_{> 0}, \cdot)$ is torsion-free,
    the group $\mathcal{G}$ is  a finitely generated free Abelian group, 
    and thus isomorphic to 
    $\mathbb{Z}^d$ for some $d \ge 0$.
    Let $f$ be this isomorphism from $\mathbb{Z}^{d}$ to
    the multiplicative group.
    By flipping $\pm 1$ in $\mathbb{Z}$ we may assume that this isomorphism maps the basis
    elements of $\mathbb{Z}^{d}$ to some elements
    $\{g_{t}\}_{t \in [d]}$ such that $g_{t} > 1$
    for all $t \in [d]$.
    The set $\{g_{t}\}_{t \in [d]}$ is
    a generating set.
\end{proof}

We now use \cref{lemma: generatingSet} to find a generating set
for the entries $(M_{ij})_{i, j \in [q]}$ of any 
matrix $M \in \text{Sym}_{q}(\mathbb{R}_{\neq 0})$.
Note that this generating set need not be unique. 
However, with respect to a fixed generating set, 
for any $M_{ij}$, there are unique integers 
$e_{ij0} \in \{0, 1\}$, and
$e_{ij1}, \dots, e_{ijd} \in \mathbb{Z}$,
such that
\begin{equation}\label{equation: generatingM}
    M_{ij} = (-1)^{e_{ij0}} \cdot g_{1}^{e_{ij1}} \cdots g_{d}^{e_{ijd}}.
\end{equation}

\begin{remark*}
    It should be noted that since $M$ is symmetric, $M_{ij} = M_{ji}$
    for all $i, j \in [q]$.
    The uniqueness of the integers $e_{ijt}$
    in \cref{equation: generatingM} then implies that
    for all $i, j \in [q]$,
    $e_{ijt} = e_{jit}$ for all $t \in [d]$.
\end{remark*}

\begin{lemma}\label{lemma: MequivalentCM}
    Let $M \in \text{Sym}_{q}(\mathbb{R}_{\neq 0})$ with a generating set  $\{g_{t}\}_{t \in [d]}$ 
   for its 
   entries. There
    exists an $N = cM \in \text{Sym}_{q}(\mathbb{R}_{\neq 0})$,
    for some $c \in \mathbb{R}_{>0}$, such that
    $\PlEVAL(N) \equiv \PlEVAL(M)$,
    and $\{g_{t}\}_{t \in [d]}$  is also a generating set for the entries of $N$,
        with
    unique integers $e_{ij0} \in \{0, 1\}$
    and $e_{ijt} \in \mathbb{Z}_{\geq 0}$ satisfying
    \begin{equation}\label{N-entries-eijt}
    N_{ij} = (-1)^{e_{ij0}} \cdot 
    g_{1}^{e_{ij1}} \cdots g_{d}^{e_{ijd}}.
    \end{equation}
\end{lemma}
\begin{proof}
    For any $c \neq 0$,
    and any planar graph $G = (V, E)$, we have
    $$Z_{c M}(G) = \sum_{\sigma: V \rightarrow [q]}
    \prod_{(u, v) \in E}(c \cdot M)_{\sigma(u)\sigma(v)}
    = c^{|E|} \cdot \sum_{\sigma: V \rightarrow [q]}
    \prod_{(u, v) \in E}M_{\sigma(u)\sigma(v)}
    = c^{|E|}Z_{M}(G).$$
    Therefore, $\PlEVAL(M) \equiv \PlEVAL(c M)$
    for all $c \neq 0$.
    As the entries of $M$
    are generated by $\{g_{t}\}_{t \in [d]}$, 
    we have unique 
    integers $e'_{ij0} \in \{0, 1\}$, and
    $e'_{ijt} \in \mathbb{Z}$, such that
    $$M_{ij} = (-1)^{e'_{ij0}}g_{1}^{e'_{ij1}}
    \cdots g_{d}^{e'_{ijd}}.$$
    Now let
    $c = (g_{1} \cdots g_{d})^{- e} \in \mathbb{R}_{>0}$,
    where $e = \min_{i, j \in [q],\ t \in [d]}\{e'_{ijt}\}$.

    Clearly,  $\{g_{t}\}_{t \in [d]}$ is also a generating set
    for the entries of $N$.
    If we let
    $e_{ij0} = e'_{ij0}$, and
    $e_{ijt} = e'_{ijt} - e$
    for all $i, j \in [q]$, and $t \in [d]$,
    we see that $e_{ijt} \geq 0$, and that
    $$(c M)_{ij} = (-1)^{e_{ij0}} \cdot 
    g_{1}^{e_{ij1}} \cdots g_{d}^{e_{ijd}}$$
    for all $i, j \in [d]$.
%
%
\end{proof}

\begin{remark*}
    For $\PlEVAL(M)$, given by $M \in
    \text{Sym}(\mathbb{R}_{\neq 0})$
     with generating set
    $\{g_{t}\}_{t \in [d]}$,
    \cref{lemma: MequivalentCM} allows us to 
    replace $M$ with the matrix $N = c M$
    whose entries are generated by
    $\{g_{t}\}_{t \in [d]}$
    such that $e_{ijt} \geq 0$ for all $i, j \in [q]$
    and $t \in [d]$. In the following we will often make this
    substitution when convenient.
\end{remark*}

\begin{definition}\label{definition: mathcalT}
    Given $\PlEVAL(M)$ defined by  $M \in \text{Sym}_{q}(\mathbb{R}_{\neq 0})$
    with entries $(M_{ij})_{i, j \in q}$, 
    we assume a generating set
    $\{g_{t}\}_{t \in [d]}$ is chosen and the replacement of $N = c M$ in \cref{lemma: MequivalentCM}
    has been
    made so that the integers $e_{ijt} \ge 0$ in \cref{equation: generatingM}.
    We  define the function
    $\mathcal{T}_{M}: \mathbb{R}^{d} \rightarrow
    \text{Sym}_{q}(\mathbb{R})$ such that
    $$\mathcal{T}_{M}(\mathbf{p})_{ij}
    = (-1)^{e_{ij0}} \cdot p_{1}^{e_{ij1}} \cdots
    p_{d}^{e_{ijd}}$$
    is a signed monomial in $\mathbf{p} = (p_1, \ldots, p_d)$ for all
    $i, j \in [q]$.
\end{definition}

\begin{lemma}\label{lemma: thickeningInterpolation}
    Let $M \in \text{Sym}_{q}(\mathbb{R}_{\neq 0})$,
    with entries $(M_{ij})_{i, j \in q}$
    generated by some $\{g_{t}\}_{t \in [d]}$.
    Then, $\PlEVAL(\mathcal{T}_{M}(\mathbf{p})) \leq \PlEVAL(M)$ 
    for all $\mathbf{p} \in \mathbb{R}^{d}$.
\end{lemma}
\begin{proof}
    Replacing $M$ by $N=cM$ as in \cref{lemma: MequivalentCM}, we may assume
    $\mathcal{T}_{M}(\mathbf{p})$ is defined in \cref{definition: mathcalT} with  all  $e_{ijt} \geq 0$
    in \cref{equation: generatingM}.
    For any $n \geq 1$, and any given 
    graph $G = (V, E)$,
    \begin{equation}\label{Z_M-expression}
    Z_{M}(T_{n}G) = \sum_{x \in X(G)}x^{n} \cdot \text{\#}_{M}(G, x),
    \end{equation}
    where $X(G)$ and $\text{\#}_{M}(G, x)$ are given in \cref{equation: full_thickening-X(G)-set,equation: full_thickening-number-of-mappings}.
    %
    By definition $X(G)$ is a set, and so each $x \in X(G)$ is distinct, and $|X(G)| \leq |E|^{O(1)}$. 
    By oracle access to $\PlEVAL(M)$ we can get $Z_{M}(T_{n}G)$ for $n \in [|X(G)|]$. 
    Then \cref{Z_M-expression} is a full rank Vandermonde system of linear equations, 
    which can be solved in polynomial time 
    to find $\text{\#}_{M}(G, x)$ for all $x \in X(G)$.
    
    Now, let us consider the set $X(G)$ more closely.
    Given any $x \in X(G)$, we see that $x = \prod M_{ij}^{k_{ij}}$
    for some (not necessarily unique) $\mathbf{k} \in
    \mathcal{P}_{q^{2}}(|E|)$. 
    Since $\{g_{t}\}_{t \in [d]}$ is a generating set for the entries of $M$,
    any $x \in X(G)$ can be represented \emph{uniquely} as
    $$x = (-1)^{e^{x}_{0}}g_{1}^{e^{x}_{1}} \cdots g_{d}^{e^{x}_{d}},$$
    with exponents $e^{x}_{0} \in \{0, 1\}$, 
    and $e^{x}_{1}, \dots, e^{x}_{d} \in \mathbb{Z}_{\geq 0}$.
    
    Fix any $\mathbf{p} =
    (p_{1}, \dots, p_{d}) \in \mathbb{R}^{d}$.
    We define the function
    $\widehat{y}: X(G) \rightarrow \mathbb{R}$, 
    such that
    $$\widehat{y}(x) = (-1)^{e^{x}_{0}} \cdot p_{1}^{e^{x}_{1}} \cdots p_{d}^{e^{x}_{d}}.$$
    By definition of $X(G)$, we see that
    for any $x \in X(G)$, there exist $\mathbf{k} \in
    \mathcal{P}_{q^{2}}(|E|)$, such that
    $$x = \prod_{i, j \in [q]}M_{ij}^{k_{ij}}
    = (-1)^{\sum_{i, j \in [q]}k_{ij}e_{ij0}} \cdot 
    g_{1}^{\sum_{i, j \in [q]}k_{ij}e_{ij1}} \cdots
    g_{d}^{\sum_{i, j \in [q]}k_{ij}e_{ijd}}.$$
    By definition of $\widehat{y}$, this means that
    $\widehat{y}(x) =
    \prod_{i, j \in [q]}\mathcal{T}_{M}(\mathbf{p})_{ij}^{k_{ij}}$.
    Now, let
    $$Y(G) =
    \left\{\prod_{i,j\in [q]}\mathcal{T}_{M}(\mathbf{p})_{ij}^{k_{ij}}
    \hspace{0.08cm}\Big|\hspace{0.1cm}
    \mathbf{k} = (k_{ij})_{i, j \in [q]} \in
    \mathcal{P}_{q^{2}}(|E|) \right\}.$$
    
    Consider any $\sigma: V \rightarrow [q]$.
    For a given $\sigma$, we define $\mathbf{k} \in \mathcal{P}_{q^{2}}(|E|)$
    such that $k_{ij} = |\{(u, v) \in E:
    \sigma(u) = i, \sigma(v) = j\}|$, for all $i, j \in [q]$.
   Then
    $$\widehat{y}\left(\prod_{(u, v) \in E}
    M_{\sigma(u)\sigma(v)}\right)
    = \widehat{y}\left(\prod_{i, j \in [q]}
    M_{ij}^{k_{ij}}\right)
    = \prod_{i, j \in [q]}
    \mathcal{T}_{M}(\mathbf{p})_{ij}^{k_{ij}}
    = \prod_{(u, v) \in E}
    \mathcal{T}_{M}(\mathbf{p})_{\sigma(u)\sigma(v)}.$$    
    This implies that for any $y \in Y(G)$,
    $$\left\{\sigma:V\rightarrow [q]\hspace{0.08cm}:
    \hspace{0.08cm} \prod_{(u, v) \in E}
    \mathcal{T}_{M}(\mathbf{p})_{\sigma(u), \sigma(v)} = y\right\}
    = \bigsqcup_{\substack{x \in X(G):\\ \widehat{y}(x) = y}}
    \left\{\sigma:V\rightarrow [q]\hspace{0.08cm}:
    \hspace{0.08cm} \prod_{(u, v) \in E}
    M_{\sigma(u), \sigma(v)} = x\right\}.$$
    Therefore,  for any $y \in Y(G)$,
    \begin{align*}
        \text{\#}_{\mathcal{T}_{M}(\mathbf{p})}(G, y)
        &= \Big|\big\{\sigma:V\rightarrow [q]\hspace{0.08cm}:
        \hspace{0.08cm} \prod_{(u, v) \in E}
        \mathcal{T}_{M}(\mathbf{p})_{\sigma(u), \sigma(v)} = y\big\}\Big|\\
        &= \sum_{x \in X(G): ~\widehat{y}(x) = y}
        \Big|\big\{\sigma:V\rightarrow [q]\hspace{0.08cm}: \hspace{0.08cm}
        \prod_{(u, v) \in E} M_{\sigma(u), \sigma(v)} = x\big\}\Big|
        &= \sum_{x \in X(G): ~\widehat{y}(x) = y}\text{\#}_{M}(G, x).
    \end{align*}
    
    Having already obtained $\text{\#}_{M}(G, x)$ for all $x \in X(G)$, we can compute
    in polynomial time
    $$\sum_{x \in X(G)}\widehat{y}(x) \cdot \text{\#}_{M}(G, x)
    = \sum_{y \in Y(G)}y \cdot \sum_{x \in X(G): \widehat{y}(x) = y}
    \text{\#}_{M}(G, x)
    = \sum_{y \in Y(G)}y \cdot \text{\#}_{\mathcal{T}_{M}(\mathbf{p})}(G, y)
    = Z_{\mathcal{T}_{M}(\mathbf{p})}(G).$$
    Therefore, $\PlEVAL(\mathcal{T}_{M}(\mathbf{p})) \leq \PlEVAL(M)$.
\end{proof}

We will now need the following theorem and corollary
from \cite{vertigan2005computational}:

\begin{theorem}\label{theorem: tutte}
    For $x, y \in \mathbb{C}$, evaluating the Tutte polynomial 
    at $(x, y)$ is $\#$P-hard over planar graphs unless 
    $(x - 1)(y - 1) \in \{1, 2\}$ or $(x, y) \in \{(1, 1), (-1,-1), (\omega,\omega^{2}), (\omega^{2}, \omega)\}$, 
    where $\omega = e^{\nicefrac{2\pi i}{3}}$. 
    In each exceptional case, the problem is 
    in polynomial time.
\end{theorem}
\begin{corollary}\label{corollary: tutteHardness}
    The $q$-state Potts Model
    $\PlEVAL({\tt{Potts}}_{q}(x))$ is $\#$P-hard for any integer $q \geq 3$,
    and real $x \neq 1$,
    where ${\tt{Potts}}_{q}(x) \in \text{Sym}_{q}(\mathbb{R})$
    is the matrix with entries
    $({\tt{Potts}}_{q}(x)_{ij})$ such that
    ${\tt{Potts}}_{q}(x)_{ij} = 1$ if $i \neq j$,
    and ${\tt{Potts}}_{q}(x)_{ij} = x$ otherwise.
\end{corollary}

Note that $\PlEVAL({\tt{Potts}}_{q}(0))$ is the problem of counting vertex coloring with $q$ colors
on planar graphs.
\cref{theorem: tutte} allows us to prove our first hardness result.

\begin{lemma}\label{lemma: thickeningBasic}
    Let $M \in \text{Sym}_{q}(\mathbb{R}_{\neq 0})$,
    and let $\mathcal{T}_{M}$ be the function  defined
    in \cref{definition: mathcalT}.
    Furthermore, assume for all $i \in [q]$ there exists some 
    (not necessarily distinct) $t(i) \in \{1, \dots, d\}$, 
    such that $e_{iit(i)} > 0, \text{ and } e_{jkt(i)} = 0$
    for all $j \neq k$.
    Then $\PlEVAL(M)$ is $\#$P-hard.
\end{lemma}
\begin{proof}
    We apply \cref{lemma: thickeningInterpolation}.
    Let $\mathbf{p}^{*} \in \mathbb{R}^{d}$
    defined by  $p^{*}_{t} = 0$ for all $t \in [d]$,
    such that $t = t(i)$ for some $i \in [q]$,
    and $p^{*}_{t} = 1$ for all other $t \in [d]$.
    Then, we see that $\mathcal{T}_{M}(\mathbf{p}^{*})_{ii} = 0$
    for all $i \in [q]$,
    and $\mathcal{T}_{M}(\mathbf{p}^{*})_{ij} = \pm 1$
    for all $i \neq j \in [q]$.
    Therefore, 
    $T_{2}(\mathcal{T}_{M}(\mathbf{p}^{*})) = {\tt{Potts}}_{q}(0)$.
    
    From \cref{lemma: thickeningInterpolation}, we get
    $\PlEVAL(\mathcal{T}_{M}(\mathbf{p}^{*})) \leq \PlEVAL(M)$.
    Therefore,
    $$ \PlEVAL({\tt{Potts}}_{q}(0)) \leq
    \PlEVAL(\mathcal{T}_{M}(\mathbf{p}^{*})) \leq \PlEVAL(M).$$
    It follows from  \cref{corollary: tutteHardness} that
    $\PlEVAL(M)$ is  $\#$P-hard.
\end{proof}
\section{Lattice on Eigenvalues}\label{sec: lattice_eigenvalues}

In this section we focus on full ranked matrices.
Using stretching,
we shall prove the hardness of a more
interesting class of matrices than we were able to do
in \cref{lemma: thickeningBasic}.
Consider any $M \in \text{Sym}_{q}(\mathbb{R})$.
There exists some (not necessarily unique) real
orthogonal matrix $H$,
and
a real diagonal matrix $D =  \text{diag}(\lambda_{1}, \dots, \lambda_{q})$
such that
$$M = HDH^{\tt T},$$
where $(\lambda_{1}, ..., \lambda_{q})$ are the eigenvalues of $M$
and the columns of $H$ are the corresponding eigenvectors.
In the rest of the paper, when we refer to $M = HDH^{\tt{T}}$,
it is to be understood that we refer to such an orthogonal matrix $H$,
and diagonal matrix $D$.

From the decomposition  $M = HDH^{\texttt{T}}$,
we have  $M^n = HD^{n}H^{\texttt{T}}$, and
$$(M^{n})_{ij} = (H_{i1}H_{j1}) \lambda_{1}^{n} + \dots +
(H_{iq}H_{jq}) \lambda_{q}^{n}.$$
It follows that
\begin{equation}\label{eqn: stretching}
    Z_{M^{n}}(G) = \sum_{\mathbf{k} \in \mathcal{P}_{q}(|E|)}c_{H}(G, \mathbf{k}) \cdot \left(\lambda_{1}^{k_{1}} \cdots \lambda_{q}^{k_{q}} \right)^{n},
\end{equation}
where
$$c_{H}(G, \mathbf{k}) = \sum_{\sigma: V \rightarrow [q]} \left( \sum_{\substack{E_{1} \sqcup \dots \sqcup E_{q} = E\\|E_{i}| = k_{i}}} \left( \prod_{i \in [q]} \prod_{(u, v) \in E_{i}}H_{\sigma(u)i}H_{\sigma(v)i} \right) \right)$$
depends only on $G$ and the orthogonal matrix $H$, but not on $D$.

Before we can analyze \cref{eqn: stretching} in greater detail,
we will need a few more definitions.

\begin{definition}\label{definition: latticeSet}
    Let $\mathcal{A} = (a_{1}, \dots, a_{n})$ be a tuple 
    of non-zero real numbers (not necessarily distinct). 
    The lattice  of $\mathcal{A}$ consists of the set defined as
    $$\mathcal{L}(\mathcal{A}) = \left\{(x_{1}, \dots, x_{n}) \in \mathbb{Z}^{n}\hspace{0.08cm}\Big|\hspace{0.1cm} \sum_{i = 1}^{n}x_{i} = 0, \prod_{i = 1}^{n}a_{i}^{x_{i}} = 1 \right\},$$
    with addition in $\mathbb{Z}^{n}$.
\end{definition}

\begin{lemma}\label{lemma: latticeVector}
    Let $\mathcal{A} = (a_{1}, \dots, a_{n})$ be a 
    tuple of non-zero real numbers. 
    The set  $\mathcal{L}(\mathcal{A})$ forms a lattice and is
    isomorphic to $\mathbb{Z}^{d}$ for some unique
    integer $0 \leq d \leq n$.
\end{lemma}
\begin{proof}
    We note that $\mathcal{L}(\mathcal{A})$ is a
    subgroup of the finitely generated (discrete)
    Abelian group $\mathbb{Z}^{n}$.
    So, $\mathcal{L}(\mathcal{A})$ is itself a
    finitely generated (discrete) Abelian group,
    and is  torsion-free.
    It follows that it is a lattice and is
    isomorphic to $\mathbb{Z}^{d}$ for some unique
    $0 \leq d \leq n$.
\end{proof}

\begin{definition}\label{definition: latticeDegree}
    Let $\mathcal{A} = (a_{1}, \dots, a_{n})$ be a 
    tuple of non-zero real numbers.
    The lattice dimension of $\mathcal{A}$, denoted
    by $\dim(\mathcal{L}(\mathcal{A}))$,
    is the unique integer
    $d \geq 0$ such that
    $\mathcal{L}(\mathcal{A}) \cong \mathbb{Z}^{d}$.
    A set $\{\mathbf{x}_{1}, \dots, \mathbf{x}_{d}\}
    \subset \mathcal{L}(\mathcal{A})$,
    is called a lattice basis of $\mathcal{A}$ if there
    exists an isomorphism from $\mathcal{L}(\mathcal{A})$
    to $\mathbb{Z}^{d}$ that maps this to a
    basis of $\mathbb{Z}^{d}$.
\end{definition}

\begin{remark*}
    Here we note the known fact that if
    $\mathbb{Z}^{d} \cong \mathbb{Z}^{d'}$, then
    $d = d'$.
    This fact, together with \cref{lemma: latticeVector}
    guarantees that the lattice dimension of any given
    tuple of non-zero reals is well-defined.
\end{remark*}

With the help of these new definitions, we can now go back
to studying the effects of the stretching gadget.

\begin{lemma}\label{lemma: stretchingBasic}
    Let $M \in \text{Sym}^{\tt{F}}_{q}(\mathbb{R})$,
    such that $M = HDH^{\tt{T}}$, where
    $D = \text{\rm diag}(\lambda_{1}, \dots, \lambda_{q})$.
    Then $\PlEVAL(H\Delta H^{\tt{T}}) \leq \PlEVAL(M)$ 
    for any diagonal matrix
    $\Delta = \text{diag}(\Delta_{1}, \dots, \Delta_{q})$,
    such that
    $\Delta_{i} \in \mathbb{R}_{\neq 0}$ for all $i \in [q]$,
    and $\mathcal{L}(\lambda_{1}, \dots, \lambda_{q})
    \subseteq \mathcal{L}(\Delta_{1}, \dots, \Delta_{q})$.
\end{lemma}
\begin{proof}
    We recall that \cref{eqn: stretching} states that for any graph
    $G = (V, E)$,
    $$Z_{M}(S_{n}G) =
    \sum_{\mathbf{k} \in \mathcal{P}_{q}(|E|)}c_{H}(G, \mathbf{k}) \cdot
    \left(\lambda_{1}^{k_{1}} \cdots \lambda_{q}^{k_{q}} \right)^{n},$$
    where
    $$c_{H}(G, \mathbf{k}) = \sum_{\sigma: V \rightarrow [q]}
    \left( \sum_{\substack{E_{1} \sqcup \dots \sqcup E_{q}= E\\|E_{i}| = k_{i}}}
    \left( \prod_{i \in [q]} \prod_{(u, v) \in E_{i}}
    H_{\sigma(u)i}H_{\sigma(v)i} \right) \right).$$

    We now define the set
    $$\Lambda_{D}(G) = \left\{\prod_{i \in [q]} \lambda_{i}^{k_{i}}
    \Big|\hspace{0.1cm}
    \mathbf{k} \in \mathcal{P}_{q}(|E|)\right\}.$$
    For each $\mu \in \Lambda_{D}(G)$, we then define
    $$X_{D}(G, \mu) =
    \left\{\mathbf{k} \in \mathcal{P}_{q}(|E|)\Big|\hspace{0.1cm}
    \prod_{i \in [q]}\lambda_{i}^{k_{i}} = \mu \right\}
    ~~\text{ and }~~ c_{H, D}(G, \mu) =
    \sum_{\mathbf{k} \in X_{D}(G, \mu)}c_{H}(G, \mathbf{k}).$$

    Putting everything together, we see that
    $$Z_{M}(S_{n}G) = \sum_{\mu \in \Lambda_{D}(G)}c_{H, D}(G, \mu) \cdot \mu^{n}.$$
    Since $\Lambda_{D}(G)$ is a  set, each $\mu \in \Lambda_{D}(G)$ is distinct,
    and $|\Lambda_{D}(G)| \leq |E|^{O(1)}$.
    With oracle access to $\PlEVAL(M)$  we can get  $Z_{M}(S_{n}G)$ for $n \in [|\Lambda_{D}(G)|]$.
    Then we have a full rank Vandermonde system of linear equations,
    which can be solved in polynomial time to find
    $c_{H, D}(G, \mu)$ for all $\mu \in \Lambda_{D}(G)$.

    Now, we consider
    $$Z_{H\Delta H^{\tt{T}}}(G) = \sum_{\nu \in \Lambda_{\Delta}(G)}c_{H, \Delta}(G, \nu) \cdot \nu.$$
    Consider $\mathbf{k}, \mathbf{l} \in X_{D}(G, \mu)$
    for some $\mu \in \Lambda_{D}(G)$.
    By definition, this implies that
    $$\prod_{i \in [q]}\lambda_{i}^{k_{i}} =
    \prod_{i \in [q]}\lambda_{i}^{l_{i}} = \mu.$$
    Therefore,
    $(k_{1} - l_{1}, \dots, k_{q} - l_{q})
    \in \mathcal{L}(\lambda_{1}, \dots, \lambda_{q})
    \subseteq \mathcal{L}(\Delta_{1}, \dots, \Delta_{q})$
    by our choice of $\Delta$.
    So, it follows that there exists some $\nu \in \Lambda_{\Delta}(G)$
    such that
    $$\prod_{i \in [q]}\Delta_{i}^{k_{i}} =
    \prod_{i \in [q]}\Delta_{i}^{l_{i}} = \nu.$$
    Therefore, given any $\mu \in \Lambda_{D}(G)$, there exists
    some $\nu \in \Lambda_{\Delta}(G)$, such that
    $X_{D}(G, \mu) \subseteq X_{\Delta}(G, \nu)$.
    Now, we consider some
    $\nu \in \Lambda_{\Delta}(G)$.
    Let $\mathbf{k} \in X_{\Delta}(G, \nu)$.
    Then, we let $\mu = \prod_{i \in [q]}\lambda_{i}^{k_{i}}$.
    We see that $\mathbf{k} \in X_{D}(G, \mu)$.
    This implies that given any $\nu \in \Lambda_{\Delta}(G)$,
    there exist some $\mu_{1}, \dots, \mu_{t} \in \Lambda_{D}(G)$,
    such that
    $$X_{\Delta}(G, \nu) = X_{D}(G, \mu_{1}) \sqcup \dots \sqcup X_{D}(G, \mu_{t}).$$
    
    Now, we know that
    $$c_{H, \Delta}(G, \nu) = \sum_{\mathbf{k} \in X_{\Delta}(\nu)}c_{H}(G, \mathbf{k})
    = \sum_{i \in [t]}\left( \sum_{\mathbf{l} \in X_{D}(\mu_{i})}c_{H}(G, \mathbf{l})\right) = \sum_{i \in [t]}c_{H, D}(G, \mu_{i}).$$
    Note that we have already computed $c_{H, D}(G, \mu)$ for all $\mu \in \Lambda_{D}(G)$. Therefore, we can compute $c_{H, \Delta}(G, \nu)$ for each
    $\nu \in \Lambda_{\Delta}(G)$, and therefore,
    $Z_{H \Delta H^{\tt{T}}}(G)$ can also be computed.
    Therefore, $\PlEVAL(H \Delta H^{\tt{T}}) \leq \PlEVAL(M)$.
\end{proof}

\begin{corollary}\label{corollary: stretchingDegreeZero}
    If $M = HDH^{\tt{T}} \in \text{Sym}^{\tt{F}}_{q}(\mathbb{R})$
    is such that its eigenvalues
    $(\lambda_{1}, \dots, \lambda_{q})$
    have lattice dimension zero,
    then $\PlEVAL(H\Delta H^{\tt{T}}) \leq \PlEVAL(M)$ 
    for any diagonal matrix
    $\Delta = \text{diag}(\Delta_{1}, \dots, \Delta_{q})$
    such that $\Delta_{i} \in \mathbb{R}_{\neq 0}$
    for all $i \in [q]$.
\end{corollary}
\begin{proof}
    If the lattice dimension of $(\lambda_{1}, \dots, \lambda_{q})$ is zero,
    it implies that
    $\mathcal{L}(\lambda_{1}, \dots, \lambda_{q}) = \{\mathbf{0}\}$.
    Therefore, $\mathcal{L}(\lambda_{1}, \dots, \lambda_{q}) \subseteq
    \mathcal{L}(\Delta_{1}, \dots, \Delta_{q})$
    for any $\Delta = (\Delta_{1}, \dots, \Delta_{q})$
    such that $\Delta_{i} \in \mathbb{R}_{\neq 0}$.
    The result then follows from \cref{lemma: stretchingBasic}.
\end{proof}

We will now prove that there exists some $\Delta$ such that $\PlEVAL(H\Delta H^{\tt{T}})$ is $\#$P-hard.

\begin{lemma}\label{lemma: transcendentKappaExists}
    Let $\mathcal{A} = \{a_{1}, \dots, a_{n}\}$, 
    and $\mathcal{B} = \{b_{1}, \dots, b_{m}\}$
    be finite sets of positive real numbers.
    There exists some $\kappa \in \mathbb{R}$
    such that $\kappa + a_{i} > 1$
    for all $i \in [n]$, $\kappa$ is \emph{transcendental}
    to the field $\mathbf{F} = \mathbb{Q}(\mathcal{B})$,
    and for all $(e_{1}, \dots, e_{n}) \in \mathbb{Z}^{n}$,
    $$(\kappa + a_{1})^{e_{1}} \cdots (\kappa + a_{n})^{e_{n}} = 1
    \implies (e_{1}, \dots, e_{n}) = \mathbf{0}.$$
\end{lemma}
\begin{proof}
    For each $\mathbf{e} = (e_{1}, \dots, e_{n}) \neq \mathbf{0}
    \in \mathbb{Z}^{n}$,
    we define the polynomial $f_{\mathbf{e}}:
    \mathbb{R} \rightarrow \mathbb{R}$ as
    $$f_{\mathbf{e}}(x) = 
    \prod_{i \in [n]: e_{i} > 0}(x + a_{i})^{e_{i}} -
    \prod_{i \in [n]: e_{i} < 0}(x + a_{i})^{-e_{i}}.$$
    We can see that no such $f_{\mathbf{e}}$ is the zero polynomial
    as each $a_{i} \in \mathcal{A}$ is a distinct
    element in a unique factorization domain.
    Therefore, for each $f_{\mathbf{e}}$, the set
    $\emptyset_{\mathbf{e}} = \{x \in \mathbb{R}:
    f_{\mathbf{e}} (x) = 0\}$ is finite.
    We also see that the set $Z \subset \mathbb{R}$ of
    all the algebraic numbers over the field $\mathbf{F}
    = \mathbb{Q}(\mathcal{B})$ is countable, since
    $\mathcal{B}$ is finite.
    Therefore, $(\cup_{\mathbf{e}} \emptyset_{\mathbf{e}}) \cup Z$
    is a countable set.
    Therefore, we can pick some $\kappa \in \mathbb{R}
    \setminus ((\cup_{\mathbf{e}} \emptyset_{\mathbf{e}}) \cup Z)$
    such that $\kappa + a_{i} > 1$ for all $i \in [n]$.
    This $\kappa$ satisfies all the requirements
    of the lemma.
\end{proof}

\begin{lemma}\label{lemma: latticeHardnessDelta}
    Let $M = HDH^{\texttt{T}} \in \text{Sym}^{\tt{F}}_{q}
    (\mathbb{R}_{\neq 0})$ (for $q \geq 3$).
    Then there exists a diagonal matrix $\Delta = D + \kappa I$
    for some $\kappa \in \mathbb{R}$,
    such that $\PlEVAL(H\Delta H^{\tt{T}})$ is $\#$P-hard.
\end{lemma}
\begin{proof}
    Let $(M_{ij})_{i, j \in [q]}$ be the entries of the matrix $M$
    generated by some $\{g_{t}\}_{t \in [d]}$.
    We will replace $M$ with the matrix $N$
    guaranteed by \cref{lemma: MequivalentCM}.
    So, we may assume that $e_{ijt} \geq 0$
    for all $i, j \in [q]$, and $t \in [d]$.
    Let $\mathcal{A} = \{M_{ii}: i \in [q]\}$ be the set of diagonal elements (with
    duplicates removed). 
    Without loss of generality, we let
    $\mathcal{A} 
    = \{M_{11}, \dots, M_{rr}\}$
    for some $1 \le r \leq q$. Let
    $\mathcal{B} = \{g_{t}\}_{t \in [d]}$.
    With this choice of $\mathcal{A}$ and
    $\mathcal{B}$, we can let $\kappa \in \mathbb{R}$ be the
    number whose existence is guaranteed by
    \cref{lemma: transcendentKappaExists}.
    We will now let $\Delta = D + \kappa I$, a diagonal matrix with all diagonal elements positive.

    Clearly, $H\Delta H^{\tt{T}} = M + \kappa I \in \text{Sym}^{\tt{F}}_{q}
    (\mathbb{R}_{\neq 0})$. All its non-diagonal entries are the same as that of $M$, and can be
    represented as a product of
    non-negative integer powers of
    the generating set $\{g_{t}\}_{t \in [d]}$, up to a $\pm$ sign.
    The diagonal entries of $H\Delta H^{\tt{T}}$
    can also be trivially represented as
    non-negative integer powers
    of 
    the set $\{\kappa + M_{11}, \dots,
    \kappa + M_{rr}\}$.
    We will now show that each entry of
    $H\Delta H^{\tt{T}}$ can in fact be expressed
    \emph{uniquely} as a product of integer powers
    of $\{g_{t}\}_{t \in [d]} \cup
    \{\kappa + M_{tt}: t \in [r]\}$, up to a $\pm$ sign.
    To show that, we only need to prove that for any 
   $(e_{1}, \dots, e_{d},
    e'_{1}, \dots, e'_{r}) \in \mathbb{Z}^{d+ r}$,
    if 
    \begin{equation}\label{g-kappa-generating-set-eqn}
    g_{1}^{e_{1}} \cdots g_{d}^{e_{d}} \cdot
    (\kappa + M_{11})^{e'_{1}} \cdots (\kappa + M_{rr})^{e'_{r}} = 1,
    \end{equation}
    then $(e_{1}, \dots, e_{d},
    e'_{1}, \dots, e'_{r}) = \mathbf{0}$.

    First, we assume
    $(e_{1}, \dots, e_{d}) \neq \mathbf{0}$.    
    In \cref{g-kappa-generating-set-eqn}, 
    since $\{g_{t}\}_{t \in [d]}$ is a generating set for $M$,
    and $(e_{1}, \dots, e_{d}) \neq \mathbf{0}$, 
    we have $\prod_{t \in [d]}g_{t}^{e_{t}} \neq 1$. 
    Therefore, 
    $(e'_{1}, \dots, e'_{r}) \neq  \mathbf{0}$  by \cref{g-kappa-generating-set-eqn}. 
    Separating out positive and negative $e'_{t}$'s, we have
    \begin{equation}\label{eqn: generatingSetDecomp}
        \left(\prod_{t \in [d]}g_{t}^{e_{t}}\right) \cdot
        \prod_{t \in [r]: e'_{t} > 0}(\kappa + M_{tt} )^{e'_{t}} =
        \prod_{t \in [r]: e'_{t} < 0}(\kappa + M_{tt} )^{-e'_{t}}.
    \end{equation}
    Since $(e'_{1}, \dots, e'_{r}) \neq  \mathbf{0}$,
    at least one side
    is a non-constant polynomial in $\kappa$ over the field ${\bf F} =
    \mathbb{Q}(\{g_{t}\}_{t \in [d]})$, yet both sides have  different leading coefficients.
    This contradicts our assumption that $\kappa$ is 
    transcendental to ${\bf F}$.
    Therefore, we must have $(e_{1}, \dots, e_{d}) = \mathbf{0}$.
    But then,
    $\prod_{t \in [d]}g_{t}^{e_{t}} = 1$, which implies that
    $$\prod_{t \in [r]}(\kappa + M_{tt})^{e'_{t}} = 1.$$
    \cref{{lemma: transcendentKappaExists}} implies that
    for $\kappa$,
   this is only possible if
    $(e'_{1}, \dots, e'_{r}) = \mathbf{0}$, 
    so $(e_{1}, \dots, e_{d},
    e'_{1}, \dots, e'_{r}) = \mathbf{0}$.
    
    This proves that $\{g_{t}\}_{t \in [d]} \cup
    \{\kappa + M_{tt}: t \in [r]\}$ is
    a generating set of the entries of
    $H\Delta H^{\tt{T}}$.
    Importantly, if we let $e_{ijt}$ and $e'_{ijt}$
    be the unique integers such that
    $$(H\Delta H^{\tt{T}})_{ij}
    = (-1)^{e_{ij0}}g_{1}^{e_{ij1}} \cdots g_{d}^{e_{ijd}}
    \cdot (\kappa + M_{11})^{e'_{ij1}} \cdots 
    (\kappa + M_{r'})^{e'_{ijr}},$$
    we see that $e_{ijt}, e'_{ijt'} \geq 0$ for all
    $i, j \in [q]$, $t \in [d]$, and $t' \in [r]$.
    Moreover,
    for every $i \in [q]$, we know that there exists
    some $t(i) \in [r]$ such that
    $e'_{iit(i)} = 1$, and $e'_{jkt(i)} = 0$ for
    all $j \neq k$.
    So, from \cref{lemma: thickeningBasic}, 
    we conclude that $\PlEVAL(H\Delta H^{\tt{T}})$ is $\#$P-hard.
\end{proof}

\cref{corollary: stretchingDegreeZero} and
\cref{lemma: latticeHardnessDelta} immediately allow us
to prove the following theorem.
\begin{theorem}\label{theorem: latticeHardnessSimple}
    If $M \in \text{Sym}_{q}^{\tt{F}}(\mathbb{R}_{\neq 0})$
    (for $q \geq 3$) has eigenvalues with lattice dimension $0$,
    then $\PlEVAL(M)$ is $\#$P-hard.
\end{theorem}

With a little more effort however, \cref{lemma: latticeHardnessDelta}
allows us to prove a slightly stronger version
of \cref{theorem: latticeHardnessSimple}. Given $i \neq j \in [q]$,
we define $\delta_{ij} \in \mathbb{Z}^{q}$ such that
$$\delta_{ij}(k) = \begin{cases}
    1 & \text{if } k = i,\\
    -1 & \text{if } k = j,\\
    0 & \text{otherwise}.
\end{cases}$$
We then let

\begin{equation}\label{equation: mathcalD}
    \mathcal{D} = \{\delta_{ij}: i \neq j \in [q]\}
\end{equation}

\begin{theorem}\label{theorem: latticeHardness}
    If $M \in \text{Sym}_{q}^{\tt{F}}(\mathbb{R}_{\neq 0})$
    (for $q \geq 3$) has
    eigenvalues $(\lambda_{1}, \dots, \lambda_{q})$
    with a lattice basis $\mathcal{B}$ such that
    $\mathcal{B} \subseteq \mathcal{D}$,
    then $\PlEVAL(M)$ is $\#$P-hard.
\end{theorem}
\begin{proof}
    From \cref{lemma: latticeHardnessDelta}, we know that given
    $M = HDH^{\tt{T}}$, there exists a diagonal matrix
    $\Delta = D + \kappa I$ for some $\kappa \in \mathbb{R}$
    such that $\PlEVAL(H\Delta H^{\tt{T}})$ is $\#$P-hard.
    Now, if $\delta_{ij} \in \mathcal{B}$ for some $i < j$,
    we know that $\lambda_{i}^{+1}\cdot\lambda_{j}^{-1} = 1$.
    This implies that $\lambda_{i} = \lambda_{j}$.
    Therefore, $\lambda_{i} + \kappa = \lambda_{j} + \kappa$.
    So, $\Delta_{i}^{+1}\cdot\Delta_{j}^{-1} = 1$.
    Therefore, $\delta_{ij} \in
    \mathcal{L}(\Delta_{1}, \dots, \Delta_{q})$.
    Since this is true for all $\delta_{ij} \in \mathcal{B}$,
    it follows that $\mathcal{B} \subseteq
    \mathcal{L}(\Delta_{1}, \dots, \Delta_{q})$.
    Therefore, $\mathcal{L}(\lambda_{1}, \dots, \lambda_{q})
    \subseteq \mathcal{L}(\Delta_{1}, \dots, \Delta_{q})$.
    Now, \cref{lemma: stretchingBasic} implies that
    $\PlEVAL(H \Delta H^{\tt{T}}) \leq \PlEVAL(M)$.
    By our choice of $\Delta$, this implies that
    $\PlEVAL(M)$ is $\#$P-hard.
\end{proof}

\begin{remark*}
    It should be noted that the identity matrix
    $I \in \text{Sym}_{q}^{\tt{F}}(\mathbb{R})$ has
    eigenvalues $(1, \dots, 1)$,
    which trivially has a lattice basis
    $\mathcal{B} = \{\delta_{1j}: 2 \leq j \leq q\}
    \subseteq \mathcal{D}$.
    However, we are not claiming that $\PlEVAL(I)$ is $\#$P-hard (in fact the problem 
     $\PlEVAL(I)$  is clearly in {\rm P}).
    This is because the identity matrix does not satisfy the
    crucial property that all entries of the matrix
    belong to $\mathbb{R}_{\neq 0}$.
    In general, \cref{theorem: latticeHardness} implies that
    if $M \in \text{Sym}_{q}^{\tt{F}}(\mathbb{R})$ 
    has eigenvalues with a lattice basis $\mathcal{B} \subseteq \mathcal{D}$,
    then either $\PlEVAL(M)$ is $\#$P-hard, or $M$ is guaranteed
    to have zero entries.
\end{remark*}
\section{Proof of Hardness}\label{sec: proof_of_hardness}

We will now try to understand the lattice
$\mathcal{L}(\lambda_{1}, \dots, \lambda_{q})$
of the eigenvalues of matrices
$M \in \text{Sym}_{q}^{\tt{F}}(\mathbb{R}_{> 0})$
which do not satisfy the conditions of
\cref{theorem: latticeHardness}.
Let
$$\chi_{q} = \left\{\mathbf{x} = (x_{1}, \dots, x_{q})
\in \mathbb{Z}^{q} \Big|\hspace{0.1cm}
x_{1} + \cdots + x_{q} = 0\right\}.$$
Consider some $\mathbf{0}  \neq  \mathbf{x}  
\in \chi_{q}$.
We will use this $\mathbf{x}$ to define a polynomial
$\phi_{\mathbf{x}}: \mathbb{R}^{q} \rightarrow \mathbb{R}$ as
\begin{equation}\label{equation: phiX}
    \phi_{\mathbf{x}}(\alpha_{1}, \dots, \alpha_{q}) =
    \prod_{i \in [q]:\ x_{i} > 0} \alpha_{i}^{x_{i}} -
    \prod_{i \in [q]:\ x_{i} < 0} \alpha_{i}^{-(x_{i})}.
\end{equation}
Since $\mathbf{x} \neq \mathbf{0}$, we can see
that $\phi_{\mathbf{x}}$ is not the trivial
constant zero function.
Moreover, since $\sum_{i \in [q]}x_{i} = 0$, we see that
$$\sum_{i \in [q]:\ x_{i} > 0} x_{i} = 
\sum_{i \in [q]:\ x_{i} < 0} -(x_{i}).$$
So, we see that $\phi_{\mathbf{x}}$ is a
homogeneous polynomial.
Moreover, by construction, we see that
$$\mathbf{x} \in \mathcal{L}(\alpha_{1}, \dots, \alpha_{q}) \iff
\phi_{\mathbf{x}}(\alpha_{1}, \dots, \alpha_{q}) = 0.$$
We will now define the polynomial
$\Phi_{\mathbf{x}}: \mathbb{R}^{q} \rightarrow \mathbb{R}$
as follows:
\begin{equation}\label{eqn: Phi}
    \Phi_{\mathbf{x}}(\alpha_{1}, \dots, \alpha_{q}) = 
    \prod_{\sigma \in S_{q}}\phi_{\mathbf{x}}(\alpha_{\sigma(1)}, \dots, \alpha_{\sigma(q)}),
\end{equation}
where $S_{q}$ is the symmetric group over $[q]$.
By construction, we see that $\Phi_{\mathbf{x}}$ is a
symmetric homogeneous polynomial.
In the next section, we will explore this polynomial
in greater detail, and exploit its symmetry to prove
some useful results.
For now, the most important property of this polynomial
is that if $\mathbf{x} \in \mathcal{L}(\alpha)$
for some $\alpha = (\alpha_{1}, \dots, \alpha_{q})$,
then $\Phi_{\mathbf{x}}(\alpha) = 0$.
However, because we constructed $\Phi_{\mathbf{x}}$ to be
symmetric, the converse does not quite hold.

\begin{lemma}\label{lemma: permutatePhi}
    Let $\mathbf{x}, \mathbf{y} \in \chi_{q}$.
    Let $\sigma \in S_{q}$ such that $x_{i} = y_{\sigma(i)}$
    for all $i \in [q]$.
    Then, $\Phi_{\mathbf{x}} = \Phi_{\mathbf{y}}$.
\end{lemma}
\begin{proof}
    By definition, given any $\alpha \in \mathbb{R}^{q}$, we know that
    \begin{align*}
        \phi_{\mathbf{x}}(\alpha_{\sigma(1)}, \dots, \alpha_{\sigma(q)})
        &= \prod_{i \in [q]:\ x_{i} > 0} \alpha_{\sigma(i)}^{x_{i}} -
        \prod_{i \in [q]:\ x_{i} < 0} \alpha_{\sigma(i)}^{-(x_{i})}\\
        &= \prod_{i \in [q]:\ y_{\sigma(i)} > 0} \alpha_{\sigma(i)}^{y_{\sigma(i)}} -
        \prod_{i \in [q]:\ y_{\sigma(i)} < 0} \alpha_{\sigma(i)}^{-(y_{\sigma(i)})}\\
        &= \prod_{j \in [q]:\ y_{j} > 0} \alpha_{j}^{y_{j}} -
        \prod_{j \in [q]:\ y_{j} < 0} \alpha_{j}^{-(y_{j})}\\
        &= \phi_{\mathbf{y}}(\alpha_{1}, \dots, \alpha_{q}).
    \end{align*}
    Therefore,
    \begin{align*}
        \Phi_{\mathbf{y}}(\alpha_{1}, \dots, \alpha_{q})
        &= \prod_{\sigma' \in S_{q}}\phi_{\mathbf{y}}(\alpha_{\sigma'(1)}, \dots, \alpha_{\sigma'(q)})\\
        &= \prod_{\sigma' \in S_{q}}\phi_{\mathbf{x}}(\alpha_{\sigma(\sigma'(1))}, \dots, \alpha_{\sigma(\sigma'(q))})\\
        &= \prod_{\tau \in S_{q}}\phi_{\mathbf{x}}(\alpha_{\tau(1)}, \dots, \alpha_{\tau(q)})\\
        &= \Phi_{\mathbf{x}}(\alpha_{1}, \dots, \alpha_{q}).
        \end{align*}
\end{proof}

\cref{lemma: permutatePhi} implies that if
$\mathbf{x} \in \mathcal{L}(\lambda_{1}, \dots, \lambda_{q})$,
then $\Phi_{\mathbf{y}}(\lambda_{1}, \dots, \lambda_{q}) = 0$
for all $\mathbf{y} \in \chi_{q}$ such that $y_{\sigma(i)} = x_{i}$
for some $\sigma \in S_{q}$.
To navigate through some of the problems introduced by this,
we will now introduce some notation.

\begin{notation}
    Let $X \subset \chi_{q}$.
    Then, given any $\sigma \in S_{q}$, we define
    $$X^{\sigma} = \{\mathbf{x} \in \chi_{q}: (x_{\sigma(1)}, \dots, x_{\sigma(q)}) \in X \}.$$
    We also define
    $$\overline{X} = \bigcup_{\sigma \in S_{q}}X^{\sigma}.$$
\end{notation}

\begin{lemma}\label{lemma: PhiZero}
    Let $\alpha = (\alpha_{1}, \dots, \alpha_{q})$
    be non-zero reals.
    Then,
    $$\mathbf{x} \in \overline{\mathcal{L}(\alpha)} \iff 
    \Phi_{\mathbf{x}}(\alpha) = 0.$$
\end{lemma}
\begin{proof}
    If $\mathbf{x} \in \overline{\mathcal{L}(\alpha)}$,
    we know that there exists some $\sigma \in S_{q}$ such that
    $\mathbf{x} \in \mathcal{L}(\alpha)^{\sigma}$.
    If we now define $\mathbf{y}\in \chi_{q}$, such that
    $y_{i} = x_{\sigma(i)}$ for all $i \in [q]$,
    we see that $\mathbf{y} \in \mathcal{L}(\alpha)$.
    Now, \cref{lemma: permutatePhi} implies that
    $\Phi_{\mathbf{x}}(\alpha) = \Phi_{\mathbf{y}}(\alpha) = 0$.

    If $\Phi_{\mathbf{x}}(\alpha) = 0$, then we see that
    $\phi_{\mathbf{x}}(\alpha_{\sigma(1)}, \dots, \alpha_{\sigma(q)})
    = 0$ for some $\sigma \in S_{q}$.
    Therefore, $\phi_{\mathbf{y}}(\alpha) = 0$, where
    $y_{i} = x_{\sigma^{-1}(i)}$ for all $i \in [q]$.
    So, $\mathbf{y} \in \mathcal{L}(\alpha)$, which implies that
    $\mathbf{x} \in \mathcal{L}(\alpha)^{\sigma^{-1}} \subset
    \overline{\mathcal{L}(\alpha)}$.
\end{proof}

We recall from \cref{eqn: Phi} that given any
$\mathbf{0}  \neq  \mathbf{x}
 \in \chi_{q}$,
the polynomial $\Phi_{\mathbf{x}}$ is a symmetric,
homogeneous polynomial.
We now recall the Fundamental Theorem
of Symmetric Polynomials, stated below:

\begin{theorem}\label{theorem: symmetricFundamental}
    Let $\Phi(\alpha_{1}, \dots, \alpha_{q})$ be a symmetric
    polynomial on $q$ variables. Let $e_{i}$ for $i \in [q]$
    represent the elementary symmetric polynomials, such that
    $$e_{i}(\alpha_{1}, \dots, \alpha_{q}) =
    \sum_{1 \leq t_{1} < \dots < t_{i} \leq q}
    \alpha_{t_{1}} \cdots \alpha_{t_{i}}.$$
    Then, there exists a polynomial function
    $\widehat{\Phi}: \mathbb{R}^{q} \rightarrow \mathbb{R}$ such that
    $$\Phi(\alpha_{1}, \dots, \alpha_{q}) = 
    \widehat{\Phi}\Big(e_{1}(\alpha_{1}, \dots, \alpha_{q}), \dots,
    e_{q}(\alpha_{1}, \dots, \alpha_{q})\Big).$$
\end{theorem}

Now, consider some $\mathbf{0}  \neq  \mathbf{x} \in \chi_{q}$,
and $M \in \text{Sym}_{q}(\mathbb{R})$
with eigenvalues $(\lambda_{1}, \dots, \lambda_{q})$.
\cref{theorem: symmetricFundamental} implies that
there exists a polynomial $\widehat{\Phi_{\mathbf{x}}}$
such that
$$\Phi_{\mathbf{x}}(\lambda_{1}, \dots, \lambda_{q}) =
\widehat{\Phi_{\mathbf{x}}}\Big(e_{1}(\lambda_{1}, \dots, \lambda_{q}),
\dots, e_{q}(\lambda_{1}, \dots, \lambda_{q})\Big).$$
We will now use the fact that  these
$e_{i}(\lambda_{1}, \dots, \lambda_{q})$ are
all just coefficients of the characteristic polynomial of $M$. Thus, each $e_i$ can be written
as a homogeneous polynomial
of the entries of $M$.

\begin{theorem}\label{theorem: charPolyElementarySym}
    Let $M \in \text{Sym}_{q}(\mathbb{R})$
    with eigenvalues $(\lambda_{1}, \dots, \lambda_{q})$.
    For any elementary symmetric polynomial $e_{k}$
    for $k \in [q]$, there exists a homogeneous polynomial
    $s_{k}: \text{Sym}_{q}(\mathbb{R}) \rightarrow \mathbb{R}$
    of degree $k$ such that
    $$e_{k}(\lambda_{1}, \dots, \lambda_{q}) =
    s_{k}\left((M_{ij})_{i \leq j \in [q]}\right).$$
\end{theorem}

We will finally define the polynomial $\Psi_{\mathbf{x}}:
\text{Sym}_{q}(\mathbb{R}) \rightarrow \mathbb{R}$ as
\begin{equation}\label{eqn: Psi}
    \Psi_{\mathbf{x}}(M) = \widehat{\Phi_{\mathbf{x}}}\Big(
    s_{1}\left((M_{ij})_{i \leq j \in [q]}\right), \dots, 
    s_{q}\left((M_{ij})_{i \leq j \in [q]}\right)\Big).
\end{equation}

\begin{definition}\label{definition: matrixPoly}
    A function $F: \text{Sym}_{q}(\mathbb{R})
    \rightarrow \mathbb{R}$ is called
    a $\text{Sym}_{q}(\mathbb{R})$-polynomial if it
    is a homogeneous polynomial in the
    entries of the matrix.
\end{definition}

\begin{remark*}
    Note that we require $\text{Sym}_{q}(\mathbb{R})$-polynomials
    to be homogeneous.
    If $F$ is a $\text{Sym}_{q}(\mathbb{R})$-polynomial
    of degree $d$, then
    $F(c M) = c^{d} \cdot F(M)$.
    In particular, this would allow us to conclude
    that $F(M) = 0$ if and only if $F(c M) = 0$.
    So, we can safely replace any $M$
    with $N = c M$ as guaranteed by
    \cref{lemma: MequivalentCM} without
    changing whether $F(M) = 0$ for any
    $\text{Sym}_{q}(\mathbb{R})$
    polynomial $F$.
\end{remark*}

Clearly, given any $\mathbf{0}  \neq  \mathbf{x} \in \chi_{q}$,
we see that $\Psi_{\mathbf{x}}(M)$ is 
polynomial in the entries of $M$.
Moreover, we note that for each $i \in [q]$,
$e_{i}$ and $s_{i}$ are homogeneous polynomials of the
same degree.
Since $\Phi_{\mathbf{x}}$ is a homogeneous polynomial,
by construction of $\widehat{\Phi}$, it follows that
$\Psi_{\mathbf{x}}$ is also a homogeneous polynomial,
of the same degree as $\Phi_{\mathbf{x}}$.
In summation, we have managed to prove the following
lemma.

\begin{lemma}\label{lemma: PsiExists}
    Let $\mathbf{0}  \neq  \mathbf{x} \in \chi_{q}$.
    There exists a $\text{Sym}_{q}(\mathbb{R})$-polynomial
    $\Psi_{\mathbf{x}}$, such that
    given any $M \in \text{Sym}^{\tt{F}}_{q}(\mathbb{R})$
    with eigenvalues $(\lambda_{1}, \dots, \lambda_{q})$,
    $$\Psi_{\mathbf{x}}(M) = 0 \iff \mathbf{x} \in
    \overline{\mathcal{L}(\lambda_{1}, \dots, \lambda_{q})}.$$
\end{lemma}

The proof immediately follows from the construction of
$\Psi_{\mathbf{x}}$, and \cref{lemma: PhiZero}.
Before we can proceed further, we will need to
establish a few important properties of
$\overline{\mathcal{L}(\lambda_{1}, \dots, \lambda_{q})}$.

\begin{lemma}\label{lemma: dimReduction}
    Let $\mathcal{L}_{1} = \mathcal{L}(\alpha_{1}, \dots, \alpha_{q})$
    and $\mathcal{L}_{2} = \mathcal{L}(\beta_{1}, \dots, \beta_{q})$
    be lattices such that $\beta_{i} > 0$ for all $i \in [q]$.
    Let $d_{1} = \dim(\mathcal{L}_{1})$ and
    $d_{2} = \dim(\mathcal{L}_{2})$.
    If $\overline{\mathcal{L}_{2}} \subsetneq
    \overline{\mathcal{L}_{1}}$,
    then $d_{2} < d_{1}$.
\end{lemma}
\begin{proof}
    We will first define the rational span of a
    lattice $\mathcal{L}$ of dimension $d$
    with a lattice basis $\mathcal{B} =
    \{\mathbf{y}_{1}, \dots, \mathbf{y}_{d}\}$ to be
    $$\text{Qsp}(\mathcal{L}) = \{c_{1}\cdot \mathbf{y}_{1}
    + \cdots + c_{d} \cdot \mathbf{y}_{d} \Big|\hspace{0.1cm}
    c_{1}, \dots, c_{d} \in \mathbb{Q}\}.$$
    It is easily seen that $\text{Qsp}(\mathcal{L})$ is a 
    $\mathbb{Q}$-vector space of the same dimension $d$, as the
    lattice dimension of $\mathcal{L}$.
    We also note that given any $\sigma \in S_{q}$,
    $\mathcal{L}^{\sigma}$ is also a lattice 
    of dimension $d$, with $\mathcal{B}^{\sigma}$
    as a lattice basis.
    Therefore,
    $$\overline{\text{Qsp}(\mathcal{L})} =
    \bigcup_{\sigma \in S_{q}} \text{Qsp}(\mathcal{L}^{\sigma})$$
    for any lattice $\mathcal{L}$.

    Now, let us consider any $\mathbf{y} \in
    \overline{\text{Qsp}(\mathcal{L}_{2})}$.
    There exists some $\sigma \in S_{q}$,
    such that $\mathbf{y} \in
    \text{Qsp}(\mathcal{L}_{2}^{\sigma})$.
    So, there exists some
    $c \in \mathbb{Z}_{> 0}$, such that
    $c \cdot \mathbf{y} \in \mathcal{L}_{2}^{\sigma}
    \subset \overline{\mathcal{L}_{2}}$.
    From our assumptions, we know that
    $\overline{\mathcal{L}_{2}} \subset
    \overline{\mathcal{L}_{1}}$.
    Therefore,  there exists some
    $\tau \in S_{q}$ such that
    $c \cdot \mathbf{y} \in \mathcal{L}_{1}^{\tau}$.
    So, $\mathbf{y} \in
    \text{Qsp}(\mathcal{L}_{1}^{\tau}) \subset
    \overline{\text{Qsp}(\mathcal{L}_{1})}$.
    Therefore, we have shown that
    $$\overline{\text{Qsp}(\mathcal{L}_{2})} \subseteq
    \overline{\text{Qsp}(\mathcal{L}_{1})}.$$
    Hence, 
    $$\text{Qsp}(\mathcal{L}_{2}) \subseteq
    \overline{\text{Qsp}(\mathcal{L}_{2})} \subseteq
    \overline{\text{Qsp}(\mathcal{L}_{1})} = 
    \bigcup_{\sigma \in S_{q}} \text{Qsp}(\mathcal{L}_{1}^{\sigma}).$$
    Since this is a finite union of $\mathbb{Q}$-vector spaces,
    we know that there must be some $\tau \in S_{q}$,
    such that $\text{Qsp}(\mathcal{L}_{2}) \subseteq
    \text{Qsp}(\mathcal{L}_{1}^{\tau})$.
    In particular, this implies that
    $d_{2} = \dim(\text{Qsp}(\mathcal{L}_{2})) \leq
    \dim(\text{Qsp}(\mathcal{L}_{1}^{\tau})) = d_{1}.$
    Now, for a contradiction assume   $d_{2} = d_{1}$.
    This implies that in fact,
    $\text{Qsp}(\mathcal{L}_{2}) =
    \text{Qsp}(\mathcal{L}_{1}^{\tau})$.
    
    Now, we recall that there exists some
    $\mathbf{x} \in \overline{\mathcal{L}_{1}}
    \setminus \overline{\mathcal{L}_{2}}$. Being outside $\mathcal{L}_{2}$, clearly $\mathbf{x} \neq \mathbf{0}$.
    There exists some $\sigma \in S_{q}$,
    such that $\mathbf{x} \in \mathcal{L}_{1}^{\sigma}$.
    If we let $\mathbf{y} \in \chi_{q}$ such that
    $y_{i} = x_{\tau^{-1}(\sigma(i))}$ for all $i \in [q]$,
    we see that $\mathbf{y} \in \mathcal{L}_{1}^{\tau}
    \subset \text{Qsp}(\mathcal{L}_{1}^{\tau}) = 
    \text{Qsp}(\mathcal{L}_{2})$.
    Therefore, there exists some $c \in \mathbb{Z}_{> 0}$,
    such that $c \cdot \mathbf{y} \in \mathcal{L}_{2}$.
    We recall that $\mathcal{L}_{2} =
    \mathcal{L}(\beta_{1}, \dots, \beta_{q})$.
    Therefore, by definition, this implies that
    $$\beta_{1}^{c \cdot y_{1}} \cdots \beta_{q}^{c \cdot y_{q}} 
    = \left(\beta_{1}^{y_{1}} \cdots \beta_{q}^{y_{q}}\right)^{c} = 1
    \implies \beta_{1}^{y_{1}} \cdots \beta_{q}^{y_{q}}= \pm 1.$$
    Since $\beta_{i} > 0$ for all $i \in [q]$,
    it must be the case that
    $\beta_{1}^{y_{1}} \cdots \beta_{q}^{y_{q}} = 1$,
    which means that $\mathbf{y} \in \mathcal{L}_{2}$.
    By the construction of $\mathbf{y}$,
    it follows that $\mathbf{x} \in
    \mathcal{L}_{2}^{\tau^{-1} \sigma} \subset
    \overline{\mathcal{L}_{2}}$.
    But this contradicts our assumption about
    $\mathbf{x}$ that $\mathbf{x} \notin
    \overline{\mathcal{L}_{2}}$.
    Therefore, our assumption that $d_{2} = d_{1}$
    must be false.
\end{proof}

We will now see how \cref{lemma: dimReduction}
allows us to prove the $\#$P-hardness of
$\PlEVAL(M)$.
We will need just one more definition.

\begin{definition}\label{definition: reduct}
    Let $M \in \text{Sym}_{q}^{\tt{F}}(\mathbb{R}_{\neq 0})$
    have eigenvalues $(\lambda_{1}, \dots, \lambda_{q})$.
    Let $\mathcal{F}$ be a countable set of
    $\text{Sym}_{q}(\mathbb{R})$-polynomials.
    A matrix $N \in \text{Sym}_{q}^{\tt{F}}(\mathbb{R}_{\neq 0})$
    is called a reduct of $(M, \mathcal{F})$ if
    $\PlEVAL(N) \leq \PlEVAL(M)$, and
    $F(N) \neq 0$ for all $F \in \mathcal{F}$.
    We  denote the set of reducts of $(M, \mathcal{F})$ by
    $$\mathfrak{R}(M, \mathcal{F}) =
    \Big\{N \in \text{Sym}_{q}^{\tt{F}}(\mathbb{R}_{\neq 0})
    \Big|\hspace{0.1cm} N \text{ is a reduct of }
    (M, \mathcal{F})\Big\}.$$
\end{definition}

Consider any 
$M \in
\text{Sym}_{q}^{\tt{pd}}(\mathbb{R}_{\neq 0})$,
with eigenvalues $(\lambda_{1}, \dots, \lambda_{q})$.
We define
\begin{equation}\label{equation: FM}
    \mathcal{F}_{M} = \Big\{\Psi_{\mathbf{y}}:
    \mathbf{0}  \neq \mathbf{y} \in \chi_{q}, \Psi_{\mathbf{y}}(M) \neq 0 \Big\}.
\end{equation}
Since $\chi_{q}$ is a countable set, it follows that
$\mathcal{F}_{M}$ is also a countable set
of $\text{Sym}_{q}(\mathbb{R})$-polynomials.
From \cref{lemma: PsiExists}, we  see that
$\mathcal{F}_{M}$ consists of precisely the polynomials  indexed by vectors in  $\chi_{q} \setminus
\overline{\mathcal{L}(\lambda_{1}, \dots, \lambda_{q})}$.

Now, if the eigenvalues $(\lambda_{1}, \dots, \lambda_{q})$
of $M$ have a lattice basis $\mathcal{B}$
such that $\mathcal{B} \subset \mathcal{D}$, we know from
\cref{theorem: latticeHardness} that $\PlEVAL(M)$
is $\#$P-hard.
Otherwise, for any lattice basis $\mathcal{B}$
of the eigenvalues, there exists some
$\mathbf{x} \in \mathcal{B} \setminus \mathcal{D}$.
Our goal will be to prove that there exists some
$$N_{1} \in \mathfrak{R}(M, \mathcal{F}_{M} \cup
\{\Psi_{\mathbf{x}}\}) \cap
\text{Sym}_{q}^{\tt{pd}}(\mathbb{R}_{\neq 0}).$$
If such an $N_1$  exists, let its
 eigenvalues be $(\mu_{1}, \dots, \mu_{q})$.
From \cref{definition: reduct}, 
for all $\mathbf{0}  \neq  \mathbf{y} \in \chi_{q}$, if
$\Psi_{\mathbf{y}}(M) \neq 0$, then we also have 
$\Psi_{\mathbf{y}}(N_{1}) \neq 0$.
In other words,
for all $\mathbf{0}  \neq  \mathbf{y}  \in \chi_{q}$,
$\Psi_{\mathbf{y}}(N_{1}) = 0$ implies that
$\Psi_{\mathbf{y}}(M) = 0$.
So, \cref{lemma: PsiExists} now implies that
$$\overline{\mathcal{L}(\mu_{1}, \dots, \mu_{q})} \subseteq
\overline{\mathcal{L}(\lambda_{1}, \dots, \lambda_{q})}.$$
Moreover, by our choice of $N_{1}$,
we know that $\Psi_{\mathbf{x}}(N_{1}) \neq 0$, 
i.e., $\mathbf{x} \not \in \overline{\mathcal{L}(\mu_{1}, \dots, \mu_{q})} $, again by \cref{lemma: PsiExists}.
But by our choice of $\mathbf{x}$,
we know that $\mathbf{x} \in \mathcal{B} \subset
\mathcal{L}(\lambda_{1}, \dots, \lambda_{q}) \subseteq
\overline{\mathcal{L}(\lambda_{1}, \dots, \lambda_{q})}$.
Therefore, we see that in fact,
$$\overline{\mathcal{L}(\mu_{1}, \dots, \mu_{q})} \subsetneq
\overline{\mathcal{L}(\lambda_{1}, \dots, \lambda_{q})}.$$
Therefore,
\cref{lemma: dimReduction} tells us that
$$\dim(\mathcal{L}(\mu_{1}, \dots, \mu_{q})) <
\dim(\mathcal{L}(\lambda_{1}, \dots, \lambda_{q})).$$
Now we can repeat this process with
$N_{1}$ in place of $M$.
Since the lattice dimension of
$\mathcal{L}(\lambda_{1}, \dots, \lambda_{q})$
is some finite $d \in \mathbb{Z}_{\geq 0}$,
we can only  repeat this process
$k$ times, for some $k \leq d$, until we
find some $N_{k} \in 
\text{Sym}_{q}^{\tt{pd}}(\mathbb{R}_{\neq 0})$
such that it has a basis $\mathcal{B} \subset \mathcal{D}$,
and thus $\PlEVAL(N_{k})$ is $\#$P-hard
by \cref{theorem: latticeHardness}.
Furthermore,
$$\PlEVAL(N_{k}) \leq \cdots \leq \PlEVAL(N_{1})
\leq \PlEVAL(M).$$
This would then allow us to prove that $\PlEVAL(M)$
is $\#$P-hard.

In other words, given a matrix $M \in
\text{Sym}_{q}^{\tt{pd}}(\mathbb{R}_{\neq 0})$,
and any non-zero $\mathbf{x} \in \chi_{q} \setminus \mathcal{D}$,
our goal will be to prove that there exists some
$N \in \mathfrak{R}(M, \mathcal{F}_{M} \cup
\{\Psi_{\mathbf{x}}\} ) \cap
\text{Sym}_{q}^{\tt{pd}}(\mathbb{R}_{\neq 0})$.
In subsequent sections, we will
find larger and larger classes of matrices $M$
and $\mathbf{x}$ for which we can find such an
$N$.


\section{Lattice on Diagonal Entries}\label{sec: lattice_diagonal}

Let $M \in \text{Sym}_{q}^{\tt{F}}(\mathbb{R}_{\neq 0})$,
and let
$\mathcal{F}$ be a countable set of
$\text{Sym}_{q}(\mathbb{R})$-polynomials.
Since our goal is to show that for some
choices of $M$ and $\mathcal{F}$,
$\mathfrak{R}(M, \mathcal{F}) \neq \emptyset$,
we will first prove some sufficient conditions for this
set to be non-empty.
We will recall from \cref{lemma: generatingSet}
that given any $M \in 
\text{Sym}_{q}(\mathbb{R}_{\neq 0})$, there exists
a generating set $\{g_{1}, \dots, g_{d}\}$
for the entries of $M$.
We also recall that \cref{lemma: MequivalentCM}
allows us to replace this matrix $M$ with a
matrix $N = c M$ satisfying $\PlEVAL(M)
\equiv \PlEVAL(N)$ and
such that the entries of $N$ are generated by
$\{g_{t}\}_{t \in [d]}$ with  exponents $e_{ijt} \geq 0$
for all $i, j \in [q]$, and $t \in [d]$.
Therefore, \cref{lemma: thickeningInterpolation}
tells us that $\PlEVAL(\mathcal{T}_{M}(\mathbf{p})) \leq
\PlEVAL(M)$ for all $\mathbf{p} \in \mathbb{R}^{d}$,
where $\mathcal{T}_{M}$ is as defined in
\cref{definition: mathcalT}, and satisfies
$\mathcal{T}_{M}(g_{1}, \dots, g_{d}) = M$.

\begin{lemma}\label{lemma: thickeningWorks}
    Let $M \in \text{Sym}_{q}^{\tt{F}}(\mathbb{R}_{\neq 0})$
    be a matrix,
    whose entries are generated by $\{g_{t}\}_{t \in [d]}$.
    Let $\mathcal{F}$ be a countable set of 
    $\text{Sym}_{q}(\mathbb{R})$-polynomials,
    such that for each $F \in \mathcal{F}$,
    there exists some $\mathbf{p}_{F} \in \mathbb{R}^{d}$,
    such that
    $F(\mathcal{T}_{M}(\mathbf{p}_{F})) \neq 0$.
    Then, there exists  $N
    = \mathcal{T}_{M}(\mathbf{p}^{*})\in 
    \mathfrak{R}(M, \mathcal{F}) \cap
    \text{Sym}_{q}^{\tt{F}}(\mathbb{R}_{\neq 0})$
    for some $\mathbf{p}^{*} \in \mathbb{R}^{d}$.
    Moreover,
    \begin{itemize}
        \item if $M \in
            \text{Sym}_{q}^{\tt{F}}(\mathbb{R}_{> 0})$,
            we can ensure that $N \in
            \text{Sym}_{q}^{\tt{F}}(\mathbb{R}_{> 0})$,
        \item if $M \in
            \text{Sym}_{q}^{\tt{pd}}(\mathbb{R}_{\neq 0})$,
            we can ensure that $N \in
            \text{Sym}_{q}^{\tt{pd}}(\mathbb{R}_{\neq 0})$,
        \item if $M \in
            \text{Sym}_{q}^{\tt{pd}}(\mathbb{R}_{> 0})$,
            we can ensure that $N \in
            \text{Sym}_{q}^{\tt{pd}}(\mathbb{R}_{> 0})$.
    \end{itemize}
\end{lemma}
\begin{proof}
    We will first replace $M$ with $c M$
    as guaranteed by \cref{lemma: MequivalentCM}, for some $c >0$.
  Since $M \in \text{Sym}_{q}^{\tt{F}}(\mathbb{R}_{\neq 0})$,
    it follows that $c M \in
    \text{Sym}_{q}^{\tt{F}}(\mathbb{R}_{\neq 0})$ as well.
    If $M \in \text{Sym}_{q}^{\tt{F}}(\mathbb{R}_{> 0})$
    (similarly, $\text{Sym}_{q}^{\tt{pd}}(\mathbb{R}_{\neq 0})$
    and $\text{Sym}_{q}^{\tt{pd}}(\mathbb{R}_{> 0})$), then
    so is $c M$.
    We note from the definition of
    $\mathcal{T}_{M}$ in 
    \cref{definition: mathcalT},
    that the entries of the matrix
    $\mathcal{T}_{M}(\mathbf{p})$ are all polynomials in
    $\mathbf{p}$.
    Since each $F \in \mathcal{F}$ is a 
    $\text{Sym}_{q}(\mathbb{R})$-polynomial,
    it follows that
    $F(\mathcal{T}_{M}(\mathbf{p}))$
    is a polynomial in $\mathbf{p}$ for
    all $F \in \mathcal{F}$.
    Moreover, by our choice of $\mathcal{F}$, we know
    that for each $F \in \mathcal{F}$, there exists
    some $\mathbf{p}_{F} \in \mathbb{R}^{d}$
    such that $F(\mathcal{T}_{M}(\mathbf{p}_{F})) \neq 0$.
    We also note that
    $\det(\mathcal{T}_{M}(g_{1}, \dots, g_{d})) = \det(M) \neq 0$,
    since $M$ is full rank.
    Therefore, $F(\mathcal{T}_{M}(\mathbf{p})):
    \mathbb{R}^{d} \rightarrow \mathbb{R}$
    is a non-zero polynomial for all $F \in \mathcal{F} \cup \{\det\}$.
    
    Let us now assume that $M$ is positive definite.
    So, the eigenvalues
    $(\lambda_{1}, \dots, \lambda_{q})$ of
    $\mathcal{T}_{M}(g_{1}, \dots, g_{d}) = M$ will all be positive.
    We note that each entry of the matrix
    $\mathcal{T}_{M}(\mathbf{p})$ is a continuous function
    of $\mathbf{p}$.
    So, each of the coefficients of the 
    characteristic polynomial of $\mathcal{T}_{M}(\mathbf{p})$
    is a continuous function of $\mathbf{p}$.
    Since the eigenvalues of $\mathcal{T}_{M}(\mathbf{p})$
    are simply the roots of the characteristic
    polynomial, we can use the well-known
    fact that the roots of a polynomial
    are continuous in the coefficients \cite{harris1987shorter}
    to see that eigenvalues of
    $\mathcal{T}_{M}(\mathbf{p})$ are also continuous
    as functions of $\mathbf{p}$.
    In other words, there exist open intervals $I_{1}, \dots, I_{d}$
    such that $g_{t} \in I_{t}$ for all $t \in [d]$,
    and $\mathcal{T}_{M}(\mathbf{p})$ is positive definite
    for all $\mathbf{p} \in
    (I_{1} \times \cdots \times I_{d}) = U$.
    Since $g_{t} > 1$ for all $t \in [q]$, we may further
    assume that $U \subset (\mathbb{R}_{> 1})^{d}$.
    By construction, we note that $U$ has positive measure.
    If on the other hand, $M$ is not positive definite,
    we can simply let
    $U = (\mathbb{R}_{> 1})^{d}$.
    
    We note that for each $F \in \mathcal{F} \cup \{\det \}$,
    the set $\emptyset_{F} =
    \{\mathbf{p}: F(\mathcal{T}_{M}(\mathbf{p})) = 0\}$
    has  measure $0$, since $F(\mathcal{T}_{M}(\mathbf{p}))$
    is a non-zero polynomial.
    The measure of a countable union of measure $0$
    sets is also $0$.
    Therefore, there exists some $\mathbf{p}^{*} = 
    (p^{*}_{1}, \dots, p^{*}_{d}) \in U 
    \setminus \cup_{F} \emptyset_{F}$.
    If we let $N = \mathcal{T}_{M}(\mathbf{p}^{*})$,
    we now see that $F(N) \neq 0$ for all
    $F \in \mathcal{F} \cup \{ \det\}$.
    In particular, this implies that $\det(N) \neq 0$.
    Moreover, since $\mathbf{p}^{*} \in U \subseteq
    (\mathbb{R}_{> 1})^{d}$,
    we see that $N \in
    \text{Sym}_{q}^{\tt{F}}(\mathbb{R}_{\neq 0})$.
    So, $N \in \mathfrak{R}(M, \mathcal{F}) \cap
    \text{Sym}_{q}^{\tt{F}}(\mathbb{R}_{\neq 0})$
    is the required matrix.

    Moreover, if $M \in \text{Sym}_{q}(\mathbb{R}_{> 0})$,
    we know from \cref{equation: generatingM}
    that $e_{ij0} = 0$ for all $i, j \in [q]$.
    Since $\mathbf{p}^{*} \in U$, this implies that
    $N \in \text{Sym}_{q}(\mathbb{R}_{> 0})$ as well.
    Finally, if $M \in \text{Sym}_{q}^{\tt{pd}}(\mathbb{R}_{\neq 0})$,
    we note that by our choice of $U$,
    $N \in \text{Sym}_{q}^{\tt{pd}}(\mathbb{R}_{\neq 0})$
    as well.
    Therefore, we see that this $N$ satisfies
    all the requirements of the lemma.
\end{proof}

\begin{corollary}\label{corollary: thickeningSufficient}
    Let $M \in \text{Sym}_{q}^{\tt{pd}}(\mathbb{R}_{> 0})$.
    Let $\mathcal{F}$ be a countable set of
    $\text{Sym}_{q}(\mathbb{R})$-polynomials
    such that $F(M) \neq 0$ for all $F \in \mathcal{F}$.
    Let $\mathbf{0} \neq  \mathbf{x}  \in \chi_{q}$
    such that $\Psi_{\mathbf{x}}(T_{n}M) \neq 0$
    for some integer $n \geq 1$.
    Then, there exists some $N \in 
    \mathfrak{R}(M, \mathcal{F} \cup 
    \mathcal{F}_{M} \cup \{\Psi_{\mathbf{x}}\})
    \cap \text{Sym}_{q}^{\tt{pd}}(\mathbb{R}_{> 0})$.
\end{corollary}
\begin{proof}
    Let $\{g_{t}\}_{t \in [d]}$ be a generating set
    of the entries of $M$.
    As we have already seen, we may replace $M$
    with $c \cdot M$ as guaranteed by
    \cref{lemma: MequivalentCM}.
    We have seen in the remark following
    \cref{definition: matrixPoly}
    that $F(M) = 0$ if and only if $F(c M) = 0$.
    Therefore, $F(c \cdot M) \neq 0$
    for all $F \in \mathcal{F}$.
    By our choice of $\mathcal{F}$, and 
    the definition of $\mathcal{F}_{M}$ in
    \cref{equation: FM},
    we see that for all $F \in \mathcal{F} \cup \mathcal{F}_{M}$,
    $F(\mathcal{T}_{M}(g_{1}, \dots, g_{d})) = F(M) \neq 0$.
    Moreover,
    we also see that
    since $M \in \text{Sym}_{q}(\mathbb{R}_{> 0})$,
    $\Psi_{\mathbf{x}}(\mathcal{T}_{M}
    ((g_{1})^{n}, \dots, (g_{d})^{n}))
    = \Psi_{\mathbf{x}}(T_{n}M) \neq 0$ as given.
    Let  $\mathcal{F}' = \mathcal{F} \cup \mathcal{F}_{M} 
    \cup \{\Psi_{\mathbf{x}}\}$. It
    is a countable set of $\text{Sym}_{q}(\mathbb{R})$-polynomials
     with the property that for every $F \in \mathcal{F}'$
    there exists some $\mathbf{p}_{F} \in \mathbb{R}^{d}$
    such that $F(\mathcal{T}_{M}(\mathbf{p}_{F})) \neq 0$.
    We are also given that $M \in \text{Sym}_{q}^{\tt{pd}}(\mathbb{R}_{> 0})$.
    Therefore, \cref{lemma: thickeningWorks}
    implies that there exists some $N \in 
    \mathfrak{R}(M, \mathcal{F}')
    \cap \text{Sym}_{q}^{\tt{pd}}(\mathbb{R}_{> 0})$.
\end{proof}

Now we will consider $M \in
\text{Sym}_{q}^{\tt{pd}}(\mathbb{R}_{> 0})$
and $\mathbf{0} \neq  \mathbf{x} \in \chi_{q}$,
such that $\Psi_{\mathbf{x}}(T_{n}M) = 0$ for all $n$.
To better understand these matrices, it will be helpful
for us to study the function $\Psi_{\mathbf{x}}:
\text{Sym}_{q}(\mathbb{R}) \rightarrow \mathbb{R}$ better.
We already know that the function $\Psi_{\mathbf{x}}$ is a
homogeneous polynomial of some degree $d \geq 1$.

Now, we will try to understand the individual terms of this polynomial.
Recall that we have defined
$$\mathcal{P}_{q^{2}}(d) = \left\{(k_{ij})_{i, j \leq q} \in \mathbb{Z}^{q^{2}}
\hspace{0.08cm}\Big|\hspace{0.1cm} k_{ij} \geq 0,
\sum_{i, j \in [q]}k_{ij} = d\right\}.$$
Given any $\mathbf{k} \in \mathcal{P}_{q^{2}}(d)$,
we will now define
$m_{\mathbf{k}}:
\text{Sym}_{q}(\mathbb{R}) \rightarrow \mathbb{R}$, such that
$$m_{\mathbf{k}}(M)
= \prod_{i, j \in [q]}M_{ij}^{k_{ij}}.$$
As
 $\Psi_{\mathbf{x}}$  had degree $d$, each $m_{\mathbf{k}}$  represents a monomial term
in $\Psi_{\mathbf{x}}$ (with possibly a 0 coefficient).
We then let
$$\mathcal{M} = \left\{m_{\mathbf{k}}\hspace{0.08cm}\Big|\hspace{0.1cm}
\mathbf{k} \in \mathcal{P}_{q^{2}}(d)\right\},$$
and we can now express $\Psi_{\mathbf{x}}$ as
$$\Psi_{\mathbf{x}}(M) =
\sum_{m \in \mathcal{M}}c_{m}(\mathbf{x})m(M),$$
where $c_{m}(\mathbf{x}) \in \mathbb{R}$ for all $m \in \mathcal{M}$.

\begin{lemma}\label{lemma: diagonalLatticeSpecial}
    Let $\mathbf{0} \neq  \mathbf{x}  \in \chi_{q}$.
    There exists $\epsilon > 0$, such that given any
    matrix $M \in \text{Sym}_{q}(\mathbb{R})$ satisfying
    the conditions that $|M_{ij} - I_{ij}| < \epsilon$
    for all $i, j \in [q]$, and
    $\Psi_{\mathbf{x}}(T_{n}M) = 0$ for all $n \geq 1$,
    then
    $$\Psi_{\mathbf{x}}(\text{diag}(M)) = 0,$$
    where ${\rm \text{diag}}(M) \in \text{Sym}_{q}(\mathbb{R})$
    is the diagonal matrix such that
    $$\text{diag}(M)_{ij} = \begin{cases}
        M_{ii} & \text{if } i = j,\\
        0 & \text{otherwise}.
    \end{cases}$$
\end{lemma}
\begin{proof}
    We will let $\epsilon = \nicefrac{1}{3d}$,
    and assume that 
    $|M_{ij} - I_{ij}| < \epsilon$ for all $i, j \in [q]$.
    By construction of $\mathcal{M}$, we see that
    $$\Psi_{\mathbf{x}}(T_{n}M) = \sum_{m \in \mathcal{M}}
    c_{m}(\mathbf{x})\left(m(M)\right)^{n}.$$
    Note that each $m \in \mathcal{M}$ is $m_{\mathbf{k}}$ for some
    $\mathbf{k} \in \mathcal{P}_{q^{2}}(d)$, where $d$ is the degree of $\Psi_{\mathbf{x}}$.
    Now, we let
    $$\mathcal{M}(M) = \left\{m(M)
    \hspace{0.08cm}\Big|\hspace{0.1cm} m \in
    \mathcal{M}\right\}$$
    be the set of all values of the monomial terms in
    $\mathcal{M}$ when evaluated at $M$.
    For all the values $v \in \mathcal{M}(M)$, we let
    $$c_{v, M}(\mathbf{x}) = \sum_{m \in \mathcal{M}:\ m(M)\ =\ v}
    c_{m}(\mathbf{x}).$$
    So, we see that
    $$\Psi_{\mathbf{x}}(T_{n}M) = \sum_{v \in \mathcal{M}(M)}c_{v, M}(\mathbf{x}) \cdot v^{n}.$$
    We can see that $|\mathcal{M}(M)|$ is $O(1)$.
    So, in particular, the equations $\Psi_{\mathbf{x}}(T_{n}M) = 0$
    form a full rank Vandermonde system of linear equations.
    This implies that $c_{v, M}(\mathbf{x}) = 0$ for all $v \in \mathcal{M}(M)$.
    
    Setting that aside for a moment,
    we also note that by construction of $\text{diag}(M)$,
    $$m_{\mathbf{k}}(\text{diag}(M)) = \begin{cases}
        m_{\mathbf{k}}(M) & \text{if } k_{ij} = 0\ \forall\ i \neq j,\\
        0 & \text{otherwise}.
    \end{cases}$$
    Therefore, if we define
    $$\text{diag}(\mathcal{M}) = \left\{m_{\mathbf{k}}
    \hspace{0.08cm}\Big|\hspace{0.1cm}
    \mathbf{k} \in \mathcal{P}_{q^{2}}(d),
    \sum_{i\in [q]}k_{ii} = d\right\},$$
    we see that $\text{diag}(\mathcal{M}) \subset \mathcal{M}$,
    and
    $$\Psi_{\mathbf{x}}(\text{diag}(M)) = 
    \sum_{m \in \text{diag}(\mathcal{M})}
    c_{m}(\mathbf{x})m(M).$$
    So, if we let
    $$\text{diag}(\mathcal{M})(M) = \left\{m(M)
    \hspace{0.08cm}\Big|\hspace{0.1cm} m \in
    \text{diag}(\mathcal{M})\right\}, ~~\text{ and  }~~
    c_{v, \text{diag}(M)}(\mathbf{x}) =
    \sum_{m \in \text{diag}(\mathcal{M}):\
    m(M)\ =\ v} c_{m}(\mathbf{x})$$
    for all $v \in \text{diag}(\mathcal{M})(M)$,
    then we see that
    $$\Psi_{\mathbf{x}}(\text{diag}(M)) = 
    \sum_{v \in \text{diag}(\mathcal{M})(M)}
    c_{v, \text{diag}(M)}(\mathbf{x}) \cdot v.$$

    Let us now consider $m_{\mathbf{k}} \in
    \text{diag}(\mathcal{M})$.
    We note that by our choice of $\epsilon$ and $M$,
    $m_{\mathbf{k}}(M) > (1 - \epsilon)^{d}$.
    On the other hand, let us now consider some $m_{\mathbf{k}'}
    \in \mathcal{M} \setminus \text{diag}(\mathcal{M})$.
    We note that  by our choice of $M$,
    $|m_{\mathbf{k}'}(M)| < (1 + \epsilon)^{d-1} \cdot \epsilon$.
    Therefore,
    $$\frac{|m_{\mathbf{k}'}(M)|}{|m_{\mathbf{k}}(M)|}
    < \frac{(1 + \epsilon)^{d-1} \cdot \epsilon}{(1 - \epsilon)^{d}}
    = \left(\frac{1 + \epsilon}{1 - \epsilon}\right)^{d}\cdot
    \frac{\epsilon}{1 + \epsilon}.$$
    Now, we note that
    $$\left(\frac{1 + \epsilon}{1 - \epsilon}\right)^{d}
    = \left(1 + \frac{2\epsilon}{1 - \epsilon}\right)^{d}
    < \left(e^{\left(\frac{2\epsilon}{1 - \epsilon}\right)}\right)^{d}.$$
    From our choice of $\epsilon = \nicefrac{1}{3d}$, we know that
    $\epsilon \leq \nicefrac{1}{3}$
    since  $d \geq 1$.
    Therefore, $1 - \epsilon \geq \nicefrac{2}{3}$. So,
    $$\frac{2\epsilon}{1 - \epsilon} \leq 3 \epsilon 
    \le \frac{1}{d}.$$
    So,
    $$\frac{|m_{\mathbf{k}'}(M)|}{|m_{\mathbf{k}}(M)|} < \left(\frac{1 + \epsilon}{1 - \epsilon}\right)^{d} \cdot 
    \frac{\epsilon}{1 + \epsilon}
    < \frac{e}{3}
    < 1.$$
    This proves that if $v \in \text{diag}(\mathcal{M})(M)$,
    then there cannot exist any
    $m \in \mathcal{M} \setminus \text{diag}(\mathcal{M})$,
    such that $v = m(M)$.
    Therefore,
    $$c_{v, \text{diag}(M)}(\mathbf{x}) =
    \sum_{m \in \text{diag}(\mathcal{M}):\
    m(M)\ =\ v} c_{m}(\mathbf{x}) = 
    \sum_{m \in \mathcal{M}:\
    m(M)\ =\ v} c_{m}(\mathbf{x}) = 
    c_{v, M}(\mathbf{x})$$
    for all $v \in \text{diag}(\mathcal{M})(M)$.
    But we already know that $c_{v, M}(\mathbf{x}) = 0$
    for all $v \in \mathcal{M}(M)$.
    Therefore, $c_{v, \text{diag}(M)}(\mathbf{x}) = 0$
    for all $v \in \text{diag}(\mathcal{M})(M)$.
    Therefore,
    $$\Psi_{\mathbf{x}}(\text{diag}(M)) = 
    \sum_{v \in \text{diag}(\mathcal{M})(M)}
    c_{v, \text{diag}(M)}(\mathbf{x}) \cdot v = 0.$$
\end{proof}

In order to make effective use of
\cref{lemma: diagonalLatticeSpecial}, we
will need one more lemma that provides a sufficient
condition to prove that $\mathfrak{R}(M, \mathcal{F})$
is not empty.

\begin{definition}\label{definition: mathcalS}
    Let $M = HDH^{\tt{T}} \in
    \text{Sym}_{q}^{\tt{pd}}(\mathbb{R})$.
    We now define the function
    $\mathcal{S}_{M}: \mathbb{R} \rightarrow
    \text{Sym}_{q}^{\tt{pd}}(\mathbb{R})$,
    such that
    $$\mathcal{S}_{M}(\theta) = HD^{\theta} H^{\tt{T}}.$$
\end{definition}

\begin{remark*}
    Let $M = HDH^{\tt{T}} \in
    \text{Sym}_{q}^{\tt{pd}}(\mathbb{R})$
    have the eigenvalues $(\lambda_{1}, \dots, \lambda_{q})$.
    Since these are all positive, $D^{\theta}$ is well-defined
    for all $\theta \in \mathbb{R}$, as the diagonal matrix
    such that $(D^{\theta})_{ii} = (\lambda_{i})^{\theta}$.
    Consequently, $\mathcal{S}_{M}$ is well-defined.
    Moreover, we note that
    $$\mathcal{S}_{M}(\theta)_{ij}
    = (H_{i1}H_{j1})e^{\theta \cdot \log(\lambda_{1})}
    + \cdots + 
    (H_{iq}H_{jq})e^{\theta \cdot \log(\lambda_{q})},$$
    for all $i, j \in [q]$.
    Therefore, $\mathcal{S}_{M}(\theta)_{ij}$
    is a real analytic function in $\theta$
    for all $i, j \in [q]$.    
\end{remark*}

\begin{lemma}\label{lemma: stretchingWorks}
    Let $M = HDH^{\tt{T}} \in
    \text{Sym}_{q}^{\tt{pd}}(\mathbb{R}_{\neq 0})$.
    Let $\mathcal{F}$ be a countable set of 
    $\text{Sym}_{q}(\mathbb{R})$-polynomials,
    such that for each $F \in \mathcal{F}$,
    there exists some $\theta_{F} \in \mathbb{R}$
    such that
    $F(\mathcal{S}_{M}(\theta_{F})) \neq 0$.
    Then, given any interval $(a, b ) \subset \mathbb{R}$,
    there exists some $\theta^{*} \in (a, b)$
    such that $\mathcal{S}_{M}(\theta^{*}) \in 
    \mathfrak{R}(M, \mathcal{F}) \cap
    \text{Sym}_{q}^{\tt{pd}}(\mathbb{R}_{\neq 0})$.
\end{lemma}
\begin{proof}
    We will let
    $$\mathcal{A} = \Big\{\mathcal{S}_{M}(\theta)_{ij}:
    i, j \in [q] \Big\} \cup
    \Big\{F(\mathcal{S}_{M}(\theta)):
    F \in \mathcal{F} \Big\}.$$
    Since each $F \in \mathcal{F}$ is polynomial
    in the entries of the input matrix, and since
    the entries of the matrix
    $\mathcal{S}_{M}(\theta)$
    are all real analytic functions of $\theta$,
    we see that each function in $\mathcal{A}$
    is a real valued analytic function of $\theta$.
    We also know that
    $(\mathcal{S}_{M}(1))_{ij} = 
    (HDH^{\tt{T}})_{ij} = M_{ij} \neq 0$
    for all $i, j \in [q]$, and that
    $F(\mathcal{S}_{M}(\theta_{F})) \neq 0$
    for all $F \in \mathcal{F}$.
    So, each function in $\mathcal{A}$ is a non-zero
    analytic function.

    Since all the functions in $\mathcal{A}$ are
    non-zero real analytic functions, the Identity
    Theorem for Real Analytic Functions
    (\cite{krantz2002primer}, Corollary 1.2.7)
    implies that
    the set of zeros of any of these functions does not
    have an accumulation point.
    In particular, each of these functions has
    only finitely many zeros within $(a, b)$.
    Since the countable union of finite sets is
    only countable, there exists some $a < \theta^{*} < b$
    such that if we let $N = \mathcal{S}_{M}(\theta^{*})$,
    then $N_{ij} \neq 0$ for all $i, j \in [q]$,
    and $F(N) \neq 0$ for all $F \in \mathcal{F}$.
    We may also assume that $\theta^{*} \neq 0$.

    We also note that since $\theta^{*} \neq 0$, and
    $\lambda_{i} > 0$ for all $i \in [q]$,
    $$\lambda_{1}^{y_{1}} \cdots \lambda_{q}^{y_{q}} = 1 \iff
    (\lambda_{1}^{y_{1}} \cdots \lambda_{q}^{y_{q}})^{\theta^{*}} = 1
    \iff (\lambda_{1}^{\theta^{*}})^{y_{1}} \cdots
    (\lambda_{q}^{\theta^{*}})^{y_{q}} = 1,$$
    for all $\mathbf{y} \in \chi_{q}$.
    This implies that
    $\mathcal{L}(\lambda_{1}, \dots, \lambda_{q}) = 
    \mathcal{L}(\lambda_{1}^{\theta^{*}}, \dots,
    \lambda_{q}^{\theta^{*}})$.
    Now, \cref{lemma: stretchingBasic} implies that
    $\PlEVAL(N) \equiv \PlEVAL(M)$.
    Therefore, $N = \mathcal{S}_{M}(\theta^{*}) \in
    \mathfrak{R}(M, \mathcal{F}) \cap
    \text{Sym}_{q}^{\tt{pd}}(\mathbb{R}_{\neq 0})$
    is the required matrix.
\end{proof}

We will now use \cref{lemma: stretchingWorks} to
prove a useful corollary.
\begin{corollary}\label{corollary: stretchingWorksPositive}
    Let $M \in
    \text{Sym}_{q}^{\tt{pd}}(\mathbb{R}_{> 0})$.
    Let $\mathcal{F}$ be a countable set of 
    $\text{Sym}_{q}(\mathbb{R})$-polynomials,
    such that for each $F \in \mathcal{F}$,
    there exists some $\theta_{F} \in \mathbb{R}$
    such that
    $F(\mathcal{S}_{M}(\theta_{F})) \neq 0$.
    Then, there exists some
    $N = \mathcal{S}_{M}(\theta^{*}) \in 
    \mathfrak{R}(M, \mathcal{F}) \cap
    \text{Sym}_{q}^{\tt{pd}}(\mathbb{R}_{> 0})$.
\end{corollary}
\begin{proof}
    We note that $M_{ij} > 0$ for all $i, j \in [q]$.
    We also note that $\mathcal{S}_{M}$
    defined in \cref{definition: mathcalS}
    is continuous as a function of $\theta$.
    Since $\mathcal{S}_{M}(1) = M$,
    there exists some $a < 1 < b$
    such that $\mathcal{S}_{M}(\theta)_{ij}
    > 0$ for all $i, j \in [q]$, and $\theta \in (a, b)$.
    We can apply \cref{lemma: stretchingWorks}
    with this choice of $a < b$ to find
    the required $N = \mathcal{S}_{M}(\theta^{*}) \in 
    \mathfrak{R}(M, \mathcal{F}) \cap 
    \text{Sym}_{q}^{\tt{pd}}(\mathbb{R}_{> 0})$.
\end{proof}

\begin{theorem}\label{theorem: positiveDefiniteDiagonal}
    Let $M \in \text{Sym}_{q}^{\tt{pd}}(\mathbb{R}_{> 0})$.
    Let $\mathcal{F}$ be a countable set of
    $\text{Sym}_{q}(\mathbb{R})$-polynomials,
    such that $F(M) \neq 0$ for all $F \in \mathcal{F}$.
    Let $\mathbf{0}  \neq  \mathbf{x}  \in \chi_{q}$.
    If $\Phi_{\mathbf{x}}(M_{11}, \dots, M_{qq}) \neq 0$,
    there exists some $N \in 
    \mathfrak{R}(M, \mathcal{F} \cup 
    \mathcal{F}_{M} \cup \{\Psi_{\mathbf{x}}\}) \cap
    \text{Sym}_{q}^{\tt{pd}}(\mathbb{R}_{> 0})$.
\end{theorem}
\begin{proof}
    We consider $\mathcal{S}_{M}: \mathbb{R} \rightarrow
    \text{Sym}_{q}^{\tt{F}}(\mathbb{R})$ as defined
    in \cref{definition: mathcalS}.
    We see that
    $F(\mathcal{S}_{M}(1)) = F(M) \neq 0$
    for all $F \in \mathcal{F} \cup \mathcal{F}_{M}$.
    We also see that the function
    $\zeta: N \mapsto \Psi_{\mathbf{x}}(\text{diag}(N))$
    is a $\text{Sym}_{q}(\mathbb{R})$-polynomial.
    We note that since $\text{diag}(N)$ is a diagonal
    matrix, its eigenvalues are precisely the
    diagonal values of $\text{diag}(N)$, i.e.,
    $(N_{11}, \dots, N_{qq})$.
    Therefore, from the construction of 
    $\Phi_{\mathbf{x}}$ and $\Psi_{\mathbf{x}}$,
    we know that
    $$\Psi_{\mathbf{x}}(\text{diag}(N)) = 0 \iff
    \Phi_{\mathbf{x}}(N_{11}, \dots, N_{qq}) = 0$$
    for all $N \in \text{Sym}_{q}^{\tt{F}}(\mathbb{R})$.
    In particular, since
    $\Phi_{\mathbf{x}}(\mathcal{S}_{M}(1)_{11},
    \dots, \mathcal{S}_{M}(1)_{qq}) = 
    \Phi_{\mathbf{x}}(M_{11}, \dots, M_{qq}) \neq 0$
    by our choice of $M$, we see that
    $\zeta(\mathcal{S}_{M}(1)) \neq 0$.
    So, if we let
    $$\mathcal{F}' = \mathcal{F} \cup \mathcal{F}_{M}
    \cup \{\zeta\},$$
    we see that for each $F \in \mathcal{F}'$,
    there exists some $\theta_{F} \in \mathbb{R}$
    such that $F(\mathcal{S}_{M}(\theta_{F})) \neq 0$.

    We now note that $\mathcal{S}_{M}(0)
    = HD^{0}H^{\tt{T}} = HH^{\tt{T}} = I$.
    Given $\mathbf{x} \in \chi_{q}$, we will now
    let $\epsilon > 0$ be the number whose
    existence is guaranteed by
    \cref{lemma: diagonalLatticeSpecial}.
    Since $\mathcal{S}_{M}(\theta)$ is continuous
    as a function of $\theta$, we know that there
    exists some $\delta > 0$, such that
    for all $0 < \theta < \delta$, for all $i, j \in [q]$,
    $|\mathcal{S}_{M}(\theta)_{ij} - I_{ij}| < \epsilon$.
    We can now use \cref{lemma: stretchingWorks}
    to find $0 < \theta^{*} < \delta$
    such that $M' = \mathcal{S}_{M}(\theta^{*})
    \in \mathfrak{R}(M, \mathcal{F}')$.
    It should be stressed here that the 
    entries of $M'$ may be negative.
    
    Since $\zeta \in \mathcal{F}'$, we know that
    $\zeta (M') = \Psi_{\mathbf{x}}(\text{diag}(M'))\neq 0$.
    Moreover, due to our choice of $\theta^{*} < \delta$,
    \cref{lemma: diagonalLatticeSpecial} allows us
    to conclude that there exists some $n \geq 1$
    such that $\Psi_{\mathbf{x}}(T_{n}M')) \neq 0$.
    We can now define the $\text{Sym}_{q}(\mathbb{R})$-polynomial $\xi: \text{Sym}_{q}(\mathbb{R}) 
    \rightarrow \mathbb{R}$ such that
    $\xi(N) = \Psi_{\mathbf{x}}(T_{n}N)$.
    Since $n \geq 1$ is an integer, the entries
    of $T_{n}(N)$ are all homogeneous polynomials
    in the entries of the matrix $N$, and therefore, $\xi$ is a 
    $\text{Sym}_{q}(\mathbb{R})$-polynomial.
    We have seen that $\xi(\mathcal{S}_{M}(\theta^{*}))
    \neq 0$.
    Therefore, since $M \in \text{Sym}_{q}^{\tt{pd}}
    (\mathbb{R}_{> 0})$, we may use
    \cref{corollary: stretchingWorksPositive} to
    find a new
    $M'' \in \mathfrak{R}(M, \mathcal{F} \cup \mathcal{F}_{M} 
    \cup \{\xi\}) \cap
    \text{Sym}_{q}^{\tt{pd}}(\mathbb{R}_{> 0})$.
    Since $\xi(M'') \neq 0$, we see
    that $\Psi_{\mathbf{x}}(T_{n}M'') \neq 0$.
    
    But then, \cref{corollary: thickeningSufficient}
    implies that there exists some $N \in
    \mathfrak{R}(M'', \mathcal{F} \cup \mathcal{F}_{M}
    \cup \{\Psi_{\mathbf{x}}\}) \cap 
    \text{Sym}_{q}^{\tt{pd}}(\mathbb{R}_{> 0})$.
    Since $\PlEVAL(M'') \leq \PlEVAL(M)$,
    this means that $N \in 
    \mathfrak{R}(M, \mathcal{F} \cup \mathcal{F}_{M} \cup
    \{\Psi_{\mathbf{x}}\}) \cap
    \text{Sym}_{q}^{\tt{pd}}(\mathbb{R}_{> 0})$
    is the required matrix.
\end{proof}
\section{Confluence and Pairwise Order Independence}\label{sec: order_indep_matrices}

In this section, we will prove the $\#$P-hardness
of $\PlEVAL(M)$ for a subset of the matrices
$M \in \text{Sym}_{q}^{\tt{pd}}(\mathbb{R}_{> 0})$.
We will do so by identifying a large class of matrices
$M \in \text{Sym}_{q}^{\tt{pd}}(\mathbb{R}_{> 0})$,
and a large class of $\mathbf{x} \in \chi_{q}$,
such that
$\Phi_{\mathbf{x}}(M_{11}, \dots, M_{qq}) \neq 0$.

\begin{definition}\label{definition: confluence}
    Let $I, J$ be finite sets.
    Let $\mathbf{x} = (x_{i})_{i \in I} \in \mathbb{Z}_{> 0}^{|I|}$,
    and $\mathbf{y} = (y_{j})_{j \in J} \in \mathbb{Z}_{> 0}^{|J|}$.
    We say that $(\mathbf{x}, \mathbf{y})$ is a \emph{confluence}
    if it satisfies the following two conditions:
    \begin{enumerate}
    \item $\sum_{i \in I}x_{i} = \sum_{j \in J}y_{j}$,
    and 
    \item for any $S_{1}, S_{2} \subseteq I$, and 
    $T_{1}, T_{2} \subseteq J$,
    $$\left[ \sum_{i \in S_{1}}x_{i} = \sum_{j \in T_{1}}y_{j}
    ~\text{ and }~ \sum_{i \in S_{2}}x_{i} =
    \sum_{j \in T_{2}}y_{j} \right] \implies
    \sum_{i \in S_{1} \cap S_{2}}x_{i} = 
    \sum_{j \in T_{1} \cap T_{2}}y_{j}.$$
    \end{enumerate}
\end{definition}

As all $x_i, y_j >0$, in a \emph{confluence} clearly $I=\emptyset$ if and only if $J=\emptyset$.
We will only be interested in nonempty $I$ and $J$.
While \cref{definition: confluence} is concise,
we will now prove that it has an alternate
equivalent definition, with some useful
properties.

\begin{lemma}\label{lemma: confluenceEquivalence}
    Consider $(\mathbf{x}, \mathbf{y})$,
    where $\mathbf{x} = (x_{i})_{i \in I} \in \mathbb{Z}_{> 0}^{|I|}$,
    and $\mathbf{y} = (y_{i})_{j \in J} \in \mathbb{Z}_{> 0}^{|J|}$
    for nonempty finite sets $I$, and $J$.
    $(\mathbf{x}, \mathbf{y})$ is a confluence if and only
    if there exist partitions
    $I = S_{1} \sqcup \cdots \sqcup S_{r}$
    and $J = T_{1} \sqcup \cdots \sqcup T_{r}$
    for some $r \geq 1$, such that for any
    $S \subseteq I$, $T \subseteq J$,
    \begin{equation}\label{lm-confluence-def}
    \sum_{i \in S}x_{i} = \sum_{j \in T}y_{j} \iff
     \left[  S = \bigsqcup_{a \in P}S_{a} ~\text{ and }~
    T = \bigsqcup_{a \in P}T_{a}, ~\text{ for some }
    P \subseteq [r] \right]. 
    \end{equation}
    Furthermore, the paired partition $\{(S_{1}, T_{1}), \dots,
    (S_{r}, T_{r})\}$ is unique up to the order of the pairs.
\end{lemma}
\begin{proof}
    We will first assume that
    $I = S_{1} \sqcup \cdots \sqcup S_{r}$,
    and $J = T_{1} \sqcup \cdots \sqcup T_{r}$ satisfies the condition \cref{lm-confluence-def}.
    If we let $P = [r]$, we immediately get that
    $\sum_{i \in I}x_{i} = \sum_{j \in J}y_{j}$.
    Moreover, for any $a \in [r]$, let $P= \{a\}$
    then $\sum_{i \in S_{a}}x_{i} = \sum_{j \in T_{a}}y_{j}$.
    
    Now if we consider any
    $S_{1}, S_{2} \subseteq I$, and
    $T_{1}, T_{2} \subseteq J$ such that
    $$\sum_{i \in S_{1}}x_{i} = \sum_{j \in T_{1}}y_{j}
    ~~\text{ and }~~ \sum_{i \in S_{2}}x_{i} =
    \sum_{j \in T_{2}}y_{j},$$
    we know that there exist some
    $P_{1}, P_{2} \subseteq [r]$ such that
    $S_{1} = \sqcup_{a \in P_{1}}S_{a}$,
    $T_{1} = \sqcup_{a \in P_{1}}T_{a}$,
    $S_{2} = \sqcup_{a \in P_{2}}S_{a}$, and
    $T_{2} = \sqcup_{a \in P_{2}}T_{a}$.
    Therefore, if we let $P_{3} = P_{1} \cap P_{2}$,
    we see that $S_{1} \cap S_{2} =
    \sqcup_{a \in P_{3}}S_{a}$, and
    $T_{1} \cap T_{2} = \sqcup_{a \in P_{3}}T_{a}$.
    Therefore,
    $$\sum_{i \in S_{1} \cap S_{2}}x_{i} = 
    \sum_{a \in P_{3}}\sum_{i \in S_{a}}x_{i} = 
    \sum_{a \in P_{3}}\sum_{j \in T_{a}}y_{j} =
    \sum_{j \in T_{1} \cap T_{2}}y_{j}.$$
    This proves that $(\mathbf{x}, \mathbf{y})$
    is a confluence.

    Conversely,  assume that $(\mathbf{x}, \mathbf{y})$
    is a confluence.
   Let
    $$\mathcal{C} = \left\{(S, T): S \subseteq I, T \subseteq J,
    ~~\sum_{i \in S}x_{i} = \sum_{j \in T}y_{j}\right\},
    ~\text{ and }~$$
    $$\mathcal{M} = \Big\{(S, T)
    \in \mathcal{C}: (S, T) \neq (\emptyset, \emptyset) 
    ~\mbox{and}~(\forall\ S' \subsetneq S, \forall\ T' \subsetneq T)
    \left[ (S', T') \in \mathcal{C} \implies S' = T' = \emptyset
    \right] \Big\}$$
    be the minimal members of $\mathcal{C}$.
    Since $\mathcal{C}$ and consequently $\mathcal{M}$
    are finite sets,
    we may assume that $\mathcal{M} = \{(S_{1}, T_{1}), \dots,
    (S_{r}, T_{r})\}$ for some $r \geq 1$.
    We will now show that
    $\{S_{1}, \dots, S_{r}\}$, and
    $\{T_{1}, \dots, T_{r}\}$ are the 
    required partitions of $I$ and $J$ respectively.

    For $a \neq b \in [r]$, we first consider
    $(S_{a}, T_{a}), (S_{b}, T_{b}) \in \mathcal{M}$.
    By definition of $\mathcal{M} \subseteq \mathcal{C}$, we know that
    $\sum_{i \in S_{a}}x_{i} = \sum_{j \in T_{a}}y_{j}$,
    and $\sum_{i \in S_{b}}x_{i} = \sum_{j \in T_{b}}y_{j}$.
    Since $(\mathbf{x}, \mathbf{y})$ is a confluence,
    this implies that $\sum_{i \in S_{a} \cap S_{b}}x_{i} 
    = \sum_{j \in T_{a} \cap T_{b}}y_{j}$.
    But $(S_{a} \cap S_{b}) \subseteq S_{a}$,
    and $(T_{a} \cap T_{b}) \subseteq T_{a}$.
    Since $(S_{a}, T_{a}) \neq (S_{b}, T_{b})$,
    we may assume without loss of generality that
    $(S_{a} \cap S_{b}) \subsetneq S_{a}$.
    But then, since $x_{i} \in \mathbb{Z}_{> 0}$
    for all $i \in I$, we see that
    $$\sum_{j \in T_{a}}y_{j} = 
    \sum_{i \in S_{a}}x_{i} >
    \sum_{i \in S_{a} \cap S_{b}}x_{i}
    = \sum_{j \in T_{a} \cap T_{b}}y_{j}.$$
    Therefore, $(T_{a} \cap T_{b}) \subsetneq T_{a}$
    as well.
    But by our definition of $\mathcal{M}$, since
    $(S_{a} \cap S_{b}, T_{a} \cap T_{b}) \in \mathcal{C}$, and $(S_a, T_a) \in \mathcal{C}$,
    it must be the case that $S_{a} \cap S_{b}
    = T_{a} \cap T_{b} = \emptyset$.
    Therefore, $\{S_1, \ldots, S_r\}$ and $\{T_1, \ldots, T_r\}$
    are both pairwise disjoint. 

    Since $(S_{a}, T_{a}) \in \mathcal{M} \subseteq \mathcal{C}$
    for all $a \in [r]$, it is trivial to see that
    given any $P \subseteq [r]$,
    $$\sum_{a \in P}\sum_{i \in S_{a}}x_{i}
    = \sum_{a \in P}\sum_{j \in T_{a}}y_{j},
    \text{ and therefore, }
    \left(\bigcup_{a \in P}S_{a}, \bigcup_{a \in P}T_{a}\right)
    \in \mathcal{C}.$$
    Now, using induction, we  will show that
    given any $(S, T) \in \mathcal{C}$, we can find some
    $P \subseteq [r]$, such that
    $S = \sqcup_{a \in P}S_{a}$, and
    $T = \sqcup_{a \in P}T_{a}$.
    This is trivially true when $|S| + |T| = 0$.
    Let us now assume that it is true when
    $|S| + |T| < k$ for some $k > 0$.
    Now, we consider some $(S, T) \in \mathcal{C}$ such that
    $|S| + |T| = k$.
    If $(S, T) \in \mathcal{M}$, we are already done.
    Otherwise, by definition of $\mathcal{M}$, we know that
    there exist some $(S', T') \in \mathcal{C}$ such that
    $S' \subsetneq S$, and $T' \subsetneq T$,
    but $S' \neq \emptyset$, and $T' \neq \emptyset$.
    But we note that
    $$\sum_{i \in S \setminus S'}x_{i}
    = \sum_{i \in S}x_{i} - \sum_{i \in S'}x_{i}
    = \sum_{j \in T}y_{j} - \sum_{j \in T'}y_{j}
    = \sum_{j \in T \setminus T'}y_{j},$$
    which implies that
    $(S \setminus S', T \setminus T') \in \mathcal{C}$.
    We note that $0 < |S'| + |T'| < k$.
    Therefore, $|S \setminus S'| + |T \setminus T'|
    = (|S| + |T|) - (|S'| + |T'|) < k$.
    Therefore, our induction hypothesis implies that
    there exist $P_{1}, P_{2} \subseteq [r]$,
    such that $S' = \cup_{a \in P_{1}}S_{a}$,
    $T' = \cup_{a \in P_{1}}T_{a}$,
    $(S \setminus S') = \cup_{b \in P_{2}}S_{b}$,
    and $(T \setminus T') = \cup_{b \in P_{2}}T_{b}$.
    Therefore, there exists
    $P = P_{1} \cup P_{2} \subseteq [r]$
    such that
    $$S = \bigcup_{a \in P}S_{a}, \text{ and }
    T = \bigcup_{a \in P}T_{a}.$$
    This completes the induction.
    In particular, since $(\mathbf{x}, \mathbf{y})$
    is a confluence, we know that
    $(I, J) \in \mathcal{C}$.
    This then means that there exists some
    $P \subseteq [r]$ such that
    $I = \sqcup_{a \in P}S_{a}$, and
    $J = \sqcup_{a \in P}T_{a}$.
    Recall that $S_{a} \cap S_{b}  = \emptyset$
    for $a \neq b$.
    Each $(S_a, T_a) \in \mathcal{M}$ satisfies $S_a \ne \emptyset$.
  If there is some $a' \not \in P$ then $S_{a'} = S_{a'} \cap I = S_{a'} \cap \left[ \bigcup_{a \in P}S_{a} \right] = \emptyset$, a contradiction.
    This implies that $P = [r]$.
    This proves that $I = S_{1} \sqcup \cdots \sqcup S_{r}$,
    and $J = T_{1} \sqcup \cdots \sqcup T_{r}$.
    This finishes the proof of the existence of the paired partition $\mathcal{M} =\{(S_{1}, T_{1}), \dots,
    (S_{r}, T_{r})\}$.

    For uniqueness, let $\mathcal{M}'= \{(S'_{1}, T'_{1}), \dots,
    (S'_{\rho}, T'_{\rho})\}$ be another such paired partition. We first 
    note that for $\mathcal{M}$,
    if $a \ne b \in [r]$ then $\sum_{i \in S_a} x_i \ne \sum_{i \in S_b} x_i$.
    Otherwise, by the confluence property,
    $\sum_{j \in T_a} y_j = \sum_{i \in S_a} x_i = \sum_{i \in S_b} x_i$
    implies $0 < \sum_{j \in T_a} y_j =  \sum_{j \in T_a \cap T_a} y_j = \sum_{i \in S_a \cap S_b} x_i
    =0$, a contradiction. The same is true for $\mathcal{M}'$.

    For any $(S_a, T_a) \in \mathcal{M}$, as $\sum_{i \in S_a} x_i
    = \sum_{j \in T_a} y_j$, there exists $P' \subseteq [\rho]$ such that
    $S_a = \cup_{c \in P'} S'_c$. As $S_a \ne \emptyset$, we have $P' \ne \emptyset$.
    Pick any $c \in P'$. Then by the same argument there exists $P \subseteq [r]$
    such that $S'_c =  \cup_{b \in P} S_b$. As $\{S_a: a \in [r]\}$ are pairwise disjoint and nonempty,
    $P'$ must be a singleton set $\{c\}$ and  then in turn $P$ for this $c$ is also a singleton set
    $\{b\}$.
    This sets up a mapping sending $a \in [r]$ to $c \in [\rho]$.
    Since $\sum_{i \in S_a} x_i = \sum_{i \in S'_c} x_i = \sum_{i \in S_b} x_i $,
    we must have $a=b$. Hence the mapping sending $a$ to $c$ is 1-1.
    Switching the role of $\mathcal{M}$ and $\mathcal{M}'$ shows that $r = \rho$ and
    the two paired partitions are in 1-1 correspondence.  This proves uniqueness.
\end{proof}

Now that we have defined confluences and understood their
properties better, we will now relate them to
$\mathbf{x} \in \chi_{q}$, and see how they are
relevant to our problem at hand.

\begin{definition}\label{definition: confluent}
    Let $\mathbf{0} \neq \mathbf{x}   \in \chi_{q}$.
    We will define $I^{+} = \{i \in [q]: x_{i} > 0\}$,
    and $I^{-} = \{i \in [q]: x_{i} < 0\}$.
    Then, we let $\mathbf{x}^{+} = (x_{i})_{i \in I^{+}}$,
    and $\mathbf{x}^{-} = (-x_{i})_{i \in I^{-}}$.
    We say that $\mathbf{x}$ is confluent
    if $(\mathbf{x}^{+}, \mathbf{x}^{-})$
    is a confluence.

    If $\mathbf{x}$ is confluent,
    let $I^{+} = S_{1} \sqcup \cdots \sqcup S_{r}$
    and $I^{-} = T_{1} \sqcup \cdots \sqcup T_{r}$
    be the partition guaranteed by
    \cref{lemma: confluenceEquivalence}.
    We say $((S_{1}, T_{1}), \dots, (S_{r}, T_{r}))$
    is the confluence basis of $\mathbf{x}$. The confluence basis
    is uniquely defined, up to the order of the pairs.
\end{definition}

An example $\mathbf{x}$ that is \emph{non}-confluent is
$(1, 1, -1, -1) \in \chi_{4}$. This will play a pivotal role later in isolating
tensor products.
Before we can study confluent $\mathbf{x}$,
we will need a few more definitions.

\begin{definition}\label{definition: 01Poly}
    We say that a polynomial $f:
    \mathbb{R}^{d} \rightarrow \mathbb{R}$ is
    a 0-1 polynomial with $q$ terms
    if it can be expressed as
    $$f(p_{1}, \dots, p_{d}) =
    \sum_{i \in [q]}\prod_{t \in [d]}p_{t}^{x_{i, t}},$$
    where each monomial $\prod_{t \in [d]}p_{t}^{x_{i, t}}$ 
    is distinct, and has a coefficient $1$.
\end{definition}

\begin{definition}\label{definition: proportionalPoly}
    Two polynomials $f, g: \mathbb{R}^{d} \rightarrow \mathbb{R}$
    are said to be proportional to each other, if
    there exist integers $x_{1}, \dots, x_{d} \geq 0$
    and $y_{1}, \dots, y_{d} \geq 0$ such that for all
    $\mathbf{p} \in \mathbb{R}^{d}$,
    $$p_{1}^{x_{1}} \cdots p_{d}^{x_{d}} \cdot f(\mathbf{p})
    = p_{1}^{y_{1}} \cdots p_{d}^{y_{d}} \cdot g(\mathbf{p}).$$
\end{definition}

We will first focus on 0-1 polynomials $f: \mathbb{R}
\rightarrow \mathbb{R}$ in one variable with $q$ terms.
Such polynomials must be of the form
$$f(p) = p^{x_{1}} + p^{x_{2}} + \cdots + p^{x_{q}},$$
where $0 \leq x_{1} < x_{2} < \cdots < x_{q}$.

\begin{lemma}\label{lemma: confluent01PolyLemmaSimple}
    Let $\mathbf{0}  \neq \mathbf{x} \in \chi_{q}$
    be confluent, with the confluence basis
    $((S_{1}, T_{1}), \dots, (S_{r}, T_{r}))$.
    Let $f_{1}, \dots, f_{q}: \mathbb{R} \rightarrow \mathbb{R}$
    be {\rm 0-1} polynomials with $q$ terms such that
    $\phi_{\mathbf{x}}(f_{1}(p), \dots, f_{q}(p)) = 0$
    for all $p \in \mathbb{R}$.
    Then, given any $i, j \in [q]$, such that
    $i, j \in S_{a} \cup T_{a}$ for some $a \in [r]$,
    we have that $f_{i}$ and $f_{j}$ are
    proportional to each other.
\end{lemma}
\begin{proof}
    Since $f_{1}, \dots, f_{q}$ are all 0-1
    polynomials with $q$ terms, we may assume that
    $$f_{i}(p) = p^{z_{i, 1}} + \cdots +p^{z_{i, q}}$$
    for integers $0 \leq z_{i, 1} < \cdots < z_{i, q}$
    for all $i \in [q]$.
    We will now let $d_{i} = z_{i, 1}$, and $y_{i, j} = 
    z_{i, j} - d_{i}$ for all $i, j \in [q]$.
    If we now let
    $$g_{i}(p) = p^{y_{i, 1}} + \cdots + p^{y_{i, q}},$$
    we see that $f_{i}(p) = p^{d_{i}}g_{i}(p)$
    for all $i \in [q]$.
    By construction, we also see that $y_{i, 1} = 0$
    for all $i \in [q]$.
    Therefore, for any $i \neq j \in [q]$, we have
    $g_{i} = g_{j}$ if and only if $f_{i}$ and $f_{j}$ are 
    proportional to each other.
    
    Since $\phi_{\mathbf{x}}(f_{1}(p), \dots, f_{q}(p)) = 0$,
    we see that
    $$\prod_{i \in I^{+}}p^{x_{i} \cdot d_{i}}
    \cdot g_{i}(p)^{x_{i}}
    = \prod_{i \in I^{+}}f_{i}(p)^{x_{i}} 
    = \prod_{i \in I^{-}}f_{i}(p)^{-x_{i}}
    = \prod_{i \in I^{-}}p^{-x_{i} \cdot d_{i}}
    \cdot g_{i}(p)^{-x_{i}}.$$
    By our construction of $g_{i}$, we see that
    $p \nmid g_{i}$ for all $i \in [q]$.
    Therefore,
    $$\prod_{i \in I^{+}}p^{x_{i} \cdot d_{i}}
    = \prod_{i \in I^{-}}p^{-x_{i} \cdot d_{i}}, \text{ and }
    \prod_{i \in I^{+}}g_{i}(p)^{x_{i}}
    = \prod_{i \in I^{-}}g_{i}(p)^{-x_{i}}.$$

    For any $i \in [q]$, if  $x_{i} \neq 0$
    then either $i \in I^{+}$ or $i \in I^{-}$.
    For any  $i \in I^{+} \cup I^{-}$, there exists a unique $a \in [r]$, and we will denote it by $\sigma(i)$, 
    such that $i \in S_{a}$ or $i \in T_{a}$, i.e., $i \in S_{\sigma(i)} \cup T_{\sigma(i)}$.

    We will now prove by induction that given any
    $i, j \in S_{a} \cup T_{a}$ for some $a \in [r]$, i.e., $\sigma(i)  = \sigma(j)$,
    we have  $$g_{i} = g_{j}.$$
    For all $i \in [q]$,
    given any $1 \leq t \leq q$, we define a truncated version of $g_i$,
    $$g_{i}|_{t}(p) = p^{y_{i, 1}} + \cdots + p^{y_{i, t}}.$$
    By the construction of $g_{i}$,
    we know that $y_{i, 1} = 0$ for all $i \in [q]$.
    So, $g_{i}|_{1}(p) = p^{y_{i, 1}} = 1$ for all $i \in [q]$.

    Now, for each $r \leq k \leq rq$, we  define the following:
    \begin{align*}
        \mathbf{Statement}\, (k):\ & \text{There exist integers
        $1 \leq t_{1}, \dots, t_{r} \leq q$ such that
        $t_{1} + \cdots + t_{r} = k$, and}\\
        &\text{some polynomials
        $h_{1}, \dots, h_{r}$ such that $g_{i}|_{t_{\sigma(i)}}
        = h_{\sigma(i)}$ for all $i \in  I^{+} \cup I^{-}$.}
    \end{align*}
    
    We will momentarily assume the following claim,
    which we shall prove shortly.

    \begin{claim}\label{claim: StatementKImpliesK'}
        If $\mathbf{Statement}\, (k)$
        is true for some $r \le k < rq$, then there exists some
        $k < k' \leq rq$, such that
        $\mathbf{Statement}\, (k')$ is true.
    \end{claim}

    We have already seen that when
    $t_{1} = \cdots = t_{r} = 1$,
    and $h_{1} = \cdots = h_{r} = 1$,
    $g_{i}|_{t_{\sigma(i)}} = 1 = h_{\sigma(i)}$
    for all $i \in [q]$.
    Therefore, $\mathbf{Statement}\, (r)$ is true.
    So, \cref{claim: StatementKImpliesK'} implies that
    $\mathbf{Statement}\, (rq)$ is true.
    In this case with $k=rq$, based on the bound on the $t_a$'s, all $t_a = q$.
    Then $g_i = g_i|_q = g_{i}|_{t_{\sigma(i)}} =  h_{\sigma(i)}$. Hence if $\sigma(i) = \sigma(j)$
    then $g_i = g_j$.
    From the construction of $g_{i}$, it now follows
    that $f_{i}$ and $f_{j}$ are proportional
    to each other, if $i, j \in S_{a} \cup T_{a}$
    for some $a \in [r]$.
\end{proof}

We will now finish the proof of \cref{lemma: confluent01PolyLemmaSimple} by proving
\cref{claim: StatementKImpliesK'}.

\begin{claimproof}{\cref{claim: StatementKImpliesK'}}
    Let us assume that there exist $1 \leq t_{1}, \dots, t_{r} \leq q$, 
    as well as polynomials $h_{1}, \dots, h_{r}$
    that satisfy $\mathbf{Statement}\, (k)$ for some 
    $r \le k < rq$.
    We will consider the following equation
    \begin{equation}\label{equation: confluent01PolyEqn}
        \prod_{i \in I^{+}}g_{i}(p)^{x_{i}}
        - \prod_{i \in I^{+}}
        \left(g_{i}|_{t_{\sigma(i)}}(p)\right)^{x_{i}}
        = \prod_{i \in I^{-}}g_{i}(p)^{-x_{i}}
        - \prod_{i \in I^{-}}
        \left(g_{i}|_{t_{\sigma(i)}}(p)\right)^{-x_{i}}.
    \end{equation}
    We can see that \cref{equation: confluent01PolyEqn}
    is true since we know that
    $$\prod_{i \in I^{+}}g_{i}(p)^{x_{i}}
    = \prod_{i \in I^{-}}g_{i}(p)^{-x_{i}},$$
    and the induction hypothesis tells us that since
    $\mathbf{x}$ is confluent,
    $$\prod_{i \in I^{+}}
    \left(g_{i}|_{t_{\sigma(i)}}(p)\right)^{x_{i}}
    = \prod_{a \in [r]}h_{a}^{\sum_{i \in S_{a}}x_{i}}
    = \prod_{a \in [r]}h_{a}^{\sum_{i \in T_{a}}(-x_{i})}
    = \prod_{i \in I^{-}}
    \left(g_{i}|_{t_{\sigma(i)}}(p)\right)^{-x_{i}}.$$
    
    Since we have assumed that $t_{1} + \cdots + t_{r}
    = k < rq$, we see that there exists some $a \in [r]$
    such that $t_{a} < q$.
    Since $\sigma : I^{+} \cup I^{-} \rightarrow [r]$ is onto, \cref{equation: confluent01PolyEqn}
    is not the trivial $0 = 0$ equation.
    Now, we will consider the least degree term
    of the LHS and RHS of \cref{equation: confluent01PolyEqn}.
    If we focus on the LHS, we notice that
    $$\prod_{i \in I^{+}}g_{i}(p)^{x_{i}} 
    = \sum_{k_{i, j}: \ (\forall i \in I^{+},  \forall j \in [x_i])  [ 1 \le k_{i, j} \le q]  }
    p^{\sum_{i \in I^{+}}\left( y_{i, k_{i, 1}} + \cdots + y_{i, k_{i, x_i}}\right)}.
    $$
    However, any term $p^{\sum_{i \in I^{+}}\left( y_{i, k_{i, 1}} + \cdots + y_{i, k_{i, x_i}}\right)}$
    for which $k_{i, j} \leq t_{\sigma(i)}$ for all $i \in I^{+}$ and $j \in [x_i]$
    would appear in 
    $$\prod_{i \in I^{+}}
    \left(g_{i}|_{t_{\sigma(i)}}(p)\right)^{x_{i}}$$
    as well, and get cancelled  in the LHS
    of \cref{equation: confluent01PolyEqn}.
    Therefore, to be a candidate for the least degree term remaining,
    it must be obtained by choosing one least degree remaining term $p^{y_{i, t_{\sigma(i)}+1}}$
    from one factor $g_{i}(p)$ and the term $p^{y_{j, 1}} = p^0 =1$ for all other factor polynomials.
    Furthermore, all such product terms from  $\prod_{i \in I^{+}}g_{i}(p)^{x_{i}}$ are not cancelled
    from $\prod_{i \in I^{+}}
    \left(g_{i}|_{t_{\sigma(i)}}(p)\right)^{x_{i}}$ in \cref{equation: confluent01PolyEqn},
    since all terms from this second product have already been used
    to cancel the corresponding terms from the first product.
    Hence, this least degree  term has  degree
    $$d^{+} 
    = \min_{i \in I^{+}}\left(y_{i, t_{\sigma(i)}+1}\right),$$
    and this occurs exactly $x_i$ times for each minimizer $i$ that achieves this minimum.
    Summing up, the least degree term of the LHS of
    \cref{equation: confluent01PolyEqn} is therefore
    $$\sum_{\substack{i \in I^{+}:\\ y_{i, t_{\sigma(i)}+1} = d^{+}}}
    x_{i} \cdot p^{d^{+}}.$$
    Similarly, if we let
    $$d^{-} = \min_{i \in I^{-}}\left(y_{i, t_{\sigma(i)}+1}\right),$$
    we find that the least degree term of the RHS of
    \cref{equation: confluent01PolyEqn} is
    $$\sum_{\substack{i \in I^{-}:\\ y_{i, t_{\sigma(i)}+1} = d^{-}}}
    (-x_{i}) \cdot p^{d^{-}}.$$
    Since the least degree terms of the LHS and the RHS
    have to be identical, this implies that
    $d^{+} = d^{-}$, and also that
    $\sum_{i \in S}x_{i} = \sum_{i \in T}(-x_{i})$,
    where $S = \{i \in I^{+}: y_{i, t_{\sigma(i)}+1} = d^{+}\}$,
    and $T = \{i \in I^{-}: y_{i, t_{\sigma(i)}+1} = d^{-}\}$.
    Since $\mathbf{x}$ is confluent,
    by definition, this means that there exists
    some $P \subseteq [r]$ such that
    $S = \sqcup_{a \in P}S_{a}$, and $T = \sqcup_{a \in P}T_{a}$.
    But this means that $y_{i, t_{\sigma(i)}+1} = d^{+} = d^{-}$
    for all $i \in \cup_{a \in P}(S_{a} \cup T_{a})$.
    
    In other words,
    we have shown that
    there exist integers
    $1 \leq t_{1}', \dots, t_{r}' \leq q$
    such that $t_{a}' = t_{a} + 1$ for
    $a \in P$, and $t_{a}' = t_{a}$ for $a \notin P$,
    as well as polynomials
    $h_{1}', \dots, h_{r}'$ such that
    $g_{i}|_{t'_{\sigma(i)}} = h'_{\sigma(i)}$
    for all $i\in  I^{+} \cup I^{-}$.
    Since $k' = t_{1}' + \cdots + t_{r}' >
    t_{1} + \cdots + t_{r} = k$,
    this finishes the proof of
    \cref{claim: StatementKImpliesK'}.
\end{claimproof}

We will now extend \cref{lemma: confluent01PolyLemmaSimple}
to be applicable for all 0-1 polynomials
$f_{i}: \mathbb{R}^{d} \rightarrow \mathbb{R}$.

\begin{lemma}\label{lemma: confluent01PolyLemma}
    Let $\mathbf{0}  \neq \mathbf{x} \in \chi_{q}$
    be confluent, with the confluence basis
    $((S_{1}, T_{1}), \dots, (S_{r}, T_{r}))$.
    Let $f_{1}, \dots, f_{q}: \mathbb{R}^{d} \rightarrow \mathbb{R}$
    be 0-1 polynomials with $q$ terms such that
    $\phi_{\mathbf{x}}(f_{1}(\mathbf{p}), \dots,
    f_{q}(\mathbf{p})) = 0$
    for all $\mathbf{p} \in \mathbb{R}^{d}$.
    Then, given any $i, j \in [q]$, such that
    $i, j \in S_{a} \cup T_{a}$ for some $a \in [r]$,
    we have that $f_{i}$ and $f_{j}$ are proportional
    to each other.
\end{lemma}
\begin{proof}
    From \cref{definition: 01Poly}, we may assume that
    $$f_{j}(\mathbf{p}) = 
    \sum_{i \in [q]} \prod_{t \in [d]}p_{t}^{x_{i, t, j}}$$
    for all $j \in [q]$.
    We let $\mathfrak{m} =
    \max_{i \in [q], t \in [d], j \in [q]}(x_{i, j, t})$.
    For all $j \in [q]$, we know from \cref{definition: 01Poly}
    that the following sequences are pairwise distinct
    $$(x_{1, 1, j}, x_{1, 2, j}, \dots, x_{1, d, j})
    \neq (x_{2, 1, j}, x_{2, 2, j}, \dots, x_{2, d, j})
    \neq \cdots \neq 
    (x_{q, 1, j}, x_{q, 2, j}, \dots, x_{q, d, j}).$$
    Therefore for any $m > 2\mathfrak{m}$, if we let
    $$z_{i, j} = \sum_{t \in [d]}m^{t} \cdot x_{i, t, j},$$
    we see that for all $j \in [q]$ they are pairwise distinct,
    $z_{1, j} \neq \cdots \neq z_{q, j}$.
    Moreover, 
    $$f_{j}(p^{m}, p^{m^{2}}, \dots, p^{m^{d}}) = 
    \sum_{i \in [q]}p^{z_{i, j}}$$
    for all $j \in [q]$.
    We can now define $g_{j}: \mathbb{R} \rightarrow \mathbb{R}$
    as
    $$g_{j}(p) = f_{j}(p^{m}, p^{m^{2}}, \dots, p^{m^{d}}).$$
    By our choice of $m > 2\mathfrak{m}$,
    we see that each $g_{j}$ is a
    0-1 polynomial with (exactly) $q$ terms, i.e., no terms get combined.

    By hypothesis,
    $$\phi_{\mathbf{x}}(f_{1}(\mathbf{p}), \dots, f_{q}(\mathbf{p}))
    = 0$$
    for all $\mathbf{p} \in \mathbb{R}^{d}$.
    Therefore,
    $$\phi_{\mathbf{x}}(g_{1}(p), \dots, g_{q}(p)) = 0$$
    for all $p \in \mathbb{R}$.
    So, \cref{lemma: confluent01PolyLemmaSimple} tells us that
    given any $i, j \in [q]$ such that $i, j \in S_{a} \cup T_{a}$
    for some $a \in [r]$, we have that
    $g_{i}$ and $g_{j}$ are proportional to each other.
    But from our construction of $g_{i}$ and $g_{j}$,
    this means that there exists some $d_{i, j} \in \mathbb{Z}$
    such that
    $$\frac{p^{z_{1, i}} + \cdots + p^{z_{q, i}}}
    {p^{z_{1, j}} + \cdots + p^{z_{q, j}}} = 
    \frac{g_{i}(p)}{g_{j}(p)} = 
    p^{d_{i, j}}.$$
    Therefore, there exists some $\sigma \in S_{q}$
    such that
    $$p^{z_{1i} - z_{\sigma(1)j}} = 
    p^{z_{2i} - z_{\sigma(2)j}} =
    \cdots = p^{z_{qi} - z_{\sigma(q)j}}
    = p^{d_{i, j}}.$$
    From the construction of $z_{1i}, \dots, z_{qi}$, and
    $z_{1j}, \dots, z_{qj}$, this implies that
    $$\sum_{t \in [d]}m^{t} \cdot 
    (x_{1, t, i} - x_{\sigma(1), t, j})
    = \sum_{t \in [d]}m^{t} \cdot 
    (x_{2, t, i} - x_{\sigma(2), t, j})
    = \cdots = \sum_{t \in [d]}m^{t} \cdot 
    (x_{q, t, i} - x_{\sigma(q), t, j}).$$
    We note that $-\mathfrak{m} \leq 
    (x_{1ti} - x_{\sigma(1)tj}), \dots,
    (x_{qti} - x_{\sigma(q)tj}) \leq \mathfrak{m}$
    for all $t \in [d]$.
    So, our choice of $m > 2\mathfrak{m}$ implies that
    for all $t \in [d]$ there exists some $r_{t} \in \mathbb{Z}$
    such that,
    $$(x_{1, t, i} - x_{\sigma(1), t, j}) = \cdots =
    (x_{q, t, i} - x_{\sigma(q), t, j}) = r_{t}.$$
    Therefore,
    $$\prod_{t \in [d]} p_{t}^{x_{1, t, i}} = 
    \prod_{t \in [d]} p_{t}^{r_{t}} \cdot
    p_{t}^{x_{\sigma(1), t, j}}, ~~~\dots, ~~~
    \prod_{t \in [d]} p_{t}^{x_{q, t, i}} = 
    \prod_{t \in [d]} p_{t}^{r_{t}} \cdot
    p_{t}^{x_{\sigma(q), t, j}}.$$
    So,
    $$f_{i}(\mathbf{p}) = 
    \prod_{t \in [d]}p_{t}^{x_{1, t, i}} + \cdots + 
    \prod_{t \in [d]}p_{t}^{x_{q, t, i}}
    = \prod_{t \in [d]} p_{t}^{r_{t}} \cdot \left(
    \prod_{t \in [d]}p_{t}^{x_{1, t, j}} + \cdots + 
    \prod_{t \in [d]}p_{t}^{x_{q, t, j}}\right)
    = \prod_{t \in [d]} p_{t}^{r_{t}} \cdot f_{j}(\mathbf{p}).$$
    This proves that $f_{i}$ and $f_{j}$ are 
    proportional to each other.
\end{proof}

We are now finally ready to study how confluent
$\mathbf{x} \in \chi_{q}$ interact with
matrices $M \in \text{Sym}_{q}^{\tt{F}}(\mathbb{R}_{\neq 0})$.
We will initially restrict ourselves to a
very special family of matrices.

\begin{definition}\label{definition: fullMatrix}
    Let $M \in \text{Sym}_{q}(\mathbb{R})$.
    We say that $M$ is a row full matrix if the elements of each row of $M$ are pairwise distinct:
    $$\forall i \in [q],\ 
    M_{i1} \neq M_{i2} \neq \cdots \neq M_{iq}.$$
\end{definition}

\begin{definition}\label{definition: pairwiseIndepMatrix}
    Let $M \in \text{Sym}_{q}(\mathbb{R})$.
    We say that the rows $i \neq j$ of $M$
    are order dependent 
    if there exists some  $\sigma \in S_{q}$ such that 
    $(M_{i1}, M_{i2}, \dots, M_{iq})$ and 
$(M_{j\sigma(1)}, M_{j\sigma(2)}, \dots, M_{j\sigma(q)})$ are linearly dependent.
    We say that $M \in \text{Sym}_{q}(\mathbb{R})$ is
    pairwise order independent (or p.o.\,independent) if 
    no rows $i \neq j$ of $M$ are order dependent.
\end{definition}

\begin{lemma}\label{lemma: fullPairwiseIndepTest}
    There exist $\text{Sym}_{q}(\mathbb{R})$-polynomials
    $\rho_{\tt{full}}$ and $\rho_{\tt{indep}}$ such that given any
    $M \in \text{Sym}_{q}(\mathbb{R})$,
    $\rho_{\tt{full}}(M) \neq 0$ if and only if
    $M$ is row full, and
    $\rho_{\tt{indep}}(M) \neq 0$ if and only if $M$ is
    p.o.\,independent.
\end{lemma}
\begin{proof}
    For all $i \in [q]$, we can define the
    $\text{Sym}_{q}(\mathbb{R})$-polynomial
    $\rho_{{\tt{full}}, i}(M) = \prod_{j_{1} \neq j_{2} \in [q]}
    (M_{ij_{1}} - M_{ij_{2}}).$
    We then define $\rho_{\tt{full}}$ such that
    $$\rho_{\tt{full}}(M) =
    \prod_{i \in [q]}\rho_{{\tt{full}}, i}(M).$$
    We can easily see that for any $M \in \text{Sym}_{q}(\mathbb{R})$,
    $\rho_{\tt{full}}(M) \neq 0$
    if and only if
    $M$ is a row full matrix.
    We can also see that $\rho_{\tt{full}}(M)$ is a homogeneous
    polynomial in the entries of $M$. So, it is a
    $\text{Sym}_{q}(\mathbb{R})$-polynomial.
    
    We know that given any two vectors
    $\mathbf{a}, \mathbf{b} \in \mathbb{R}^{q}$, 
    the Gram determinant $\mathfrak{g}: \mathbb{R}^{n} \times 
    \mathbb{R}^{n} \rightarrow \mathbb{R}$:
    \[\mathfrak{g}(\mathbf{a}, \mathbf{b})  = \det \begin{bmatrix}
\langle \mathbf{a}, \mathbf{a} \rangle  & \langle \mathbf{a}, \mathbf{b} \rangle \\ 
\langle \mathbf{b}, \mathbf{a} \rangle  & \langle \mathbf{b}, \mathbf{b} \rangle 
\end{bmatrix}
    \]
    is a homogeneous polynomial, where $\langle \cdot, \cdot \rangle $ denotes inner product, such that
    $\mathfrak{g}(\mathbf{a}, \mathbf{b}) \not = 0$
    if and only if $\mathbf{a}$ and $\mathbf{b}$
    are linearly independent.
    We will now define the function
    $\rho_{{\tt{indep}}, ij}$ for $i \neq j \in [q]$
    as
    \begin{equation}\label{ij-indep}
    \rho_{{\tt{indep}}, ij}(M) = 
    \prod_{\sigma \in S_{q}}\mathfrak{g}\Big(
    (M_{i1}, \dots, M_{iq}),
    (M_{j\sigma(1)}, \dots, M_{j\sigma(q)})\Big),
      \end{equation}
    and the function $\rho_{\tt{indep}}$ as
     \begin{equation}\label{po-indep}
    \rho_{{\tt{indep}}}(M) =
    \prod_{i \neq j \in [q]} \rho_{{\tt{indep}}, ij}(M).
      \end{equation}

    From the construction, we can easily see that
    $\rho_{\tt{indep}}(M) \neq 0$ if and only if
    $M$ is p.o.\,independent.
    We can also see that $\rho_{\tt{indep}}(M)$ is a homogeneous
    polynomial in the entries of $M$. So, it is a
    $\text{Sym}_{q}(\mathbb{R})$-polynomial.
\end{proof}

We will now study how confluent $\mathbf{x} \in \chi_{q}$
interact with row full, p.o.\,independent matrices in
$\text{Sym}_{q}^{\tt{pd}}(\mathbb{R}_{> 0})$.

\begin{theorem}\label{theorem: confluentReduction}
    Let $M \in \text{Sym}_{q}^{\tt{pd}}(\mathbb{R}_{> 0})$
    be a row full, p.o.\,independent matrix.
    Let $\mathcal{F}$ be a countable set
    of $\text{Sym}_{q}(\mathbb{R})$-polynomials such that
    $F(M) \neq 0$ for all $F \in \mathcal{F}$.
    Given any confluent $\mathbf{0} \neq \mathbf{x} \in \chi_{q}$,
    there exists some row full, p.o.\,independent
    $N \in \mathfrak{R}(M, \mathcal{F} \cup \{\Psi_{\mathbf{x}}\}) \cap
    \text{Sym}_{q}^{\tt{pd}}(\mathbb{R}_{> 0})$.
\end{theorem}
\begin{proof}
    We may assume that the entries of $M$ are generated
    by $\{g_{t}\}_{t \in [d]}$.
    We recall that we may replace $M$ with some
    $c M$ as in \cref{lemma: MequivalentCM}.
    So, we know that for $i, j \in [q]$, there exist
    unique integers $e_{ij0} \in \{0, 1\}$,
    and $e_{ij1}, \dots, e_{ijd} \geq 0$, such that
    $$M_{ij} = (-1)^{e_{ij0}} \cdot g_{1}^{e_{ij1}} \cdots g_{d}^{e_{ijd}}.$$
    Since, $M_{ij} > 0$ for all $i, j \in [q]$ by our choice
    of $M$, we see that $e_{ij0} = 0$ for all $i, j \in [q]$.
    So, from \cref{definition: mathcalT}, we see that
    $$\mathcal{T}_{M}(\mathbf{p})_{ij} =
    p_{1}^{e_{ij1}} \cdots p_{d}^{e_{ijd}}$$
    is a polynomial in $\mathbf{p}$ for all $i, j \in [q]$.
    So, we see that for all $i \in [q]$,
    $$\left(\mathcal{T}_{M}(\mathbf{p})^{2}\right)_{ii}
    = \sum_{j \in [q]}
    \left(\mathcal{T}_{M}(\mathbf{p})_{ij} \right)^{2}
    = \sum_{j \in [q]}\prod_{t \in [d]}p_{t}^{2e_{ijt}}.$$

    Since $M$ is a row full matrix, we know that for all $i \in [q]$,
    $M_{i1} \neq M_{i2} \neq \cdots \neq M_{iq}$
    are pairwise distinct.
    Since $\{g_{t}\}_{t \in [d]}$ is a generating set,
    this implies that for each $i \in [q]$,
    $(e_{i11}, e_{i12}, \dots, e_{i1d}) \neq 
    (e_{i21}, e_{i22}, \dots, e_{i2d}) \neq \cdots \neq
    (e_{iq1}, e_{iq2}, \dots, e_{iqd})$
    are also pairwise distinct.
    Therefore, for each $i \in [q]$, we see that
    $\left(\mathcal{T}_{M}(\mathbf{p})^{2}\right)_{ii}$
    is a 0-1 polynomial with exactly $q$ terms.
    Since $\mathbf{x}$ is confluent, we know
    that it has a confluence basis $((S_{1}, T_{1}),
    \dots, (S_{r}, T_{r}))$.
    Since $r \geq 1$, and $S_{1}, T_{1} \neq \emptyset$,
    we may pick some $i \in S_{1}$, and $j \in T_{1}$.
    We will now assume that
    $\phi_{\mathbf{x}}\left((\mathcal{T}_{M}(\mathbf{p})^{2})_{11},
    \dots, (\mathcal{T}_{M}(\mathbf{p})^{2})_{qq}\right) = 0$
    for all $\mathbf{p} \in \mathbb{R}^{d}$.
    \cref{lemma: confluent01PolyLemma} now implies
    that by our choice of $i, j \in [q]$,
    $(\mathcal{T}_{M}(\mathbf{p})^{2})_{ii}$ and
    $(\mathcal{T}_{M}(\mathbf{p})^{2})_{jj}$ are
    proportional to each other.
    But this means that there exists some
    $c_{1}, \dots, c_{d} \in \mathbb{Z}$ such that
    $$\frac{\prod_{t \in [d]} p_{t}^{2e_{i1t}} + \cdots 
    + \prod_{t \in [d]} p_{t}^{2e_{iqt}}}
    {\prod_{t \in [d]} p_{t}^{2e_{j1t}} + \cdots 
    + \prod_{t \in [d]} p_{t}^{2e_{jqt}}}
    = \prod_{t \in [d]}p_{t}^{c_{t}}.$$
    Therefore, there exists some $\sigma \in S_{q}$
    such that
    $$\prod_{t \in [d]}p_{t}^{2e_{i1t}}
    = \prod_{t \in [d]}p_{t}^{c_{t}} \cdot
    \prod_{t \in [d]}p_{t}^{2e_{j\sigma(1)t}}, ~~~\dots,~~~
    \prod_{t \in [d]}p_{t}^{2e_{iqt}}
    = \prod_{t \in [d]}p_{t}^{c_{t}} \cdot
    \prod_{t \in [d]}p_{t}^{2e_{j\sigma(q)t}}$$
    as an identity in $\mathbf{p}$.
    In particular, when we evaluate these
    expressions at $(g_{1}, \dots, g_{d})$,
    we see that
    there exists a constant $c =
    \prod_{t \in [d]}g_{t}^{c_{t}} \in \mathbb{R}$
    such that
    $$\Big((M_{i1})^{2}, \dots, (M_{iq})^{2} \Big) = c \cdot
    \Big((M_{j\sigma(1)})^{2}, \dots, (M_{j\sigma(q)})^{2} \Big).$$
    Since $M_{ij} > 0$ for all $i, j \in [q]$ by
    our choice of $M$, this implies that in fact
    $$\Big(M_{i1}, \dots, M_{iq} \Big) = \sqrt{c} \cdot
    \Big(M_{j\sigma(1)}, \dots, M_{j\sigma(q)} \Big).$$
    This contradicts our assumption that $M$ is
    p.o.\,independent, since $i \ne j$.
    Therefore, our assumption that
    $\phi_{\mathbf{x}}\left(\mathcal{T}_{M}(\mathbf{p})^{2})_{11},
    \dots, (\mathcal{T}_{M}(\mathbf{p})^{2})_{qq}\right) = 0$
    must be false.

    Since our only restriction on $\mathbf{x}$
    was that it be confluent,
    $\phi_{\mathbf{y}}((\mathcal{T}_{M}(\mathbf{p})^{2})_{11},
    \dots, (\mathcal{T}_{M}(\mathbf{p})^{2})_{qq})$
    must be a non-zero polynomial in $\mathbf{p}$
    for all confluent $\mathbf{y} \in \chi_{q}$.
    In particular, given any $\sigma \in S_{q}$,
    we can construct $\mathbf{y} \in \chi_{q}$ such that
    $y_{i} = x_{\sigma(i)}$ for all $i \in [q]$.
    Since $\mathbf{x}$ is confluent, it is trivial
    to see that $\mathbf{y}$ is also confluent.
    Therefore, it follows from the construction of
    $\Phi_{\mathbf{x}}$ that
    $$\Phi_{\mathbf{x}}((\mathcal{T}_{M}(\mathbf{p})^{2})_{11},
    \dots, (\mathcal{T}_{M}(\mathbf{p})^{2})_{qq}) = 
    \prod_{\substack{\mathbf{y} \in \chi_{q}:
    \exists\ \sigma \in S_{q}\\
    \forall\ i \in [q], y_{i} = x_{\sigma(i)}}}
    \phi_{\mathbf{y}}((\mathcal{T}_{M}(\mathbf{p})^{2})_{11},
    \dots, (\mathcal{T}_{M}(\mathbf{p})^{2})_{qq})$$
    is a product of non-zero polynomials, and is therefore,
    a non-zero polynomial.
    So, there exists some $\mathbf{p}_{\mathbf{x}}
    \in \mathbb{R}^{d}$ such that
    $\Phi_{\mathbf{x}}(
    (\mathcal{T}_{M}(\mathbf{p}_{\mathbf{x}})^{2})_{11}, \dots, 
    (\mathcal{T}_{M}(\mathbf{p}_{\mathbf{x}})^{2})_{qq}) \neq 0$.

    We will now define $\zeta: N \mapsto \Phi_{\mathbf{x}}
    ((N^{2})_{11}, \dots, (N^{2})_{qq})$
    for all $N \in \text{Sym}_{q}(\mathbb{R})$.
    We can see that since the entries of the matrix $N^{2}$
    are homogeneous polynomials in the entries of the matrix $N$,
    $\zeta$ is in fact a $\text{Sym}_{q}(\mathbb{R})$-polynomial.
    From \cref{lemma: fullPairwiseIndepTest}, we also know that
    there exist $\text{Sym}_{q}(\mathbb{R})$-polynomials
    $\rho_{\tt{full}}$, and $\rho_{\tt{indep}}$.
    We will now let
    $\mathcal{F}' = \mathcal{F} \cup  \{\zeta, \rho_{\tt{full}},
    \rho_{\tt{indep}}\}$.
    We know that  for every $F \in \mathcal{F}$,
    $F(\mathcal{T}_{M}(g_{1}, \dots, g_{d})) = F(M) \neq 0$.
    We have also seen that
    $\zeta(\mathcal{T}_{M}(\mathbf{p}_{\mathbf{x}})) \neq 0$.
    Since $M$ is row full and p.o.\,independent,
    we also see that $\rho_{\tt{full}}(\mathcal{T}_{M}(g_{1}, \dots, g_{d}))
    = \rho_{\tt{full}}(M) \neq 0$, and    
    $\rho_{\tt{indep}}(\mathcal{T}_{M}(g_{1}, \dots, g_{d}))
    = \rho_{\tt{indep}}(M) \neq 0$.
    Therefore, \cref{lemma: thickeningWorks}
    tells us that there exists some
    $M' \in \mathfrak{R}
    (M, \mathcal{F}') \cap \text{Sym}_{q}^{\tt{pd}}
    (\mathbb{R}_{> 0})$.

    Therefore, $F(M') \neq 0$ for all
    $F \in \mathcal{F}$,
    and $\Phi_{\mathbf{x}}
    (((M')^{2})_{11}, \dots, ((M')^{2})_{qq}) \neq 0$.
    We will now define $\xi: N \mapsto
    \Phi_{\mathbf{x}}(N_{11}, \dots, N_{qq})$
    for all $N \in \text{Sym}_{q}(\mathbb{R})$,
    and let $\mathcal{F}'' = \mathcal{F}
    \cup \{\rho_{\tt{full}}, \rho_{\tt{indep}}, \xi\}$.
    From our construction of $M'$, we know that
    $F(M') \neq 0$ for all $F \in \mathcal{F} \cup 
    \{\rho_{\tt{full}}, \rho_{\tt{indep}}\}$,
    and $\xi((M')^{2}) = \zeta(M') \neq 0$.
    (We note that as $M' \in \text{Sym}_{q}^{\tt{pd}}
    (\mathbb{R}_{> 0})$, we have $(M')^2 = \mathcal{S}_{M'}(2)$.)
    Therefore, \cref{corollary: stretchingWorksPositive}
    implies that there exists some 
    $M'' \in \mathfrak{R}(M', \mathcal{F''})
    \cap \text{Sym}_{q}^{\tt{pd}}(\mathbb{R}_{> 0})$.
    Since $\PlEVAL(M') \leq \PlEVAL(M)$, this means that
    $M'' \in \mathfrak{R}(M, \mathcal{F}'')$.
    Finally, since
    $\Phi_{\mathbf{x}}((M'')_{11}, \dots, (M'')_{qq}) =\xi(M'') \neq 0$,
    \cref{theorem: positiveDefiniteDiagonal}
    implies that there exists some $N \in
    \mathfrak{R}(M'', \mathcal{F} \cup
    \{\rho_{\tt{full}}, \rho_{\tt{indep}}, \Psi_{\mathbf{x}}\})
    \cap \text{Sym}_{q}^{\tt{pd}}(\mathbb{R}_{> 0})$.
    This $N$ is our  required matrix
    \[N \in \mathfrak{R}(M, \mathcal{F}
    \cup \{\rho_{\rho_{\tt{full}}, \tt{indep}}, \Psi_{\mathbf{x}}\}) \cap
    \text{Sym}_{q}^{\tt{pd}}(\mathbb{R}_{> 0}).\]
\end{proof}
\section{Diagonal  and Pairwise Order Distinctness}\label{sec: distinct_matrices}

We will now try to understand the requirement that
$M$ be a p.o.\,independent matrix.

\begin{definition}\label{definition: pairwiseDistinct}
    Let $M \in \text{Sym}_{q}(\mathbb{R})$.
    We say that the rows $i \neq j$ of the matrix $M$
    are order identical, if there exists
    some $\sigma \in S_{q}$ such that
    $$(M_{i1}, \dots, M_{iq}) =
    (M_{j\sigma(1)}, \dots, M_{j\sigma(q)}).$$
    We say that $M \in \text{Sym}_{q}(\mathbb{R})$ is
    pairwise order distinct (or p.o.\,distinct) if 
    no two rows $i \neq j$ of $M$ are order identical.
\end{definition}

\begin{lemma}\label{lemma: fullPairwiseDistTest}
    There exists a $\text{Sym}_{q}(\mathbb{R})$-polynomial
    $\rho_{\tt{dist}}$, such that for any
    $M \in \text{Sym}_{q}(\mathbb{R})$,
    $\rho_{\tt{dist}}(M) \neq 0$ if and only if
    $M$ is p.o.\,distinct.
\end{lemma}
\begin{proof}    
    We know that given any two vectors
    $\mathbf{a}, \mathbf{b} \in \mathbb{R}^{q}$,
    $$\mathbf{a} = \mathbf{b} \iff
    \sum_{k \in [q]}(a_{k} - b_{k})^{2} = 0.$$
    So, we can define the function
    $\rho_{{\tt{dist}}, ij}$ for $i \neq j \in [q]$
    as
    $$\rho_{{\tt{dist}}, ij}(M) = 
    \prod_{\sigma \in S_{q}}\sum_{k \in [q]}
    (M_{ik} - M_{j\sigma(k)})^{2},$$
    and the function $\rho_{\tt{dist}}$ as
    $$\rho_{{\tt{dist}}}(M) = 
    \prod_{i \neq j \in [q]} \rho_{{\tt{dist}}, ij}(M).$$

    From the construction, we can easily see that
    $\rho_{\tt{dist}}(M) \neq 0$ if and only if
    $M$ is p.o.\,distinct.
    We can also see that $\rho_{\tt{dist}}(M)$ is a homogeneous
    polynomial in the entries of $M$. So, it is a 
    $\text{Sym}_{q}(\mathbb{R})$-polynomial.
\end{proof}

\begin{lemma}\label{lemma: DistImpliesIndep}
    Let $M \in \text{Sym}_{q}^{\tt{pd}}({\mathbb{R}_{> 0}})$
    be a row full, p.o.\,distinct matrix.
    Let $\mathcal{F}$ be a countable set of
    $\text{Sym}_{q}(\mathbb{R})$-polynomials
    such that $F(M) \neq 0$ for all
    $F \in \mathcal{F}$.
    Given any $i \neq j \in [q]$,
    there exists some row full, p.o.\,distinct
    $N \in \mathfrak{R}(M, \mathcal{F} \cup \{\rho_{{\tt{indep}}, ij}\})
    \cap \text{Sym}_{q}^{\tt{pd}}(\mathbb{R}_{> 0})$,
    where $\rho_{{\tt{indep}}, ij}$
    is as defined in \cref{ij-indep} in the proof of \cref{lemma: fullPairwiseIndepTest}.
\end{lemma}
\begin{proof}
    Without loss of generality, we may assume that
    $i = 1$, and $j = 2$.
    We may also assume that $M = HDH^{\tt{T}}$.
    We will consider the function
    $\mathcal{S}_{M}$ defined in \cref{definition: mathcalS}.
    If there exists some $\theta$ such that
    $\rho_{{\tt{indep}}, 12}(\mathcal{S}_{M}(\theta)) 
    \neq 0$,
    then we can immediately use
    \cref{corollary: stretchingWorksPositive} to find some
    $N \in \mathfrak{R}(M, \mathcal{F} \cup
    \{\rho_{\tt{full}}, \rho_{{\tt{indep}}, 12}\}) \cap
    \text{Sym}_{q}^{\tt{pd}}(\mathbb{R}_{> 0})$,
    and we would be done.
    Let us now assume that
    $\rho_{{\tt{indep}}, 12}(\mathcal{S}_{M}(\theta)) = 0$
    for all $\theta \in \mathbb{R}$.
    
    Since $M$ is row full and p.o.\,distinct,
    it follows that
    $\rho_{\tt{full}}(\mathcal{S}_{M}(1))
    = \rho_{\tt{full}}(M) \neq 0$, and
    $\rho_{\tt{dist}}(\mathcal{S}_{M}(1))
    = \rho_{\tt{dist}}(M) \neq 0$.
    We will now define
    $\xi: N \mapsto \rho_{\tt{full}}(T_{2}N)$.
    Since the entries of the matrix $T_{2}N$ are 
    homogeneous polynomials in the entries of $N$,
    $\xi$ is clearly a $\text{Sym}_{q}(\mathbb{R})$-polynomial.
    Since $M \in \text{Sym}_{q}(\mathbb{R}_{> 0})$
    is a row full matrix, we note $T_{2}M$
    is also a row full matrix.
    So, $\xi(\mathcal{S}_{M}(1))
    = \rho_{\tt{full}}(T_{2}M) \neq 0$.
   
    Note that $\mathcal{S}_{M}(0) = I$.
    Since $\mathcal{S}_{M}(\theta)_{ij}$ is a
    continuous function of $\theta$,
    we know that there exists some $\delta > 0$
    such that $|\mathcal{S}_{M}(\theta)_{ij} - I_{ij}|
    < \frac{1}{3}$ for all $0 < \theta < \delta$,
    for all $i, j \in [q]$.
    With this choice of $\delta$, 
    \cref{lemma: stretchingWorks} allows us to find some
    $M' = HD^{\theta'}H^{\tt{T}} = \mathcal{S}_{M}(\theta') \in
    \mathfrak{R}(M, \mathcal{F} \cup \{\rho_{\tt{full}},
    \rho_{\tt{dist}}, \xi\}) \cap
    \text{Sym}_{q}^{\tt{pd}}(\mathbb{R}_{\neq 0})$,
    such that $0 < \theta' < \delta$.
    We should stress here that $M'$ may have negative entries.

    From the definition of $\rho_{{\tt{indep}}, 12}$,
    we see that
    $$\rho_{{\tt{indep}}, 12}(\mathcal{S}_{M}(\theta)) = 
    \prod_{\sigma \in S_{q}}\mathfrak{g}\Big(
    (\mathcal{S}_{M}(\theta)_{11},
    \dots, \mathcal{S}_{M}(\theta)_{1q}),
    (\mathcal{S}_{M}(\theta)_{2\sigma(1)}, 
    \dots, \mathcal{S}_{M}(\theta)_{2\sigma(q)})\Big).$$
    Since $\rho_{{\tt{indep}}, 12}
    (\mathcal{S}_{M}(\theta))$ is identically 0, and is the product
    of real analytic functions, it follows that there
    must be some $\sigma \in S_{q}$ such that
    $\mathfrak{g}\Big(
    (\mathcal{S}_{M}(\theta)_{11}, \dots, 
    \mathcal{S}_{M}(\theta)_{1q}),
    (\mathcal{S}_{M}(\theta)_{2\sigma(1)}, \dots, 
    \mathcal{S}_{M}(\theta)_{2\sigma(q)})\Big) = 0$
     for all $\theta$.
    This implies that for all $i, j \in [q]$,
    $M'_{1i}\cdot M'_{2\sigma(j)} = M'_{1j} \cdot M'_{2\sigma(i)}$.
    
    From our construction of $M'$, we have
    $M'_{11}, M'_{22} > \nicefrac{2}{3}$, and
    $|M'_{ij}| < \nicefrac{1}{3}$ for all $i \neq j$.
    Therefore, given any $t_{1} \neq 1, 
    t_{2} \neq 2 \in [q]$, we see that
    $M'_{11} > |M'_{1t_{1}}|$ and 
    $M'_{22} > |M'_{2t_{2}}|$.
    Therefore, $|M'_{11}| \cdot |M'_{22}| >
    |M'_{1t_{1}}| \cdot |M'_{2t_{2}}|$.
    But we know that for all $j \in [q]$,
    $|M'_{11}| \cdot |M'_{2\sigma(j)}| =
    |M'_{1j}| \cdot |M'_{2\sigma(1)}|$.
    In particular, it must also be true
    for $j = \sigma^{-1}(2)$.
    This is only possible if $j = 1$ or $\sigma(1) = 2$.
    These are equivalent, namely $\sigma(1) = 2$.
    Now, let us assume that
    there exists some arbitrary $\sigma' \in S_{q}$, 
    such that
    $\mathfrak{g}\Big(
    (M'_{11}, \dots, M'_{1q}),
    (M'_{2\sigma'(1)}, \dots, M'_{2\sigma'(q)})\Big) = 0$.
    By the same argument, we have $\sigma'(1) = 2$.
    But this implies that, for all $t \in [q]$,
    $$\frac{|M'_{1t}|}{|M'_{2\sigma'(t)}|}
    = \frac{|M'_{11}|}{|M'_{22}|}
    = \frac{|M'_{1t}|}{|M'_{2\sigma(t)}|}.$$
    (The ratios are well defined as $M'\in
    \text{Sym}_{q}^{\tt{pd}}(\mathbb{R}_{\neq 0})$.)
    By construction, $M'$ satisfies the property
    that $\xi(M') = \rho_{\tt{full}}(T_{2}M') \neq 0$.
    Therefore, $T_2 M'$  is row full, and thus 
    $|M'_{2\sigma'(t)}|
    = |M'_{2\sigma(t)}|$ implies that  $\sigma'(t) = \sigma(t)$, for all $t \in [q]$.
    Therefore, we  conclude that
    $$\left( (\forall\ \theta \in \mathbb{R}) \left[ \mathfrak{g}\Big(
    (\mathcal{S}_{M}(\theta)_{11}, \dots, 
    \mathcal{S}_{M}(\theta)_{1q}),
    (\mathcal{S}_{M}(\theta)_{2\sigma'(1)}, \dots, 
    \mathcal{S}_{M}(\theta)_{2\sigma'(q)})\Big) = 0
    \  \right] \right) \implies 
    \sigma' = \sigma.$$
    Therefore, we can define
    $\zeta_{\tau}: N \mapsto 
    \mathfrak{g}((N_{11}, \dots, N_{1q}),
    (N_{2\tau(1)}, \dots, N_{2\tau(q)}))$
    for all $\tau \in S_{q}$.
    We can now use
    \cref{corollary: stretchingWorksPositive} to find
    $M'' \in \mathfrak{R}(M, \mathcal{F} \cup
    \{\rho_{\tt{full}}, \rho_{\tt{dist}}\} \cup
    \{\zeta_{\tau}: \tau \neq \sigma \in S_{q}\})
    \cap \text{Sym}_{q}^{\tt{pd}}(\mathbb{R}_{> 0})$.

    Now, let the entries of this matrix $M''$
    be generated by some $\{g_{1}, \dots, g_{d}\}$.
    Recall that we may assume by replacing
    $M''$ with some $c M''$, that there
    exist unique integers $e_{ijt}$ for that
    for all $i, j \in [q]$,
    $$(M'')_{ij} = (-1)^{e_{ij0}}g_{1}^{e_{ij1}}
    \cdots g_{d}^{e_{ijd}},$$
    where $e_{ij0} \in \{0, 1\}$ and
    $e_{ijt} \geq 0$ for all $t \in [d]$,
    for all $i, j \in [q]$.
    Since $M'' \in \text{Sym}_{q}(\mathbb{R}_{> 0})$,
    we see that $e_{ij0} = 0$ for all $i, j \in [q]$.
    Moreover, $M''$ is a row full matrix by construction.
    Therefore, for each $i \in [q]$, we have pairwise distinct tuples
    $$(e_{i11}, e_{i12}, \ldots, e_{i1d}) \neq
    (e_{i21}, e_{i22}, \ldots, e_{i2d}) \neq \cdots \neq 
    (e_{iq1}, e_{iq2}, \ldots, e_{iqd}).$$
    Let $\mathfrak{m} =
    \max_{i, j \in [q], t \in [d]}(e_{ijt})$.
    Pick any integer $m  > \mathfrak{m}$, and
    let
    $$z_{ij} = \sum_{t \in [d]}m^{t} \cdot e_{ijt}.$$
    We see that because of our choice of $m$,
    $z_{i1} \neq z_{i2} \neq \cdots \neq
    z_{iq}$ are pairwise distinct for all $i \in [q]$.
    We also observe that for all $i, j \in [q]$,
    $$\mathcal{T}_{M''}(p^{m}, \dots, p^{m^{d}})_{ij}
    = p^{z_{ij}}.$$
    Moreover, we may assume that $\zeta_{\sigma}
    (\mathcal{T}_{M''}(p^{m}, \dots, p^{m^{d}})) = 0$
    for all $p \in \mathbb{R}$, since
    otherwise, we would be able to
    use \cref{lemma: thickeningWorks} to find the required
    $N \in \mathfrak{R}(M'', \mathcal{F} \cup
    \{\rho_{\tt{full}}, \rho_{\tt{dist}}\} 
    \cup \{\zeta_{\tau}: \tau \in S_{q}\}) 
    \cap \text{Sym}_{q}^{\tt{pd}}(\mathbb{R}_{> 0}) = 
    \mathfrak{R}(M'', \mathcal{F} \cup
    \{\rho_{\tt{full}}, \rho_{\tt{dist}}, \rho_{{\tt{indep}}, 12}\}) 
    \cap \text{Sym}_{q}^{\tt{pd}}(\mathbb{R}_{> 0})$.
    Therefore, the two row vectors
    $$(p^{z_{11}}, \ldots,  p^{z_{1q}}) ~~\mbox{and}~~
    (p^{z_{2\sigma(1)}},  \ldots,  p^{z_{2\sigma(q)}})$$
    are linearly
    dependent for all $p$, and so one is a scalar multiple of the other
    for every $p$. This  scalar multiple is a function of $p$,
    and is obtainable as a ratio of two monomials, and thus is itself a power of $p$.
    Therefore, 
    there exists some integer $c \in \mathbb{Z}$
    such that $z_{11} - z_{2\sigma(1)} = \cdots =
    z_{1q} - z_{2\sigma(q)} = c$.
    We will now define one more function
    $\mathcal{T}_{M}^{*}: \mathbb{R} \rightarrow 
    \text{Sym}_{q}(\mathbb{R})$ such that
    $$\mathcal{T}_{M}^{*}(p)
    = (\mathcal{T}_{M''}(p^{m}, \dots, p^{m^{d}}))^{2}.$$

    Let us first assume that
    $\zeta_{\sigma}(\mathcal{T}_{M}^{*}(p)) = 0$ for all
    $p \in \mathbb{R}$.
    This implies that
    $$(\mathcal{T}_{M}^{*}(p)^{2})_{11} \cdot 
    (\mathcal{T}_{M}^{*}(p)^{2})_{2 \sigma (t)}
    = (\mathcal{T}_{M}^{*}(p)^{2})_{2 \sigma(1)} \cdot
    (\mathcal{T}_{M}^{*}(p)^{2})_{1t}$$
    for all $t \in [q]$.
    We recall that $\sigma(1) = 2$.
    Therefore, for all $t \in [q]$, we will define
    $f_{t}: \mathbb{R} \rightarrow \mathbb{R}$ as
    $f_{t}(p) = (\mathcal{T}_{M}^{*}(p)^{2})_{11} \cdot 
    (\mathcal{T}_{M}^{*}(p)^{2})_{2 \sigma (t)}
    - (\mathcal{T}_{M}^{*}(p)^{2})_{22} \cdot
    (\mathcal{T}_{M}^{*}(p)^{2})_{1t}$,
    and we find that
    \begin{align*}
        f_{t}(p)
        &= \left(\sum_{i \in [q]}p^{2z_{i1}}\right) \cdot
        \left(\sum_{i \in [q]}p^{z_{i2} + z_{i\sigma(t)}}\right) -
        \left(\sum_{i \in [q]}p^{2z_{i2}}\right) \cdot
        \left(\sum_{i \in [q]}p^{z_{i1} + z_{it}}\right)\\
        &= \left(\sum_{i, j \in [q]}
        p^{2z_{i1} + z_{j2} + z_{j\sigma(t)}} \right) - 
        \left(\sum_{i, j \in [q]}
        p^{2z_{i2} + z_{j1} + z_{jt}} \right).
    \end{align*}
    Since we know that $f_{t}(p) = 0$ for all
    $p \in \mathbb{R}$,
    it follows that $\frac{df_{t}}{dp}(p) = 0$
    for all $p \in \mathbb{R}$ as well.
    Specifically, we note that
    \begin{align*}
        \frac{df_{t}}{dp}(1)
        &= \left(\sum_{i, j \in [q]}
        2z_{i1} + z_{j2} + z_{j\sigma(t)} \right) - 
        \left(\sum_{i, j \in [q]}
        2z_{i2} + z_{j1} + z_{jt} \right)\\
        &= q\left(2\sum_{i \in [q]}(z_{i1}) +
        \sum_{i \in [q]}(z_{i2}) +
        \sum_{i \in [q]}(z_{i\sigma(t)}) \right) -
        q\left(2\sum_{i \in [q]}(z_{i2}) +
        \sum_{i \in [q]}(z_{i1}) +
        \sum_{i \in [q]}(z_{it}) \right)\\
        &= q\left(\sum_{i \in [q]}(z_{i1}) -
        \sum_{i \in [q]}(z_{i2}) + 
        \sum_{i \in [q]}(z_{i\sigma(t)}) - 
        \sum_{i \in [q]}(z_{it})\right).
    \end{align*}

    We recall 
    that $z_{11} - z_{2\sigma(1)} = 
    \cdots = z_{1q} - z_{2\sigma(q)} = c$
    for some integer $c \in \mathbb{Z}$.
    This means that $\sum_{i \in [q]}(z_{i1}) -
    \sum_{i \in [q]}(z_{i2}) = cq$.
    Therefore, we have found that
    $$\sum_{i \in [q]}(z_{it}) - 
    \sum_{i \in [q]}(z_{i\sigma(t)}) = cq.$$
    Since this is true for all $t \in [q]$,
    by the same argument, we also get that
    $$\sum_{i \in [q]}(z_{i\sigma(t)}) - 
    \sum_{i \in [q]}(z_{i\sigma(\sigma(t))}) = cq.$$
    This can be repeated. Since $\sigma \in S_{q}$, we know that
    there exists some $2 \leq n \leq q$ such that
    $\sigma^{n}(1) = 1$.
    So, we have shown that
    $$0 = \sum_{k = 1}^{n}\left(
    \sum_{i \in [q]}(z_{i\sigma^{k-1}(1)}) - 
    \sum_{i \in [q]}(z_{i\sigma^{k}(1)}) \right) = cnq.$$
    Since $n > 0$ and $q > 0$, this implies that $c = 0$.
    But from our construction of $z_{ij}$, this implies
    that $e_{1jt} = e_{2 \sigma(j) t}$, for all $j \in [q]$ and $t \in [d]$.  This implies
    that $M''_{1j} = M''_{2\sigma(j)}$ for all
    $j \in [q]$.
    This contradicts our assumption that
    $\rho_{\tt{dist}}(M'') \neq 0$.
    Therefore, our assumption that
    $\zeta_{\sigma}(\mathcal{T}_{M}^{*}(p)) = 0$
    for all $p \in \mathbb{R}$ must be false.
    We can therefore use \cref{lemma: thickeningWorks}
    to find $M''' \in \mathfrak{R}(M'', \mathcal{F}
    \cup \{\zeta_{\tau}: \tau \neq \sigma\}
    \cup \{\rho_{\tt{full}}, \rho_{\tt{dist}}, \zeta^{*}\}) \cap
    \text{Sym}_{q}^{\tt{pd}}(\mathbb{R}_{> 0})$,
    where $\zeta^{*}: N \mapsto \zeta_{\sigma}(N^{2})$.
    Since $\zeta_{\tau}(M''') \neq 0$,
    for all $\tau \neq \sigma$,
    and $\zeta_{\sigma}((M''')^{2}) \neq 0$
    we can then use
    \cref{corollary: stretchingWorksPositive}
    to find the required
    $N \in \mathfrak{R}(M''', \mathcal{F}
    \cup \{\rho_{\tt{full}}, \rho_{\tt{dist}}, \rho_{{\tt{indep}}, 12}\})
    \cap \text{Sym}_{q}^{\tt{pd}}(\mathbb{R}_{> 0})$.
\end{proof}

\begin{corollary}\label{corollary: pairwiseDistImpliesIndep}
    Let $M \in \text{Sym}_{q}^{\tt{pd}}(\mathbb{R}_{> 0})$
    be a row full, p.o.\,distinct matrix.
    Let $\mathcal{F}$ be a countable set of $\text{Sym}_{q}(\mathbb{R})$-polynomials such that
    $F(M) \neq 0$ for all $F \in \mathcal{F}$.
    There exists some row full, p.o.\,independent
    $N \in \mathfrak{R}(M, \mathcal{F}) \cap
    \text{Sym}_{q}^{\tt{pd}}(\mathbb{R}_{> 0})$.
\end{corollary}
\begin{proof}
    Since $M$ is row full, and p.o.\,distinct,
    \cref{lemma: DistImpliesIndep} implies that
    there exists some $M^{1, 2} \in \mathfrak{R}
    (M, \mathcal{F} \cup \{\rho_{\tt{full}},
    \rho_{\tt{dist}}, \rho_{{\tt{indep}}, 12}\}) \cap
    \text{Sym}_{q}^{\tt{pd}}(\mathbb{R}_{> 0})$.
    We can now use \cref{lemma: DistImpliesIndep} again to find
    $M^{1, 3} \in \mathfrak{R}
    (M^{1, 2}, \mathcal{F} \cup \{\rho_{\tt{full}},
    \rho_{\tt{dist}}, \rho_{{\tt{indep}}, 12},
    \rho_{{\tt{indep}}, 13}\}) \cap
    \text{Sym}_{q}^{\tt{pd}}(\mathbb{R}_{> 0})$ which is contained in
    \[
    \mathfrak{R} (M, \mathcal{F} \cup \{\rho_{\tt{full}},
    \rho_{\tt{dist}}, \rho_{{\tt{indep}}, 12},
    \rho_{{\tt{indep}}, 13}\}) \cap
    \text{Sym}_{q}^{\tt{pd}}(\mathbb{R}_{> 0}).\]
    By repeating this up to $\binom{q}{2}$ times,
    we can find the required 
    \begin{eqnarray*} 
    N & \in  & \mathfrak{R}
    (M, \mathcal{F} \cup \{\rho_{\tt{full}},
    \rho_{\tt{dist}}\} \cup
    \{\rho_{{\tt{indep}}, ij}: i \neq j \in [q]\}) \cap
    \text{Sym}_{q}^{\tt{pd}}(\mathbb{R}_{> 0}) = \\
    &  &
    \mathfrak{R}(M, \mathcal{F} \cup \{\rho_{\tt{full}},
    \rho_{\tt{dist}}, \rho_{\tt{indep}}\}) \cap
    \text{Sym}_{q}^{\tt{pd}}(\mathbb{R}_{> 0}),
    \end{eqnarray*}
    where $\rho_{\tt{indep}}$ is defined in \cref{po-indep}.
\end{proof}

\begin{lemma}\label{lemma: pairwiseDistHardness}
    Let $M \in \text{Sym}_{q}^{\tt{pd}}(\mathbb{R}_{> 0})$
    be a row full, p.o.\,distinct matrix.
    If $\mathbf{x}$ is confluent for all
    $\mathbf{0} \neq \mathbf{x} \in \chi_{q}$ such that
    $\Psi_{\mathbf{x}}(M) = 0$, then
    $\PlEVAL(M)$ is $\#$P-hard.
\end{lemma}
\begin{proof}
    Let $\mathcal{F}_{M}$ be as defined in \cref{equation: FM}.
    Let $M' \in \mathfrak{R}(M, \mathcal{F}_{M}) \cap
    \text{Sym}_{q}^{\tt{pd}}(\mathbb{R}_{> 0})$
    be the row full, p.o.\,independent matrix
    whose existence is guaranteed by
    \cref{corollary: pairwiseDistImpliesIndep}.
    
    Let $\mathcal{L}_{M}$ be the lattice formed by the
    eigenvalues of $M$, and let $\mathcal{L}_{M'}$ be the
    lattice formed by the eigenvalues of $M'$.
    By our assumption about $M$, we know that all
    $\mathbf{0} \neq \mathbf{x} \in \mathcal{L}_{M}$
    are confluent.
    From the definition of $\overline{\mathcal{L}_{M}}$,
    we see that this immediately implies that all
    $\mathbf{0} \neq \mathbf{x} \in \overline{\mathcal{L}_{M}}$
    are confluent.
    By our construction of $M'$, and our 
    choice of $\mathcal{F}_{M}$ as defined in 
    \cref{equation: FM}, 
    we see that
    $\overline{\mathcal{L}_{M'}} \subseteq \overline{\mathcal{L}_{M}}$.
    So, all $\mathbf{0} \neq \mathbf{x} \in
    \overline{\mathcal{L}_{M'}}$ are confluent.

    Let $d = \dim \mathcal{L}_{M'}$, the lattice dimension of $\mathcal{L}_{M'}$.
    If $d > 0$, we can pick some
    $\mathbf{x}_{1} \neq \mathbf{0}$ from this lattice.
    We know that $\mathbf{x}_{1}$ is confluent.
    Since $M'$ is row full and p.o.\,independent, 
    we can use \cref{theorem: confluentReduction} to
    find some $N_{1} \in \mathfrak{R}
    (M', \mathcal{F}_{M'} \cup \{\Psi_{\mathbf{x}_{1}}\}) 
    \cap \text{Sym}_{q}^{\tt{pd}}(\mathbb{R}_{> 0})$.
    By \cref{lemma: PsiExists}, $\mathbf{x}_{1} \not \in \overline{\mathcal{L}_{N_{1}}}$,
    where $\mathcal{L}_{N_{1}}$ is the
    lattice formed by the eigenvalues of $N_{1}$.
    We also have $\overline{\mathcal{L}_{N_{1}}} \subsetneq
    \overline{\mathcal{L}_{M'}}$. 
   Hence, by \cref{lemma: dimReduction}, $\dim \mathcal{L}_{N_{1}} < d$.
    If the lattice dimension of $\mathcal{L}_{N_{1}}$
    is non-zero, we can repeat this process by picking
    some $\mathbf{0} \neq \mathbf{x}_{2} \in \mathcal{L}_{N_{1}}$,
    and then using \cref{theorem: confluentReduction} once again
    to find some $N_{2} \in \mathfrak{R}
    (N_{1}, \mathcal{F}_{N_1} \cup \{\Psi_{\mathbf{x}_{1}},
    \Psi_{\mathbf{x}_{2}}\})
    \cap \text{Sym}_{q}^{\tt{pd}}(\mathbb{R}_{> 0}) \subseteq
    \mathfrak{R}(M', \mathcal{F}_{N_1} \cup \{\Psi_{\mathbf{x}_{1}},
    \Psi_{\mathbf{x}_{2}}\}) \cap 
    \text{Sym}_{q}^{\tt{pd}}(\mathbb{R}_{> 0})$.
    We can therefore repeat this process
    up to $d$ times to find some matrix $N
    \in \text{Sym}_{q}^{\tt{F}}(\mathbb{R}_{> 0})$
    such that $\PlEVAL(N) \leq \PlEVAL(M') \leq \PlEVAL(M)$, and
    the lattice dimension of $\mathcal{L}_{N}$ is $0$.
    \cref{theorem: latticeHardness} then proves
    that $\PlEVAL(N)$ is $\#$P-hard, 
    which implies that $\PlEVAL(M)$ is also
    $\#$P-hard.
\end{proof}

We will now prove that the requirement that the matrix
$M$ be p.o.\,distinct is sufficient, and
that the row fullness requirement is redundant.
In order to do that, we will make use
of an intermediary condition.

\begin{definition}\label{definition: diagonalDistinct}
    Let $M \in \text{Sym}_{q}(\mathbb{R})$.
    We say that it is diagonal distinct, if
    $M_{11} \neq M_{22} \neq \cdots \neq M_{qq}$
    are pairwise distinct.
\end{definition}

\begin{lemma}\label{lemma: DiagTest}
    There exists a $\text{Sym}_{q}(\mathbb{R})$-polynomial
    $\rho_{\tt{diag}}$ such that for any
    $M \in \text{Sym}_{q}(\mathbb{R})$,
    $\rho_{\tt{diag}}(M) \neq 0$ if and only if
    $M$ is diagonal distinct.
\end{lemma}
\begin{proof}
    We will define $\rho_{\tt{diag}}:
    \text{Sym}_{q}(\mathbb{R}) \rightarrow \mathbb{R}$ such that
    $$\rho_{\tt{diag}}(N) = \prod_{i \neq j \in [q]}
    (N_{ii} - N_{jj}).$$
    It is easily seen that $\rho_{\tt{diag}}(N)$
    is a $\text{Sym}_{q}(\mathbb{R})$-polynomial such that
    $\rho_{\tt{diag}}(N) \neq 0$ if and only if 
    $N$ is  diagonal distinct.
\end{proof}

\begin{lemma}\label{lemma: pairwiseDistinctImpliesDiagonal}
    Let $M \in \text{Sym}_{q}^{\tt{pd}}(\mathbb{R}_{> 0})$
    be a row full, p.o.\,distinct matrix.
    Let $\mathcal{F}$ be a family of
    $\text{Sym}_{q}(\mathbb{R})$-polynomial,
    such that $F(M) \neq 0$ for all
    $F \in \mathcal{F}$.
    There exists some
     diagonal distinct $N \in \mathfrak{R}
    (M, \mathcal{F}) \cap 
    \text{Sym}_{q}^{\tt{pd}}(\mathbb{R}_{> 0})$.
\end{lemma}
\begin{proof} 
    If $M$ were itself  diagonal distinct,
    we would be done, so let us assume otherwise.
    We may assume that $M_{ii} = M_{jj}$ for
    some $i \neq j$.
    We assume that the entries of $M$ are generated
    by some $\{g_{t}\}_{t \in [d]}$.
    As before, we may use \cref{lemma: MequivalentCM}
    to replace $M$ with some $c M$,
    and consider the function $\mathcal{T}_{M}$
    as defined in \cref{definition: mathcalT}.
    Since $M \in \text{Sym}_{q}(\mathbb{R}_{> 0})$,
    we see that $e_{ij0} = 0$ for all $i, j \in [q]$.
    
    We will first assume that $\rho_{\tt{diag}}
    (\mathcal{T}_{M}(\mathbf{p})^{2}) = 0$
    for all $\mathbf{p} \in \mathbb{R}^{d}$.
    Now, we let $\mathfrak{m} =
    \max_{i, j \in [q], t \in [d]} (e_{ijt})$, 
    pick some $m > \mathfrak{m}$, and
    define $z_{ij} = \sum_{t \in [d]}m^{t} \cdot e_{ijt}$
    for all $i, j \in [q]$.
    We note that
    $$\mathcal{T}_{M}(p^{m}, p^{m^{2}}, \dots, p^{m^{d}})_{ij}
    = p^{z_{ij}}$$
    for all $i, j \in [q]$,
    and because of our choice of $m$,
    $z_{ij} = z_{i'j'}$ if and only if
    $(e_{ij1}, e_{ij2}, \dots, e_{ijd}) = 
    (e_{i'j'1}, e_{i'j'2}, \dots, e_{i'j'd})$.
    We recall that by our assumption,
    $\rho_{\tt{diag}}(\mathcal{T}_{M}
    (p^{m}, \dots, p^{m^{d}})^{2}) = 0$,
    for all $p \in \mathbb{R}$.
    But that implies that for some $i \neq j \in [q]$,
    $(\mathcal{T}_{M}(p^{m}, \dots, p^{m^{d}})^{2})_{ii} = 
    (\mathcal{T}_{M}(p^{m}, \dots, p^{m^{d}})^{2})_{jj}$.
    This implies that
    $$p^{2z_{i1}} + p^{2z_{i2}} + \cdots + p^{2z_{iq}}
    = p^{2z_{j1}} + p^{2z_{j2}} + \cdots + p^{2z_{jq}}$$
    for all $p \in \mathbb{R}$.
    This implies that there exists some 
    $\sigma \in S_{q}$ such that
    $z_{it} = z_{j\sigma(t)}$ for all $t \in [q]$.
    But from our construction of $z_{ij}$,
    this implies that
    $M_{it} = M_{j\sigma(t)}$ for all $t \in [q]$.
    This contradicts our assumption that
    $M$ is p.o.\,distinct.
    Therefore, our assumption that $\rho_{\tt{diag}}
    (\mathcal{T}_{M}(\mathbf{p})^{2}) = 0$
    for all $\mathbf{p} \in \mathbb{R}^{d}$
    must be false.

    So, $\rho_{\tt{diag}}(\mathcal{T}_{M}(\mathbf{p}^{*})^{2})
    \neq 0$ for some $\mathbf{p}^{*} \in \mathbb{R}^{d}$.
    But then, \cref{lemma: thickeningWorks} allows us to 
    first find $M' \in \text{Sym}_{q}^{\tt{pd}}(\mathbb{R}_{> 0})$
    such that $F(M') \neq 0$ for all $F \in \mathcal{F}$,
    and $\rho_{\tt{diag}}((M')^2) \neq 0$.
    Then, \cref{corollary: stretchingWorksPositive}
    allows us to find the required $N
    \in \mathfrak{R}(M, \mathcal{F} \cup \{\rho_{\tt{diag}} \}) \cap 
    \text{Sym}_{q}^{\tt{pd}}(\mathbb{R}_{> 0})$.
\end{proof}

Somewhat informally, we have proved that
$$\text{row full } + \text{ p.o.\,distinct} \implies
\text{diagonal distinct}.$$
We will prove this is actually an equivalence.
To prove the other direction of this equivalence will require
us to set up some more machinery.
Specifically, we need to make use of a new edge gadget.

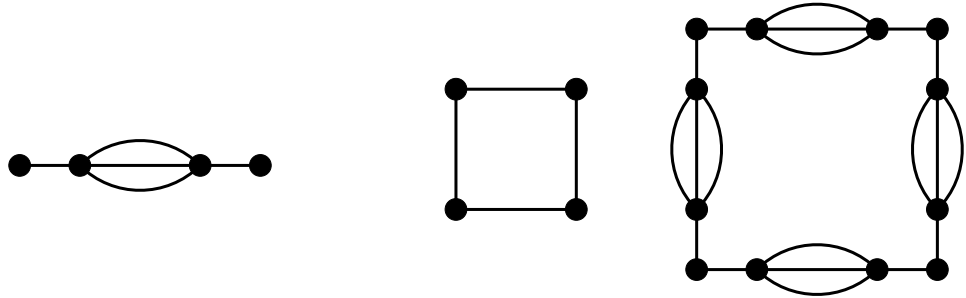
\begin{figure}[b]
    \centering
    \begin{subfigure}{0.45\textwidth}
        \centering
        \raisebox{1.5\height}{
            \scalebox{0.8}{\begin{tikzpicture}[line join=miter, draw opacity=1]

\node[circle, draw=black, fill=black] (c) at (-2, 0){};
\node[circle, draw=black, fill=black] (c) at (-1, 0){};
\node[circle, draw=black, fill=black] (c) at (1, 0){};
\node[circle, draw=black, fill=black] (c) at (2, 0){};

\draw[line width=0.5mm, black] (-2, 0) -- (2, 0);
\draw[line width=0.5mm, black] (-1, 0) to [in=135, out=45] (1, 0);
\draw[line width=0.5mm, black] (-1, 0) to [in=225, out=-45] (1, 0);

\end{tikzpicture}}
        }
	\caption{The edge gadget $R_{3}$.}
	\label{fig: midStretchGadget}
    \end{subfigure}
    \begin{subfigure}{0.45\textwidth}
        \centering
        \scalebox{0.8}{\begin{tikzpicture}[line join=miter, draw opacity=1]

\node[circle, draw=black, fill=black] (c) at (-1, -1){};
\node[circle, draw=black, fill=black] (c) at (-1, 1){};
\node[circle, draw=black, fill=black] (c) at (1, 1){};
\node[circle, draw=black, fill=black] (c) at (1, -1){};

\draw[line width=0.5mm, black] (-1, -1) -- (-1, 1);
\draw[line width=0.5mm, black] (-1, 1) -- (1, 1);
\draw[line width=0.5mm, black] (1, 1) -- (1, -1);
\draw[line width=0.5mm, black] (1, -1) -- (-1, -1);

\begin{scope}[xshift = 5cm]

\node[circle, draw=black, fill=black] (c) at (-2, -2){};
\node[circle, draw=black, fill=black] (c) at (-2, 2){};
\node[circle, draw=black, fill=black] (c) at (2, 2){};
\node[circle, draw=black, fill=black] (c) at (2, -2){};

\node[circle, draw=black, fill=black] (c) at (-2, -1){};
\node[circle, draw=black, fill=black] (c) at (-2, 1){};
\node[circle, draw=black, fill=black] (c) at (-1, 2){};
\node[circle, draw=black, fill=black] (c) at (1, 2){};
\node[circle, draw=black, fill=black] (c) at (2, -1){};
\node[circle, draw=black, fill=black] (c) at (2, 1){};
\node[circle, draw=black, fill=black] (c) at (-1, -2){};
\node[circle, draw=black, fill=black] (c) at (1, -2){};

\draw[line width=0.5mm, black] (-2, -2) -- (-2, 2);
\draw[line width=0.5mm, black] (-2, 2) -- (2, 2);
\draw[line width=0.5mm, black] (2, 2) -- (2, -2);
\draw[line width=0.5mm, black] (2, -2) -- (-2, -2);

\draw[line width=0.5mm, black] (-2, -1) to [in=225, out=135] (-2, 1);
\draw[line width=0.5mm, black] (-2, -1) to [in=-45, out=45] (-2, 1);

\draw[line width=0.5mm, black] (-1, 2) to [in=135, out=45] (1, 2);
\draw[line width=0.5mm, black] (-1, 2) to [in=225, out=-45] (1, 2);

\draw[line width=0.5mm, black] (2, 1) to [in=45, out=-45] (2, -1);
\draw[line width=0.5mm, black] (2, 1) to [in=135, out=225] (2, -1);

\draw[line width=0.5mm, black] (1, -2) to [in=45, out=135] (-1, -2);
\draw[line width=0.5mm, black] (1, -2) to [in=-45, out=225] (-1, -2);

\end{scope}

\end{tikzpicture}}
	\caption{An example graph $G$,
                    and the graph $R_{3}G$.}
	\label{fig: midStretchGadgetExample}
    \end{subfigure}
    \caption{The edge gadget $R_{n}$}
    \label{fig: midStretchGadgetAll}
\end{figure}

Given any graph $G = (V, E)$, and any $n \in \mathbb{Z}_{> 0}$,
we will construct the graph $R_{n}G$
by replacing each edge of $G$ with a path of length
$3$, and then replacing the middle edge of each such
path with $n$ parallel edges.
Clearly, this gadget preserves planarity.
Moreover, we note that for any $M \in
\text{Sym}_{q}(\mathbb{R})$,
given any graph $G = (V, E)$,
$$Z_{M}(R_{n}G) = \sum_{\sigma: V \rightarrow [q]}
\prod_{(u, v) \in E}
\sum_{a, b \in [q]}M_{\sigma(u)a}(M_{ab})^{n}M_{\sigma(v)b}
= \sum_{\sigma: V \rightarrow [q]} \prod_{(u, v) \in E}
(R_{n}M)_{\sigma(u)\sigma(v)},$$
where $R_{n}(M) \in \text{Sym}_{q}(\mathbb{R})$ such that
for all $i, j \in [q]$,
$$(R_{n}M)_{ij} = \sum_{a, b \in [q]}
(M_{ab})^{n} \cdot M_{ia}M_{jb}.$$

\begin{definition}\label{definition: RInterpolation}
    Let $M \in \text{Sym}_{q}(\mathbb{R}_{> 0})$.
    We can now define a function $\mathcal{R}_{M}:
    \mathbb{R} \rightarrow \text{Sym}_{q}(\mathbb{R}_{> 0})$,
    such that for all $i, j \in [q]$,
    $$\mathcal{R}_{M}(\theta)_{ij} = 
    \sum_{a, b \in [q]}(M_{ab})^{\theta} \cdot 
    M_{ia}M_{jb}.$$
\end{definition}
In particular, for any  $M \in \text{Sym}_{q}(\mathbb{R}_{> 0})$ and integer $n \ge 1$,
\begin{equation}\label{Rn-and-cal-RM}
\mathcal{R}_{M}(n) = R_{n}M.
\end{equation}
\begin{lemma}\label{lemma: RInterpolationReduction}
    Let $M \in \text{Sym}_{q}(\mathbb{R}_{> 0})$.
    Then, $\PlEVAL(\mathcal{R}_{M}(\theta)) \leq 
    \PlEVAL(M)$ for all $\theta \in \mathbb{R}$.
\end{lemma}
\begin{proof}
    We know that for any $n \geq 1$, and given any
    planar graph $G = (V, E)$,
    $$Z_{M}(R_{n}G) = 
    \sum_{\sigma: V \rightarrow [q]} \prod_{(u, v) \in E}
    \sum_{a, b \in [q]}M_{\sigma(u)a}(M_{ab})^{n}M_{\sigma(v)b}.$$
    We recall that
    $$\mathcal{P}_{m}(n) = \left\{\mathbf{x} = (x_{1}, \dots, x_{m})
    \in \mathbb{Z}^{m} \hspace{0.08cm}\Big|\hspace{0.1cm} 
    (\forall\ i \in [m])\ [x_{i} \geq 0],~~\mbox{and}~~
    x_{1} + \cdots + x_{m} = n\right\}.$$
    We  define the \emph{set} (with no duplicates)
    $$X_{M}(G) = \left\{ \prod_{i, j \in [q]}
    M_{ij}^{k_{ij}} \Big|\hspace{0.1cm}
    \mathbf{k} \in \mathcal{P}_{q^{2}}(|E|) \right\}.$$
    For each $\mathbf{k} \in \mathcal{P}_{q^{2}}(|E|)$,
    we will define $x_{M}(\mathbf{k}) = 
    \prod_{i, j \in [q]}M_{ij}^{k_{ij}}$, and    
    $$Y_{M}(\mathbf{k}) = 
    \left\{(\mathbf{a}, \mathbf{b}) \in [q]^{|E|}
    \hspace{0.08cm}\Big|\hspace{0.1cm} 
    (\forall\ 
    i, j \in [q])~\left[\ k_{ij} = |\{t: a_{t} = i, b_{t} = j\}| \ \right]\  \right\}.$$
    Then, for each $x \in X_{M}(G)$, we define
    $$c_{M}(G, x) = \sum_{\sigma: V \rightarrow [q]}\left(
    \sum_{\substack{\mathbf{k} \in \mathcal{P}_{q^{2}}(|E|):\\
    x_{M}(\mathbf{k}) = x}}\left(
    \sum_{\mathbf{a}, \mathbf{b} \in Y_{M}(\mathbf{k})}\left(
    \prod_{(u, v) \in E} M_{\sigma(u)a_{u}}M_{\sigma(v)b_{v}}
    \right)\right)\right).$$
    Finally, we note that
    $$Z_{M}(R_{n}G) = \sum_{x \in X_{M}(G)}x^{n} \cdot
    c_{M}(G, x).$$
    Since each $x \in X_{M}(G)$ is distinct, and
    since $|X_{M}(G)| \leq |E|^{O(1)}$, we can use
    oracle access to $\PlEVAL(M)$ to compute
    $Z_{M}(R_{n}G)$ for $n \in [|X_{M}(G)|]$.
    We will then have a full rank Vandermonde system of
    linear equations, which can be solved in
    polynomial time to find
    $c_{M}(G, x)$ for all $x \in X_{M}(G)$.

    We will now note that each $x \in X_{M}(G)$ is the product
    of entries of the matrix $M$.
    Since $M \in \text{Sym}_{q}(\mathbb{R}_{> 0})$,
    this implies that every element of $X_{M}(G)$
    is positive.
    Moreover, we see that for every $\mathbf{k} \in
    \mathcal{P}_{q^{2}}(|E|)$, and for every
    $\theta \in \mathbb{R}$,
    $$(x_{M}(\mathbf{k}))^{\theta}
    = \prod_{i, j \in [q]}(M_{ij}^{k_{ij}})^{\theta}
    = \prod_{i, j \in [q]}(M_{ij}^{\theta})^{k_{ij}}$$
    is well defined.
    Therefore, for any $\theta \in \mathbb{R}$, 
    we can compute the quantity
    $$\sum_{x \in X_{M}(G)} x^{\theta} \cdot c_{M}(G, x)
    = \sum_{\sigma: V \rightarrow [q]} \prod_{(u, v) \in E}
    \sum_{a, b \in [q]}M_{\sigma(u)a}(M_{ab})^{\theta}M_{\sigma(v)b}
    = \sum_{\sigma: V \rightarrow [q]} \prod_{(u, v) \in E}
    \mathcal{R}_{M}(\theta)_{\sigma(u)\sigma(v)}$$
    in polynomial time. This proves that
    $\PlEVAL(\mathcal{R}_{M}(\theta)) \leq 
    \PlEVAL(M)$ for all $\theta \in \mathbb{R}$.
\end{proof}

\begin{lemma}\label{lemma: RInterpolationWorks}
    Let $M \in \text{Sym}_{q}^{\tt{pd}}(\mathbb{R}_{> 0})$.
    Let $\mathcal{F}$ be a countable set of
    $\text{Sym}_{q}(\mathbb{R})$-polynomial, such that
    for all $F \in \mathcal{F}$, there exists some
    $\theta_{F} \in \mathbb{R}$ such that
    $F(\mathcal{R}_{M}(\theta_{F})) \neq 0$.
    Then, there exists some $N = \mathcal{R}_{M}(\theta^{*})
    \in \mathfrak{R}(M, \mathcal{F}) \cap
    \text{Sym}_{q}^{\tt{pd}}(\mathbb{R}_{> 0})$.
\end{lemma}
\begin{proof}
    Since each $F \in \mathcal{F}$ is polynomial in the
    entries of the input matrix, and since
    each entry of the matrix
    $\mathcal{R}_{M}(\theta)$ is a real analytic function of
    $\theta$, we see that
    $F(\mathcal{R}_{M}(\theta))$ is a real analytic function
    of $\theta$ for all $F \in \mathcal{F}$.
    Moreover, since $F(\mathcal{R}_{M}(\theta_{F})) \neq 0$
    for all $F \in \mathcal{F}$,
    we see that these are non-zero analytic functions.

    From the definition of $\mathcal{R}_{M}(\theta)$,
    we can see that for all $i, j \in [q]$,
    $\mathcal{R}_{M}(\theta)_{ij} > 0$.
    We also note that $\mathcal{R}_{M}(1) = M^{3}$, and since
    $M$ is positive definite, it follows that
    $M^{3}$ is also positive definite.
    We note that the eigenvalues of the matrix
    $\mathcal{R}_{M}(\theta)$ are continuous functions of
    the entries of the matrix, and consequently, 
    are continuous functions of $\theta$.
    This implies that there exists some interval
    $(a, b) \subset \mathbb{R}$ such that
    $1 \in (a, b)$, and $\mathcal{R}_{M}(\theta)
    \in \text{Sym}_{q}^{\tt{pd}}(\mathbb{R}_{> 0})$.

    Now, since we know that for each $F \in \mathcal{F}$,
    $F(\mathcal{R}_{M}(\theta))$ is a non-zero real analytic
    function, we know that for each $F \in \mathcal{F}$,
    there are only finitely many zeros within the interval
    $(a, b)$ (\cite{krantz2002primer}, Corollary 1.2.7).
    Since a countable union of finite sets is countable,
    there must exist some $\theta^{*} \in (a, b)$
    such that $F(\mathcal{R}_{M}(\theta^{*})) \neq 0$.
    From our choice of $(a, b)$,
    we also see that $\mathcal{R}_{M}(\theta^{*}) \in
    \text{Sym}_{q}^{\tt{pd}}(\mathbb{R}_{> 0})$.
    So, it is the required matrix.
\end{proof}

We can now use this gadget to first show that if $M \in 
\text{Sym}_{q}^{\tt{pd}}(\mathbb{R}_{> 0})$ is 
diagonal distinct, we can find a  \emph{row full} and diagonal distinct  
$N \in \text{Sym}_{q}^{\tt{pd}}(\mathbb{R}_{> 0})$
such that $\PlEVAL(N) \leq \PlEVAL(M)$.
Later, we will show that if $M \in 
\text{Sym}_{q}^{\tt{pd}}(\mathbb{R}_{> 0})$
is both  row full and diagonal distinct, then we can in fact find a row full,
p.o.\,distinct $N \in 
\text{Sym}_{q}^{\tt{pd}}(\mathbb{R}_{> 0})$.

\begin{lemma}\label{lemma: diagDistinctImpliesRowFull}
    Let $M \in \text{Sym}_{q}^{\tt{pd}}(\mathbb{R}_{> 0})$
    be a  diagonal distinct matrix.
    Let $\mathcal{F}$ be a countable set of
    $\text{Sym}_{q}(\mathbb{R})$-polynomials such that
    $F(M) \neq 0$ for all $F \in \mathcal{F}$.
    Then there exists a row full,  diagonal distinct
    $N \in \mathfrak{R}(M, \mathcal{F})
    \cap \text{Sym}_{q}^{\tt{pd}}(\mathbb{R}_{> 0})$.
\end{lemma}
\begin{proof}
    Since $\mathcal{S}_{M}(0) = I$, 
    there exists some $\delta > 0$
    such that $|\mathcal{S}_{M}(\theta)_{ij} - I_{ij}|
    \leq \frac{1}{3}$ for all $i, j \in [q]$,
    and $0 < \theta < \delta$.
    Therefore, we can use
    \cref{lemma: stretchingWorks} to find
    $M' = \mathcal{S}_{M}(\theta^{*}) \in \mathfrak{R}(M, \{\rho_{\tt{diag}}\}) \cap
    \text{Sym}_{q}^{\tt{pd}}(\mathbb{R}_{\neq 0})$,
    for some $\theta^{*} \in \mathbb{R}$,
    such that $|(M')_{ij} - I_{ij}| \leq \frac{1}{3}$
    for all $i, j \in [q]$.

    We can now consider $R_{n}(M')$ for $n \geq 1$.
    From the definition, we see that
    $$R_{n}(M')_{ij} = \sum_{a, b \in [q]}
    (M'_{ab})^{n} \cdot (M')_{ia}(M')_{jb}.$$
    We may now let $X = \{(M')_{ab}: a, b \in [q]\}$,
    and for each $x \in X$, define
    $$c_{ij}(x) = \sum_{a, b \in [q]: (M')_{ab} = x}
    (M')_{ia}(M')_{jb}.$$
    So, we can express $R_{n}(M')_{ij}$ as
    $$R_{n}(M')_{ij} = \sum_{x \in X}x^{n} \cdot c_{ij}(x)$$
    for all $i, j \in [q]$.
    Therefore, for any $i \in [q]$, and $j \neq j' \in [q]$,
    we see that
    $$R_{n}(M')_{ij} - R_{n}(M')_{ij'}
    = \sum_{x \in X}
    x^{n} \cdot (c_{ij}(x) - c_{ij'}(x)).$$
    Since each $x \in X$ is distinct, and $|X| = O(1)$,
    this forms a full rank Vandermonde system of
    linear equations.
    If $R_{n}(M')_{ij} - R_{n}(M')_{ij'} = 0$ for all
    $n \geq 1$, that would imply
    that $c_{ij}(x) - c_{ij'}(x) = 0$ for all
    $x \in X$.

    In particular, we note that by our choice of $M'$,
    $\rho_{\tt{diag}}(M') \neq 0$.
    So, $(M')_{jj}$ is not equal to any other diagonal entry
    of $M'$. By our choice of $M'$, we also know that
    the diagonal entries of $M'$ have a higher absolute value
    than any non-diagonal entry.
    Therefore, $(M')_{ab} = (M')_{jj}$ implies that
    $(a, b) = (j, j)$.
    So,
    $$c_{ij}((M')_{jj}) - c_{ij'}((M')_{jj})
    = (M')_{ij}(M')_{jj} - (M')_{ij}(M')_{jj'} = 0.$$
    Since $M' \in \text{Sym}_{q}(\mathbb{R}_{\neq 0})$,
    we know that $(M')_{ij} \neq 0$.
    So, we have shown that $(M')_{jj} = (M')_{jj'}$,
    which contradicts our assumption that
    the diagonal entries of $M'$ have a higher absolute value
    than any non-diagonal entry.
    This implies that it cannot be
    possible that $R_{n}(M')_{ij} - R_{n}(M')_{ij'} = 0$
    for all $n \geq 1$.

    Therefore, given any $i, j \neq j' \in [q]$, there
    exists some $n_{i, j, j'} \in \mathbb{Z}_{> 0}$
    such that $(R_{n_{i, j, j'}}(M'))_{ij} -
    (R_{n_{i, j, j'}}(M'))_{ij'} \neq 0$.
    We can now define the function
   $\zeta_{i, j, j'}: N \mapsto (R_{n_{i, j, j'}}N)_{ij} -
   (R_{n_{i, j, j'}}N)_{ij'}$.
   From the definition of $R_{n_{i, j, j'}}$
    in \cref{definition: RInterpolation},
    we can see that $\zeta_{i, j, j'}$ is a 
    $\text{Sym}_{q}(\mathbb{R})$-polynomial.

    Recall that by construction, $M' = \mathcal{S}_{M}(\theta^{*})$
    for some $\theta^{*} \in \mathbb{R}$.
    So, we have seen that for all $i, j \neq j' \in [q]$,
    $\zeta_{i, j, j'}(\mathcal{S}_{M}(\theta^{*}))
    \neq 0$.
    We can now let
    $$\mathcal{F}' = \left\{F': N \mapsto F(N^{3})
    \hspace{0.08cm}\Big|\hspace{0.1cm}
    F \in \mathcal{F} \cup \{\rho_{\tt{diag}} \}\right\}
    \cup \{\zeta_{i, j, j'} \}_{i, j \neq j' \in [q]}.$$
    We note that $(\mathcal{S}_{M}(\frac{1}{3}))^{3}
    = M$, and $F(M) \neq 0$ for all $F \in \mathcal{F}
    \cup \{\rho_{\tt{diag}} \}$,
    and $\mathcal{S}_{M}(\theta^{*}) = M'$.
    Therefore, \cref{corollary: stretchingWorksPositive}
    allows us to find
    $M'' \in \mathfrak{R}(M, \mathcal{F}') \cap 
    \text{Sym}_{q}^{\tt{pd}}(\mathbb{R}_{> 0})$.
    Since $M'' \in \text{Sym}_{q}(\mathbb{R}_{> 0})$,
    we can now consider the function
    $\mathcal{R}_{M''}$.
    We see that by construction of $M''$,
    $\mathcal{R}_{M''}(1) = (M'')^{3}$.
    So, for all $F \in \mathcal{F} \cup \{\rho_{\tt{diag}}\}$,
    we see that $F(\mathcal{R}_{M''}(1)) = F((M'')^{3}) \neq 0$ 
    by the fact that $M'' \in \mathfrak{R}(M, \mathcal{F}')$.
    
    Finally, consider the functions $\xi_{i, j, j'} : N \mapsto N_{ij} -  N_{ij'}$.
    Clearly $\xi_{i, j, j'}$ are
    $\text{Sym}_{q}(\mathbb{R})$-polynomials.
    From our choice of $\mathcal{F}'$, and $\mathcal{R}_{M''}(n_{i, j, j'}) = R_{n_{i, j, j'}} (M'')$ 
    by \cref{Rn-and-cal-RM},
    we know that for every $i, j \neq j' \in [q]$, there exists $n_{i, j, j'}$, such that
$$  \xi_{i, j, j'}(\mathcal{R}_{M''}(n_{i, j, j'})) =  \xi_{i, j, j'}(R_{n_{i, j, j'}} (M''))
= (R_{n_{i, j, j'}} (M''))_{ij} -
    (R_{n_{i, j, j'}} (M''))_{ij'} = \zeta_{i, j, j'}(M'')
    \neq 0.$$
    Therefore, we apply \cref{lemma: RInterpolationWorks} to the function set
    $\mathcal{F} \cup \{\rho_{\tt{diag}} \} \cup \{\xi_{i, j, j'}: i \in [q], j \neq j' \in [q]\}$, to get
   a row full,  diagonal distinct
    $N \in \mathfrak{R}(M'', \mathcal{F})
    \cap \text{Sym}_{q}^{\tt{pd}}(\mathbb{R}_{> 0})$.
    Note that $\rho_{\tt{full}} = \prod_{i \in [q]} \prod_{j \ne j'  \in [q]} \xi_{i, j, j'} $
    Since $\PlEVAL(M'') \leq \PlEVAL(M)$, this is the
    required matrix.
\end{proof}

\begin{lemma}\label{lemma: diagDistinctImpliesOrderDist}
    Let $M \in \text{Sym}_{q}^{\tt{pd}}(\mathbb{R}_{> 0})$
    be a  diagonal distinct matrix.
    Let $\mathcal{F}$ be a countable set of 
    $\text{Sym}_{q}(\mathbb{R})$-polynomials such that
    $F(M) \neq 0$ for all $M \in \mathcal{F}$.
    There exists some row full, p.o.\,distinct
    $N \in \mathfrak{R}(M, \mathcal{F})
    \cap \text{Sym}_{q}^{\tt{pd}}(\mathbb{R}_{> 0})$.
\end{lemma}
\begin{proof}
    We will first use \cref{lemma: diagDistinctImpliesRowFull}
    to find a row full,  diagonal distinct matrix 
    $$M' \in 
    \mathfrak{R}(M, \mathcal{F} \cup \{\rho_{\tt{full}}, \rho_{\tt{diag}}\})
    \cap \text{Sym}_{q}^{\tt{pd}}(\mathbb{R}_{> 0}).$$

    We may assume that $M' = HDH^{\tt{T}}$.
    Since $\mathcal{S}_{M'}(0) = I$,
    we know that there exists some $\delta > 0$
    such that $|\mathcal{S}_{M'}(\theta)_{ij} - I_{ij}|
    < \frac{1}{3}$ for all $i, j \in [q]$,
    for all $0 < \theta < \delta$.
    We can now use
    \cref{lemma: stretchingWorks} to find
    some $M'' = HD^{\theta^{*}}H^{\tt{T}} \in 
    \mathfrak{R}(M', \mathcal{F} \cup 
    \{\rho_{\tt{diag}},\rho_{\tt{full}}\})$,
    where $0 < \theta^{*} < \delta$.
    From our choice of $\theta^{*}$,
    we can see that $(M'')_{ii} > (M'')_{jk}$
    for all $i \in [q]$, and $j \neq k \in [q]$.
    Moreover, since $(M'')$ is
     diagonal distinct, we also see that
    $(M'')_{ii} \neq (M'')_{jj}$ for any $i \neq j$.
    Therefore, given any $i \neq j$, we see that
    $(M'')_{ii} \neq (M'')_{jk}$ for all $k \in [q]$.
    In particular, this implies that there does
    not exist any $\sigma \in S_{q}$ such that
    $((M'')_{i1}, \dots, (M'')_{iq}) =
    ((M'')_{j\sigma(1)}, \dots, (M'')_{j\sigma(q)})$.
    This proves that $\rho_{\tt{dist}}
    (HD^{\theta^{*}}H^{\tt{T}}) \neq 0$.
    Now, \cref{corollary: stretchingWorksPositive}
    allows us to find the required
    $N \in \mathfrak{R}(M', \mathcal{F}\cup \{\rho_{\tt{full}},\rho_{\tt{dist}}\})
    \cap \text{Sym}_{q}^{\tt{pd}}(\mathbb{R}_{> 0}) 
    \subseteq \mathfrak{R}(M, \mathcal{F}\cup \{\rho_{\tt{full}},\rho_{\tt{dist}}\})
    \cap \text{Sym}_{q}^{\tt{pd}}(\mathbb{R}_{> 0})$.
\end{proof}

Informally,  \cref{lemma: diagDistinctImpliesRowFull}  and
\cref{lemma: diagDistinctImpliesOrderDist} together
imply the following:
$$\text{diagonal distinct} \implies \text{row full } + 
\text{ p.o.\,distinct}.$$
Together with \cref{lemma: pairwiseDistinctImpliesDiagonal},
this means that 
$$\text{diagonal distinct} \iff \text{row full } + 
\text{ p.o.\,distinct}.$$
\cref{lemma: pairwiseDistHardness} now lets us
immediately prove the following.

\begin{corollary}\label{corollary: diagDistHardness}
    Let $M \in \text{Sym}_{q}^{\tt{pd}}(\mathbb{R}_{> 0})$
    be a  diagonal distinct matrix.
    If $\mathbf{x}$ is confluent for all
    $\mathbf{0} \neq \mathbf{x} \in \chi_{q}$ such that
    $\Psi_{\mathbf{x}}(M) = 0$, then
    $\PlEVAL(M)$ is $\#$P-hard.
\end{corollary}

Our next goal will be to prove the following:
$$\text{diagonal distinct} \iff \text{ p.o.\,distinct}.$$
We note that 
\cref{lemma: diagDistinctImpliesOrderDist} already
implies that
$$\text{diagonal distinct} \implies \text{ p.o.\,distinct}.$$
So, we only need to show that the converse is also true.

\begin{lemma}\label{lemma: OrderedDistImpliesDiagDist}
    Fix any $i \ne j \in [q]$. Let $M \in \text{Sym}_{q}^{\tt{pd}}(\mathbb{R}_{> 0})$
    be a matrix such that rows $i$ and $j$ are
    not order identical.
    Let $\mathcal{F}$ be a countable set of
    $\text{Sym}_{q}(\mathbb{R})$-polynomials
    such that $F(M) \neq 0$ for all
    $F \in \mathcal{F}$.
    Then there exists some $N \in 
    \mathfrak{R}(M, \mathcal{F}) \cap
    \text{Sym}_{q}^{\tt{pd}}(\mathbb{R}_{> 0})$,
    such that $N_{ii} \neq N_{jj}$.
\end{lemma}
\begin{proof}
    We can assume without loss of generality that
    $i = 1$ and $j = 2$.
    We will define $\xi: N \mapsto N_{11} - N_{22}$.
    If there is some $\theta$ such that
    $\xi(\mathcal{S}_{M}(\theta)) \neq 0$,
    then we are done by \cref{corollary: stretchingWorksPositive}.
    So we may assume that
    $\xi(\mathcal{S}_{M}(\theta)) = 0$
    for all $\theta \in \mathbb{R}$.
    We also recall from \cref{lemma: fullPairwiseDistTest}
    that there exists a $\text{Sym}_{q}(\mathbb{R})$-polynomial
    $\rho_{{\tt{dist}}, 12}$
    such that $\rho_{{\tt{dist}}, 12}(N) \neq 0$ on any matrix $N$
    if and only if rows $1$ and $2$ of $N$ are not
    order identical.    
    We will now define $\zeta: N \mapsto
    \rho_{{\tt{dist}}, 12}(T_{4}N)$.
    Since $\rho_{{\tt{dist}}, 12}(M) \neq 0$,
    and $M \in \text{Sym}_{q}(\mathbb{R}_{> 0})$,
    it follows that $\zeta(M) =
    \rho_{{\tt{dist}}, 12}(T_{4}M) \neq 0$ as well.
    
    We will now let $\mathcal{F}' =
    \{F': N \mapsto F(N^{3})\ |\ F \in \mathcal{F}\}$.
    We see that $F'(\mathcal{S}_{M}(\frac{1}{3})) = F(\mathcal{S}_{M}(1)) = F(M)
    \neq 0$ for all $F' \in \mathcal{F}'$.
    We also know that $\zeta(\mathcal{S}_{M}(1))
    = \zeta(M) \neq 0$.
    Since $\mathcal{S}_{M}(0) = I$, we
    know that there exists some $\delta > 0$
    such that $|\mathcal{S}_{M}(\theta)_{ij} - I_{ij}|
    < \frac{1}{3}$ for all $0 < \theta < \delta$,
    for all $i, j \in [q]$.
    So, we can use \cref{lemma: stretchingWorks}
    to find
    $M' = \mathcal{S}_{M}(\theta^{*}) \in
    \mathfrak{R}(M, \mathcal{F}' \cup \{\zeta\})
    \cap \text{Sym}_{q}^{\tt{pd}}(\mathbb{R}_{\neq 0})$
    where $0 < \theta^{*} < \delta$.

    We know that $\xi(M') = 0$.
    So, we see that $(M')_{11} = (M')_{22}$.
    We will now consider $R_{n}(T_{m}M')$ for
    $n \geq  1$ and $m \geq 1$.
    Let us first assume that $\xi(R_{n}(T_{m}M')) = 0$
    for all $n, m \geq 1$.
    We will now define an equivalence relation
    on $[q]$ such that $i \sim j$ if and only if
    $(R_{n}(T_{m}M'))_{ii} = (R_{n}(T_{m}M'))_{jj}$
    for all $n, m \geq 1$.
    By our assumption, we see that $1 \sim 2$.
    This equivalence relation defines a partition
    $[q] = S_{1} \sqcup \cdots \sqcup S_{r}$
    for some $r < q$.
    We may assume without loss of generality that
    $1, 2 \in S_{1}$.
    We will now fix some \emph{odd} $m \geq 1$, and consider
    $\xi(R_{n}(T_{m}M'))$ for all $n \geq 1$.
    We know that
    $$\xi(R_{n}(T_{m}M')) =
    \sum_{a, b \in [q]}(M'_{ab})^{nm} \cdot 
    ((M'_{1a})^{m}(M'_{1b})^{m} - 
    (M'_{2a})^{m}(M'_{2b})^{m}).$$
    We now let $X = \{(M'_{ab}): a, b \in [q]\}$,
    and
    $$c_{m, 11}(x) = \sum_{a, b \in [q]: M'_{ab} = x}
    (M'_{1a}M'_{1b})^{m}, ~~\text{ and }~~
    c_{m, 22}(x) = \sum_{a, b \in [q]: M'_{ab} = x}
    (M'_{2a}M'_{2b})^{m}.$$
    So, we see that
    $$\xi(R_{n}(T_{m}M')) = 
    \sum_{x \in X}x^{mn} \cdot (c_{m, 11}(x)
    - c_{m, 22}(x)).$$
    Since $m$ is odd by our choice, we see that
    $x^{m} = (x')^{m}$ for $x, x' \in X$ implies that
    $x = x'$.
    So, if $\xi(R_{n}(T_{m}M')) = 0$ for all $n \geq 1$,
    this forms a full rank Vandermonde system
    of linear equations of size $O(1)$.
    This implies that $c_{m, 11}(x)
    - c_{m, 22}(x) = 0$ for all $x \in X$.
    
    Now, we note that by our construction of $M'$,
    we ensured that every diagonal entry is 
    greater than the absolute values of all non-diagonal entries.
    Now, given any $i \in [q]$, we know that
    since $[q] = S_{1} \sqcup \cdots S_{r}$,
    there exists some $t \in [r]$ such that
    $i \in S_{t}$.
    We now see that
    $$c_{m, 11}((M')_{ii}) - c_{m, 22}((M')_{ii}) = 
    \sum_{j \in S_{t}} \left((M'_{1j})^{2m} -
    (M'_{2j})^{2m} \right).$$
    Since we can do this for all $i \in [q]$,
    we see that for all $t \in [r]$,
    $$\sum_{j \in S_{t}} (M'_{1j})^{2m} 
    = \sum_{j \in S_{t}} (M'_{2j})^{2m}.$$
    So,
    $$\sum_{j \in [q]} (M'_{1j})^{2m} = 
    \sum_{t \in [r]}\sum_{j \in S_{t}}(M'_{1j})^{2m}
    = \sum_{t \in [r]}\sum_{j \in S_{t}}(M'_{2j})^{2m}
    = \sum_{j \in [q]} (M'_{2j})^{2m}.$$

    Since this is true for all odd $m = 2k-1$ with $k  \geq 1$, we have
    a Vandermonde system of equations of the form
    $$\begin{pmatrix}
        (M'_{11})^{2} & (M'_{12})^{2} & \cdots & (M'_{1q})^{2}
        & (M'_{21})^{2} & (M'_{22})^{2} & \cdots & (M'_{2q})^{2}\\
        (M'_{11})^{6} & (M'_{12})^{6} & \cdots & (M'_{1q})^{6}
        & (M'_{12})^{6} & (M'_{22})^{6} & \cdots & (M'_{2q})^{6}\\
        \vdots & \vdots & \ddots &\vdots
        & \vdots & \vdots & \ddots & \vdots\\
        (M'_{11})^{4k + 2} & (M'_{12})^{4k + 2} & \cdots & (M'_{1q})^{4k + 2}
        & (M'_{12})^{4k + 2} & (M'_{22})^{4k + 2} & \cdots & (M'_{2q})^{4k + 2}\\
    \end{pmatrix} \times \begin{pmatrix}
        1\\
        1\\
        \vdots\\
        1\\
        -1\\
        -1\\
        \vdots\\
        -1
    \end{pmatrix} = \begin{pmatrix}
        0\\
        \vdots\\
        0
    \end{pmatrix}.$$
    This is only possible if the following condition is satisfied:
    For any $v$, let $n_v = | \{i\ | \ |M'_{1i}| =v \}|$ 
    and $n'_v =  | \{i\ | \ |M'_{2i}| =v \}|$. Then $n_v = n'_v$ for all $v$.  
    This can be seen by first 
    ordering the elements of the following multisets by magnitude 
    \[\{|M'_{11}|,  |M'_{12}|,  \ldots  |M'_{1q}|\}, ~~\mbox{ and }~~
    \{|M'_{21}|,  |M'_{22}|,  \ldots  |M'_{2q}|\},\]
    and then taking a sufficiently large $k$.
    Thus  the entries of $|M'_{1i}|$ and $|M'_{2i}|$ 
    can be matched in a 1-1 correspondence.
    Hence, there  exists some
    $\sigma \in S_{q}$ such that
    $(M'_{1i})^{4} = (M'_{2\sigma(i)})^{4}$ for
    all $i \in [q]$.
    But from our construction of $M'$, we ensured
    that $\zeta(M') \neq 0$,
    where $\zeta$ was defined such that
    $\zeta: N \mapsto \rho_{{\tt{dist}}, 12}(T_{4}N)$.
    This contradiction therefore implies that
    it is not possible that
    $\xi(R_{n}(T_{m}(M'))) = 0$ for all $n, m \geq 1$.

    So, we may assume that there exists 
    some $n \geq 1$ and $m \geq 1$
    such that $\xi(R_{n}(T_{m}M')) \neq 0$.
    We can now define $\xi': N \mapsto 
    \xi(R_{n}(T_{m}N))$,
    and we see that
    $\xi'(\mathcal{S}_{M}(\theta^{*})) \neq 0$.
    So, \cref{corollary: stretchingWorksPositive}
    allows us to find some
    $M'' \in \mathfrak{R}(M, \mathcal{F}'
    \cup \{\xi'\}) \cap
    \text{Sym}_{q}^{\tt{pd}}(\mathbb{R}_{> 0})$.
    We can assume that the entries of $M''$
    are generated by some $\{g''_{t}\}_{t \in [d]}$.
    We also know that using \cref{lemma: MequivalentCM},
    we can replace $M''$ with some $c  M''$.
    If we now let $\xi'': N \mapsto
    (R_{n}(N))_{11} - (R_{n}(N))_{22}$,
    we can see that since $M'' \in
    \text{Sym}_{q}(\mathbb{R}_{> 0})$,
    $\xi''(\mathcal{T}_{M''}((g''_{1})^{m},
    \dots, (g''_{d})^{m})) = \xi''(T_{m}M'') 
    = \xi'(M'') \neq 0$.
    So, \cref{lemma: thickeningWorks} allows
    us to find some
    $M''' \in \mathfrak{R}(M'', \mathcal{F}'
    \cup \{\xi''\}) \cap
    \text{Sym}_{q}^{\tt{pd}}(\mathbb{R}_{> 0})$.
    We note that since $M''' \in 
    \text{Sym}_{q}(\mathbb{R}_{> 0})$,
    $\xi(\mathcal{R}_{M'''}(n)) = 
    \xi(R_{n}(M''')) = \xi''(M''') \neq 0$.
    Also, $\mathcal{R}_{M'''}(1) = (M''')^3$, and so, for all $F \in \mathcal{F}$,
    $F(\mathcal{R}_{M'''}(1)) = F((M''')^3) = F'(M''') \ne 0$, where $F' \in \mathcal{F}'$
    is the corresponding function to $F$.
    Finally, \cref{lemma: RInterpolationWorks}
    allows us to find the required $N \in 
    \mathfrak{R}(M''', \mathcal{F} \cup 
    \{\xi\}) \cap
    \text{Sym}_{q}^{\tt{pd}}(\mathbb{R}_{> 0})$.
\end{proof}

We can repeatedly apply \cref{lemma: OrderedDistImpliesDiagDist} to all distinct
pairs $i \ne j \in [q]$, while each time incorporating an additional polynomial $\xi$
in the set $\mathcal{F}$ that represents the last pair that is ensured to be not order identical. 
More formally, we can prove the following corollary.

\begin{corollary}\label{corollary: ordDistDiagDistReduct}
    Let $M \in \text{Sym}_{q}^{\tt{pd}}(\mathbb{R}_{> 0})$
    be a p.o.\,distinct matrix, that
    is not necessarily row full.
    Let $\mathcal{F}$ be a countable set of
    $\text{Sym}_{q}(\mathbb{R})$-polynomials
    such that $F(M) \neq 0$ for all $F \in \mathcal{F}$.
    There exists some  diagonal distinct
    $N \in \mathfrak{R}(M, \mathcal{F}) \cap
    \text{Sym}_{q}^{\tt{pd}}(\mathbb{R}_{> 0})$.
\end{corollary}
\begin{proof}
    We can define the $\text{Sym}_{q}(\mathbb{R})$-polynomial 
    $\rho_{{\tt{diag}}, ij}$ such that
    $\rho_{{\tt{diag}}, ij}(N) : N \mapsto N_{ii} - N_{jj}$.
    Since $M$ is p.o.\,distinct,
    \cref{lemma: OrderedDistImpliesDiagDist}
    allows us to find
    $M^{(12)} \in \mathfrak{R}(M, \mathcal{F} \cup
    \{\rho_{\tt{dist}}, \rho_{{\tt{diag}}, 12}\}) \cap
    \text{Sym}_{q}^{\tt{pd}}(\mathbb{R}_{> 0})$.
    We can repeat this process again with $M^{(12)}$
    to find $M^{(13)} \in \mathfrak{R}(M^{(12)}, \mathcal{F}
    \cup \{\rho_{\tt{dist}}, \rho_{{\tt{diag}}, 12},
    \rho_{{\tt{diag}}, 13}\}) \cap
    \text{Sym}_{q}^{\tt{pd}}(\mathbb{R}_{> 0})$.

    After repeating this for all $i \neq j \in [q]$,
    we obtain $N \in \mathfrak{R}(M, \mathcal{F} \cup
    \{\rho_{\tt{dist}}\} \cup
    \{\rho_{{\tt{diag}}, ij}: i \neq j \in [q] \}) \cap
    \text{Sym}_{q}^{\tt{pd}}(\mathbb{R}_{> 0})$.
    So, we see that
    $N \in \mathfrak{R}(M, \mathcal{F}_{M} \cup
    \{\rho_{\tt{diag}}\}) \cap
    \text{Sym}_{q}^{\tt{pd}}(\mathbb{R}_{> 0})$
    is the required matrix.
\end{proof}

This finishes our proof that
$$\text{row full } + \text{ p.o.\,distinct} 
\iff \text{diagonal distinct} 
\iff \text{p.o.\,distinct}.$$
We can therefore also prove the following
corollary.

\begin{corollary}\label{corollary: ordDistHardness}
    Let $M \in \text{Sym}_{q}^{\tt{pd}}(\mathbb{R}_{> 0})$
    be a p.o.\,distinct matrix.
    If $\mathbf{x}$ is confluent for all
    $\mathbf{0} \neq \mathbf{x} \in \chi_{q}$ such that
    $\Psi_{\mathbf{x}}(M) = 0$, then
    $\PlEVAL(M)$ is $\#$P-hard.
\end{corollary}
\begin{proof}
    \cref{corollary: ordDistDiagDistReduct} implies that
    there exists some  diagonal distinct
    $N \in \mathfrak{R}(M, \mathcal{F}_{M})
    \cap \text{Sym}_{q}^{\tt{pd}}(\mathbb{R}_{> 0})$.
    \cref{corollary: diagDistHardness} then proves that
    $\PlEVAL(N)$ is $\#$P-hard, which then implies that
    $\PlEVAL(M)$ is also $\#$P-hard.
\end{proof}
\section[Hardness for Diagonal Distinct 4 x 4 matrices]{Hardness for Diagonal Distinct $4 \times 4$ matrices}\label{sec: 4x4DiagDist}

We are now ready to use all this quite elaborate
machinery that we have built to prove a complexity
dichotomy in the case where $q = 4$.

\begin{lemma}\label{lemma: almostAllConfluent}
    Let $\lambda_{1}, \dots, \lambda_{4} > 0$.
    If $\overline{\mathcal{L}(\lambda_{1}, \dots, \lambda_{4})}$
    contains any $ \mathbf{x}  \neq \mathbf{0} $ that is non-confluent, then
    such an $ \mathbf{x} $ must satisfy  $|x_{1}| = |x_{2}| = |x_{3}| = |x_{4}|$,
    and in this case,  $(1, 1, -1, -1) \in
    \overline{\mathcal{L}(\lambda_{1}, \dots, \lambda_{4})}$.
\end{lemma}
\begin{proof}
    By definition $\mathbf{0} \neq \mathbf{x} 
    \in \chi_{q}$, and $\sum_{i} x_i =0$.
    Let $I^{+} = \{i \in [4]: x_{i} > 0\}$,
    and $I^{-} = \{i \in [4]: x_{i} < 0 \}$.
    Since $\mathbf{x} \neq \mathbf{0}$, we know
    that both $I^{+}, I^{-} \neq \emptyset$.
    If $|I^{+}| = 1$ \emph{or} $|I^{-}| = 1$, $\mathbf{x}$
    must be confluent, by  \cref{lemma: confluenceEquivalence}, since the only choice for nonempty
    $S \subseteq I^{+}$ and $T \subseteq I^{-}$ such that
    $\sum_{i \in S}x_{i} = \sum_{j \in T}(-x_{j})$
    is in fact $S = I^{+}$ and $T = I^{-}$.
    This leaves us to consider as the only non-trivial
    case, the scenario where $|I^{+}| = |I^{-}| = 2$.
    We note that if $\mathbf{x} \in 
    \overline{\mathcal{L}
    (\lambda_{1}, \dots, \lambda_{4})}$,
    then $\mathbf{y} \in \overline{\mathcal{L}
    (\lambda_{1}, \dots, \lambda_{4})}$ as well
    for any $\mathbf{y}$ such that $y_{i} = x_{\sigma(i)}$
    for some $\sigma \in S_{4}$.
    So, we may assume without loss of generality
    that $I^{+} = \{1, 4\}$, and $I^{-} = \{2, 3\}$.
    
    Since $\mathbf{x}$ is not confluent, there must be
    $S_{1}, S_{2} \subseteq \{1, 4\}$, and
    $T_{1}, T_{2} \subseteq \{2, 3\}$ such that
    $$\sum_{i \in S_{1}}x_{i} = \sum_{j \in T_{1}}(-x_{j})
    ~\text{ and }~
    \sum_{i \in S_{2}}x_{i} = \sum_{j \in T_{2}}(-x_{j})
    ~~\text{ but }~~
    \sum_{i \in S_{1} \cap S_{2}}x_{i} \neq 
    \sum_{j \in T_{1} \cap T_{2}}(-x_{j}).$$
    Since all $x_i \ne 0$, if one of $S_1, S_2, T_1$ or $ T_2 = \emptyset$,
    so is the corresponding set, violating the last inequality.
    So, we may rule out any empty sets, and also the case where $S_{1} = S_{2} = \{1, 4\}$
    and the case where  $T_{1} = T_{2} = \{2, 3\}$.
    Without loss of generality, 
    $S_{1} \neq \{1, 4\}$,
    and we may assume that
    $S_{1} = \{1\}$.
    This implies that $T_{1} \neq \{2, 3\}$ since
    we know that $x_{1} < x_{1} + x_{4} = (-x_{2}) + (-x_{3})$.
    Once again, without loss of generality, we may
    assume that $T_{1} = \{2\}$.
    
    Now, if  $S_{2} = \{1, 4\}$,
    that would force $T_{2} = \{2, 3\}$
    since $x_{1} + x_{4} > (-x_{2})$ and
    $x_{1} + x_{4} > (-x_{3})$.
    But then, $S_{1} \cap S_{2} = S_{1}$ and 
    $T_{1} \cap T_{2} = T_{1}$, which implies that
    $\sum_{i \in S_{1} \cap S_{2}}x_{i} =
    \sum_{j \in T_{1} \cap T_{2}}(-x_{j})$, a contradiction.
    Therefore, 
    $S_{2} \not = \{1, 4\}$. So the only possibilities are $S_2 = \{1\}$ or  $S_2 = \{4\}$.
    By symmetry, the only  possibilities for $T_2$ are $T_2 = \{2\}$ or  $T_2 = \{3\}$.

    If $S_2 = \{1\}$, then we claim $T_2 = \{3\}$, for otherwise  $T_2 = \{2\}$ and that
    would lead to an equality $\sum_{i \in S_{1} \cap S_{2}}x_{i} = x_1 = -x_2 =
    \sum_{j \in T_{1} \cap T_{2}}(-x_{j})$, a contradiction. 
    Then we have $x_1 = -x_2$ and $x_1 = -x_3$ which leads to $x_4 = (x_1 + x_4) - x_1
    = (-x_2 + -x_3) - x_1 = x_1$.
    This proves that $|x_{1}| = |x_{2}| = |x_{3}| = |x_{4}|$,
    and also that
    $(x, -x, -x, x) \in \mathcal{L}(\lambda_{1}, \dots, 
    \lambda_{4})^{\sigma}$ for some
    $x \in \mathbb{Z}_{> 0}$, and $\sigma \in S_{4}$.
    This means that
    $(\lambda_{\sigma(1)})^{x}(\lambda_{\sigma(2)})^{-x}
    (\lambda_{\sigma(3)})^{-x}(\lambda_{\sigma(4)})^{x} = 1$.
    Since $\lambda_{1}, \dots, \lambda_{4} > 0$,
    this implies that $(\lambda_{\sigma(1)})(\lambda_{\sigma(2)})^{-1}
    (\lambda_{\sigma(3)})^{-1}(\lambda_{\sigma(4)}) = 1$.
    Therefore, $(1, -1, -1, 1) 
    \in \overline{\mathcal{L}(\lambda_{1}, \dots, \lambda_{4})}$.

    Finally, if $S_2 = \{4\}$, then $S_1 \cap S_2 = \emptyset$.
    Then $T_2 = \{2\}$, for otherwise $T_2 = \{3\}$, would give  $T_1 \cap T_2 = \emptyset$
    and then we would have an equality $\sum_{i \in S_{1} \cap S_{2}}x_{i} = 0 =
    \sum_{j \in T_{1} \cap T_{2}}(-x_{j})$, a contradiction. 
    Then we have  $x_1 = -x_2$ and $x_4 = -x_2$, which also leads to $-x_3 =   x_1$ as above.
    Hence, even in this case, we find that
    $|x_{1}| = |x_{2}| = |x_{3}| = |x_{4}|$, and that
    $(x, -x, -x, x)$ in $\overline{\mathcal{L}(\lambda_{1}, \dots, \lambda_{4})}$, for some $x \in \mathbb{Z}_{> 0}$.
    This once again implies that $(1, -1, -1, 1) 
    \in \overline{\mathcal{L}(\lambda_{1}, \dots, \lambda_{4})}$. 
\end{proof}

Our goal will be to now show that given a
diagonal distinct matrix
$M \in \text{Sym}_{4}^{\tt{pd}}(\mathbb{R}_{> 0})$,
we can find some $N \in \mathfrak{R}(M, \mathcal{F}_{M})$
with eigenvalues $(\mu_{1}, \dots, \mu_{4})$
such that $(1, -1, -1, 1) \notin
\overline{\mathcal{L}(\mu_{1}, \dots, \mu_{4})}$.
\cref{corollary: diagDistHardness} then
implies that $\PlEVAL(N)$ is $\#$P-hard,
which means that $\PlEVAL(M)$ is also
$\#$P-hard.
On the other hand, if we are unable to find
such a matrix $N$, then we will show that
$M$ is isomorphic to $A \otimes B$
for some $A, B \in \text{Sym}_{2}(\mathbb{R}_{> 0})$.
Here for any  $\sigma \in S_{4}$  we define the matrix
$M^{\sigma}$ to be 
    such that $(M^{\sigma})_{ij} = M_{\sigma(i)\sigma(j)}$.
    We say that  $M^{\sigma}$ is isomorphic to $M$.

\begin{lemma}\label{lemma: tensorPolynomial}
    There exist
    $\text{Sym}_{4}(\mathbb{R})$-polynomials,
    $\varrho_{\tt{tensor}}$ and
    $\rho_{\tt{tensor}}$ such that given any
    $M \in \text{Sym}_{4}(\mathbb{R}_{\neq 0})$,
    \begin{itemize}
    \item
    $\varrho_{\tt{tensor}}(M) = 0$ if and only if
    $M = A \otimes B$ for some $A, B \in
    \text{Sym}_{2}(\mathbb{R})$, and
    \item
    $\rho_{\tt{tensor}}(M) = 0$ if and only if
    $M$ is isomorphic 
    to $A \otimes B$ for some
    $A, B \in \text{Sym}_{2}(\mathbb{R})$.
    \end{itemize}
\end{lemma}
\begin{proof}
    We will first define $\varrho_{\tt{tensor}}$ as
    $$\varrho_{\tt{tensor}}(N) = (N_{14} - N_{23})^{4} +
    (N_{11}N_{44} - N_{22}N_{33})^{2} +
    \sum_{i \in [4]}(N_{i1}N_{i4} - N_{i2}N_{i3})^{2}.$$
    We can see that $\varrho_{\tt{tensor}}$ is a
    homogeneous polynomial in the entries of the matrix.
    So, it is a $\text{Sym}_{4}(\mathbb{R})$-polynomial.
    Note that if we index the rows and columns not from $\{1, 2, 3, 4\}$ but rather
    from $D= \{00, 01, 10, 11\}$, the polynomial $\varrho_{\tt{tensor}}(N)$ takes the form
 {\small   
 $$ \varrho_{\tt{tensor}}(N) = (N_{00, 11} - N_{01, 10})^{4} +
    (N_{00, 00}N_{11, 11} - N_{01, 01}N_{10, 10})^{2} +
    \sum_{ab \in D}(N_{ab,00}N_{ab,11} - N_{ab,01}N_{ab,10})^{2}. $$
    }
    From this form, it is clear that $\varrho_{\tt{tensor}}(N)$ is invariant if we flip the bit
    $a$ or $b$ in the matrix $N$.
    
    We see that $\varrho_{\tt{tensor}}(M) = 0$
    if and only if $M_{14} = M_{23}$, 
    $M_{11}M_{44} - M_{22}M_{33}$, and
    $$M_{11}M_{14} = M_{12}M_{13},\quad 
    M_{12}M_{24} = M_{22}M_{23},\quad 
    M_{13}M_{34} = M_{23}M_{33},\quad 
    M_{14}M_{44} = M_{24}M_{34}.$$

    Let $M \in \text{Sym}_{4}(\mathbb{R}_{\neq 0})$.
    We will now let $x = \frac{M_{13}}{M_{11}}$,
    and $y = \frac{M_{33}}{M_{11}}$. (As we assume the entries of $M$ are non-zero, the divisions are well defined.)
    Since $M_{11}M_{14} = M_{12}M_{13}$,
    this means that $M_{14} = x \cdot M_{12}$.
    Similarly, since $M_{11}M_{44} = M_{22}M_{33}$,
    we see that $M_{44} = y \cdot M_{22}$.
    Since (by symmetry) $M_{23} = M_{14} = x \cdot M_{12}$,
    and $M_{12}M_{24} = M_{22}M_{23}$,
    we see that $M_{24} = x \cdot M_{22}$.
    Since $M_{33} = y \cdot M_{11}$, and
    $M_{13}M_{34} = M_{23}M_{33}$,
    we see that $M_{34} =
    \frac{xy \cdot M_{11}M_{12}}{x \cdot M_{11}}
    = y \cdot M_{12}$.
    Putting everything together, we see that
    $$M = \begin{pmatrix}
        M_{11} & M_{12} & xM_{11} & xM_{12}\\
        M_{12} & M_{22} & xM_{12} & xM_{22}\\
        xM_{11} & xM_{12} & yM_{11} & yM_{12}\\
        xM_{12} & xM_{22} & yM_{12} & yM_{22}\\
    \end{pmatrix} = \begin{pmatrix}
        1 & x\\
        x & y\\
    \end{pmatrix} \otimes \begin{pmatrix}
        M_{11} & M_{12}\\
        M_{12} & M_{22}\\
    \end{pmatrix}.$$
    So, we see that $\varrho_{\tt{tensor}}(M) = 0$
    implies that $M = A \otimes B$ for some
    $A, B \in \text{Sym}_{2}(\mathbb{R})$.

    Now, let us assume that $M = A \otimes B$ for some
    $A, B \in \text{Sym}_{2}(\mathbb{R})$.
    So, we see that
    $$M = \begin{pmatrix}
        A_{11}B_{11} & A_{11}B_{12} &
        A_{12}B_{11} & A_{12}B_{12}\\
        A_{11}B_{12} & A_{11}B_{22} &
        A_{12}B_{12} & A_{12}B_{22}\\
        A_{12}B_{11} & A_{12}B_{12} & 
        A_{22}B_{11} & A_{22}B_{12}\\
        A_{12}B_{12} & A_{12}B_{22} &
        A_{22}B_{12} & A_{22}B_{22}\\
    \end{pmatrix}$$
    But now we can verify that
    $\varrho_{\tt{tensor}}(M) = 0$.
    This proves that $\varrho_{\tt{tensor}}(M) = 0$
    if and only if $M = A \otimes B$ for some
    $A, B \in \text{Sym}_{2}(\mathbb{R})$.

    Given any $M \in \text{Sym}_{4}(\mathbb{R}_{\neq 0})$,
    and $\sigma\in S_{4}$,
    we can define $M^{\sigma}$ to be the matrix
    such that $(M^{\sigma})_{ij} = M_{\sigma(i)\sigma(j)}$.
    Matrices isomorphic to $M$ take the form $M^{\sigma}$ under a simultaneous row
    and column permutation by some $\sigma$.
    Now, we can define
    $$\rho_{\tt{tensor}}(M) = \prod_{\sigma \in S_{4}}
    \varrho_{\tt{tensor}}(M^{\sigma}).$$
    We can see that for each $\sigma \in S_{4}$,
    $\varrho_{\tt{tensor}}(M^{\sigma})$ is a 
    homogeneous polynomial in the entries of $M$.
    So, $\rho_{\tt{tensor}}$ is a 
    $\text{Sym}_{q}(\mathbb{R})$-polynomial,
    such that $\rho_{\tt{tensor}}(M) = 0$ implies
    that $M$ is isomorphic to $A \otimes B$
    for some $A, B \in \text{Sym}_{2}(\mathbb{R})$.
    Similarly, if $M$ is isomorphic to some
    $A \otimes B$, it follows that there exists some
    $\sigma \in S_{4}$ such that $M^{\sigma} = A \otimes B$.
    This implies that $\rho_{\tt{tensor}}(M) = 0$.
\end{proof}

\begin{remark*}
    For  $M \in \text{Sym}_{4}(\mathbb{R}_{> 0})$,
    being expressible as $M = A \otimes B$ for some $A, B \in
    \text{Sym}_{2}(\mathbb{R})$ is equivalent to  being 
    expressible as $M = A \otimes B$ for some $A, B \in
    \text{Sym}_{2}(\mathbb{R}_{> 0})$. This is because if
    $A$ or $B$ have a zero entry, there would be a zero entry in $M$
    as well, and if $A$ and $B$ have any entries $a$ and $b$ of
    the opposite signs, then $ab < 0$ would be an entry of $M$. Then all entries of $A$
    have the same sign, as well as that of $B$, and their signs are the same. Finally, if they
    are both $-$, then replace $A$ and $B$ by $-A$ and $-B$.
\end{remark*}

We will now show that given a diagonal distinct matrix
$M \in \text{Sym}_{4}^{\tt{pd}}(\mathbb{R}_{> 0})$,
if we cannot find some $N \in \mathfrak{R}(M, \mathcal{F}_{M})$
such that $\Psi_{(1, -1, -1, 1)}(N) \neq 0$,
then $M$ needs to satisfy more and more
conditions, until we are able to prove
that $\rho_{\tt{tensor}}(M) = 0$.
We will start with a lemma that is applicable for all
$M \in \text{Sym}_{4}^{\tt{pd}}(\mathbb{R}_{> 0})$,
and not just diagonally distinct matrices.

\begin{lemma}\label{lemma: notConfluentImpliesM14M23}
    Let $M \in \text{Sym}_{4}^{\tt{pd}}(\mathbb{R}_{> 0})$.
    Let $\mathcal{F}$ be a countable set of 
    $\text{Sym}_{4}(\mathbb{R})$-polynomials
    such that $F(M) \neq 0$ for all
    $F \in \mathcal{F}$.
    Let $\xi$ be the $\text{Sym}_{4}(\mathbb{R})$-polynomial
    such that $\xi: N \mapsto \phi_{(1, -1, -1, 1)}
    (N_{11}, \dots, N_{44})$.
    Then there exists some $N \in \mathfrak{R}(M,
    \mathcal{F} \cup \{\xi\}) \cap
    \text{Sym}_{4}^{\tt{pd}}(\mathbb{R}_{> 0})$,
    unless $M_{11} \cdot M_{44} = M_{22} \cdot M_{33}$,
    and $M_{14} = M_{23}$.
\end{lemma}
\begin{proof}
    If $\xi(M) \neq 0$, then we are done already,
    so we may assume that $\xi(M)= 0$.
    This immediately implies that
    $M_{11} \cdot M_{44} = M_{22} \cdot M_{33}$.
    So, our goal now is just to show that if
    no $N \in \mathfrak{R}(M,
    \mathcal{F} \cup \{\xi\}) \cap
    \text{Sym}_{4}^{\tt{pd}}(\mathbb{R}_{> 0})$
    can be found, then $M_{14} = M_{23}$.
    
    Let the entries of $M$ be generated by some
    $\{g_{t}\}_{t \in [d]}$.
    We know from \cref{lemma: MequivalentCM} that
    we can replace $M$ with some $c \cdot M$
    such that the entries of $M$ are generated
    by $\{g_{t}\}_{t \in [d]}$ with $e_{ijt} \geq 0$
    for all $i, j \in [4]$, and $t \in [d]$.
    Since $M \in \text{Sym}_{4}(\mathbb{R}_{> 0})$,
    we also know that $e_{ij0} = 0$ for all
    $i, j \in [4]$.
    We will now let $\mathfrak{m} = \max_{i, j \in [4],
    t \in [d]}(e_{ijt})$, and pick some
    $m > \mathfrak{m}$.
    We then let
    $$z_{ij} = \sum_{t \in [d]}m^{t}e_{ijt}$$
    for all $i, j \in [4]$.
    We will now define
    $\mathcal{T}_{M}^{*}: \mathbb{R} \rightarrow 
    \text{Sym}_{4}(\mathbb{R})$ such that
    $$\mathcal{T}_{M}^{*}(p)_{ij}
    = \mathcal{T}_{M}(p^{m}, \dots, p^{m^{d}})_{ij}
    = p^{z_{ij}}$$
    for all $i, j \in [4]$.
    We will now define the 
    $\text{Sym}_{4}(\mathbb{R})$-polynomial $\xi_{2}$
    such that $\xi_{2}: N \mapsto \xi(N^{2})$.
    
    If $\xi(\mathcal{T}_{M}^{*}(p)) \neq 0$
    for some $p \in \mathbb{R}$, then
    \cref{lemma: thickeningWorks} allows
    us to find some $N \in \mathfrak{R}(
    M, \mathcal{F} \cup \{\xi\}) \cap
    \text{Sym}_{4}^{\tt{pd}}(\mathbb{R}_{> 0})$,
    and we will be done.
    Similarly, if $\xi_{2}(\mathcal{T}_{M}^{*}(p)) \neq 0$
    for some $p \in \mathbb{R}$, we can first use
    \cref{lemma: thickeningWorks} allows
    us to find some $M' \in \mathfrak{R}(M,
    \mathcal{F} \cup \{\xi_{2}\}) \cap
    \text{Sym}_{4}^{\tt{pd}}(\mathbb{R}_{> 0})$.
    Then, since $F(M') \neq 0$ for all $F \in \mathcal{F}$,
    and $\xi((M')^{2}) \neq 0$,
    \cref{corollary: stretchingWorksPositive}
    allows us to once again find the required $N \in \mathfrak{R}(
    M'', \mathcal{F} \cup \{\xi\}) \cap
    \text{Sym}_{4}^{\tt{pd}}(\mathbb{R}_{> 0})$.

    On the other hand, let us assume that
    $\xi(\mathcal{T}_{M}^{*}(p)) = 0$, and
    $\xi_{2}(\mathcal{T}_{M}^{*}(p)) = 0$
    for all $p \in \mathbb{R}$.
    We will now define
    $\zeta, \zeta_{2}: \mathbb{R} \rightarrow \mathbb{R}$
    such that $\zeta(p) = \xi(\mathcal{T}_{M}^{*}(p))$, and
    $\zeta_{2}(p) = \xi_{2}(\mathcal{T}_{M}^{*}(p))$ for all $p \in \mathbb{R}$.
    We note that
    $$\zeta(p) = \mathcal{T}_{M}^{*}(p)_{11} \cdot
    \mathcal{T}_{M}^{*}(p)_{44} - \mathcal{T}_{M}^{*}(p)_{22}
    \cdot \mathcal{T}_{M}^{*}(p)_{33}
    = p^{z_{11} + z_{44} - z_{22} - z_{33}}, \text{ and}$$
    \begin{align*}
        \zeta_{2}(p)
        &= (\mathcal{T}_{M}^{*}(p)^{2})_{11} \cdot
        (\mathcal{T}_{M}^{*}(p)^{2})_{44} - 
        (\mathcal{T}_{M}^{*}(p)^{2})_{22} \cdot
        (\mathcal{T}_{M}^{*}(p)^{2})_{33}\\
        &= \big(\sum_{i \in [4]}p^{2z_{1i}}\big) \cdot
        \big(\sum_{j \in [4]}p^{2z_{4j}}\big) - 
        \big(\sum_{i \in [4]}p^{2z_{2i}}\big) \cdot
        \big(\sum_{j \in [4]}p^{2z_{3j}}\big)\\
        &= \sum_{i, j \in [4]}\left(p^{2z_{1i} + 2z_{4j}} - 
        p^{2z_{2i} + 2z_{3j}}\right).
    \end{align*}
    From our assumptions above, we know that $\zeta(p) = 0$,
    and $\zeta_{2}(p) = 0$ for all $p \in \mathbb{R}$.
    Therefore $\frac{d}{dp}(\zeta(p)) = 0$, and
    $\frac{d}{dp}(\zeta_{2}(p)) = 0$ for all
    $p \in \mathbb{R}$.
    Specifically, when evaluated at $p = 1$, we see that
    $$\frac{d\zeta}{dp}(1) =
    z_{11} + z_{44} - z_{22} - z_{33} = 0.$$
    This implies that $z_{11} + z_{44} = z_{22} + z_{33}$.
    Similarly,
    \begin{align*}
        \frac{d\zeta_{2}}{dp}(1)
        &= \sum_{i, j \in [4]}(2z_{i1} + 2z_{4j}
        - 2z_{2i} - 2z_{3j})\\
        &= (2z_{11} + 2z_{12} + 2z_{13} + 2z_{14})
        + (2z_{14} + 2z_{24} + 2z_{34} + 2z_{44})\\
        & \quad - (2z_{12} + 2z_{22} + 2z_{23} + 2z_{24})
        - (2z_{13} + 2z_{23} + 2z_{33} + 2z_{34})\\
        &= 2z_{11} + 4z_{14} + 2z_{44} - 2z_{22} - 4z_{23} - 2z_{33}
        = 0.
    \end{align*}
    Therefore, $2(z_{11} + z_{44} - z_{22} - z_{33})
    + 4(z_{14} - z_{23}) = 0$.
    Since we already saw that $z_{11} + z_{44} = z_{22} + z_{33}$,
    this implies that $z_{14} = z_{23}$.

    We recall that by our choice of $m > \mathfrak{m}$,
    and construction of $z_{ij}$,
    $z_{14} = z_{23}$ implies that
    $(e_{141}, \dots, e_{14d}) = (e_{231}, \dots, e_{23d})$.
    But this means that $M_{14} = M_{23}$,
    which finishes the proof.
\end{proof}

\begin{lemma}\label{lemma: notConfluentImpliesTensor}
    Let $M \in \text{Sym}_{4}^{\tt{pd}}(\mathbb{R}_{> 0})$
    be a diagonal distinct matrix.
    Let $\mathcal{F}$ be a countable set of
    $\text{Sym}_{4}(\mathbb{R})$-polynomials
    such that $F(M) \neq 0$ for all
    $F \in \mathcal{F}$.
    Let $\xi$ be the $\text{Sym}_{4}(\mathbb{R})$-polynomial
    such that $\xi: N \mapsto \phi_{(1, -1, -1, 1)}
    (N_{11}, \dots, N_{44})$.
    Then there exists some diagonal distinct
    $N \in \mathfrak{R}(M, \mathcal{F}
    \cup \{\xi\}) \cap
    \text{Sym}_{4}^{\tt{pd}}(\mathbb{R}_{> 0})$,
    unless $M = A \otimes B$ for some $A, B \in
    \text{Sym}_{2}(\mathbb{R}_{> 0})$.
\end{lemma}
\begin{proof}
    We will assume that $M \neq A \otimes B$
    for any $A, B \in \text{Sym}_{2}(\mathbb{R}_{> 0})$.
    From \cref{lemma: tensorPolynomial}, we know that
    $\varrho_{\tt{tensor}}(M) \neq 0$.
    So we can add $\rho_{\tt{diag}}$ and $\varrho_{\tt{tensor}}$ to $\mathcal{F}$ assumed for $M$.
    Our goal will be to show that we can construct
    some diagonal distinct
    $N' \in \mathfrak{R}(M, \mathcal{F}) \cap
    \text{Sym}_{4}^{\tt{pd}}(\mathbb{R}_{> 0})$
    such that $N'_{11} \cdot N'_{44} \neq N'_{22} \cdot N'_{33}$
    or $N'_{14} \neq N'_{23}$.
    Then, \cref{lemma: notConfluentImpliesM14M23}
    would allow us to find the required
    $N \in \mathfrak{R}(N', \mathcal{F}
    \cup \{\rho_{\tt{diag}}, \xi\}) \cap
    \text{Sym}_{4}^{\tt{pd}}(\mathbb{R}_{> 0})$.

    We recall that since $\mathcal{S}_{M}(0) = I$,
    there exists some $\delta > 0$ such that
    $|\mathcal{S}_{M}(\theta)_{ij} - I_{ij}| < \frac{1}{3}$
    for all $0 < \theta < \delta$ for all $i, j \in [4]$.
    We will now let
    $$\mathcal{F}' = \{F': N \mapsto F(N^{3})
    \hspace{0.08cm}\Big|\hspace{0.1cm}
    F \in \mathcal{F}\}.$$
    We note that $F'(\mathcal{S}_{M}(\frac{1}{3}))
    \neq 0$ for all $F' \in \mathcal{F}'$, and
    $\rho_{\tt{diag}}(\mathcal{S}_{M}(1))
    = \rho_{\tt{diag}}(M) \neq 0$,
    and $\varrho_{\tt{tensor}}(\mathcal{S}_{M}(1)) =
    \varrho_{\tt{tensor}}(M) \neq 0$.
    We can therefore use \cref{lemma: stretchingWorks}
    to find $M' = \mathcal{S}_{M}(\theta^{*}) \in
    \mathfrak{R}(M, \mathcal{F}' \cup 
    \{\rho_{\tt{diag}}, \varrho_{\tt{tensor}}\}) \cap
    \text{Sym}_{4}^{\tt{pd}}(\mathbb{R}_{\neq 0})$
    for some $0 < \theta^{*} < \delta$.

    If $(M')_{14} \neq (M')_{23}$ or 
    $M'_{11}M'_{44} \neq M'_{22}M'_{33}$, we will
    be done. So, we may assume otherwise.
    We will first assume that
    $(R_{n}(M'))_{14} - (R_{n}(M'))_{23} = 0$
    for all $n \geq 1$.
    We note that
    $$R_{n}(M')_{14} - R_{n}(M')_{23}
    = \sum_{a, b \in [4]}(M'_{ab})^{n} \cdot
    (M'_{1a}M'_{4b} - M'_{2a}M'_{3b}).$$
    We may now let $X = \{(M')_{ab}: a, b \in [4]\}$,
    and for each $x \in X$, define
    $$c_{14}(x) = \sum_{a, b \in [4]: (M')_{ab} = x}
    (M')_{1a}(M')_{4b}, ~~\text{ and }~~
    c_{23}(x) = \sum_{a, b \in [4]: (M')_{ab} = x}
    (M')_{2a}(M')_{3b}.$$
    So, we see that
    $$R_{n}(M')_{14} - R_{n}(M')_{23}
    = \sum_{x \in X}x^{n} \cdot
    (c_{14}(x) - c_{23}(x)).$$
    Since each $x \in X$ is distinct, and $|X| \leq O(1)$,
    this forms a full rank Vandermonde system of
    linear equations.
    Since we have assumed that 
    $(R_{n}(M'))_{14} - (R_{n}(M'))_{23} = 0$
    for all $n \geq 1$, this implies that
    $c_{14}(x) - c_{23}(x) = 0$ for all $x \in X$.

    Now, by our choice of $M'$, we know that all the
    diagonal terms $M'$ of are distinct. Also any diagonal element of  $M' = \mathcal{S}_{M}(\theta^{*})$
    is greater than the absolute value of any off diagonal element, by our choice of
     $0 < \theta^{*} < \delta$.
    Therefore, $(M')_{ii} = (M')_{ab}$ implies
    that $(a, b) = (i, i)$ for all $i \in [4]$.
    This implies that
    $$c_{14}((M')_{ii}) - c_{23}((M')_{ii}) = 
    (M')_{1i}(M')_{4i} - (M')_{2i}(M')_{3i} = 0$$
    for all $i \in [4]$.
    We also already know that
    $(M')_{14} = (M')_{23}$, and
    $M'_{11}M'_{44} = M'_{22}M'_{33}$.
    But this implies that $\varrho_{\tt{tensor}}(M') = 0$,
    which is a contradiction to  $M' \in
    \mathfrak{R}(M, \mathcal{F}' \cup 
    \{\rho_{\tt{diag}}, \varrho_{\tt{tensor}}\})$.

    This means that there exists some $n^{*} \geq 1$,
    such that $(R_{n^{*}}(M'))_{14} -
    R_{n^{*}}(M'))_{23} \neq 0$.
    This implies that
    we can define the $\text{Sym}_{4}(\mathbb{R})$-polynomial
    $\xi_{n^{*}}: N \mapsto
    (R_{n^{*}}(N))_{14} - (R_{n^{*}}(N))_{23}$,
    and we see that 
    $\xi_{n^{*}}(\mathcal{S}_{M}(\theta^{*}))
    = \xi_{n^{*}}(M') \neq 0$.
    So, using $\mathcal{S}_{M}(\frac{1}{3})$, \cref{corollary: stretchingWorksPositive}
    allows us to first find some
    $M'' \in \mathfrak{R}(M, \mathcal{F}'
    \cup \{\rho_{\tt{diag}}^{*},  \xi_{n^{*}}\}) 
    \cap \text{Sym}_{4}^{\tt{pd}}(\mathbb{R}_{> 0})$,
    where $\rho_{\tt{diag}}^{*}(N) = \rho_{\tt{diag}}(N^{3})$.
    By construction, we see that
    $F(\mathcal{R}_{M''}(1)) \neq 0$ for all
    $F \in \mathcal{F}$, and
    $\rho_{\tt{diag}}(\mathcal{R}_{M''}(1)) \neq 0$, as $\mathcal{R}_{M''}(1) = (M'')^3$.
    Moreover, we see that
    $(\mathcal{R}_{M''}(n^{*}))_{14} - 
    (\mathcal{R}_{M''}(n^{*}))_{23} \neq 0$.
    Therefore, \cref{lemma: RInterpolationWorks}
    allows us to find some
    $N' \in \mathfrak{R}(M'', \mathcal{F} \cup
    \{\rho_{\tt{diag}}\}) \cap 
    \text{Sym}_{4}^{\tt{pd}}(\mathbb{R}_{> 0})$,
    such that $N'_{14} \neq N'_{23}$.
    \cref{lemma: notConfluentImpliesM14M23}
    then allows us to find the required
    $N \in \mathfrak{R}(N', \mathcal{F}
    \cup \{\rho_{\tt{diag}}, \xi\}) \cap
    \text{Sym}_{4}^{\tt{pd}}(\mathbb{R}_{> 0})$.
\end{proof}

We are now ready to state a hardness criterion
that applies to all $\PlEVAL(M)$ where
$M$ is diagonal distinct.

\begin{theorem}\label{theorem: diagDistHardness}
    Let $M \in \text{Sym}_{4}^{\tt{pd}}(\mathbb{R}_{> 0})$
    be a diagonal distinct matrix.
    Then $\PlEVAL(M)$ is $\#$P-hard unless
    $M$ is isomorphic to $A \otimes B$ for some
    $A, B \in \text{Sym}_{2}(\mathbb{R}_{> 0})$.
\end{theorem}
\begin{proof}
    Let us assume that $M$ is not isomorphic to
    any $A \otimes B$ for any $A, B \in 
    \text{Sym}_{2}(\mathbb{R})$.
    (Note that for $M \in \text{Sym}_{4}(\mathbb{R}_{> 0})$,
    $M$ is isomorphic to $A \otimes B$ for some
    $A, B \in \text{Sym}_{2}(\mathbb{R})$ is equivalent to  being isomorphic to $A \otimes B$ for some
    $A, B \in \text{Sym}_{2}(\mathbb{R}_{> 0})$.)
    So, \cref{lemma: tensorPolynomial} tells us
    that $\rho_{\tt{tensor}}(M) \neq 0$.
    We will now construct
    $$\mathcal{F} = \mathcal{F}_{M} \cup \{
    \rho_{\tt{diag}}, \rho_{\tt{tensor}}\},$$
    where $\mathcal{F}_{M}$ is as defined
    in \cref{equation: FM}.
    We can see that $F(M) \neq 0$ for all
    $F \in \mathcal{F}$.
    We recall from \cref{equation: phiX} that
    $$\Phi_{(1, -1, -1, 1)}(N_{11}, \dots, N_{44})
    = \prod_{\sigma \in S_{4}}\phi_{(1, -1, -1, 1)}
    (N_{\sigma(1)\sigma(1)}, \dots, N_{\sigma(4)\sigma(4)}).$$
    So, we will define the
    $\text{Sym}_{4}(\mathbb{R})$-polynomials
    $\xi_{\sigma}$ for $\sigma \in S_{4}$ as
    $$\xi_{\sigma}(N) = \phi_{(1, -1, -1, 1)}
    (N_{\sigma(1)\sigma(1)}, \dots, N_{\sigma(4)\sigma(4)}),$$
    and the $\text{Sym}_{4}(\mathbb{R})$-polynomial $\xi$
    such that
    $\xi: N \mapsto \prod_{\sigma \in S_{4}}\xi_{\sigma}(N)$.

    For each $\sigma \in S_{4}$, 
    we may define $M^{\sigma} \in
    \text{Sym}_{4}^{\tt{pd}}(\mathbb{R}_{> 0})$
    as in \cref{lemma: tensorPolynomial} to be
    such that $(M^{\sigma})_{ij} = M_{\sigma(i)\sigma(j)}$.
    We will also let
    $S_{4} = \{\sigma_{1}, \dots, \sigma_{24}\}$.
    Now, \cref{lemma: notConfluentImpliesTensor} implies
    that since $\varrho_{\tt{tensor}}(M^{\sigma_{1}}) \neq 0$,
    we can find some
    $M_{1} \in \mathfrak{R}(M, \mathcal{F} \cup
    \{\xi_{\sigma_{1}} \}) \cap
    \text{Sym}_{4}^{\tt{pd}}(\mathbb{R}_{> 0})$.
    Since $F(M_{1}) \neq 0$ for all $F \in \mathcal{F}$,
    starting with $M_{1}$ instead of $M$, we can repeat this
    to find $M_{2} \in \mathfrak{R}(M_{1}, \mathcal{F} \cup
    \{\xi_{\sigma_{1}}, \xi_{\sigma_{2}} \}) \cap
    \text{Sym}_{4}^{\tt{pd}}(\mathbb{R}_{> 0})$.

    After repeating this process for all
    $\sigma \in S_{4}$, we end up with some
    $$M' = M_{24} \in \mathfrak{R}(M, \mathcal{F} \cup
    \{\xi\}) \cap \text{Sym}_{4}^{\tt{pd}}(\mathbb{R}_{> 0}).$$
    
    Since $\Phi_{(1, -1, -1, 1)}(M'_{11}, \dots, M'_{44})
    \neq 0$, \cref{theorem: positiveDefiniteDiagonal}
    allows us to find some
    $N \in \mathfrak{R}(M', \mathcal{F}_{M}
    \cup \{\rho_{\tt{diag}}, \Psi_{(1, -1, -1, 1)}\})
    \cap \text{Sym}_{4}^{\tt{pd}}(\mathbb{R}_{> 0})$.

    Let $(\mu_{1}, \dots, \mu_{4})$ be
    the eigenvalues of $N$.
    Let $\mathcal{L}_{N}$ be the lattice 
    formed by these eigenvalues.
    By construction, we see that $(1, -1, -1, 1)
    \notin \overline{\mathcal{L}_{N}}$.
    So, from \cref{lemma: almostAllConfluent},
    we see that all $\mathbf{0} \neq \mathbf{x} \in
    \overline{\mathcal{L}_{N}}$ are confluent.
    Since $N$ is diagonal distinct,
    \cref{corollary: diagDistHardness} then
    proves that $\PlEVAL(N)$ is $\#$P-hard.
    Since $\PlEVAL(N) \leq \PlEVAL(M)$,
    we see that $\PlEVAL(M)$ is also
    $\#$P-hard.
\end{proof}

\section[Hardness for non-Diagonal Distinct 4 x 4 matrices]
{Hardness for non-Diagonal Distinct 
$4 \times 4$ matrices}\label{sec: 4x4nonDiagDist}

In this section, we will deal with the matrices $M \in
\text{Sym}_{4}^{\tt{pd}}(\mathbb{R}_{> 0})$
that are not diagonal distinct.
We will first show that in this case, there
are only finitely many forms that the matrix $M$
must take.
Then, we will show that for each such form,
$\PlEVAL(M)$ is $\#$P-hard, unless $M$ 
is isomorphic to $A \otimes B$ for some
$A, B \in  \text{Sym}_{2}(\mathbb{R}_{\neq 0})$.

\begin{lemma}\label{lemma: order211Forms}
    Let $M \in \text{Sym}_{4}^{\tt{pd}}(\mathbb{R}_{> 0})$
    such that $M_{11} = M_{22}$, $M_{22} \neq M_{33} \neq M_{44}$
    are pairwise distinct, and
    $\rho_{\tt{tensor}}(M) \neq 0$.
    Unless $M$ is of one of the two forms below (Forms (I) or (II)),
    $\PlEVAL(M)$ is $\#$P-hard.
    \begin{figure*}[ht]
    \centering
    \begin{subfigure}{0.45\textwidth}
        \centering
        \[
        \begin{pmatrix}
        M_{11} & M_{12} & M_{13} & M_{14}\\
        M_{12} & M_{11} & M_{13} & M_{14}\\
        M_{13} & M_{13} & M_{33} & M_{34}\\
        M_{14} & M_{14} & M_{34} & M_{44}\\
        \end{pmatrix} = \begin{pmatrix}
            a & x & y & z\\
            x & a & y & z\\
            y & y & b & t\\
            z & z & t & c
        \end{pmatrix}
        \]
        \caption*{Form (I)}
    \end{subfigure}
    \hfill
    \begin{subfigure}{0.45\textwidth}
        \centering
        \[
        \begin{pmatrix}
        M_{11} & M_{12} & M_{13} & M_{14}\\
        M_{12} & M_{11} & M_{14} & M_{13}\\
        M_{13} & M_{14} & M_{33} & M_{34}\\
        M_{14} & M_{13} & M_{34} & M_{44}\\
        \end{pmatrix} = \begin{pmatrix}
            a & x & y & z\\
            x & a & z & y\\
            y & z & b & t\\
            z & y & t & c
        \end{pmatrix}
        \]
        \caption*{Form (II)}
    \end{subfigure}
\end{figure*}
\end{lemma}
\begin{proof}
    We apply  \cref{lemma: OrderedDistImpliesDiagDist} to $M$.
    Suppose rows $1$ and $2$ are not
    order identical.
    The pairwise distinctness of  $\{M_{22}, M_{33}, M_{44}\}$
    and of $\{M_{11}, M_{33}, M_{44}\}$ are conditions expressible as nonvanishing 
    of $\text{Sym}_{q}(\mathbb{R})$-polynomials
    and thus maintained by  \cref{lemma: OrderedDistImpliesDiagDist},
    and hence its conclusion $N_{11} \ne N_{22}$ in addition to the above  
    pairwise distinctness conditions
    gives diagonal distinctness. 
    Therefore, from \cref{lemma: OrderedDistImpliesDiagDist},
    unless rows $1$ and $2$ are
    order identical, we get
    some diagonal distinct
    $N \in \mathfrak{R}(M, \mathcal{F}_{M}
    \cup \{\rho_{\tt{tensor}} \}) \cap
    \text{Sym}_{4}^{\tt{pd}}(\mathbb{R}_{> 0})$.
    \cref{theorem: diagDistHardness} then
    proves that $\PlEVAL(N)$, and consequently,
    $\PlEVAL(M)$ are $\#$P-hard.

    So, we may assume  rows $1$ and $2$ are 
    order identical. There  exists some $\sigma \in S_{4}$
    such that $M_{1i} = M_{2\sigma(i)}$ for all
    $i \in [4]$. We are given 
    $M_{11} = M_{22}$. We also have   $M_{12} = M_{21}$ by $M$ being symmetric. 
    From the identical multisets $\{M_{11}, M_{12}, M_{13}, M_{14}\}$ and 
        $\{M_{21}, M_{22},  M_{23}, M_{24}\}$ if we remove the element pairs 
        $M_{11} = M_{22}$ and  $M_{12} = M_{21}$, we still have an equal multiset
        $\{M_{13}, M_{14}\}=\{M_{23}, M_{24}\}$. 
    So, we may assume that
    ($\sigma(1) = 2$, $\sigma(2) = 1$, and)
    $\sigma(3) = 3$ or
    $\sigma(3) = 4$.
    Now, if $\sigma(3) = 3$, that means that
    $\sigma(4) = 4$, which forces
    $M_{13} = M_{23}$, and $M_{14} = M_{24}$.
    In that case, we see that $M$ is of 
    Form (I) above.
    On the other hand, if $\sigma(3) = 4$, that
    means $\sigma(4) = 3$, which forces $M$
    to be of Form (II) above.
\end{proof}

\begin{remark*}
    When we say that a matrix
    is of Form (I), we do not require that
    $M_{13} \neq M_{14}$ for example.
    In general, for all the Forms that we will describe,
    we allow the possibility that some of the values denoted by distinct symbols
    may be equal to each other.
    All we require is that  
    entries denoted by the same symbol are equal.
    So, it may be possible that a matrix is of more than one Form.
\end{remark*}

\begin{lemma}\label{lemma: order22Forms}
    Let $M \in \text{Sym}_{4}^{\tt{pd}}(\mathbb{R}_{> 0})$
    such that $M_{11} = M_{22} \neq M_{33} = M_{44}$, and
    $\rho_{\tt{tensor}}(M) \neq 0$.
    Unless $M$ is isomorphic to a matrix of the Form (I) or Form (II), or
    is of the form below (Form (III)),
    $\PlEVAL(M)$ is $\#$P-hard.
    \begin{figure*}[ht]
    \centering
    \begin{subfigure}{0.45\textwidth}
        \centering
        \[
        \begin{pmatrix}
        M_{11} & M_{12} & M_{13} & M_{14}\\
        M_{12} & M_{11} & M_{14} & M_{13}\\
        M_{13} & M_{14} & M_{33} & M_{34}\\
        M_{14} & M_{13} & M_{34} & M_{33}\\
        \end{pmatrix} = \begin{pmatrix}
            a & x & y & z\\
            x & a & z & y\\
            y & z & b & t\\
            z & y & t & b
        \end{pmatrix}
        \]
        \caption*{Form (III)}
    \end{subfigure}
\end{figure*}
\end{lemma}
\begin{proof}
    Let us first assume that it is not the case that rows $1$ and $2$
    are order identical, and that rows $3$ and $4$ are order identical.
    By symmetry, we may assume that rows $3$ and $4$ are not
    order identical.
    We will also assume that $M$ is not isomorphic
    to a matrix of Form (I) or of Form (II).
    So, if we define $\zeta: N \mapsto (N_{11} - N_{33})
    (N_{11} - N_{44})(N_{22} - N_{33})(N_{22} - N_{44})$,    
    $\zeta_{1}: N \mapsto (N_{13} - N_{23})^{2}
    + (N_{14} - N_{24})^{2}$, and
    $\zeta_{2}: N \mapsto (N_{13} - N_{24})^{2} +
    (N_{14} - N_{23})^{2}$, we see that
    $\zeta(M) \neq 0$, $\zeta_{1}(M) \neq 0$, and $\zeta_{2}(M) \neq 0$.    
    From \cref{lemma: OrderedDistImpliesDiagDist},
    we know that since rows $3$ and $4$ are not
    order identical, we can construct
    some $N \in \mathfrak{R}(M, \mathcal{F}_{M}
    \cup \{\rho_{\tt{tensor}}, \zeta, \zeta_{1}, \zeta_{2}\}) \cap
    \text{Sym}_{4}^{\tt{pd}}(\mathbb{R}_{> 0})$
    such that $N_{22} \neq N_{33} \neq N_{44}$, and
    $N_{11} \neq N_{33} \neq N_{44}$.
    If $N_{11} \neq N_{22}$, then $N$ is diagonal distinct,
    and \cref{theorem: diagDistHardness} implies that
    $\PlEVAL(N) \leq \PlEVAL(M)$ is $\#$P-hard.
    Moreover, by our construction of $N$,
    we ensured that even if $N_{11} = N_{22}$,
    $N$ is not of Form (I) or Form (II).
    So, we can see from
    \cref{lemma: order211Forms} that $\PlEVAL(N) \leq \PlEVAL(M)$
    must be $\#$P-hard.
    
    So, we may assume that rows $1$ and $2$ are
    order identical, and that rows
    $3$ and $4$ are also order identical.
    So, there must exist some $\sigma_{1},
    \sigma_{2} \in S_{4}$ such that 
    $M_{1i} = M_{2\sigma_{1}(i)}$, and
    $M_{3i} = M_{4\sigma_{2}(i)}$ for all
    $i \in [4]$. We already know that
    $M_{11} = M_{22}$, and that $M_{12} = M_{21}$.
    So, we may assume that 
    $\sigma_{1}(1) = 2$, and $\sigma_{1}(2) = 1$.
    Similarly, we may assume that $\sigma_{2}(3) = 4$, and
    $\sigma_{2}(4) = 3$.
    Now, if $\sigma_{1}(3) = 3$, that means that
    $\sigma_{1}(4) = 4$, which forces
    $M_{13} = M_{23}$, and $M_{14} = M_{24}$.
    But then irrespective of whether
    $\sigma_{2}(1) = 1$ or $\sigma_{2}(1) = 2$,
    we see that $M_{31} = M_{4\sigma_{2}(1)}$ which
    implies that $M_{13} = M_{14} = M_{23} = M_{24}$.
    In that case, we see that $M$ is of 
    Form (III) above (with $M_{13} = M_{14}$).
    On the other hand, if $\sigma_{1}(3) = 4$, that
    means $\sigma_{1}(4) = 3$, which forces $M$
    to be of Form (III) above.
\end{proof}

\begin{lemma}\label{lemma: order31Forms}
    Let $M \in \text{Sym}_{4}^{\tt{pd}}(\mathbb{R}_{> 0})$
    such that $M_{11} = M_{22} = M_{33} \neq M_{44}$,
    and $\rho_{\tt{tensor}}(M) \neq 0$.
    Unless $M$ is isomorphic to a matrix of Form (I) or Form (II),
    or is of the two forms below,
    $\PlEVAL(M)$ is $\#$P-hard.
    \begin{figure*}[ht]
    \centering
    \begin{subfigure}{0.45\textwidth}
        \centering
        \[
        \begin{pmatrix}
        M_{11} & M_{13} & M_{13} & M_{14}\\
        M_{13} & M_{11} & M_{13} & M_{14}\\
        M_{13} & M_{13} & M_{11} & M_{14}\\
        M_{14} & M_{14} & M_{14} & M_{44}\\
        \end{pmatrix} = \begin{pmatrix}
            a & x & x & z\\
            x & a & x & z\\
            x & x & a & z\\
            z & z & z & b
        \end{pmatrix}
        \]
        \caption*{Form (IV)}
    \end{subfigure}
    \hfill
    \begin{subfigure}{0.45\textwidth}
        \centering
        \[
        \begin{pmatrix}
        M_{11} & M_{12} & M_{13} & M_{14}\\
        M_{12} & M_{11} & M_{14} & M_{13}\\
        M_{13} & M_{14} & M_{11} & M_{12}\\
        M_{14} & M_{13} & M_{12} & M_{44}\\
        \end{pmatrix} = \begin{pmatrix}
            a & x & y & z\\
            x & a & z & y\\
            y & z & a & x\\
            z & y & x & b
        \end{pmatrix}
        \]
        \caption*{Form (V)}
    \end{subfigure}
\end{figure*}
\end{lemma}
\begin{proof}
    We will first suppose rows $1$, $2$ and $3$
    are not all order identical to each other. By the symmetry of our assumptions so far among $\{1, 2, 3\}$ 
    in this lemma, without loss of generality suppose rows $2$ and $3$ are not order identical.
    We can also assume that $M$ is not isomorphic
    to a matrix of Form (I) or Form (II).
    We can now define $\zeta: N \mapsto (N_{11} - N_{44})
    (N_{22} - N_{44})(N_{33} - N_{44})$,    
    $\zeta_{1}: N \mapsto (N_{13} - N_{23})^{2}
    + (N_{14} - N_{24})^{2}$,
    $\zeta_{2}: N \mapsto (N_{13} - N_{24})^{2} +
    (N_{14} - N_{23})^{2}$,
    $\zeta_{3}: N \mapsto (N_{12} - N_{23})^{2}
    + (N_{14} - N_{34})^{2}$, and
    $\zeta_{4}: N \mapsto (N_{12} - N_{34})^{2} +
    (N_{14} - N_{23})^{2}$.
    (Here, given $M_{11}=M_{22}$, being not in Form (I) or Form (II)
    implies that $\zeta_{1}(M) \ne 0$  and $\zeta_{2}(M) \ne 0$.
    Note that $\zeta_{3}$ and $\zeta_{4}$ are obtained from $\zeta_{1}$  and $\zeta_{2}$, respectively,
    by exchanging rows and columns indexed by $2$ and $3$.)
    We see that
    $\zeta(M) \neq 0$, and $\zeta_{i}(M) \neq 0$
    for all $i \in [4]$.
    By \cref{lemma: OrderedDistImpliesDiagDist}, we can obtain some $N \in \mathfrak{R}(M, \mathcal{F}_{M} \cup
    \{\rho_{\tt{tensor}}, \zeta, \zeta_{1}, \zeta_{2},
    \zeta_{3}, \zeta_{4} \}) \cap
    \text{Sym}_{4}^{\tt{pd}}(\mathbb{R}_{> 0})$
    satisfying $N_{22} \ne N_{33}$.
    If $N_{11} \neq N_{22} \neq N_{33} \neq N_{44}$ are all
    pairwise distinct, then \cref{theorem: diagDistHardness}
    implies that $\PlEVAL(N) \leq \PlEVAL(M)$ is $\#$P-hard.
    If however, $N_{11} = N_{22}$, since $\zeta(N) \neq 0$,
    it must be the case that $N_{11} = N_{22}$, 
    but $N_{22} \neq N_{33} \neq N_{44}$ are pairwise distinct.
    So, \cref{lemma: order211Forms} implies that
    $\PlEVAL(N) \leq \PlEVAL(M)$ is $\#$P-hard unless
    $N$ is of Form (I) or (II).
    But, by construction, $\zeta_{1}(N) \neq 0$,
    and $\zeta_{2}(N) \neq 0$.
    So, $N$ is not of Form (I) or (II), which proves that
    $\PlEVAL(N) \leq \PlEVAL(M)$ is $\#$P-hard.
    Similarly, if $N_{11} = N_{33}$, then since
    $\zeta(N) \neq 0$, we see that $N_{22} \neq N_{33} \neq N_{44}$
    are pairwise distinct.
    Moreover, since $\zeta_{3}(N) \neq 0$, and
    $\zeta_{4}(N) \neq 0$, we see that if
    we switch rows and columns $2 \leftrightarrow 3$ in $N$,
    we have a matrix $N'$ that is isomorphic to $N$ that
    satisfies the condition that
    $(N')_{11} = (N')_{22}$, but 
    $(N')_{22} \neq (N')_{33} \neq (N')_{44}$
    are pairwise distinct, and the condition that 
    $\zeta_{3}(N) \neq 0$ and $\zeta_{4}(N) \neq 0$ translates to
    $N'$ is not of Form (I) or (II).
    So, \cref{lemma: order211Forms} implies that 
    $\PlEVAL(N') \equiv \PlEVAL(N) \leq \PlEVAL(M)$ is $\#$P-hard.
    
    So, rows $1$, $2$ and $3$
    are  all order identical to each other, and  there  exist some $\sigma_{1}, \sigma_{2}
    \in S_{4}$ such that
    $M_{1i} = M_{2\sigma_{1}(i)}$, and
    $M_{1i} = M_{3\sigma_{2}(i)}$ for all
    $i \in [4]$.
    We already know that
    $M_{11} = M_{22}$, and that $M_{12} = M_{21}$.
    So, we may assume that $\sigma_{1}(1) = 2$,
    and $\sigma_{1}(2) = 1$.
    Similarly, we may assume that $\sigma_{2}(1) = 3$, and
    $\sigma_{2}(3) = 1$.
So far we have the following 
    \[
        \begin{pmatrix}
            a & x & y & *\\
            x & a & * & *\\
            y & * & a & *\\
            * & * & * & *
        \end{pmatrix}
        \]
There are two possibilities:   $\sigma_{1}(3) = 3$  or  $\sigma_{1}(3) = 4$.  
Suppose $\sigma_{1}(3) = 3$, then
    $\sigma_{1}(4) = 4$, which forces
    $M_{13} = M_{23}$, and $M_{14} = M_{24}$, and we have the following
    \[
        \begin{pmatrix}
            a & x & y & z\\
            x & a & y & z\\
            y & y & a & *\\
            z & z & * & *
        \end{pmatrix}
        \]
Since rows 1 and 3 are order identical,  the $y$ entry at $M_{31}$ is either $y=x$ or $y=z$.
The two cases are listed in Form (IV) and Form (V) (with $y=z$) respectively.

Now we may assume $\sigma_{1}(3) = 4$. Then we have the following
    \[
        \begin{pmatrix}
            a & x & y & z\\
            x & a & z & y\\
            y & z & a & *\\
            z & y & * & *
        \end{pmatrix}
        \]
Again since row 1 and 3 are order identical, the $M_{34}$ entry must be $x$, this gives Form (V).
%
%
\end{proof}

\begin{lemma}\label{lemma: order4Forms}
    Let $M \in \text{Sym}_{4}^{\tt{pd}}(\mathbb{R}_{> 0})$
    such that $M_{11} = M_{22} = M_{33} = M_{44}$, and
    $\rho_{\tt{tensor}}(M) \neq 0$.
    Unless $M$ is isomorphic to a matrix of Forms (I) - (V), 
    or is of the form below,
    $\PlEVAL(M)$ is $\#$P-hard.
    \begin{figure*}[ht]
    \centering
    \begin{subfigure}{0.45\textwidth}
        \centering
        \[
        \begin{pmatrix}
        M_{11} & M_{12} & M_{13} & M_{14}\\
        M_{12} & M_{11} & M_{14} & M_{13}\\
        M_{13} & M_{14} & M_{11} & M_{12}\\
        M_{14} & M_{13} & M_{12} & M_{11}\\
        \end{pmatrix} = \begin{pmatrix}
            a & x & y & z\\
            x & a & z & y\\
            y & z & a & x\\
            z & y & x & a
        \end{pmatrix}
        \]
        \caption*{Form (VI)}
    \end{subfigure}
\end{figure*}
\end{lemma}
\begin{proof}
    We first suppose rows $1$, $2$, $3$ and $4$
    are not all order identical to each other. By  symmetry, 
    without loss of generality suppose rows $3$ and $4$ are not  order identical.
    We may also assume that $M$ is not of Forms (I) - (V).
    We can now define the
    $\text{Sym}_{4}(\mathbb{R})$-polynomial
    $\zeta_{1}$ such that
    $\zeta_{1}(N) = 0$ if and only if
    $N$ is isomorphic to a matrix of Form (I).
    Similarly, we can define $\zeta_{2}$ such that
    $\zeta_{2}(N) = 0$ if and only if
    $N$ is isomorphic to a matrix of Form (II).
    We can also define $\zeta_{3}, \zeta_{4}, \zeta_{5}$
    similarly for Forms (III) - (V).
    By our assumption, $\zeta_{i}(M) \neq 0$ for all $i \in [5]$.
    
    Since rows $3$ and $4$ are not order identical,
    by \cref{lemma: OrderedDistImpliesDiagDist}, we can obtain some $N \in \mathfrak{R}(M, \mathcal{F}_{M} \cup
    \{\rho_{\tt{tensor}}, \zeta_{1}, \dots, \zeta_{5}\}) \cap
    \text{Sym}_{4}^{\tt{pd}}(\mathbb{R}_{> 0})$ satisfying $N_{33} \ne N_{44}$.
    By our construction, we also know that $N$ is not
    isomorphic to any matrix of Forms (I) - (V).
    Now, the set
    of diagonal elements of $N$ (after removal of duplicates) has cardinality 2 or 3 or 4.
   The cardinality  4 case
        is diagonal distinct, and then
    \cref{theorem: diagDistHardness} allows us to prove
    that $\PlEVAL(M)$ is $\#$P-hard.
    The cases of cardinality 2 or 3 fall into the cases of \cref{lemma: order211Forms}, \cref{lemma: order22Forms} or \cref{lemma: order31Forms},
    and in either case, we see that $\PlEVAL(N) \leq \PlEVAL(M)$
    is $\#$P-hard, since $N$ is not isomorphic to any matrix
    of Form (I) - (V).

    So suppose rows $1$, $2$, $3$ and $4$
    are  all order identical to each other. 
    We have the  following setting
    \[
        \begin{pmatrix}
            a & x & y & z\\
            x & a & * & *\\
            y & * & a & *\\
            z & * & * & a
        \end{pmatrix}
        \]
    As in the proof of  \cref{lemma: order31Forms}, there are two possibilities:   $\sigma_{1}(3) = 3$  or  $\sigma_{1}(3) = 4$.  
    Suppose $\sigma_{1}(3) = 3$, then
    $\sigma_{1}(4) = 4$, which forces
    $M_{13} = M_{23}$, and $M_{14} = M_{24}$, and we have the following
    \[
        \begin{pmatrix}
            a & x & y & z\\
            x & a & y & z\\
            y & y & a & *\\
            z & z & * & a
        \end{pmatrix}
        \]
    Then $y$ from the 3rd row must match either $x$ or $z$.
    If $y=x$, then the 3rd row is $(y,y,a,z)$ and then the 4th row is $(z,z,z,a)$, and thus
    $x=y=z$, and we have Form (VI) (with $x=y=z$).
    If $y=z$ then we have 
    \[
        \begin{pmatrix}
            a & x & y & y\\
            x & a & y & y\\
            y & y & a & *\\
            y & y & * & a
        \end{pmatrix}
        \]
    Then $M_{34} = x$ and we have  Form (VI) (with $y=z$).
    Finally if $\sigma_{1}(3) = 4$, we also have  Form (VI). 
\end{proof}

Now that we have listed all the different Forms
that the matrices must take, we will show
one by one that for each of the Forms,
matrices $M$ belonging to that Form must either
be isomorphic to some $A \otimes B$, or that
$\PlEVAL(M)$ is $\#$P-hard.
We will first show that Form (II) and Form (V)
can actually be reduced to one of the other forms.

\begin{lemma}\label{lemma: formIIReduction}
    Let $M \in \text{Sym}_{4}^{\tt{pd}}(\mathbb{R}_{> 0})$
    be of Form (II), such that
    $M_{11} = M_{22}$ and $M_{22} \neq M_{33} \neq M_{44}$ are pairwise distinct, and
    $\rho_{\tt{tensor}}(M) \neq 0$.
    Then, either $\PlEVAL(M)$ is $\#$P-hard, or
    $M$ is also isomorphic to a matrix
    of  Form (I).
\end{lemma}
\begin{proof}
    We first let $\zeta: N \mapsto (N_{11} - N_{33})(N_{11} - N_{44})
    (N_{22} - N_{33})(N_{22} - N_{44})(N_{33} - N_{44})$.
    We note that $\zeta(M) \neq 0$ by our choice of $M$.
    Moreover, we note that if $\zeta(N) \neq 0$, then
    $N$ cannot be isomorphic to any matrix of 
    Form (III), (IV), (V), or (VI),
    as there are at least three distinct diagonal elements. 
    We have already seen that there exists some 
    $\text{Sym}_{4}(\mathbb{R})$-polynomial
    $\zeta_{1}$ such that $\zeta_{1}(N) = 0$
    if and only if $N$ is isomorphic to a matrix of Form (I).
    Similarly, there exist $\text{Sym}_{4}(\mathbb{R})$-polynomials $\zeta_{2}, \dots, \zeta_{6}$
    for Forms (II) - (VI), respectively.
    By our assumption about $M$, we have 
    $\zeta_{i}(M) \neq 0$ for all $i \in \{3, 4, 5, 6\}$.
    We are given that $M$ has Form (II), thus $\zeta_{2}(M) = 0$.
    We will assume $\zeta_{1}(M) \neq 0$, and show  that $\PlEVAL(M)$ is $\#$P-hard.

    We will now define $\xi: N \mapsto (N_{14} - N_{23})^{2}
    + (N_{13} - N_{24})^{2}$.
    Since $M$ has Form (II), we have $\xi(M) = 0$.
    Moreover, for any $N$,  $\zeta(N) \neq 0$ and 
    $\xi(N) \neq 0$ implies that
    $N$ is not isomorphic to any matrix of Form (II).
    Indeed, to be isomorphic to a matrix of Form (II) and having the pairwise distinctness
    of $\{M_{11}, M_{33}, M_{44}\}$ and of $\{M_{22}, M_{33}, M_{44}\}$ given by  $\zeta(N) \neq 0$
    we only need to consider $N$ under possibly the permutations $1 \leftrightarrow 2$
    and $3 \leftrightarrow 4$; but these permutations do not change $\xi(N) \neq 0$. 
    In other words, $[\zeta(N) \neq 0]\wedge [\xi(N) \neq 0]$ implies that
    $\zeta_{2}(N) \neq 0$.
    So, if there exists some $\theta \in \mathbb{R}$
    such that $\xi(\mathcal{S}_{M}(\theta)) \neq 0$,
    we will be able to immediately find $N \in
    \mathfrak{R}(M, \{\rho_{\tt{tensor}},
    \zeta, \zeta_{1}, \xi\}) \cap
    \text{Sym}_{4}^{\tt{pd}}(\mathbb{R}_{> 0}) \subseteq
    \mathfrak{R}(M, \{\rho_{\tt{tensor}},
    \zeta, \zeta_{1}, \zeta_{2}\}) \cap
    \text{Sym}_{4}^{\tt{pd}}(\mathbb{R}_{> 0})$,
    by \cref{corollary: stretchingWorksPositive}.
    Since $\zeta(N) \neq 0$, we see that $N$ must either
    be diagonal distinct, or $N_{11} = N_{22}
    \neq N_{33} \neq N_{44}$ pairwise distinct.
    But since $\zeta_{i}(N) \neq 0$ for $i \in \{1, 2\}$,
    we see from \cref{theorem: diagDistHardness},
    and \cref{lemma: order211Forms}, that
    $\PlEVAL(N) \leq \PlEVAL(M)$ is $\#$P-hard.
    So, we may assume $\xi(\mathcal{S}_{M}(\theta)) = 0$ 
    for all $\theta \in \mathbb{R}$.

    We will now let
    $\mathcal{F}' = \{F': N \mapsto F(N^{3}) \ |\
    F \in \{\rho_{\tt{tensor}}, \zeta, \zeta_{1}\} \}$.
    Since $\mathcal{S}_{M}(0) = I$, we know that there
    exists some $\delta > 0$ such that
    $|\mathcal{S}_{M}(\theta)_{ij} - I_{ij}| < \frac{1}{3}$
    for all $i, j \in [4]$, for all
    $0 < \theta < \delta$.
    We can now use
    \cref{lemma: stretchingWorks} to find some
    $M' = \mathcal{S}_{M}(\theta^{*})
    \in \mathfrak{R}(M, \mathcal{F}' \cup 
    \{\rho_{\tt{tensor}}, \zeta, \zeta_{1}\})
    \cap \text{Sym}_{4}^{\tt{pd}}(\mathbb{R}_{\neq 0})$
    such that $0 < \theta^{*} < \delta$.
    We note that by our assumption,
    $\xi(M') = 0$.

    Now, since $\zeta(M') \neq 0$, we see that
    $(M')_{11} \neq (M')_{33} \neq (M')_{44}$ and $(M')_{22} \neq (M')_{33} \neq (M')_{44}$
    are both pairwise distinct.
    If $(M')_{11} \neq (M')_{22}$, $M'$ would be
    diagonal distinct, and
    $\PlEVAL(M') \leq \PlEVAL(M)$ is $\#$P-hard,
    due to \cref{theorem: diagDistHardness}.
    On the other hand, suppose $(M')_{11} = (M')_{22}$.
    If $M'$ is not of Form (II), we see that
    since $\zeta_{1}(M') \neq 0$, $\PlEVAL(M') \leq 
    \PlEVAL(M)$ is $\#$P-hard, due to
    \cref{lemma: order211Forms}.
    So, we may assume that $M'$ is of Form (II), such that
    $(M')_{11} = (M')_{22} \neq (M')_{33} \neq (M')_{44}$
    are pairwise distinct, and $\xi(M') = 0$.

    We will now consider $\mathcal{R}_{M'}(n)$
    for all $n \geq 1$.
    Let us first assume that $\xi(R_{n}(M')) = 0$
    for all $n \geq 1$.
    We note that
    $$(R_{n}(M'))_{14} - (R_{n}(M'))_{23}
    = \sum_{a, b \in [4]}(M'_{ab})^{n} \cdot
    (M'_{1a}M'_{4b} - M'_{2a}M'_{3b}).$$
    We will let $X = \{M'_{ab}: a, b \in [4]\}$, and define
    $c_{ij}(x) = \sum_{a, b \in [4]: M'_{ab} = x}
    (M'_{ia}M'_{jb})$ for all $i, j \in [4]$.
    So, we see that
    $$(R_{n}(M'))_{14} - (R_{n}(M'))_{23} = 
    \sum_{x \in X} x^{n} \cdot
    (c_{14}(x) - c_{23}(x)).$$
    Since each $x \in X$ is distinct, the equations
    $\xi(R_{n}(M')) = 0$ form a full rank
    Vandermonde system of size $O(1)$.
    This implies that $c_{14}(x) - c_{23}(x) = 0$
    for all $x \in X$.
    But since $\zeta(M') \neq 0$, and $0 < \theta^{*} < \delta$ in the definition of $M'$, we ensured
    that $M'_{44} = M'_{ab}$ if and only if
    $(a, b) = (4, 4)$, and
    we have also ensured that $M'_{33} = M'_{ab}$
    if and only if $(a, b) = (3, 3)$.
    Now, we see that
    $$c_{14}(M'_{44}) - c_{23}(M'_{44}) = 
    M'_{14}M'_{44} - M'_{24}M'_{34}.$$
    $$c_{14}(M'_{33}) - c_{23}(M'_{33}) = 
    M'_{13}M'_{34} - M'_{23}M'_{33}.$$
    Since $\xi(M') = \xi(\mathcal{S}_{M}(\theta^{*})) = 0$, we also know that
    $M'_{24} = M'_{13}$ and $M'_{23} =M'_{14} $.
    So, we also have
    $$c_{14}(M'_{44}) - c_{23}(M'_{44})  = M'_{14}M'_{44} - M'_{13}M'_{34}.$$
    $$c_{14}(M'_{33}) - c_{23}(M'_{33}) = 
    M'_{13}M'_{34} - M'_{14}M'_{33}.$$
    Then, it follows that 
    $$M'_{14}M'_{44} = M'_{13}M'_{34} = M'_{14}M'_{33},$$
    which implies that $M'_{33} = M'_{44}$, since $M'_{14} \ne 0$.
    But by our construction of $M'$, we have $\zeta(M') \ne 0$.
    This is a contradiction.

    Therefore, our assumption that $\xi(R_{n}(M')) = 0$
    for all $n \geq 1$ must be false.
    Let $\xi \circ R_{n}$ denote the composition of  $\xi$ and  $R_{n}$.
    There exists some $n^{*} \geq 1$ such that
    $(\xi \circ R_{n^{*}})(M') = 
    (\xi \circ R_{n^{*}})(\mathcal{S}_{M}(\theta^{*})) \neq 0$.
    So, \cref{corollary: stretchingWorksPositive} allows
    us to find some $M'' \in \mathfrak{R}(M,
    \mathcal{F}' \cup \{(\xi \circ R_{n^{*}})\}) \cap
    \text{Sym}_{4}^{\tt{pd}}(\mathbb{R}_{> 0})$.
    Since $F(\mathcal{R}_{M''}(1))
    = F((M'')^{3}) \neq 0$ for all
    $F \in \{\rho_{\tt{tensor}}, \zeta, \zeta_{1}\}$,
    and $\xi(\mathcal{R}_{M''}(n^{*})) \neq 0$,
    \cref{lemma: RInterpolationWorks} allows
    us to find some $N \in \mathfrak{R}(M, \{\rho_{\tt{tensor}},
    \zeta, \zeta_{1}, \xi\}) \cap
    \text{Sym}_{4}^{\tt{pd}}(\mathbb{R}_{> 0}) \subseteq
    \mathfrak{R}(M, \{\rho_{\tt{tensor}},
    \zeta, \zeta_{1}, \zeta_{2}\}) \cap
    \text{Sym}_{4}^{\tt{pd}}(\mathbb{R}_{> 0})$.

    Since $\zeta(N) \neq 0$, we see that $N$ must either be
    diagonal distinct, or $N_{11} = N_{22} \neq N_{33} \neq N_{44}$ pairwise distinct.
    If $N$ is diagonal distinct, \cref{theorem: diagDistHardness}
    implies that $\PlEVAL(N) \leq \PlEVAL(M)$ is
    $\#$P-hard.
    Otherwise,  since $\zeta_{i}(N) \neq 0$ for all
    $i \in \{1, 2\}$, \cref{lemma: order211Forms} implies
    that $\PlEVAL(N) \leq \PlEVAL(M)$ is $\#$P-hard.
\end{proof}

\begin{lemma}\label{lemma: formVReduction}
    Let $M \in \text{Sym}_{4}^{\tt{pd}}(\mathbb{R}_{> 0})$
    be of Form (V), such that
    $M_{11} = M_{22} = M_{33} \neq M_{44}$, and
    $\rho_{\tt{tensor}}(M) \neq 0$.
    Then, either $\PlEVAL(M)$ is $\#$P-hard, or
    $M$ is also isomorphic to a matrix
    of Form (I), or (IV).
\end{lemma}
\begin{proof}
    We will first let $\zeta: N \mapsto (N_{11} - N_{44})
    (N_{22} - N_{44})(N_{33} - N_{44})$.
    By our choice of $M$, we see that $\zeta(M) \neq 0$.
    Moreover, we note that $\zeta(N) \neq 0$ implies that
    $N$ cannot be isomorphic to any matrix of Form (III) or (VI),
    since in these two Forms, every diagonal element coincides with another diagonal element.
    We have already seen that there exist  
    $\text{Sym}_{4}(\mathbb{R})$-polynomials
    $\zeta_{1},  \dots, \zeta_{6}$ such that $\zeta_{i}(N) = 0$, for $i\in [6]$,
    if and only if $N$ is isomorphic to a matrix of the Form (I) - (VI), respectively.
    From our choice of $M$, we note that since
    $M_{11} = M_{22} = M_{33} \neq M_{44}$, it immediately
    follows that $\zeta_{3}(M) \neq 0$, and
    $\zeta_{6}(M) \neq 0$.
    We are given that $M$ is in Form (V) and hence $\zeta_{5}(M) = 0$.
    We will assume $\zeta_{i}(M) \neq 0$ for all
    $i \in \{1, 4\}$, and show that 
    $\PlEVAL(M)$ is $\#$P-hard.

    We will now define $\xi: N \mapsto (N_{14} - N_{23})^{2}
    + (N_{13} - N_{24})^{2} + (N_{12} - N_{34})^{2}$.
    We see that since $M$ is of Form (V), $\xi(M) = 0$.
    Moreover, if $\zeta(N) \neq 0$, and 
    $\xi(N) \neq 0$, that implies that
    $N$ is not isomorphic to any matrix of Form (V).
    Indeed, since  $\zeta(N) \neq 0$, $N_{44}$ is distinct from the other
    diagonal elements, and so the only permutations that could be  isomorphisms of $N$
    with a matrix of Form (V) must fix $4$. However the group $S_3$ of permutations on
    $\{1,2, 3\}$  is generated by $1 \leftrightarrow 2$
    and $1 \leftrightarrow 3$, both of which keep $\xi(N)$ invariant. 
    In other words, $[\zeta(N) \neq 0] \wedge [ \xi(N) \neq 0]$ implies that
    $\zeta_{5}(N) \neq 0$.
    So, if there exists some $\theta \in \mathbb{R}$
    such that $\xi(\mathcal{S}_{M}(\theta)) \neq 0$,
    we will be able to immediately find $N \in
    \mathfrak{R}(M, \{\rho_{\tt{tensor}},
    \zeta, \zeta_{1}, \zeta_{4}, \xi\}) \cap
    \text{Sym}_{4}^{\tt{pd}}(\mathbb{R}_{> 0}) \subseteq
    \mathfrak{R}(M, \{\rho_{\tt{tensor}},
    \zeta, \zeta_{1}, \zeta_{4}, \zeta_{5}\}) \cap
    \text{Sym}_{4}^{\tt{pd}}(\mathbb{R}_{> 0})$,
    by \cref{corollary: stretchingWorksPositive}.
    Since $\zeta(N) \neq 0$, we see that 
    $N_{11}, N_{22}, N_{33} \neq N_{44}$.
    If $N$ is diagonal distinct, then
    \cref{theorem: diagDistHardness} implies that
    $\PlEVAL(N) \leq \PlEVAL(M)$ is $\#$P-hard.
    Otherwise, if $N_{11}, N_{22}, N_{33}$ are not all identical,
    then from \cref{lemma: order211Forms}, we see that
unless $N$ is isomorphic to a matrix in  Form (I) or (II),
    $\PlEVAL(N) \leq \PlEVAL(M)$ is $\#$P-hard.
    But, we know that $\zeta_{1}(N) \neq 0$.
    So, $N$ cannot be  isomorphic to a matrix in Form (I).
    Moreover, from \cref{lemma: formIIReduction},
    we know that even if $N$ is isomorphic to a matrix in Form (II), since
    it is not isomorphic to any matrix of Form (I),
    $\PlEVAL(N) \leq \PlEVAL(M)$ is $\#$P-hard.
    This leaves only one possibility that
    $N_{11} = N_{22} = N_{33} \neq N_{44}$.
    Again, \cref{lemma: order31Forms} implies that
    unless $N$ is  isomorphic to a matrix in  Form (I), (II), (IV), or (V),
    $\PlEVAL(N) \leq \PlEVAL(M)$ is $\#$P-hard.
    But, by our construction of $N$, $\zeta_{i}(N) \neq 0$ for
    all $i \in \{1, 4, 5\}$.
    Similarly, even if $\zeta_{2}(N) = 0$, since
    $\zeta_{1}(N) \neq 0$,
    \cref{lemma: formIIReduction} implies that
    $\PlEVAL(N) \leq \PlEVAL(M)$ is $\#$P-hard.
    Therefore, we see that when 
    $\xi(\mathcal{S}_{M}(\theta)) \neq 0$ for  some $\theta \in \mathbb{R}$,
    $\PlEVAL(M)$ is $\#$P-hard.
    So, we may assume $\xi(\mathcal{S}_{M}(\theta)) = 0$ 
    for all $\theta \in \mathbb{R}$.

    We will now let
    $\mathcal{F}' = \{F': N \mapsto F(N^{3}) \ |\
    F \in \{\rho_{\tt{tensor}}, \zeta, \zeta_{1},
    \zeta_{4}\} \}$.
    Since $\mathcal{S}_{M}(0) = I$, we know that there
    exists some $\delta > 0$ such that
    $|\mathcal{S}_{M}(\theta)_{ij} - I_{ij}| < \frac{1}{3}$
    for all $i, j \in [4]$, for all
    $0 < \theta < \delta$.
    We can now use
    \cref{lemma: stretchingWorks} to find some
    $M' = \mathcal{S}_{M}(\theta^{*})
    \in \mathfrak{R}(M, \mathcal{F}' \cup 
    \{\rho_{\tt{tensor}}, \zeta, \zeta_{1}, \zeta_{4}\})
    \cap \text{Sym}_{4}^{\tt{pd}}(\mathbb{R}_{\neq 0})$
    such that $0 < \theta^{*} < \delta$.
    We note that by our assumption,
    $\xi(M') = 0$.

    Now, since $\zeta(M') \neq 0$, we see that
    $(M')_{44}$ is distinct from each of $(M')_{11}, (M')_{22}, (M')_{33}$.
    If $(M')_{11} \neq (M')_{22} \neq (M')_{33}$ are pairwise distinct, $M'$ would be
    diagonal distinct, and
    $\PlEVAL(M') \leq \PlEVAL(M)$ is $\#$P-hard,
    due to \cref{theorem: diagDistHardness}.
    On the other hand, if the \emph{set} $\{(M')_{11}, (M')_{22}, (M')_{33}\}$
    has cardinality 2, we know from
    \cref{lemma: order211Forms} that 
    $\PlEVAL(M') \leq \PlEVAL(M)$ is $\#$P-hard,
    unless $M'$ is isomorphic to a matrix in Form (I) or (II).
    But since $\zeta_{1}(M') \neq 0$, we know that
    $M'$ is not  isomorphic to a matrix in Form (I).
    Moreover, from \cref{lemma: formIIReduction},
    we also know that even if $\zeta_{2}(M') = 0$,
    since $\zeta_{1}(M') \neq 0$,
    $\PlEVAL(M') \leq \PlEVAL(M)$ is $\#$P-hard.
    Finally, we consider the case that
    $(M')_{11} = (M')_{22} = (M')_{33}$.
    From \cref{lemma: order31Forms}, we see that
    $\PlEVAL(M') \leq \PlEVAL(M)$ is $\#$P-hard,
    unless $M'$ is  isomorphic to a matrix in  Form (I), (II), (IV), or (V).
    Since $\zeta_{i}(M') \neq 0$ for $i \in \{1, 4\}$,
    we see that  Forms (I), or (IV) are ruled out for $M'$.
    Moreover, since $\zeta_{1}(M') \neq 0$,
    we see from \cref{lemma: formIIReduction} that
    even if $\zeta_{2}(M') = 0$,
    $\PlEVAL(M') \leq \PlEVAL(M)$ is $\#$P-hard.
    So, we may assume that $M'$ is of Form (V), such that
    $(M')_{11} = (M')_{22} = (M')_{33} \neq (M')_{44}$, 
    and $\xi(M') = 0$.
    
    We will now consider $\mathcal{R}_{M'}(n)$
    for all $n \geq 1$.
    Let us first assume that $\xi(R_{n}(M')) = 0$
    for all $n \geq 1$.
    We note that
    $$(R_{n}(M'))_{14} - (R_{n}(M'))_{23}
    = \sum_{a, b \in [4]}(M'_{ab})^{n} \cdot
    (M'_{1a}M'_{4b} - M'_{2a}M'_{3b}).$$
    We will let $X = \{M'_{ab}: a, b \in [4]\}$, and define
    $c_{ij}(x) = \sum_{a, b \in [4]: M'_{ab} = x}
    (M'_{ia}M'_{jb})$ for all $i, j \in [4]$.
    So, we see that
    $$(R_{n}(M'))_{14} - (R_{n}(M'))_{23} = 
    \sum_{x \in X} x^{n} \cdot
    (c_{14}(x) - c_{23}(x)).$$
    Since each $x \in X$ is distinct, the equations
    $\xi(R_{n}(M')) = 0$ form a full rank
    Vandermonde system of size $O(1)$.
    This implies that $c_{14}(x) - c_{23}(x) = 0$
    for all $x \in X$.
    But since $\zeta(M') \neq 0$, and $0 < \theta^{*} < \delta$ in the definition of $M'$, we ensured
    that $M'_{44} = M'_{ab}$ if and only if
    $(a, b) = (4, 4)$.
    Moreover, 
    from our construction of $M'$,
    we have also ensured that $\{(a, b): M'_{ab} = M'_{33}\} =
    \{(1, 1), (2, 2), (3, 3)\}$.
    Now, we see that
    $$c_{14}(M'_{44}) - c_{23}(M'_{44}) = 
    M'_{14}M'_{44} - M'_{24}M'_{34}.$$
    Since $\xi(M') = \xi(\mathcal{S}_{M}(\theta^{*})) = 0$, we also know that
    $M'_{24} = M'_{13}$, $M'_{23} = M'_{14}$,
    and $M'_{12} = M'_{34}$.
    So, we also have
    $$c_{14}(M'_{44}) - c_{23}(M'_{44})  = M'_{14}M'_{44} - M'_{13}M'_{34}.$$
    Moreover, we also see that
    $$c_{14}(M'_{33}) - c_{23}(M'_{33}) = 
    M'_{11}M'_{14} + M'_{12}M'_{24} + M'_{13}M'_{34} - 
    M'_{12}M'_{13} - M'_{22}M'_{23} - M'_{23}M'_{33}
    = M'_{13}M'_{34} - M'_{23}M'_{33},$$
    since $M'_{22}M'_{23} = M'_{11}M'_{14}$, and
    $M'_{12}M'_{13} = M'_{12}M'_{24}$.
    Then, it follows that 
    $$M'_{14}M'_{44} = M'_{13}M'_{34} = M'_{14}M'_{33},$$
    which implies that $M'_{33} = M'_{44}$, since $M'_{14} \ne 0$.
    But by our construction of $M'$, we have $\zeta(M') \ne 0$.
    This is a contradiction.

    Therefore, our assumption that $\xi(R_{n}(M')) = 0$
    for all $n \geq 1$ must be false.
    Let $\xi \circ R_{n}$ denote the composition of  $\xi$ and  $R_{n}$.
    There exists some $n^{*} \geq 1$ such that
    $(\xi \circ R_{n^{*}})(M') = 
    (\xi \circ R_{n^{*}})(\mathcal{S}_{M}(\theta^{*})) \neq 0$.
    So, \cref{corollary: stretchingWorksPositive} allows
    us to find some $M'' \in \mathfrak{R}(M,
    \mathcal{F}' \cup \{(\xi \circ R_{n^{*}})\}) \cap
    \text{Sym}_{4}^{\tt{pd}}(\mathbb{R}_{> 0})$.
    Since $F(\mathcal{R}_{M''}(1))
    = F((M'')^{3}) \neq 0$ for all
    $F \in \{\rho_{\tt{tensor}}, \zeta, \zeta_{1}, \zeta_{4}\}$,
    and $\xi(\mathcal{R}_{M''}(n^{*})) \neq 0$,
    \cref{lemma: RInterpolationWorks} allows
    us to find some $N \in \mathfrak{R}(M, \{\rho_{\tt{tensor}},
    \zeta, \zeta_{1}, \zeta_{4}, \xi\}) \cap
    \text{Sym}_{4}^{\tt{pd}}(\mathbb{R}_{> 0}) \subseteq
    \mathfrak{R}(M, \{\rho_{\tt{tensor}},
    \zeta, \zeta_{1}, \zeta_{4}, \zeta_{5}\}) \cap
    \text{Sym}_{4}^{\tt{pd}}(\mathbb{R}_{> 0})$.

    Once again, in view of $\zeta(N) \neq 0$, we consider the 
    cardinality of the \emph{set} $\{N_{11}, N_{22}, N_{33}\}$.
    If the cardinality is 3, then $N$ is diagonal distinct, and 
    \cref{theorem: diagDistHardness}
    implies that $\PlEVAL(N) \leq \PlEVAL(M)$ is
    $\#$P-hard.
    If the cardinality is 2, then since $\zeta_1(N) \ne 0$,  by  \cref{lemma: order211Forms} 
    and \cref{lemma: formIIReduction}, $\PlEVAL(N) \leq \PlEVAL(M)$ is
    $\#$P-hard.  If the cardinality is 1, then since $\zeta_4(N) \ne 0$ and $\zeta_5(N) \ne 0$,
    and also  $\zeta_1(N) \ne 0$, by \cref{lemma: order31Forms} and by
    \cref{lemma: formIIReduction},  $\PlEVAL(N) \leq \PlEVAL(M)$ is $\#$P-hard.
\end{proof}

\cref{lemma: order211Forms}, \cref{lemma: order22Forms},
\cref{lemma: order31Forms}, \cref{lemma: order4Forms}, 
\cref{lemma: formIIReduction}, and
\cref{lemma: formVReduction}
together imply the following theorem.

\begin{theorem}\label{theorem: nonDiagAllForms}
    Let $M \in \text{Sym}_{4}^{\tt{pd}}(\mathbb{R}_{> 0})$
    be a non-diagonally distinct matrix such that
    $\rho_{\tt{tensor}}(M) \neq 0$.
    Let $\mathcal{F}$ be a countable set of
    $\text{Sym}_{4}(\mathbb{R})$-polynomials
    such that $F(M) \neq 0$ for all $F \in \mathcal{F}$.
    Unless $M$ is isomorphic to a matrix
    of Form (I), (III), (IV), or (VI), 
    $\PlEVAL(M)$ is $\#$P-hard.
    Moreover, 
    \begin{enumerate}
        \item If $M$ is isomorphic to a matrix of Form (I), but not
        isomorphic to any matrix of Form (III), (IV), or (VI),
        then either $\PlEVAL(M)$ is $\#$P-hard, or
        there exists some $N \in \mathfrak{R}(M, \mathcal{F}) \cap
        \text{Sym}_{4}^{\tt{pd}}(\mathbb{R}_{> 0})$ such that
        $N$ is isomorphic to a matrix of Form (I), and
        $N_{11} = N_{22}$, but $N_{22} \neq N_{33} \neq N_{44}$
        are pairwise distinct.

        \item If $M$ is isomorphic to a matrix of Form (III), but not
        isomorphic to any matrix of Form (VI),
        then either $\PlEVAL(M)$ is $\#$P-hard, or
        there exists some $N \in \mathfrak{R}(M, \mathcal{F}) \cap
        \text{Sym}_{4}^{\tt{pd}}(\mathbb{R}_{> 0})$ such that
        $N$ is isomorphic to a matrix of Form (III), and
        $N_{11} = N_{22} \neq N_{33} = N_{44}$.

        \item If $M$ is isomorphic to a matrix of Form (IV), but not
        isomorphic to any matrix of Form (VI),
        then either $\PlEVAL(M)$ is $\#$P-hard, or
        there exists some $N \in \mathfrak{R}(M, \mathcal{F}) \cap
        \text{Sym}_{4}^{\tt{pd}}(\mathbb{R}_{> 0})$ such that
        $N$ is isomorphic to a matrix of Form (IV), and
        $N_{11} = N_{22} = N_{33} \neq N_{44}$.
    \end{enumerate}
\end{theorem}

\begin{remark*}
    In the statement of \cref{theorem: nonDiagAllForms}, a 4th
     item not listed explicitly but is logically implied is that
    $M$ can be isomorphic to a matrix of Form (VI).
    The enumeration is in a reverse order of the last appearance of
    (I), (III), (IV), or (VI),
    when $M$ is isomorphic to a matrix
    of Form (I), (III), (IV), or (VI). 
    After the implicit 4th
     item  of Form (VI), item 3 is when $M$ is in Form (IV)
    but not  (VI), item 2 is  when $M$ is in Form (III) but not (IV) or (VI),
    and item 1 is  when $M$ is in Form (I) but not (III) or (IV) or (VI).
    We note that if a matrix $M$ is both isomorphic to a matrix of
    Form (III), and also isomorphic to a matrix of Form (IV), then
    in fact all diagonal entries are equal, and all off diagonal entries are equal.
    To see this, in Form (IV) there are three equal diagonal entries and in Form (III)
    they come in two equal pairs $(a,a)$ and $(b,b)$. Thus $a=b$ and 
    all diagonal entries are equal. For the off diagonal entries, define two
    graphs $K$ and $K'$, both a copy of  $K_4$, but with
    labeled edges  according to Forms (III) and  (IV)
    respectively. In $K$ the edges will be labeled with $t$, $x$, $y$ and $z$,
    and the list of incident edges for the 4 vertices is
    $(tyz, tyz, xyz, xyz)$.
    In $K'$ we will label them $x'$ and $z'$, and the list of incident edges
    is $(x'x'z', x'x'z', x'x'z', z'z'z')$. Being isomorphic, from $z'z'z'$
    we have $y=z$ and either $t=y$ or $x=y$. Hence there are two equal triples
    in $K$, which implies that $x'=z'$ in $K'$, and hence all labels are equal.
%
Thus, such a matrix $M$ is in the Potts model, which is a special case of Form (VI).
    Consequently, in the 2nd (respectively, the 3rd)  item,
    when we assume that $M$ is isomorphic to a matrix of Form (III)
    (respectively, Form (IV)), but not isomorphic to any matrix
    of Form (VI), it follows that $M$ is 
    not isomorphic to any matrix of Form (IV) (respectively, Form (III))
    either.  (In particular, in item 2, it is not explicitly stated that
    $M$ is not in Form (IV), as being in  Form (III) but not Form (VI) already implies this.)
\end{remark*}
\section{Forms (I), (III), (IV) and (VI)}\label{sec: all_forms}

Following \cref{theorem: nonDiagAllForms}, we now only
need to deal with the matrices of the Forms (I), (III),
(IV), and (VI), as they are in \cref{theorem: nonDiagAllForms}.
We will deal with these Forms, one by one.

\begin{remark*}
    Due to \cref{theorem: nonDiagAllForms}, from this point onwards,
    when we refer to a matrix $M$ of Form (I), we may assume that
    $M_{11} = M_{22}$, but $M_{22} \neq M_{33} \neq M_{44}$ are
    pairwise distinct.
    Similarly, if we refer to a matrix of Form (III), we may
    assume that $M_{11} = M_{22} \neq M_{33} = M_{44}$, and if we refer to
     a matrix of Form (IV) we may also assume that it  satisfies
    $M_{11} = M_{22} = M_{33} \neq M_{44}$.
\end{remark*}

\subsection{Form (I)}\label{sec: form_I}

We will now deal with matrices of Form (I)
such that $M_{11} = M_{22}$, and 
$M_{22} \neq M_{33} \neq M_{44}$ are pairwise distinct.
Our strategy will be to show that,
if the matrix
$M \in \text{Sym}_{4}^{\tt{pd}}(\mathbb{R}_{> 0})$
is of Form (I) such that
$\rho_{\tt{tensor}}(M) \neq 0$, then
for any $\mathbf{0} \neq \mathbf{x} \in \chi_{4}$, we have
$\Psi_{\mathbf{x}}(M) \neq 0$.
Before we can jump into it, we will prove that
we may assume that the matrix $M$ of Form (I)
has some additional structural properties.
\begin{figure}[ht]
\[
\begin{pmatrix}
            a & x & y & z\\
            x & a & y & z\\
            y & y & b & t\\
            z & z & t & c
        \end{pmatrix}
        \]
        \caption{\label{Form-I-in-abc} Form (I)}
\end{figure}
\begin{lemma}\label{lemma: formIDistinct}
    Let $M \in \text{Sym}_{4}^{\tt{pd}}(\mathbb{R}_{> 0})$
    be of Form (I) such that $\rho_{\tt{tensor}}(M) \neq 0$,
    where  $M_{11} = M_{22}$, but $M_{22} \neq M_{33} \neq M_{44}$ are
    pairwise distinct.
    Let $\mathcal{F}$ be a countable set of 
    $\text{Sym}_{4}(\mathbb{R})$-polynomials such that
    $F(M) \neq 0$ for all $F \in \mathcal{F}$.
    Then, either $\PlEVAL(M)$ is $\#$P-hard, or
    there exists some $N \in
    \mathfrak{R}(M, \mathcal{F} \cup \{\rho_{\tt{tensor}} \}) \cap
    \text{Sym}_{4}^{\tt{pd}}(\mathbb{R}_{> 0})$ that is
    of Form (I), such that
    $N_{11} \neq N_{33} \neq N_{44}$ and
    $N_{22} \neq N_{33} \neq N_{44}$ are both pairwise distinct,
    $N_{12} \neq N_{13} \neq N_{14}$ are pairwise distinct,
    $N_{34} \neq N_{13}, N_{14}$, and
    $N_{ii} \neq N_{jk}$ for all $i \in [4]$ and for all $j \neq k \in [4]$.
\end{lemma}
\begin{remark*}
    The matrix $N$ of  Form (I) from the conclusion of \cref{lemma: formIDistinct} 
    stipulates that all diagonal elements are distinct except $M_{11} =M_{22}$, and are distinct
    from all off diagonal elements. And all  off diagonal elements denoted by distinct letters 
    in \cref{Form-I-in-abc} are distinct except possibly $x=t$.
\end{remark*}
\begin{proof}
    We will first define $\zeta: N \mapsto (N_{11} - N_{33})
    (N_{11} - N_{44})(N_{22} - N_{33})
    (N_{22} - N_{44})(N_{33} - N_{44})$, and
    $\zeta': N \mapsto \prod_{i \in [4], \ j \neq k \in [4]}(N_{ii} - N_{jk})$.
    We note that by our choice of $M$, $\zeta(M) \neq 0$.
    While we cannot claim that $\zeta'(M) \neq 0$,
    we note that $\zeta'(I) \neq 0$.
    We will let $\mathcal{F}' = \{F': N \mapsto F(N^{3})\ |\
    F \in \mathcal{F} \cup \{\rho_{\tt{tensor}}, \zeta,
    \zeta'\}\}$.
    We will now consider $\mathcal{S}_{M}(\theta)$.
    We know that $F'(\mathcal{S}_{M}(\theta)) \neq 0$
    for all $F' \in \mathcal{F}'$ for some
    $\theta \in \{\frac{1}{3}, 0\}$.
    We also know that there exists some $\delta > 0$
    such that $|\mathcal{S}_{M}(\theta)_{ij} - I_{ij}|
    < \frac{1}{3}$ for all $i, j \in [4]$, and
    $0 < \theta < \delta$.
    We can therefore use \cref{lemma: stretchingWorks} to find
    some $M' = \mathcal{S}_{M}(\theta^{*})
    \in \mathfrak{R}(M, \mathcal{F}' \cup \{\zeta, \zeta'\})
    \cap \text{Sym}_{4}^{\tt{pd}}(\mathbb{R}_{\neq 0})$
    for some $0 < \theta^{*} < \delta$.

    We will now consider $R_{n}(M')$ for all $n \geq 1$.
    We will let $X = \{M'_{ab}: a, b \in [4]\}$,
    and for each $x \in X$, and $i, j \in [4]$, define
    $c_{ij}(x) = \sum_{a, b \in [4]: M'_{ab} = x}
    (M'_{ia}M'_{jb})$.
    Then for any $i \neq j \in \{2, 3, 4\}$, we see that
    $$(R_{n}(M'))_{1i} - (R_{n}(M'))_{1j}
    = \sum_{x \in X}x^{n} \cdot 
    (c_{1i}(x) - c_{1j}(x)).$$
    Let us assume that for some $i \neq j \in \{2, 3, 4\}$,
    $(R_{n}(M'))_{1i} - (R_{n}(M'))_{1j} = 0$
    for all $n \geq 1$.
    In that case, we see that we have a full rank
    Vandermonde system of linear equations, which implies
    that $c_{1i}(x) - c_{1j}(x) = 0$
    for all $x \in X$.

    From our construction of $M'$, we know that
    $M'_{ab} = M'_{33}$ implies that $(a, b) = (3, 3)$, and that
    $M'_{ab} = M'_{44}$ implies that $(a, b) = (4, 4)$.  This is because by choosing
    $0 < \theta^{*} < \delta$ our $M'$ is close to $I$, and thus diagonal elements are all distint
    from off diagonal elements, and $\zeta(M') \ne 0$ separates $M'_{33}$ and $M'_{44}$ from each other and
    also from the other diagonal
    elements. 
    So, we see that
    $$c_{1i}(M'_{33}) - c_{1j}(M'_{33}) = 
    M'_{13}M'_{i3} - M'_{13}M'_{j3}, ~~\text{ and }~~
    c_{1i}(M'_{44}) - c_{1j}(M'_{44}) = 
    M'_{14}M'_{i4} - M'_{14}M'_{j4}.$$
    We note that by our choice of $i \neq j \in \{2, 3, 4\}$,
    at least one of $i, j \in \{3, 4\}$.
    Without loss of generality, we may assume that it is $i$.
    If $i = 3$, that implies that
    $M'_{13}M'_{33} = M'_{13}M'_{j3}$ for some
    $j \neq 3$, and if $i = 4$, that implies that
    $M'_{14}M'_{44} = M'_{14}M'_{j4}$ for some
    $j \neq 4$.
    As $M'_{13}, M'_{14} \ne 0$,
    in either case, we get a contradiction, since
    by our choice of $M'$, no diagonal entry can
    be equal to a non-diagonal entry.
    Therefore, for each $i \neq j \in \{2, 3, 4\}$,
    there exists some $n_{ij} \geq 1$ such that
    $(R_{n_{ij}}(M'))_{1i} - (R_{n_{ij}}(M'))_{1j} \neq 0$.

    Similarly, for $i \in \{3, 4\}$, we see that
    for all $n \geq 1$,
    $$(R_{n}(M'))_{1i} - (R_{n}(M'))_{34}
    = \sum_{x \in X}x^{n} \cdot 
    (c_{1i}(x) - c_{34}(x)).$$
    Let us assume that for some $i \in \{3, 4\}$,
    $(R_{n}(M'))_{1i} - (R_{n}(M'))_{34} = 0$
    for all $n \geq 1$.
    Once again, since this forms a full rank
    Vandermonde system of linear equations, we see that
    $c_{1i}(x) - c_{34}(x) = 0$ for all $x \in X$.
    If $i = 3$, we will consider
    $c_{13}(M'_{44}) - c_{34}(M'_{44})
    = M'_{14}M'_{34} - M'_{34}M'_{44} = 0$,
    which implies that $M'_{14} = M'_{44}$, and if $i = 4$,
    we consider $c_{14}(M'_{33}) - c_{34}(M'_{33})
    = M'_{13}M'_{34} - M'_{33}M'_{34} = 0$,
    which implies that $M'_{13} = M'_{33}$.
    Since neither of these are possible, we
    conclude that for each $i \in \{3, 4\}$,
    there exists some $n'_{i} \geq 1$ such that
    $(R_{n'_{i}}(M'))_{1i} - (R_{n'_{i}}(M'))_{34} \neq 0$.

    For each $i \neq j \in \{2, 3, 4\}$, we can now
    construct a $\text{Sym}_{4}(\mathbb{R})$-polynomial
    $\xi_{ij}: N \mapsto
    (R_{n_{ij}}(N))_{1i} - (R_{n_{ij}}(N))_{1j}$.
    For each $i \in \{3, 4\}$, we can also construct
    $\xi'_{i}: N \mapsto
    (R_{n'_{i}}(N))_{1i} - (R_{n'_{i}}(N))_{34}$.
    We have seen that
    $\xi_{ij}(\mathcal{S}_{M}(\theta^{*})) \neq 0$
    for all $i \neq j \in \{2, 3, 4\}$,
    and $\xi'_{i}(\mathcal{S}_{M}(\theta^{*})) \neq 0$,
    for all $i \in\{3, 4\}$.
    So, we can use \cref{corollary: stretchingWorksPositive}
    to find some $M'' \in \mathfrak{R}(M,
    \mathcal{F}' \cup \{\xi_{23}, \xi_{24}, \xi_{34},
    \xi'_{3}, \xi'_{4}\}) \cap
    \text{Sym}_{4}^{\tt{pd}}(\mathbb{R}_{> 0})$.
    But then, we see that
    $F(\mathcal{R}_{M''}(1)) \neq 0$ for all
    $F \in \mathcal{F} \cup \{\rho_{\tt{tensor}}, \zeta, \zeta'\}$.
    Furthermore, let ${\xi}^o_{ij}: N \mapsto N_{1i} - N_{1j}$
    for all $i \neq j \in \{2, 3, 4\}$, and let
    ${\xi'}^o_{i}: N \mapsto N_{1i} - N_{34}$ for all $i \in \{3, 4\}$.
    Then
    ${\xi}^o_{ij}(\mathcal{R}_{M''}(n_{ij})) \neq 0$
    for all $i \neq j \in \{2, 3, 4\}$, and
    ${\xi'}^o_{i}(\mathcal{R}_{M''}(n'_{i})) \neq 0$
    for all $i \in \{3, 4\}$.  
    So, \cref{lemma: RInterpolationWorks} allows us
    to find the required
    $N \in \mathfrak{R}(M, \mathcal{F} \cup
    \{\rho_{\tt{tensor}}, \zeta, \zeta', \xi^o_{23}, \xi^o_{24}, \xi^o_{34},
    \xi'^o_{3}, \xi'^o_{4}\})
    \cap \text{Sym}_{4}^{\tt{pd}}(\mathbb{R}_{> 0})$.
    
    Since $\zeta(N) \neq 0$, we see that
    $N_{11} \neq N_{33} \neq N_{44}$, and
    $N_{22} \neq N_{33} \neq N_{44}$ are both pairwise distinct.
    If $N_{11} \neq N_{22}$, then $N$
    is diagonal distinct, and since
    $\rho_{\tt{tensor}}(N) \neq 0$,
    \cref{theorem: diagDistHardness}  allows
    us to conclude that $\PlEVAL(M)$ is
    $\#$P-hard.
    Otherwise, if $N$ is not of Form (I), then
    \cref{lemma: order211Forms} and \cref{lemma: formIIReduction}
    tell us once
    again that $\PlEVAL(N)$ is $\#$P-hard.
    So, the only other possibility is that $N$ is of Form (I)
    such that $N_{12} \neq N_{13} \neq N_{14}$ are pairwise distinct (because
    $\xi^o_{23}(N) \ne 0, \xi^o_{24}(N) \ne 0, \xi^o_{34}(N) \ne 0$),
    and
    $N_{34} \neq N_{13}, N_{14}$ (because 
    $\xi'^o_{3}(N) \ne 0, \xi'^o_{4}(N) \ne 0$),
    and $N_{ii} \neq N_{jk}$ for all $i \in [4]$ and for all $j \neq k \in [4]$
    (because $\zeta'(N) \ne 0$).
\end{proof}

\cref{lemma: formIDistinct} allows us to claim that
for any $M \in \text{Sym}_{4}^{\tt{pd}}(\mathbb{R}_{> 0})$
 of Form (I) (such that $\rho_{\tt{tensor}}(M) \neq 0$),
we may assume that $M_{12} \neq M_{13} \neq M_{14}$ are
pairwise distinct, and that $M_{34} \neq M_{13}, M_{14}$.
Ideally, we would like to be able to say that we can assume
that $M_{12} \neq M_{34}$ as well.
That is however, not true in general.
As it turns out, it is possible that $M_{12} = M_{34}$,
but in that case, we will be able to show that
$\rho_{\tt{tensor}}(M) = 0$.

\begin{lemma}\label{lemma: formINotTensor}
    Let $M \in \text{Sym}_{4}^{\tt{pd}}(\mathbb{R}_{> 0})$
    be of Form (I) such that $\rho_{\tt{tensor}}(M) \neq 0$.
    Let $\mathcal{F}$ be a countable set of 
    $\text{Sym}_{4}(\mathbb{R})$-polynomials such that
    $F(M) \neq 0$ for all $F \in \mathcal{F}$.
    Then, either $\PlEVAL(M)$ is $\#$P-hard, or
    there exists some $N \in
    \mathfrak{R}(M, \mathcal{F} \cup \{\rho_{\tt{tensor}} \}) \cap
    \text{Sym}_{4}^{\tt{pd}}(\mathbb{R}_{> 0})$ that is
    of Form (I), such that
    either $N_{12} \neq N_{34}$, or
    $(N_{11})^{2} \neq N_{33}N_{44}$.
\end{lemma}
\begin{proof}
    We will first define $\zeta: N \mapsto (N_{11} - N_{33})
    (N_{11} - N_{44})(N_{22} - N_{33})
    (N_{22} - N_{44})(N_{33} - N_{44})$ as
    in the proof of \cref{lemma: formIDistinct}.
    From our choice of $M$, we note that $\zeta(M) \neq 0$.
    We will now let $\mathcal{F}' = \{F': N \mapsto F(N^{3})\ |\
    F \in \mathcal{F} \cup \{\rho_{\tt{tensor}}, \zeta \}\}$.
    We note that $F'(\mathcal{S}_{M}(\frac{1}{3})) \neq 0$
    for all $F' \in \mathcal{F}'$.
    We also know that there exists some $\delta > 0$ such that
    for all $0 < \theta < \delta$,
    $|\mathcal{S}_{M}(\theta)_{ij} - I_{ij}| < \frac{1}{3}$.
    We can now use
    \cref{lemma: stretchingWorks} to find
    $M' = \mathcal{S}_{M}(\theta^{*}) \in 
    \mathfrak{R}(M, \mathcal{F}' \cup \{\zeta \}) \cap
    \text{Sym}_{4}^{\tt{pd}}(\mathbb{R}_{\neq 0})$,
    such that $0 < \theta^{*} < \delta$.
    Let us assume that $M'$ is not of Form (I).
    Since $\zeta(M') \neq 0$, we know that
    $(M')_{11} \neq (M')_{33} \neq (M')_{44}$ are
    pairwise distinct, and $(M')_{22} \neq (M')_{33} \neq (M')_{44}$
    are also pairwise distinct.
    If $(M')_{11} \neq (M')_{22}$, then
    $M'$ is diagonal distinct, and \cref{theorem: diagDistHardness}
    implies that $\PlEVAL(M') \leq \PlEVAL(M)$ is
    $\#$P-hard.
    Otherwise, \cref{lemma: order211Forms} implies that
    unless $M'$ is of Form (II),
    $\PlEVAL(M') \leq \PlEVAL(M)$ is $\#$P-hard.
    Finally, if $M'$ has  Form (II), and does not have Form (I), then \cref{lemma: formIIReduction} implies that 
    $\PlEVAL(M') \leq \PlEVAL(M)$ is $\#$P-hard.
    So, we may assume that in fact,
    $M'$ is of Form (I).
    We also note that if $M'_{12} \neq M'_{34}$, or
    if $(M'_{11})^{2} \neq M'_{33}M'_{44}$,
    we can immediately use \cref{corollary: stretchingWorksPositive}
    to find the required $N$.
    So, we may assume otherwise: $M'_{12} = M'_{34}$ and
     $(M'_{11})^{2} = M'_{33}M'_{44}$. In terms of notations in \cref{Form-I-in-abc},
     we have $x=t$ and $a^2 = bc$.

    We will now consider $R_{n}(M')$ for all $n \geq 1$.
    We will let $X = \{(M')_{ab}: a, b \in [4] \}$, and
    define $c_{ij}(x) = \sum_{a, b \in [4]: (M')_{ab} = x}
    M'_{ia}M'_{jb}$.
    We note that
    $$R_{n}(M')_{12} - R_{n}(M')_{34} = \sum_{x \ in X}
    x^{n} \cdot (c_{12}(x) - c_{34}(x)).$$
    If $R_{n}(M')_{12} - R_{n}(M')_{34} = 0$
    for all $n \geq 1$, then we have a Vandermonde
    system of linear equations of size $O(1)$.
    That implies that $c_{12}(x) - c_{34}(x) = 0$
    for all $x \in X$.
    
    From our construction of $M'$, we know that
    $M'_{ab} = M'_{33}$ implies that $(a, b) = (3, 3)$,
    $M'_{ab} = M'_{44}$ implies that $(a, b) = (4, 4)$,
    and that $M'_{ab} = M'_{11}$ implies that
    $(a, b) \in \{(1, 1), (2, 2)\}$.
    So, we see that
    \begin{align*}
        c_{12}(M'_{33}) - c_{34}(M'_{33})
        &= M'_{13}M'_{13} - M'_{33}M'_{34},\\
        c_{12}(M'_{44}) - c_{34}(M'_{44})
        &= M'_{14}M'_{14} - M'_{34}M'_{44},\\
        c_{12}(M'_{11}) - c_{34}(M'_{11})
        &= M'_{11}M'_{12} + M'_{12}M'_{11}
        - M'_{13}M'_{14}  -M'_{13}M'_{14} \\
        & =  2 M_{11}M'_{12} - 2 M'_{13}M'_{14},
    \end{align*}
    and they all equal to 0.
    (Here we used the fact that $M'_{23} = M'_{13}, M'_{24} = M'_{14}, M'_{11} = M'_{22}$ in Form (I).)
    In terms of the  notations in \cref{Form-I-in-abc} (with
    $x = t$, and $a^{2} = bc$),
     we have $y^2 = bx$, $z^2= cx$, and $ax = yz$.
Together with $x=t$ and $a^2 = bc$, we can verify that $\rho_{\tt{tensor}}(M') = 0$, where
    $\rho_{\tt{tensor}}$ is defined in \cref{lemma: tensorPolynomial}.
    (We use the switching $2 \leftrightarrow 4$ in  $\varrho_{\tt{tensor}}$ to $\rho_{\tt{tensor}}$.)

    But by construction of $M'$, we ensured that
    $\rho_{\tt{tensor}}(M') \neq 0$.
    Therefore, there must exist some $n \geq 1$ such that
    $R_{n}(M')_{12} - R_{n}(M')_{34} \neq 0$.
    So, if we define $\xi: N \mapsto R_{n}(N)_{12} - R_{n}(N)_{34}$,
    we see that $\xi(\mathcal{S}_{M}(\theta^{*}))
    = \xi(M') \neq 0$.
    Now, we can use \cref{corollary: stretchingWorksPositive}
    to find $M'' \in \mathfrak{R}(M, \mathcal{F}'
    \cup \{\xi\}) \cap
    \text{Sym}_{4}^{\tt{pd}}(\mathbb{R}_{> 0})$.
    Finally, \cref{lemma: RInterpolationWorks} allows us
    to find $N \in \mathfrak{R}(M, \mathcal{F} \cup
    \{\rho_{\tt{tensor}}, \zeta, \xi^{o} \}) \cap
    \text{Sym}_{4}^{\tt{pd}}(\mathbb{R}_{> 0})$, where
    $\xi^{o}: N \mapsto N_{12} - N_{34}$.
    Since $\zeta(N) \neq 0$, $N$ may either be
    diagonal distinct, or of the form
    $N_{11} = N_{22}$ with $N_{22} \neq N_{33} \neq N_{44}$
    pairwise distinct.
    If $N$ is diagonal distinct, then \cref{theorem: diagDistHardness}
    implies that $\PlEVAL(N) \leq \PlEVAL(M)$ is
    $\#$P-hard.
    Otherwise, \cref{lemma: order211Forms} implies that
    unless $N$ is either of Form (I) or of Form (II),
    $\PlEVAL(N) \leq \PlEVAL(M)$ is $\#$P-hard.
    Finally, \cref{lemma: formIIReduction} implies that
    unless $N$ is of Form (I), $\PlEVAL(N) \leq \PlEVAL(M)$
    is $\#$P-hard.
    Since $\xi^{o}(N) \neq 0$, we know that
    $N_{12} \neq N_{34}$.
    So, $N$ is the required matrix.
\end{proof}

\begin{remark*}
    We note that in the statement of \cref{lemma: formINotTensor},
    we claim to be able to find some $N$ such that
    $N_{12} \neq N_{34}$, or $(N_{11})^{2} \neq N_{33}N_{44}$,
    however, at the end of the proof, we were able to find
    some $N$ such that $N_{12} \neq N_{34}$.
    So, it appears as if we could eliminate the option that
    $(N_{11})^{2} \neq N_{33}N_{44}$ from the statement of the lemma.
    However, we cannot do that.
    We note that in the proof of \cref{lemma: formINotTensor},
    when we constructed $M'$, if $M'_{12} = M'_{34}$, but
    $(M'_{11})^{2} \neq M'_{33}M'_{44}$, 
    we could then have produced some $N$ such that 
    $N_{12} = N_{34}$, but 
    $(N_{11})^{2} \neq N_{33}N_{44}$.
\end{remark*}

We will now show that given a matrix
$M \in \text{Sym}_{4}^{\tt{pd}}(\mathbb{R}_{> 0})$
of Form (I), and any $\mathbf{0} \neq \mathbf{x}
\in \chi_{4}$ of support size $> 2$,
we can find some $N \in \mathfrak{R}(M, \mathcal{F}_{M}
\cup \{\phi_{\mathbf{x}}\})
\cap\text{Sym}_{4}^{\tt{pd}}(\mathbb{R}_{> 0})$.
We note that if $\mathbf{0} \neq \mathbf{x}
\in \chi_{4}$ has support size $\leq 2$, it must be
of the form $\mathbf{x} = c \cdot \delta_{ij}$
for some $\delta_{ij} \in \mathcal{D}$
(as defined in \cref{{equation: mathcalD}}).
We will show later that our ability to use
\cref{theorem: latticeHardness}
is not hindered by such $\mathbf{x} \in \chi_{4}$,
so we do not have to worry about them.

\begin{lemma}\label{lemma: formIReduction}
    Let $M \in \text{Sym}_{4}^{\tt{pd}}(\mathbb{R}_{> 0})$
    be of Form (I) such that $\rho_{\tt{tensor}}(M) \neq 0$.
    Let $\mathcal{F}$ be a countable set of 
    $\text{Sym}_{4}(\mathbb{R})$-polynomials such that
    $F(M) \neq 0$ for all $F \in \mathcal{F}$.
    Let $\xi_{\mathbf{x}}$ be a $\text{Sym}_{4}(\mathbb{R})$-polynomial
    such that $\xi_{\mathbf{x}}: N \mapsto \phi_{\mathbf{x}}
    (N_{11}, \dots, N_{44})$ for any
    $\mathbf{0} \neq \mathbf{x} \in \chi_{4}$ of 
    support size greater than $2$.
    Then, either $\PlEVAL(M)$ is $\#$P-hard, or
    there exists some $N \in \mathfrak{R}
    (M, \mathcal{F} \cup \{\rho_{\tt{tensor}}, \xi_{\mathbf{x}}\}) \cap
    \text{Sym}_{4}^{\tt{pd}}(\mathbb{R}_{> 0})$
    of Form (I).
\end{lemma}
\begin{proof}
    We will first use \cref{lemma: formINotTensor} to find
    $M' \in \mathfrak{R}(M, \mathcal{F} \cup \{\rho_{\tt{tensor}}\})
    \cap \text{Sym}_{4}^{\tt{pd}}(\mathbb{R}_{> 0})$ that is
    of Form (I), such that
    either $(M')_{12} \neq (M')_{34}$, or
    $((M')_{11})^{2} \neq (M')_{33}(M')_{44}$.  (In  \cref{Form-I-in-abc}, this is $x \ne t$ or  $a^2 \ne bc$.)
    We now define $\zeta: N \mapsto
    (N_{11} - N_{33})(N_{11} - N_{44})(N_{22} - N_{33})(N_{22} - N_{44})
    (N_{33} - N_{44})$.
    We can now construct a matrix $(M')^{\sigma}$ by permuting
    both the rows and columns of $M$ by some $\sigma \in S_{4}$
    (that switches $2 \leftrightarrow 4$). Then we see that the
    condition that either $(M')_{12} \neq (M')_{34}$, or
    $((M')_{11})^{2} \neq (M')_{33}(M')_{44}$ means exactly that
    either $(M')^{\sigma}_{14} \neq (M')^{\sigma}_{23}$, or 
    $(M')^{\sigma}_{11}(M')^{\sigma}_{44} \neq (M')^{\sigma}_{22}
    (M')^{\sigma}_{33}$. (By Form (I), $M'_{11} = M'_{22}$, which gives $(M')^{\sigma}_{11} = (M')^{\sigma}_{44}$.)
    This is precisely the condition for \cref{lemma: notConfluentImpliesM14M23}
    that allows us to find some
    $(M'')^{\sigma} \in \mathfrak{R}((M')^{\sigma},
    \mathcal{F}^{\sigma} \cup
    \{\xi_{(1, -1, -1, 1)}\})
    \cap \text{Sym}_{4}^{\tt{pd}}(\mathbb{R}_{> 0})$,
    where
    $$\mathcal{F}^{\sigma} = \left\{F^{\sigma}: N \mapsto
    F(N^{\sigma^{-1}})\ |\ F \in \mathcal{F} \cup 
    \{\rho_{\tt{tensor}}, \zeta\}\right\}.$$
    If we now look at $(M'')$ (by switching back $2 \leftrightarrow 4$),
    we see that $\xi_{(1, 1, -1, -1)}(M'') \neq 0$, and that
    $F(M'') \neq 0$ for all $F \in \mathcal{F} \cup
    \{\rho_{\tt{tensor}}, \zeta\}$. 
    From the construction of $\zeta$, we can see that
    if $M''_{11} \neq M''_{22}$, then $M''$ is diagonal
    distinct, and \cref{theorem: diagDistHardness} implies that
    $\PlEVAL(M'') \leq \PlEVAL(M)$ is $\#$P-hard.
    On the other hand, if $M''_{11} = M''_{22}$,
    then \cref{lemma: order211Forms},
    and \cref{lemma: formIIReduction} together imply that
    unless $M''$ is of Form (I), then
    $\PlEVAL(M'') \leq \PlEVAL(M)$ is $\#$P-hard.
    So, we may assume that $M''$ is in fact, of Form (I).
    
    Finally, we can use \cref{lemma: formIDistinct},
    to find $M''' \in \mathfrak{R}(M'',
    \mathcal{F} \cup \{\rho_{\tt{tensor}}, \xi_{(1, 1, -1, -1)}, \zeta\})
    \cap \text{Sym}_{4}^{\tt{pd}}(\mathbb{R}_{> 0})$ such that
    $M'''_{12} \neq M'''_{13} \neq M'''_{14}$ are pairwise distinct,
    $M'''_{34} \neq M'''_{13}, M'''_{14}$, and
    that $M'''_{ii} \neq M'''_{jk}$ for all $i \in [4]$, 
    and $j \neq k \in [4]$.
    We note that since $\zeta(M''') \neq 0$, if 
    $M'''$ is not of Form (I), then
    \cref{theorem: nonDiagAllForms} implies that
    $\PlEVAL(M)$ is $\#$P-hard.
    So, we may assume that $M'''$ is of Form (I).
    For convenience, we will rename this $M'''$ as $M$.
    
    Let us consider any $\mathbf{0} \neq  \mathbf{x} =
    (x_{1}, x_{2}, x_{3}, x_{4})  \in \chi_{4}$.
    If $\xi_{\mathbf{x}}(M) = \phi_{\mathbf{x}}(M_{11}, \dots, M_{44}) \neq 0$, then
    we are already done. So we may assume that  $\xi_{\mathbf{x}}(M) =  0$.
    Let us identify those $\mathbf{x}$
    for which $\xi_{\mathbf{x}}(M) = 0$.
    We already know that $\xi_{(1, 1, -1, -1)}(M) \neq 0$.
    For any $\mathbf{x}$
    such that $x_{1} + x_{2} = 2c$, and $x_{3} = x_{4} = -c$, for any non-zero $c \in \mathbb{Z}$,
    since we have $M_{11} = M_{22}$, this implies that
    $\xi_{\mathbf{x}}(M) = \left( \xi_{(1, 1, -1, -1)}(M)\right)^c \neq 0$.
%
    We will now make use of the following claim,
    which we shall prove shortly.
    
    \begin{claim}\label{claim: killedX}
        Let $\alpha_{1}, \dots, \alpha_{4} \in \mathbb{R}$,
        such that $\alpha_{1} = \alpha_{2}$, and
        $|\alpha_{2}| \neq |\alpha_{3}| \neq |\alpha_{4}|$
        are pairwise distinct.
        If $\phi_{\mathbf{x}}(\alpha_{1}, \dots, \alpha_{4}) = 0$
        for some $\mathbf{x} \in \chi_{4}$ with support size
        greater than $2$,
        then $x_{3} \neq 0$, $x_{4} \neq 0$, and
        $x_{1} + x_{2} \neq 0$.
    \end{claim}

    Since $M \in \text{Sym}_{4}(\mathbb{R}_{> 0})$ and is
    of Form (I), we see that
    $M_{11} = M_{22}$, and $|M_{22}| \neq |M_{33}| \neq |M_{44}|$
    are pairwise distinct. So, \cref{claim: killedX}
    lets us conclude that $\xi_{\mathbf{x}}(M) = 0$
implies $x_{3}, x_{4} \neq 0$, and
    $x_{1} + x_{2} \neq 0$.

    We will now let the entries of $M$ be
    generated by some $\{g_{t}\}_{t \in [d]}$.
    We may assume that $M$ is replaced with some
    $c M$ as in \cref{lemma: MequivalentCM} such that
    $e_{ijt} \geq 0$ for all $i, j \in [4]$, and $t \in [d]$.
    Now, we will let $\mathfrak{m} = 
    \max_{i, j \in [4], t \in [d]}(e_{ijt})$, pick
    some $m > \mathfrak{m}$, and define
    $z_{ij} = \sum_{t \in [d]}m^{t} \cdot e_{ijt}$
    for all $i, j \in [4]$.
    We will now define $\mathcal{T}_{M}^{*}:
    \mathbb{R} \rightarrow \text{Sym}_{4}(\mathbb{R})$ 
    such that
    $$\mathcal{T}_{M}^{*}(p)_{ij} =
    \mathcal{T}_{M}(p^{m}, p^{m^{2}}, \dots, p^{m^{d}})_{ij}
    = p^{z_{ij}}.$$
    From our choice of $m$, we see that
    $z_{ij} = z_{i'j'}$ if and only if
    $(e_{ij1}, \dots, e_{ijd}) =
    (e_{i'j'1}, \dots, e_{i'j'd})$.
    So, from our choice of $M$, we see that in fact,
    $z_{12} \neq z_{13} \neq z_{14}$,
    $z_{34} \neq z_{13}, z_{14}$, and
    $z_{ii} \neq z_{ij}$ for all $i \neq j \in [4]$.

    We will now consider $\xi_{\mathbf{x}}
    ((\mathcal{T}_{M}^{*}(p))^{2})$ for all $p \in \mathbb{R}$.
    Since
    $x_{1} + x_{2} \neq 0$,
    we may assume without loss of generality that
    $x_{1} + x_{2} > 0$.
    We may also assume by symmetry that $x_{3} \geq x_{4}$.
    There are now two possibilities we have to deal with:
    $x_{3}$ may be positive, in which case $x_{4}$ must
    be negative (since $x_{1} + \cdots + x_{4} = 0$),
    or both $x_{3}$ and $x_{4}$ may be negative.
    (By  \cref{claim: killedX},
     $x_{3}, x_{4} \neq 0$.)
    
    We will first deal with the case where $x_{3} > 0$.
    In this case, we see that
    \begin{multline*}
        \xi_{\mathbf{x}}(\mathcal{T}_{M}^{*}(p)^{2}) = 
        (p^{2z_{11}} + p^{2z_{12}} + 
        p^{2z_{13}} + p^{2z_{14}})^{x_{1} + x_{2}} \cdot
        (2p^{2z_{13}} + p^{2z_{33}} + p^{2z_{34}})^{x_{3}}\\
        - (2p^{2z_{14}} + p^{2z_{34}} +
        p^{2z_{44}})^{x_{1} + x_{2} + x_{3}}.
    \end{multline*}
    Here we used the fact that  $\mathcal{T}_{M}^{*}(p)$ has Form (I)
    (because $M$ does)
    and thus $(\mathcal{T}_{M}^{*}(p)^2)_{11} = (\mathcal{T}_{M}^{*}(p)^2)_{22}$. We note that by our choice of $z_{ij}$,
    the exponents $2z_{11} \neq 2z_{12} \neq 2z_{13} \neq 2z_{14}$ are
    pairwise distinct,
    $2z_{13} \neq 2z_{33} \neq 2z_{34}$ are pairwise
    distinct, and
    $2z_{14} \neq 2z_{34} \neq 2z_{44}$ are pairwise
    distinct.
    So, the following claim (which we shall prove soon)
    allows us to claim that
    $\xi_{\mathbf{x}}(\mathcal{T}_{M}^{*}(p)^{2})$ is
    not the zero polynomial.
    
    \begin{claim}\label{claim: formICase1Impossible}
        Let $n, m \in  \mathbb{Z}_{> 0}$, and $y_{11} < y_{12} < y_{13} < y_{14}$,
        $y_{21} < y_{22} < y_{23}$,
        and $y_{31} < y_{32} < y_{33} \in \mathbb{Z}_{> 0}$.
        Assume the multiset $\{\alpha_{1}, \alpha_{2}, \alpha_{3}\} =
        \{\beta_{1}, \beta_{2}, \beta_{3}\} = \{2, 1, 1\}$.
       Define the polynomial
        $f: \mathbb{R} \rightarrow \mathbb{R}$ such that
        $$f(p) = (p^{y_{11}} + p^{y_{12}} +
        p^{y_{13}} + p^{y_{14}})^{n} \cdot 
        (\alpha_{1}p^{y_{21}} + \alpha_{2}p^{y_{22}} 
        + \alpha_{3}p^{y_{23}})^{m} - 
        (\beta_{1}p^{y_{31}} + \beta_{2}p^{y_{32}} 
        + \beta_{3}p^{y_{33}})^{n+m},$$
        then $f(p)$ is not the zero polynomial.    
    \end{claim}

    On the other hand, if $x_{3} < 0$, we see that
    \begin{multline*}
        \xi_{\mathbf{x}}(\mathcal{T}_{M}^{*}(p)^{2}) = 
        (p^{2z_{11}} + p^{2z_{12}} + 
        p^{2z_{13}} + p^{2z_{14}})^{x_{1} + x_{2}} \\-
        (2p^{2z_{13}} + p^{2z_{33}} + p^{2z_{34}})^{-x_{3}}
        \cdot (2p^{2z_{14}} + p^{2z_{34}} +
        p^{2z_{44}})^{x_{1} + x_{2} + x_{3}}.
    \end{multline*}
    
    We will now make use of the following claim, which
    we shall prove shortly.
    \begin{claim}\label{claim: formICase2Impossible}
        Let $n, m \in  \mathbb{Z}_{> 0}$, and  $y_{11} < y_{12} < y_{13} < y_{14}$,
        $y_{21} < y_{22} < y_{23}$,
        and $y_{31} < y_{32} < y_{33} \in \mathbb{Z}_{> 0}$.
        Assume the multiset $\{\alpha_{1}, \alpha_{2}, \alpha_{3}\} =
        \{\beta_{1}, \beta_{2}, \beta_{3}\} = \{2, 1, 1\}$.
        Define the polynomial
        $f: \mathbb{R} \rightarrow \mathbb{R}$ such that
        $$f(p) = (p^{y_{11}} + p^{y_{12}} +
        p^{y_{13}} + p^{y_{14}})^{n+m} - 
        (\alpha_{1}p^{y_{21}} + \alpha_{2}p^{y_{22}} 
        + \alpha_{3}p^{y_{23}})^{n} \cdot 
        (\beta_{1}p^{y_{31}} + \beta_{2}p^{y_{32}} 
        + \beta_{3}p^{y_{33}})^{m},$$
        then $f(p)$ is not the zero polynomial,
        unless $n=m$.
    \end{claim}

    \cref{claim: formICase2Impossible} implies that
    if $\xi_{\mathbf{x}}(\mathcal{T}_{M}^{*}(p)^{2})$ is the zero
    polynomial only if $-x_{3} = x_{1} + x_{2} + x_{3}$.
    This implies that if $x_{3} = -c$, then $x_{1} + x_{2} = 2c$,
    and $x_{4} = -c$.
    In other words, $\mathbf{x} = (x_{1}, x_{2}, -c, -c)$ for
    some $c \in \mathbb{Z}_{\neq 0}$, where
    $x_{1} + x_{2} = 2c$.
    But by our construction of $M$, we already ensured that
    for such $\mathbf{x} \in \chi_{4}$,
    $\xi_{\mathbf{x}}(M) \neq 0$.
    So, we see that if $\xi_{\mathbf{x}}(M) = 0$, then
    $\xi_{\mathbf{x}}(\mathcal{T}_{M}^{*}(p)^{2})$ is not
    the zero polynomial.

    So, we see that in either case,
    when $\mathbf{x} \in \chi_{4}$ has a 
    support size greater than $2$,
    if $\xi_{\mathbf{x}}(M) = 0$,
    there exists some
    $p^{*} \in \mathbb{R}$ such that
    $\xi_{\mathbf{x}}((\mathcal{T}_{M}^{*}(p^{*}))^{2})
    \neq 0$.
    We will now define $\xi_{2, \mathbf{x}}: 
    N \mapsto \phi_{\mathbf{x}}(N^{2})$, and
    $\zeta: N \mapsto (N_{11} - N_{33})(N_{11} - N_{44})
    (N_{22} - N_{33})(N_{22} - N_{44})(N_{33} - N_{44})$.
    So, we can use \cref{lemma: thickeningWorks} to
    find some $M' \in \mathfrak{R}(M,
    \mathcal{F} \cup \{\rho_{\tt{tensor}},
    \xi_{2, \mathbf{x}}, \zeta\}) \cap
    \text{Sym}_{4}^{\tt{pd}}(\mathbb{R}_{> 0})$.
    Since $\xi_{2, \mathbf{x}}(M') = \xi_{\mathbf{x}}((M')^{2}) \neq 0$,
    \cref{corollary: stretchingWorksPositive}
    allows us to find the required
    $N \in \mathfrak{R}(M', \mathcal{F} \cup 
    \{\rho_{\tt{tensor}}, \xi_{\mathbf{x}}, \zeta \})
    \cap \text{Sym}_{4}^{\tt{pd}}(\mathbb{R}_{> 0})$.
    If $N$ is diagonal distinct,
    \cref{theorem: diagDistHardness} proves that
    $\PlEVAL(M)$ is $\#$P-hard.
    Otherwise, \cref{theorem: nonDiagAllForms} implies that
    either $\PlEVAL(N) \leq \PlEVAL(M)$ is $\#$P-hard, or
    $N$ is of Form (I) as required.
\end{proof}

We will now finish the proof of \cref{lemma: formIReduction} by proving
\cref{claim: killedX},
\cref{claim: formICase1Impossible} and
\cref{claim: formICase2Impossible}.

\begin{claimproof}{\cref{claim: killedX}}
For a contradiction, suppose $x_{3} = 0$ or $x_{4} = 0$ or $x_{1} + x_{2} = 0$.
    First suppose $x_{3} = 0$. If $x_{1} + x_{2} = 0$,
    that would imply that $x_{4} = 0$ as well,
    which is not possible, since the support size
    of $\mathbf{x}$ is greater than $2$ by assumption.
    So, without loss of generality,
    we may assume that $x_{1} + x_{2} > 0$, by replacing $\mathbf{x}$ with $-\mathbf{x}$.
    Now,
    $$\phi_{\mathbf{x}}(\alpha_{1}, \dots, \alpha_{4})
    = (\alpha_{1})^{x_{1} + x_{2}} - (\alpha_{4})^{x_{1} + x_{2}}.$$
    Since $|\alpha_{1}| \neq |\alpha_{4}|$, we see that
    it is not possible that $\phi_{\mathbf{x}}
    (\alpha_{1}, \dots, \alpha_{4}) = 0$.
    So $x_{3} = 0$ is impossible.
    By symmetry, we see that $x_{4} = 0$ is also impossible.
    
    Now, let $x_{1} + x_{2} = 0$.
    This implies that $x_{3} + x_{4} = 0$ as well.
    We may assume without loss of generality that $x_{3} > 0$,
    since we already know that $x_{3} \neq 0$.
    Once again, we see that $\phi_{\mathbf{x}}
    (\alpha_{1}, \dots, \alpha_{4}) = 0$ would imply that
    $(\alpha_{3})^{x_{3}} = (\alpha_{4})^{x_{3}}$,
    which is also known to be not possible
    since $|\alpha_{3}| \neq |\alpha_{4}|$.
    This finishes the proof of \cref{claim: killedX}.
\end{claimproof}

\begin{claimproof}{\cref{claim: formICase1Impossible}}
    For convenience, we will define the polynomials
    $t_{1}, t_{2}, t_{3}$ such that
\begin{eqnarray*}
        t_{1}(p) &=& (p^{y_{11}} + p^{y_{12}} +
        p^{y_{13}} + p^{y_{14}}),\\
        t_{2}(p) &=&  (\alpha_{1}p^{y_{21}} + \alpha_{2}p^{y_{22}} 
        + \alpha_{3}p^{y_{23}}), \text{ and }\\
        t_{3}(p) &=&  (\beta_{1}p^{y_{31}} + \beta_{2}p^{y_{32}} 
        + \beta_{3}p^{y_{33}}).
\end{eqnarray*}
    Then, we see that
    $f(p) = t_{1}(p)^{n} \cdot t_{2}(p)^{m}
    - t_{3}(p)^{n+m}$.

    Let us assume that $f(p) = 0$ for all $p \in \mathbb{R}$.
    This means that 
    \[t_{1}(p)^{n} \cdot t_{2}(p)^{m}
    = t_{3}(p)^{n+m}\]
    for all $p \in \mathbb{R}$.
    Now, if $\beta_{1} = 2$, that means that the
    least degree term of the RHS
    has a coefficient of $2^{n+m}$.
    But the least degree term of the LHS
    can only be $\alpha_{1}^{m} \in \{1, 2^{m}\}$.
    We are given $n >0$.
    In either case,  it is not equal
    to $2^{n+m}$, and we get a contradiction.
    So, we find that $\beta_{1} = 1$.
    Now, if $\alpha_{1} = 2$, we get a similar
    contradiction (using $m > 0$), as the least degree term
    of the RHS will be $1$, and not equal
    to $2^{m}$. So, we see that
    $\alpha_{1} = \beta_{1} = 1$.
    By reasoning about the highest degree term
    instead of the least degree term, we can also
    see that $\alpha_{3} = \beta_{3} = 1$.
    This means that $\alpha_{2} = \beta_{2} = 2$.
    
    Now, $p^{ny_{11} + my_{21}}$, and
    $p^{(n+m)y_{31}}$ are
    the least degree terms of the LHS and RHS
    respectively.
    So, they must be equal.
    This means that
    $ny_{11} +my_{21} = (n+m)y_{31}$.
    It also means that
    \begin{equation}\label{eqn: formICase1}
        t_{1}(p)^{n} \cdot t_{2}(p)^{m}
        - p^{ny_{11} + my_{21}} =
        t_{3}(p)^{n+m} - p^{(n+m)y_{31}}.
    \end{equation}
    The least degree term of the RHS of
    \cref{eqn: formICase1} will now be
    $2(n+m)p^{(n+m - 1)y_{31} + y_{32}}$
    (here the coefficient $2$ is $\beta_2$),
    while the least degree term of the LHS may be either
    $np^{(n - 1)y_{11} + y_{12} + my_{21}}$ or
    $2mp^{ny_{11} + (m - 1)y_{21} + y_{22}}$,
    or their sum (if the degrees are the same).
    However, in either case, we find that the coefficient
    of the least degree term of the LHS is
    $\leq n + 2m < 2(n+m)$.
    This implies that \cref{eqn: formICase1}
    is not true.
    So, our assumption that
    $f(p) = 0$ for all $p \in \mathbb{R}$ must be false.
\end{claimproof}

\begin{claimproof}{\cref{claim: formICase2Impossible}}
    For convenience, we will define the polynomials
    $t_{1}, t_{2}, t_{3}$ such that
\begin{eqnarray*}
        t_{1}(p) &=& (p^{y_{11}} + p^{y_{12}} +
        p^{y_{13}} + p^{y_{14}}),\\
        t_{2}(p) &=&  (\alpha_{1}p^{y_{21}} + \alpha_{2}p^{y_{22}} 
        + \alpha_{3}p^{y_{23}}), \text{ and }\\
        t_{3}(p) &=&  (\beta_{1}p^{y_{31}} + \beta_{2}p^{y_{32}} 
        + \beta_{3}p^{y_{33}}).
\end{eqnarray*}
    Then, we see that
    $f(p) = t_{1}(p)^{n+m} - t_{2}(p)^{n}
    \cdot t_{3}(p)^{m}$.

    Let us assume that $f(p) = 0$ for all $p \in \mathbb{R}$.
    This means that 
    \[t_{1}(p)^{n+m} =
    t_{2}(p)^{n} \cdot t_{3}(p)^{m}\]
    for all $p \in \mathbb{R}$.
    We see that the coefficient of the least
    degree term of the LHS is $1$.
    However, if $\alpha_{1} = 2$, or $\beta_{1} = 2$,
    the coefficient of the least degree term
    of the LHS would be either $2^{n}$ or
    $2^{m}$ or $2^{n+m}$.
    Since this is not possible, as $n, m >0$,
    we may conclude that $\alpha_{1} = \beta_{1} = 1$.
    Similarly, by reasoning about the highest
    degree terms of the LHS and the RHS, we can
    see that $\alpha_{3} = \beta_{3} = 1$.
    This means that $\alpha_{2} = \beta_{2} = 2$.

    We note that the least degree term of the LHS
    is $p^{(n+m)y_{11}}$, and the
    least degree term of the RHS is
    $p^{ny_{21} + my_{31}}$.
    These terms must therefore be equal.
    This also means that
    \begin{equation}\label{eqn: formICase2}
        t_{1}(p)^{n+m} - p^{(n+m)y_{11}}
        = t_{2}(p)^{n}t_{3}(p)^{m}
        - p^{n y_{21} + my_{31}}.
    \end{equation}
    The least degree term of the LHS of
    \cref{eqn: formICase2} will now be
    $(n+m)p^{(n+m - 1)y_{11} + y_{12}}$,
    while the least degree term of the RHS will be
    either $2np^{(n - 1)y_{21} + y_{22} + my_{31}}$
    or $2mp^{ny_{21} + (m - 1)y_{31} + y_{32}}$,
    or their sum (if the degrees are the same).
    We observe the coefficients cannot be the same if the
    degrees are the same, since
    $2n + 2m > n+m$.
    So, the only remaining possibility is that
    $2n = n+m$, or
    $2m = n+m$.
    In either case, we see that unless $n=m$,
    it is not possible that $f(p) = 0$ for all
    $p \in \mathbb{R}$.
\end{claimproof}

We are finally ready to prove that if $M$ is of Form (I) and not isomorphic
to a tensor product
then $\PlEVAL(M)$ is $\#$P-hard.

\begin{lemma}\label{lemma: formIHardness}
    Let $M \in \text{Sym}_{4}^{\tt{pd}}(\mathbb{R}_{> 0})$
    be of Form (I) such that $\rho_{\tt{tensor}}(M) \neq 0$.
    Then, $\PlEVAL(M)$ is $\#$P-hard.
\end{lemma}
\begin{proof}
    Let $(\lambda_{1}, \dots, \lambda_{4})$ be the
    eigenvalues of $M$.
    Let $\mathcal{B}$ be a lattice basis
    of the lattice $\mathcal{L}(\lambda_{1}, \dots, \lambda_{4})$.
    Let us assume that $c \cdot \delta_{ij} \in
    \mathcal{B}$ for some $c \in \mathbb{Z}$,
    and some $i \neq j \in [4]$.  Being  part of a basis, $c \ne 0$.
    This implies that
    $(\lambda_{i})^{c}(\lambda_{j})^{-c} = 1$.
    Since $(\lambda_{1}, \dots, \lambda_{4})$ are all positive,
    this implies that $\lambda_{i}\lambda_{j}^{-1} = 1$.
    In other words, we see that
    $\delta_{ij} \in \mathcal{L}(\lambda_{1}, \dots,
    \lambda_{4})$.
Since $\mathcal{B}$ is a lattice basis, we have $c = \pm 1$.   So, we may replace all such $c \cdot \delta_{ij}$ with
    $\delta_{ij}$ in $\mathcal{B}$ and still
    have a lattice basis.

    Now, let us assume that there exists some
    $\mathbf{x} \in \mathcal{B} \setminus \mathcal{D}$.
    (Recall that $ \mathcal{D}$ is defined in \cref{equation: mathcalD}.)
    As  $\mathbf{0} \neq
    \mathbf{x} \in \chi_{4}$, we know that the support of this $\mathbf{x}$
    must be greater than $2$.
    Now, for each $\sigma \in S_{4}$, we can construct
    $\mathbf{x}^{\sigma}$ such that
    $(\mathbf{x}^{\sigma})_{i} = x_{\sigma(i)}$.
    We will let $S_{4} = \{\sigma_{1}, \dots, \sigma_{24}\}$,
    and define $\xi_{i}: N \mapsto 
    \phi_{\mathbf{x}^{\sigma_{i}}}(N_{11}, \dots, N_{44})$
    for $i \in [24]$.
    We can now use \cref{lemma: formIReduction} to
    find $M_{1} \in \mathfrak{R}(M, \mathcal{F}_{M}
    \cup \{\rho_{\tt{tensor}}, \xi_{1}\}) \cap
    \text{Sym}_{4}^{\tt{pd}}(\mathbb{R}_{> 0})$
    where $\mathcal{F}_{M}$ is as defined in
    \cref{equation: FM}.
    We can now repeat this process with $M_{1}$ in place
    of $M$ to find
    $M_{2} \in \mathfrak{R}(M_{1}, \mathcal{F}_{M}
    \cup \{\rho_{\tt{tensor}}, \xi_{1}, \xi_{2}\} ) \cap
    \text{Sym}_{4}^{\tt{pd}}(\mathbb{R}_{> 0})$.
    
    After repeating this for all $i \in [24]$,
    we can find 
    $M' = M_{24} \in \mathfrak{R}(M, \mathcal{F}_{M}
    \cup \{\rho_{\tt{tensor}}, \xi_{1}, \dots, \xi_{24}\})
    \cap \text{Sym}_{4}^{\tt{pd}}(\mathbb{R}_{> 0})$.
    So, we see that
    $M' \in \mathfrak{R}(M, \mathcal{F}_{M}
    \cup \{\rho_{\tt{tensor}}, \xi\})\cap
    \text{Sym}_{4}^{\tt{pd}}(\mathbb{R}_{> 0})$,
    where $\xi: N \mapsto \Phi_{\mathbf{x}}(N_{11}, \dots, N_{44})$.
    Now, \cref{theorem: positiveDefiniteDiagonal}
    implies that we can find some
    $N_{1} \in \mathfrak{R}(M, \mathcal{F}_{M} \cup
    \{\rho_{\tt{tensor}}, \Psi_{\mathbf{x}}\} ) \cap
    \text{Sym}_{4}^{\tt{pd}}(\mathbb{R}_{> 0})$.

    If the eigenvalues of $N_{1}$ have no lattice basis
    $\mathcal{B} \subseteq \mathcal{D}$,
    we can now repeat this whole process with $N_{1}$
    instead of $M$.
    From \cref{lemma: dimReduction}, we know that
    after repeating this process at most 4 times, 
    we will have some $N \in
    \text{Sym}_{4}^{\tt{pd}}(\mathbb{R}_{> 0})$,
    such that $\PlEVAL(N) \leq \PlEVAL(M)$, and
    the eigenvalues of $N$ have a lattice basis
    $\mathcal{B} \subseteq \mathcal{D}$.
    Now, \cref{theorem: latticeHardness}
    proves that $\PlEVAL(N) \leq \PlEVAL(M)$ is $\#$P-hard.
\end{proof}
\subsection{Form (IV)}\label{sec: form_IV}

We will postpone our treatment of Form (III) for a while, and
deal with matrices of Form (IV) first.

\begin{figure}[ht]
\[
\begin{pmatrix}
            a & x & x & z\\
            x & a & x & z\\
            x & x & a & z\\
            z & z & z & b
        \end{pmatrix}
        \]
        \caption{\label{Form-IV-in-abc} Form (IV)}
\end{figure}

\begin{lemma}\label{lemma: formIVReduction}
    Let $M \in \text{Sym}_{4}^{\tt{pd}}(\mathbb{R}_{> 0})$
    be of Form (IV) such that $\rho_{\tt{tensor}}(M) \neq 0$.
    Let $\mathcal{F}$ be a countable set of 
    $\text{Sym}_{4}(\mathbb{R})$-polynomials such that
    $F(M) \neq 0$ for all $F \in \mathcal{F}$.
    Let $\mathbf{0} \neq \mathbf{x} \in \chi_{4}$ have 
    support size greater than $2$.
    Then, either $\PlEVAL(M)$ is $\#$P-hard, or
    there exists some $N \in \mathfrak{R}
    (M, \mathcal{F} \cup \{\rho_{\tt{tensor}}, \Psi_{\mathbf{x}}\}) \cap
    \text{Sym}_{4}^{\tt{pd}}(\mathbb{R}_{> 0})$
    of Form (IV).
\end{lemma}
\begin{proof}
    We will first define $\zeta: N \mapsto
    (N_{11} - N_{44})(N_{22} - N_{44})(N_{33} - N_{44})$.
    By our choice of $M$, we know that
    $\zeta(M) \neq 0$.
    We will now consider $\mathcal{S}_{M}(\theta)$.
    We know that there exists some $\delta > 0$
    such that $|\mathcal{S}_{M}(\theta)_{ij} - I_{ij}|
    < \frac{1}{3}$ for all $i, j \in [4]$,
    for all $0 < \theta < \delta$.
    We can now use \cref{lemma: stretchingWorks}
    to find some $M' = \mathcal{S}_{M}(\theta^{*})
    \in \mathfrak{R}(M, \mathcal{F} \cup
    \{\rho_{\tt{tensor}}, \zeta \}) \cap
    \text{Sym}_{4}^{\tt{pd}}(\mathbb{R}_{\ne 0})$,
    such that $0 < \theta^{*} < \delta$.

    By construction of $M'$, we know that
    $M'_{44} \neq M'_{11}, M'_{22}, M'_{33}$.
    If it is not the case that
    $M'_{11} = M'_{22} = M'_{33}$,
    then using \cref{lemma: order211Forms},
    we can find some $M'' \in
    \mathfrak{R}(M, \mathcal{F} \cup \{\rho_{\tt{tensor}} \}
    \cap \text{Sym}_{4}^{\tt{pd}})(\mathbb{R}_{> 0})$
    that is diagonal distinct, or isomorphic to a
    matrix of Form (I).
    If $M''$ is diagonal distinct, then
    \cref{theorem: diagDistHardness} tells us that
    $\PlEVAL(M'')$ and therefore $\PlEVAL(M)$
    are $\#$P-hard,
    If it is isomorphic to a matrix of Form (I),
    then \cref{lemma: formIHardness} similarly
    implies that $\PlEVAL(M)$ is $\#$P-hard.
    So, we may assume that $(M')_{11} = (M')_{22}
    = (M')_{33} \neq (M')_{44}$.
    \cref{theorem: nonDiagAllForms} now tells us
    that unless $M'$ is of Form (IV),
    $\PlEVAL(M)$ is again $\#$P-hard.
    So, we may assume that $M'$ is of Form (IV).

    We will now consider $T_{2n}(M')$ for all
    $n \geq 1$.
    We let
    $$T_{2n}(M') = \begin{pmatrix}
        (z_{11})^{n} & (z_{12})^{n} & (z_{12})^{n} & (z_{14})^{n}\\
        (z_{12})^{n} & (z_{11})^{n} & (z_{12})^{n} & (z_{14})^{n}\\
        (z_{12})^{n} & (z_{12})^{n} & (z_{11})^{n} & (z_{14})^{n}\\
        (z_{14})^{n} & (z_{14})^{n} & (z_{14})^{n} & (z_{44})^{n}\\
    \end{pmatrix}$$
    where $z_{11} = (M'_{11})^{2}$, $z_{12} = (M'_{12})^{2}$,
    $z_{14} = (M'_{14})^{2}$, and $z_{44} = (M'_{44})^{2}$.
    As squares of non-zero numbers they are all positive.
    It can be verified that the eigenvalues of $T_{2n}(M')$ are
    (with multiplicity):
    \begin{equation}\label{eqn: formIVEigenvalues}
    \begin{aligned}
        \lambda_{1}(T_{2n}(M'))
            &= (z_{11})^{n} - (z_{12})^{n},\\
        \lambda_{2}(T_{2n}(M'))
            &= (z_{11})^{n} - (z_{12})^{n},\\
        \lambda_{3}(T_{2n}(M'))
            &= \frac{1}{2}
            \left((z_{11})^{n} + 2(z_{12})^{n} + (z_{44})^{n}
            - \sqrt{((z_{11})^{n} + 2(z_{12})^{n} - (z_{44})^{n})^{2} 
            + 12((z_{14})^{n})^{2}}\right),\\
        \lambda_{4}(T_{2n}(M'))
            &= \frac{1}{2}
            \left((z_{11})^{n} + 2(z_{12})^{n} + (z_{44})^{n}
            + \sqrt{((z_{11})^{n} + 2(z_{12})^{n} - (z_{44})^{n})^{2} 
            + 12((z_{14})^{n})^{2}}\right).
    \end{aligned}
    \end{equation}

    From our construction of $M'$, we know that
    $z_{11} > z_{12}$. 
    So, we can see that
    $$\lim\limits_{n \rightarrow \infty}
    \frac{\lambda_{1}(T_{2n}M')}{(z_{11})^{n}} = 
    \lim\limits_{n \rightarrow \infty}
    \frac{\lambda_{2}(T_{2n}M')}{(z_{11})^{n}} = 
    \lim\limits_{n \rightarrow \infty}
    \frac{(z_{11})^{n} - (z_{12})^{n}}{(z_{11})^{n}} = 1.$$
    We also note that by our construction of $M'$,
    $z_{11} \neq z_{44}$.
    We shall prove the following claim shortly:
    
    \begin{claim}\label{claim: formIVLimit}
        If $z_{11} > z_{44}$, then
        $$\lim\limits_{n \rightarrow \infty}
        \frac{\lambda_{3}(T_{2n}M')}{(z_{44})^{n}} = 
        \lim\limits_{n \rightarrow \infty}
        \frac{\lambda_{4}(T_{2n}M')}{(z_{11})^{n}} = 1.$$
        If $z_{11}  < z_{44}$, then
        $$\lim\limits_{n \rightarrow \infty}
        \frac{\lambda_{3}(T_{2n}M')}{(z_{11})^{n}} = 
        \lim\limits_{n \rightarrow \infty}
        \frac{\lambda_{4}(T_{2n}M')}{(z_{44})^{n}} = 1.$$
    \end{claim}

    So, we see that for large enough values of $n$,
    $\lambda_{i}(T_{2n}M') > 0$ for all $i \in [4]$.
    We can now define the function
    $\varphi_{\mathbf{x}}:
    (\mathbb{R}_{\neq 0})^{4} \rightarrow \mathbb{R}$,
    such that
    $\varphi_{\mathbf{x}}(\alpha_{1}, \dots, \alpha_{4})
    = \prod_{i \in [4]}(\alpha_{i})^{x_{i}}$.
    For large enough values of $n$, we see that
    $\varphi_{\mathbf{x}}(\lambda_{1}(T_{2n}M'), \dots, 
    \lambda_{4}(T_{2n}M'))$ is well-defined.
    Moreover, we note that
    $$\phi_{\mathbf{x}}(\lambda_{1}(T_{2n}M'), \dots, 
    \lambda_{4}(T_{2n}M')) = 0 \iff
    \varphi_{\mathbf{x}}(\lambda_{1}(T_{2n}M'), \dots, 
    \lambda_{4}(T_{2n}M')) = 1.$$
    We will now study the behavior of the function
    $\phi_{\mathbf{x}}(\lambda_{1}(T_{2n}M'), \dots, 
    \lambda_{4}(T_{2n}M'))$, by studying the function
    $\varphi_{\mathbf{x}}(\lambda_{1}(T_{2n}M'), \dots, 
    \lambda_{4}(T_{2n}M'))$.

    We first note from \cref{eqn: formIVEigenvalues}
    that $\lambda_{1}(T_{2n}M') = \lambda_{2}(T_{2n}M')$
    for all $n \geq 1$.
    We also note that $\lambda_{4}(T_{2n}M') > 
    \lambda_{1}(T_{2n}M'), \lambda_{3}(T_{2n}M')$
    for large enough $n$.
    If $\lambda_{1}(T_{2n}M') = \lambda_{3}(T_{2n}M')$,
    that implies that
    $$(z_{44})^{n} - (z_{11})^{n} + 4(z_{12})^{n}
    = \sqrt{((z_{11})^{n} + 2(z_{12})^{n} - (z_{44})^{n})^{2} 
    + 12((z_{14})^{n})^{2}}.$$
    Squaring both sides, we find that
    \begin{multline*}
        ((z_{44})^{n} - (z_{11})^{n})^{2} + 16(z_{12})^{2n}
    + 8(z_{12})^{n}((z_{44})^{n} - (z_{11})^{n}) =
    ((z_{44})^{n} - (z_{11})^{n})^{2} + 4(z_{12})^{2n}\\
    - 4(z_{12})^{n}((z_{44})^{n} - (z_{11})^{n}) + 12(z_{14})^{2n}.
    \end{multline*}
    On rearranging terms, we find that
    $$12(z_{12})^{2n} + 12(z_{12}z_{44})^{n}
    =  12(z_{14})^{2n} + 12(z_{12}z_{11})^{n}.$$
    If we write this as $A^n + B^n = C^n + D^n$, where $A = (z_{12})^{2}$,
    $B= z_{12}z_{44}$, $C= (z_{14})^{2}$ and $D = z_{12}z_{11}$, then we know that
    $A < D$ and $B \ne D$.
    (1) If $C \ge D$, then the RHS has order $C^n$. 
    As $A<D \le C$, to match the leading order, $B=C$ (in this case, $C=D$ is impossible). 
    But then $A^n = D^n$, 
    contradicting $A<D$. (2) If $C<D$, then the RHS has order $D^n$. 
    But this cannot be since $A < D$ and $B \ne D$.
    So, this equation cannot be true for large enough
    values of $n$.
    So, we see that
    $\lambda_{1}(T_{2n}M') = \lambda_{2}(T_{2n}M')$
    and $|\lambda_{1}(T_{2n}M')| \neq |\lambda_{3}(T_{2n}M')|
    \neq |\lambda_{4}(T_{2n}M')|$ are pairwise distinct for large 
    values of $n$.
    So, \cref{claim: killedX} immediately
    implies that if $x_{3} = 0$, or $x_{4} = 0$, or $x_{1} + x_{2} = 0$, then
    $\phi_{\mathbf{x}}(\lambda_{1}(T_{2n}M'), \dots, 
    \lambda_{4}(T_{2n}M')) \neq 0$
    for large enough $n$.

    We will now assume that $x_{3} \neq 0$, and $x_{4} \neq 0$.
    Now, we find that when $z_{11} > z_{44}$,
    \begin{multline*}
        \varphi_{\mathbf{x}}(\lambda_{1}(T_{2n}M'),
        \dots, \lambda_{4}(T_{2n}M'))
        = \left((z_{11})^{n(x_{1} + x_{2} + x_{4})}(z_{44})^{n(x_{3})}\right)
        \cdot \left(\frac{\lambda_{1}(T_{2n}M')}
        {(z_{11})^{n}} \right)^{x_{1}} \cdot
        \left(\frac{\lambda_{2}(T_{2n}M')}
        {(z_{11})^{n}} \right)^{x_{2}} \\ \cdot
        \left(\frac{\lambda_{3}(T_{2n}M')}
        {(z_{44})^{n}} \right)^{x_{3}} \cdot
        \left(\frac{\lambda_{4}(T_{2n}M')}
        {(z_{11})^{n}} \right)^{x_{4}}
    \end{multline*}
    So, we see that
    $$\lim\limits_{n \rightarrow \infty}
    \varphi_{\mathbf{x}}(\lambda_{1}(T_{2n}M'), \dots, 
    \lambda_{4}(T_{2n}M')) = 
    \lim\limits_{n \rightarrow  \infty}
    \left(\frac{z_{44}}{z_{11}}\right)^{nx_{3}},$$
    which is either $0$ or $\infty$ depending
    on whether $x_{3} > 0$, or $x_{3} < 0$.
    Similarly, when $z_{11} < z_{44}$, we find that
    \begin{multline*}
        \varphi_{\mathbf{x}}(\lambda_{1}(T_{2n}M'),
        \dots, \lambda_{4}(T_{2n}M'))
        = \left((z_{11})^{n(x_{1} + x_{2} + x_{3})}(z_{44})^{n(x_{4})}\right)
        \cdot \left(\frac{\lambda_{1}(T_{2n}M')}
        {(z_{11})^{n}} \right)^{x_{1}} \cdot
        \left(\frac{\lambda_{2}(T_{2n}M')}
        {(z_{11})^{n}} \right)^{x_{2}} \\ \cdot
        \left(\frac{\lambda_{3}(T_{2n}M')}
        {(z_{11})^{n}} \right)^{x_{3}} \cdot
        \left(\frac{\lambda_{4}(T_{2n}M')}
        {(z_{44})^{n}} \right)^{x_{4}}
    \end{multline*}
    So, we see that
    $$\lim\limits_{n \rightarrow \infty}
    \varphi_{\mathbf{x}}(\lambda_{1}(T_{2n}M'), \dots, 
    \lambda_{4}(T_{2n}M')) = 
    \lim\limits_{n \rightarrow}
    \left(\frac{z_{44}}{z_{11}}\right)^{nx_{4}},$$
    which is either $0$ or $\infty$ depending
    on whether $x_{4} < 0$, or $x_{4} > 0$.

    So, we see that in either case,
    $\lim\limits_{n \rightarrow \infty}
    \varphi_{\mathbf{x}}(\lambda_{1}(T_{2n}M'), \dots, 
    \lambda_{4}(T_{2n}M'))$ is either $0$ or $\infty$, and so it is away from $1$.   
    This proves that for all $\mathbf{x} \in \chi_{4}$ of
    support size greater than $2$,
    for large enough $n \geq 1$,
    $\phi_{\mathbf{x}}(\lambda_{1}(T_{2n}M'),
    \dots, \lambda_{4}(T_{2n}M')) \neq 0$.
    In particular, given any $\sigma \in S_{4}$, we can see that
    $\phi_{\mathbf{y}}(\lambda_{1}(T_{2n}M'),
    \dots, \lambda_{4}(T_{2n}M')) \neq 0$
    for all $\mathbf{y}$ such that
    $y_{\sigma(i)} = x_{i}$, for some large enough value of $n$.
    This proves there exists some $n^{*}$, such that
    $$\Phi_{\mathbf{x}}(\lambda_{1}(T_{2n^{*}}M'),
    \dots, \lambda_{4}(T_{2n^{*}}M')) =
    \Psi_{\mathbf{x}}(T_{2n^{*}}M') \neq 0.$$

    We will now define the $\text{Sym}_{4}(\mathbb{R})$-polynomial
    $\xi: N \mapsto \Psi_{\mathbf{x}}(T_{2n^{*}}N)$.
    Since $M' = \mathcal{S}_{M}(\theta^{*})$ by construction,
    we see that $\xi(\mathcal{S}_{M}(\theta^{*})) \neq 0$.
    So, we can use \cref{corollary: stretchingWorksPositive} to
    find some $M'' \in \mathfrak{R}(M, \mathcal{F} 
    \cup \{\rho_{\tt{tensor}}, \zeta, \xi \})
    \cap \text{Sym}_{4}^{\tt{pd}}(\mathbb{R}_{> 0})$.
 We now see that $\Psi_{\mathbf{x}}(T_{2n^{*}}M'') = \xi(M'') \neq 0$.
    So, \cref{corollary: thickeningSufficient} allows
    us to find the required
    $N \in \mathfrak{R}(M, \mathcal{F} 
    \cup \{\rho_{\tt{tensor}}, \zeta, \Psi_{\mathbf{x}} \})
    \cap \text{Sym}_{4}^{\tt{pd}}(\mathbb{R}_{> 0})$.
\end{proof}

\begin{claimproof}{\cref{claim: formIVLimit}}
    First, we note that $\lambda_{3}(T_{2n}M')$ and
    $\lambda_{4}(T_{2n}M')$ are symmetric if we exchange
    $z_{11}$ and $z_{44}$.
    So, it will be sufficient for us to prove this claim
    when $z_{11} > z_{44}$.
    We note that
    $$(z_{11})^{n} + 2(z_{12})^{n} + (z_{44})^{n} = 
    (z_{11})^{n} \left(1 + 2\left(\frac{z_{12}}{z_{11}}\right)^{n} 
    + \left(\frac{z_{44}}{z_{11}}\right)^{n}\right).$$
    So, 
    $$\lim\limits_{n \rightarrow \infty}
    \frac{(z_{11})^{n} + 2(z_{12})^{n} + (z_{44})^{n}}{(z_{11})^{n}} = 1.$$
    We also note that
    $$\sqrt{((z_{11})^{n} + 2(z_{12})^{n} - (z_{44})^{n})^{2} 
    + 12((z_{14})^{n})^{2}}
    = (z_{11})^{n} \cdot
    \sqrt{\left(1 + 2\left(\frac{z_{12}}{z_{11}}\right)^{n}
    - \left(\frac{z_{44}}{z_{11}}\right)^{n}\right)^{2} 
    + 12\left(\left(\frac{z_{14}}{z_{11}}\right)^{n}\right)^{2}}.$$
    So, 
    $$\lim\limits_{n \rightarrow \infty}
    \frac{\sqrt{((z_{11})^{n} + 2(z_{12})^{n} - (z_{44})^{n})^{2} 
    + 12((z_{14})^{n})^{2}}}{(z_{11})^{n}} = 1.$$
    This implies that
    $$\lim\limits_{n \rightarrow \infty}
    \frac{\lambda_{4}(T_{2}M')}{(z_{11})^{n}} = 
    1.$$

    We now note that
    \begin{align*}
        \lambda_{3}(T_{2n}M')
        &= \frac{1}{2}\left(\frac
        {\left((z_{11})^{n} + 2(z_{12})^{n} + (z_{44})^{n} \right)^{2}
        - ((z_{11})^{n} + 2(z_{12})^{n} - (z_{44})^{n})^{2} 
        - 12((z_{14})^{n})^{2}}{(z_{11})^{n} + 2(z_{12})^{n} + (z_{44})^{n}
        + \sqrt{((z_{11})^{n} + 2(z_{12})^{n} - (z_{44})^{n})^{2} 
        + 12((z_{14})^{n})^{2}}}\right)\\
        &= \frac{1}{2}\left(\frac
        {4(z_{44})^{n}((z_{11})^{n} + (2z_{12})^{n}) - 12(z_{14})^{2n}}{(z_{11})^{n} + 2(z_{12})^{n} + (z_{44})^{n}
        + \sqrt{((z_{11})^{n} + 2(z_{12})^{n} - (z_{44})^{n})^{2} 
        + 12((z_{14})^{n})^{2}}}\right)\\
        &= \frac{(z_{44})^{n}}{2}\left(\frac
        {4 + 8\left(\frac{z_{12}}{z_{11}}\right)^{n} - 
        12\left(\frac{(z_{14})^{2}}{z_{44}z_{11}}\right)^{n}}
        {1 + 2\left(\frac{z_{12}}{z_{11}}\right)^{n} 
        + \left(\frac{z_{44}}{z_{11}}\right)^{n} + 
        \sqrt{\left(1 + 2\left(\frac{z_{12}}{z_{11}}\right)^{n}
        - \left(\frac{z_{44}}{z_{11}}\right)^{n}\right)^{2} 
        + 12\left(\left(\frac{z_{14}}{z_{11}}\right)^{n}\right)^{2}}}
        \right).
    \end{align*}
    
    We already know that
    $$\lim\limits_{n \rightarrow \infty} \left(
    1 + 2\left(\frac{z_{12}}{z_{11}}\right)^{n} 
    + \left(\frac{z_{44}}{z_{11}}\right)^{n} + 
    \sqrt{\left(1 + 2\left(\frac{z_{12}}{z_{11}}\right)^{n}
    - \left(\frac{z_{44}}{z_{11}}\right)^{n}\right)^{2} 
    + 12\left(\left(\frac{z_{14}}{z_{11}}\right)^{n}\right)^{2}}
    \right) = 2.$$
    So,
    $$\lim\limits_{n \rightarrow \infty} \frac
    {\lambda_{3}(T_{2n}M')}{(z_{44})^{n}} = 
    \frac{1}{2}\left(\frac{4}{2}\right) = 1.$$
\end{claimproof}

We can now prove that $\PlEVAL(M)$, when $M$ is of Form (IV) that is not 
isomorphic to a tensor product,
must also be $\#$P-hard.

\begin{lemma}\label{lemma: formIVHardness}
    Let $M \in \text{Sym}_{4}^{\tt{pd}}(\mathbb{R}_{> 0})$
    be of Form (IV) such that $\rho_{\tt{tensor}}(M) \neq 0$.
    Then, $\PlEVAL(M)$ is $\#$P-hard.
\end{lemma}
\begin{proof}
    The proof is very similar to the proof of
    \cref{lemma: formIHardness}.
    Let $(\lambda_{1}, \dots, \lambda_{4})$ be the
    eigenvalues of $M$.
    Let $\mathcal{B}$ be a lattice basis
    of the lattice $\mathcal{L}(\lambda_{1}, \dots, \lambda_{4})$.
    If there is some $c \cdot \delta_{ij} \in
    \mathcal{B}$ for some non-zero $c \in \mathbb{Z}$, 
    we may replace all such $c \cdot \delta_{ij}$ with
    $\delta_{ij}$ in $\mathcal{B}$ and still
    have a lattice basis.

    Now, let us assume that there exists some
    $\mathbf{x} \in \mathcal{B} \setminus \mathcal{D}$.
    We know that the support of this $\mathbf{x}$
    must be greater than $2$.
    Unless $\PlEVAL(M)$ is $\#$P-hard, we can now use
    \cref{lemma: formIVReduction} to find some
    $N_{1} \in \mathfrak{R}(M, \mathcal{F}_{M} \cup
    \{\rho_{\tt{tensor}}, \Psi_{\mathbf{x}}\} \cap
    \text{Sym}_{4}^{\tt{pd}}(\mathbb{R}_{> 0})$ of
    Form (IV).

    If the eigenvalues of $N_{1}$ have no lattice basis
    $\mathcal{B} \subseteq \mathcal{D}$,
    we can now repeat this whole process with $N_{1}$
    instead of $M$.
    From \cref{lemma: dimReduction}, we know that
    after repeating this process finitely many times,
    we will have some $N \in
    \text{Sym}_{4}^{\tt{pd}}(\mathbb{R}_{> 0})$,
    such that $\PlEVAL(N) \leq \PlEVAL(M)$, and
    the eigenvalues of $N$ have a lattice basis
    $\mathcal{B} \subseteq \mathcal{D}$.
    Now, \cref{theorem: latticeHardness}
    proves that $\PlEVAL(N) \leq \PlEVAL(M)$ is $\#$P-hard.
\end{proof}
\subsection{Form (III)}\label{sec: form_III}

We will now deal with matrices of Form (III).
We will once again need some setup, by proving that
the matrix $M$ of Form (III) may be assumed to have
some additional structure.
Note that by a simultaneous permutation we can permute the rows and columns of $M$ 
such that $M_{11} = M_{22} > M_{33} = M_{44}$, and
that $M_{13} \geq M_{14}$.
For the rest of this section, we will assume this is the case.

\begin{figure}[ht]
\[
\begin{pmatrix}
            a & x & y & z\\
            x & a & z & y\\
            y & z & b & t\\
            z & y & t & b
        \end{pmatrix}
        \]
        \caption{\label{Form-III-in-abc} Form (III)}
\end{figure}

\begin{lemma}\label{lemma: DiagLarge}
    Let $M \in \text{Sym}_{4}^{\tt{pd}}(\mathbb{R}_{> 0})$.
    Let $\mathcal{F}$ be a countable set of 
    $\text{Sym}_{4}(\mathbb{R})$-polynomials such that
    $F(M) \neq 0$ for all $F \in \mathcal{F}$.
    Then there exists some $N \in \mathfrak{R}
    (M, \mathcal{F}\}) \cap
    \text{Sym}_{4}^{\tt{pd}}(\mathbb{R}_{> 0})$, such that
    $N_{ii} > N_{jk}$ for all $i \in [4]$, and
    $j \neq k \in [4]$.
\end{lemma}
\begin{proof}
    We may assume that the entries of $M$ are generated
    by some $\{g_{t}\}_{t \in [d]}$.
    We can use \cref{lemma: MequivalentCM} to replace
    $M$ with some other $c M$ such that
    $e_{ijt} \geq 0$ for all $i, j \in [4]$, $t \in [d]$.
    We will now let $\mathcal{F}' = \{F': N \mapsto F(T_{2}N):
    F \in \mathcal{F}\}$.
    We note that since $F(M) \neq 0$ for all $F \in 
    \mathcal{F}$,
    $F'(\mathcal{T}_{M}(g_{1}^{\nicefrac{1}{2}}, \dots,
    g_{d}^{\nicefrac{1}{2}})) \neq 0$ for all
    $F' \in \mathcal{F}'$.
    So, \cref{lemma: thickeningWorks} allows us to
    pick some $M' \in \mathfrak{R}(M, \mathcal{F}') \cap
    \text{Sym}_{4}^{\tt{pd}}(\mathbb{R}_{> 0})$.
    
    We will now consider $\mathcal{S}_{M'}(\theta)$.
    We know that there exists some $\delta > 0$ such that
    $|\mathcal{S}_{M'}(\theta)_{ij} - I_{ij}| < \frac{1}{3}$
    for all $i, j \in [4]$, for all $0 < \theta < \delta$.
    So, we can use \cref{lemma: stretchingWorks} to find
    some $M'' = \mathcal{S}_{M'}(\theta^{*}) \in \mathfrak{R}
    (M', \mathcal{F}') \cap
    \text{Sym}_{4}^{\tt{pd}}(\mathbb{R}_{\neq 0})$, such that
    $0 < \theta^{*} < \delta$.

    We will now let $N = T_{2}(M'')$.
    Clearly, $N \in \text{Sym}_{4}(\mathbb{R}_{> 0})$.
    By our choice of $M''$, we know that $N_{ii} > \frac{4}{9}$, and
    $N_{jk} < \frac{1}{9}$ for all $i \in [4]$, and $j \neq k \in [4]$.
    By the \emph{\bf Gershgorin Circle Theorem} \cite{gershgorin1931uber},
    we know that if we let $r_{i} = \sum_{j \in [4]\setminus \{i\}}|N_{ij}|$,
    then all the eigenvalues of $N$ lie within
    one of the intervals $[N_{ii} - r_{i}, N_{ii} + r_{i}]$.
    But as we have seen, $N_{ii} - r_{i} >
    \left(\frac{4}{9} - 3 \cdot\frac{1}{9}\right) > 0$
    for all $i \in [4]$.
    This implies that all eigenvalues of $N$ are positive.
    So, $N \in \text{Sym}_{4}^{\tt{pd}}(\mathbb{R}_{> 0})$.
    Moreover, by our choice of $\mathcal{F}'$, we know that
    since $F'(M'') \neq 0$ for all $F' \in \mathcal{F}'$,
    $F(N) \neq 0$ for all $F \in \mathcal{F}$.
    So, $N \in \mathfrak{R}(M, \mathcal{F}) \cap
    \text{Sym}_{4}^{\tt{pd}}(\mathbb{R}_{> 0})$
    is the required matrix.
\end{proof}

As it turns out, it is possible for matrices of Form (III)
to be isomorphic to $A \otimes B$ for some
$A, B \in \text{Sym}_{2}(\mathbb{R})$.
We will now show that if $\rho_{\tt{tensor}}(M) \neq 0$,
then $M$ has sufficient structure, that
we can prove that $\PlEVAL(M)$ is $\#$P-hard.

\begin{lemma}\label{lemma: FormIIIM12M34}
    Let $M \in \text{Sym}_{4}^{\tt{pd}}(\mathbb{R}_{> 0})$
    be of Form (III) such that $\rho_{\tt{tensor}}(M) \neq 0$.
    Let $\mathcal{F}$ be a countable set of 
    $\text{Sym}_{4}(\mathbb{R})$-polynomials such that
    $F(M) \neq 0$ for all $F \in \mathcal{F}$.
    Then, either $\PlEVAL(M)$ is $\#$P-hard, or
    there exists some $N \in \mathfrak{R}
    (M, \mathcal{F} \cup \{\rho_{\tt{tensor}}\}) \cap
    \text{Sym}_{4}^{\tt{pd}}(\mathbb{R}_{> 0})$
    of Form (III), such that
    $N_{12} \neq N_{34}$.
\end{lemma}
\begin{proof}
    If $M_{12} \neq M_{34}$, we are already done,
    so we may assume otherwise.
    We will let $\zeta: N \mapsto (N_{11} - N_{33})(N_{11} - N_{44})
    (N_{22} - N_{33})(N_{22} - N_{44})$.
    We know that $\zeta(M) \neq 0$.
    We will let $\mathcal{F}'
    = \{F': N \mapsto F(N^{3})\ |\ F \in \mathcal{F} \cup
    \{\rho_{\tt{tensor}}, \zeta\}\}$.
    We also know that there exists some $\delta > 0$
    such that for all $0 < \theta < \delta$,
    $|\mathcal{S}_{M}(\theta)_{ij} - I_{ij}| < \frac{1}{3}$.
    We can now use \cref{lemma: stretchingWorks} to find
    $M' = \mathcal{S}_{M}(\theta^{*})
    \in \mathfrak{R}(M, \mathcal{F} \cup \mathcal{F}' \cup \{\zeta,
    \rho_{\tt{tensor}}\}) \cap 
    \text{Sym}_{4}^{\tt{pd}}(\mathbb{R}_{> 0})$, for some
    $0 < \theta^{*} < \delta$.
    If $M'$ is not of Form (III), that would imply that
    it is either diagonal distinct, or
    isomorphic to a matrix of Form (I) or Form (II), or
    satisfies $M'_{11} = M'_{22} \neq M'_{33} = M'_{44}$, but is
    not of Form (III).
    In either case, 
    since $\rho_{\tt{tensor}}(M') \neq 0$,
    we see that $\PlEVAL(M') \leq \PlEVAL(M)$ is
    $\#$P-hard, due to either
    \cref{theorem: diagDistHardness},
    or \cref{lemma: formIIReduction} and
    \cref{lemma: formIHardness}, or
    \cref{lemma: order22Forms}.
    So, we may assume that $M'$ is of Form (III).

    We will now consider $(R_{n}M')_{12} - (R_{n}M')_{34}$
    for all $n \geq 1$.
    We let $X = \{(M')_{ab}: a, b \in [4]\}$, and
    $c_{ij}(x) = \sum_{a, b \in [4]: M'_{ab} = x}(M')_{ia}(M')_{jb}$.
    We see that
    $$(R_{n}M')_{12} - (R_{n}M')_{34} = \sum_{x \in X}x^{n}
    \cdot (c_{12}(x) - c_{34}(x))$$
    for all $n \geq 1$.
    So, the equations $(R_{n}M')_{12} - (R_{n}M')_{34} = 0$
    form a full rank Vandermonde system of equations
    of size $O(1)$.
    This implies that $c_{12}(x) - c_{34}(x) = 0$
    for all $x \in X$.
    By construction of $M'$, we know that
    $(M')_{ab} = (M')_{11}$ implies that $(a, b) \in \{(1, 1), (2, 2)\}$,
    and that $(M')_{ab} = (M'')_{33}$ implies that
    $(a, b) = \{(3, 3), (4, 4)\}$.
    So, we see that (using the fact that $M'$ has Form (III))
    $$c_{12}((M')_{11}) - c_{34}((M')_{11}) 
    = 2 (M')_{11}(M')_{12} - 2 (M')_{13}(M')_{14}  = 0, \text{ and }$$
    $$c_{12}((M')_{33}) - c_{34}((M')_{33}) 
    = 2 (M')_{13}(M')_{14}  - 2 (M')_{33}(M')_{34}  = 0.$$
    Hence, $(M')_{11}(M')_{12} = (M')_{33}(M')_{34}$.
    Since we know that $(M')_{11} \neq (M')_{33}$,
    this implies that $(M')_{12} \neq (M')_{34}$.

    So, if $(R_{n}M')_{12} - (R_{n}M')_{34} = 0$
    for all $n \geq 1$, then $M' \in \mathfrak{R}(M,
    \mathcal{F} \cup \{\rho_{\tt{tensor}}\}) \cap
    \text{Sym}_{4}^{\tt{pd}}(\mathbb{R}_{> 0})$ is the
    required matrix.
    On the other hand, suppose there exists some $n \geq 1$
    such that $(R_{n}M')_{12} - (R_{n}M')_{34} \neq 0$.
    In that case, we will construct
    $\xi': N \mapsto (R_{n}N)_{12} - (R_{n}N)_{34}$.
    We see that $\xi'(\mathcal{S}_{M}(\theta^{*})) \neq 0$.
    So, \cref{corollary: stretchingWorksPositive} lets us
    find some $M'' \in \mathfrak{R}(M, \mathcal{F}' \cup \{\xi'\})
    \cap \text{Sym}_{4}^{\tt{pd}}(\mathbb{R}_{> 0})$.
    Now, if we let $\xi: N \mapsto N_{12} - N_{34}$,
    we see that $\xi(\mathcal{R}_{M''}(n)) \neq 0$.
    We also see that $F(\mathcal{R}_{M''}(1)) = F((M'')^3) \neq 0$
    for all $F \in \mathcal{F} \cup \{\rho_{\tt{tensor}}, \zeta\}$.
    So, we can use \cref{lemma: RInterpolationWorks} to find
    some $N \in \mathfrak{R}(M, \mathcal{F} \cup
    \{\rho_{\tt{tensor}}, \zeta, \xi\})
    \cap \text{Sym}_{4}^{\tt{pd}}(\mathbb{R}_{> 0})$.

    Since $\zeta(N) \neq 0$, if $N$ is not of Form (III), then
    it must either be diagonal distinct, or isomorphic to 
    a matrix of Form (I) or Form (II), or satisfies
    $N_{11} = N_{22} \neq N_{33} = N_{44}$, but
    is not of Form (III).
    In any case, since $\rho_{\tt{tensor}}(N) \neq 0$,
    we see that $\PlEVAL(N) \leq \PlEVAL(M)$ is
    $\#$P-hard, due to either
    \cref{theorem: diagDistHardness},
    or \cref{lemma: formIIReduction} and
    \cref{lemma: formIHardness}, or
    \cref{lemma: order22Forms}.
    On the other hand, if $N$ is of Form (III), it is the
    required matrix.
\end{proof}

\begin{lemma}\label{lemma: formIIICase2Impossible}
    Let $M \in \text{Sym}_{4}^{\tt{pd}}(\mathbb{R}_{> 0})$
    be of Form (III) such that $\rho_{\tt{tensor}}(M) \neq 0$.
    Let $\mathcal{F}$ be a countable set of 
    $\text{Sym}_{4}(\mathbb{R})$-polynomials such that
    $F(M) \neq 0$ for all $F \in \mathcal{F}$.
    Then, either $\PlEVAL(M)$ is $\#$P-hard, or
    there exists some $N \in \mathfrak{R}
    (M, \mathcal{F} \cup \{\rho_{\tt{tensor}}\}) \cap
    \text{Sym}_{4}^{\tt{pd}}(\mathbb{R}_{> 0})$
    of Form (III), such that
    $N_{11}N_{33}N_{12}N_{34} - (N_{13}N_{14})^{2} \neq 0$.
\end{lemma}
\begin{proof}
    We let $\xi: N \mapsto (N_{11}N_{33}N_{12}N_{34})
    - (N_{13}N_{14})^{2}$.
    If $\xi(M) \neq 0$, we are already done, so we may assume otherwise.
    We also let $\zeta: N \mapsto (N_{11} - N_{33})(N_{11} - N_{44})
    (N_{22} - N_{33})(N_{22} -N_{44})$.
    We note that $\zeta(M) \neq 0$.
    We will first use \cref{lemma: FormIIIM12M34} to find
    $M' \in \mathfrak{R}(M, \mathcal{F} 
    \cup \{\rho_{\tt{tensor}}, \zeta\}) \cap
    \text{Sym}_{4}^{\tt{pd}}(\mathbb{R}_{> 0})$
    of Form (III)
    such that $(M')_{12} \neq (M')_{34}$.
    We can now use \cref{lemma: DiagLarge} to
    find $M'' \in \mathfrak{R}(M', \mathcal{F} 
    \cup \{\rho_{\tt{tensor}}, \zeta\}) \cap
    \text{Sym}_{4}^{\tt{pd}}(\mathbb{R}_{> 0})$, such that
    $(M'')_{ii} > (M'')_{jk}$ for all $i \in [4]$, and
    $j \neq k \in [4]$.
    Since $\zeta(M'') \neq 0$, if $M''$ is not of Form (III), then
    it must either be diagonal distinct, or isomorphic to 
    a matrix of Form (I) or Form (II), or satisfies
    $(M'')_{11} = (M'')_{22} \neq (M'')_{33} = (M'')_{44}$, but
    is not of Form (III).
    In any case, since $\rho_{\tt{tensor}}(M'') \neq 0$,
    we see that $\PlEVAL(N) \leq \PlEVAL(M)$ is
    $\#$P-hard, due to either
    \cref{theorem: diagDistHardness},
    or \cref{lemma: formIIReduction} and
    \cref{lemma: formIHardness}, or
    \cref{lemma: order22Forms}.
    On the other hand,
    we may assume that $M''$ is    
    of Form (III).
    For convenience, we may rename this $M''$ as $M$.
    We note that we can also simultaneously permute the rows and columns of $M$ 
    (first by a possible switch $\{1,2\} \leftrightarrow \{3,4\}$ and then a possible flip $3 \leftrightarrow 4$)
    such that $M_{11} > M_{33}$, and
    $M_{13} \geq M_{14}$ without loss of generality.  
    
    We may assume that the entries of $M$ are generated
    by some $\{g_{t}\}_{t \in [d]}$.
    We can also replace this $M$ with some $c M$
    as guaranteed by \cref{lemma: MequivalentCM}
    such that $e_{ijt} \geq 0$ for all $i, j \in [4]$, and
    $t \in [d]$.
    We will now define a function $\widehat{\mathcal{T}_{M}}:
    \mathbb{R}^{d + 1} \rightarrow  \text{Sym}_{4}(\mathbb{R})$
    such that
    \begin{equation}\label{equation: wideHatM}
        \widehat{\mathcal{T}_{M}}(p, z_{1}, \dots, z_{d})_{ij}
        = \mathcal{T}_{M}(p^{z_{1}}, \dots, p^{z_{d}})_{ij}
    \end{equation}
    for all $i, j \in [4]$.
    By our choice of $M$, we know that
    $\widehat{\mathcal{T}_{M}}(e, \log(g_{1}), \dots, \log(g_{d}))
    = M$.
    We note that $\widehat{\mathcal{T}_{M}}(p, z_{1}, \dots, z_{d})_{ij}$
    is continuous as a function of all its variables,
    for all $i, j \in [4]$.
    So, there exist non-empty intervals $I_{1}, \dots, I_{d} \subseteq
    \mathbb{R}_{> 0}$
    such that for all $(z_{1}, \dots, z_{d}) \in
    I_{1} \times \cdots \times I_{d}$,
    $\widehat{\mathcal{T}_{M}}(e, z_{1}, \dots, z_{d})$
    also satisfies the properties of $M$ that
    $$M_{11} > M_{33}, M_{13} \geq M_{14}, M_{12} \neq M_{34},
    \text{ and } M_{ii} > M_{jk},$$
    for all $i \in [4]$, for $j \neq k \in [4]$.
    But now, we can pick some $(z_{1}^{*}, \dots, z_{d}^{*})
    \in \mathbb{Q}^{d} \cap (I_{1} \times \cdots \times I_{d})$.
    Since $z_{t}^{*} > 0$ for all $t \in [d]$, there exists
    some $Z^{*} \in \mathbb{Z}_{> 0}$, such that
    $Z \cdot z_{t}^{*} \in \mathbb{Z}_{> 0}$ for all $t \in [d]$.
    We will now define $\mathcal{T}_{M}^{*}: \mathbb{R} \rightarrow
    \text{Sym}_{4}(\mathbb{R})$ such that
    \begin{equation}\label{equation: MStar}
        \mathcal{T}_{M}^{*}(p)_{ij} =
        \widehat{\mathcal{T}_{M}}(p, Z^{*} \cdot z_{i}^{*}, \dots,
        Z^{*} \cdot z_{d}^{*})_{ij} = p^{z_{ij}}
    \end{equation}
    for some $z_{ij} = Z^{*} \sum_{t \in [d]}e_{ijt}z_{t}^{*}
    \in \mathbb{Z}_{\geq 0}$ for all $i, j \in [4]$.
    We also see that since $\mathcal{T}_{M}^{*}(e)_{ii} >
    \mathcal{T}_{M}^{*}(e)_{jk}$ for all $i \in [4]$, and
    $j \neq k \in [4]$,
    it must be the case that $z_{ii} > z_{jk}$ for all 
    $i \in [4]$, and $j \neq k \in [4]$.
    If $\xi(\mathcal{T}_{M}^{*}(p)) \neq 0$ 
    for some $p \in \mathbb{R}$, 
    then from the construction of $\mathcal{T}_{M}^{*}$,
    we see that there must exist some
    $\mathbf{p} \in \mathbb{R}^{d}$ such that
    $\xi(\mathcal{T}_{M}(\mathbf{p})) = 
    \xi(\mathcal{T}_{M}^{*}(p)) \neq 0$.
    Therefore, \cref{lemma: thickeningWorks} allows
    us to find some $N \in \mathfrak{R}(M, \mathcal{F} \cup
    \{\rho_{\tt{tensor}}, \zeta, \xi \}) \cap
    \text{Sym}_{4}^{\tt{pd}}(\mathbb{R}_{> 0})$.
    Since $\zeta(N) \neq 0$, if $N$ is not of
    Form (III), it must either be diagonal distinct, or
    isomorphic to a matrix of Form (I) or (II), or it
    should satisfy the equation $N_{11} = N_{22} 
    \neq N_{33} = N_{44}$, but not be of Form (III).
    In any of these cases, since $\rho_{\tt{tensor}}(N) \neq 0$,
    it follows from \cref{theorem: diagDistHardness},
    or \cref{lemma: formIIReduction} and
    \cref{lemma: formIHardness}, or
    \cref{lemma: order22Forms}, that
    $\PlEVAL(N) \leq \PlEVAL(M)$ is $\#$P-hard.
    On the other hand, if $N$ is of Form (III), then
    we note that it is the required matrix,
    and we are done.

    We may now assume that $\xi(\mathcal{T}_{M}^{*}(p)) = 0$
    for all $p \in \mathbb{R}$.
    This means that
    $p^{z_{11} + z_{33} + z_{12} + z_{34}} = p^{2z_{13} + 2z_{14}}$ for all $p \in \mathbb{R}$.
    So, we see that
    \begin{equation}\label{formIIIequality}
    z_{13} + z_{14} = \frac{(z_{11}  + z_{12})  +(z_{33}  + z_{34})}{2},
    \end{equation}    
    i.e., $z_{13} + z_{14}$ is the average of the other two sums.
    In particular, we have three possibilities,
    \begin{description}
        \item{Case 0.}
        $z_{11}  + z_{12} = z_{13} + z_{14} = z_{33}  + z_{34}$; or
        \item{Case 1.} 
        $z_{11}  + z_{12} > z_{13} + z_{14} > z_{33}  + z_{34}$; or
        \item{Case 2.} 
        $z_{11}  + z_{12} < z_{13} + z_{14} < z_{33}  + z_{34}$.
    \end{description}

    We will now consider
    $$\xi((\mathcal{T}_{M}^{*}(p))^{2}) = 
    ((\mathcal{T}_{M}^{*}(p))^{2})_{11}
    ((\mathcal{T}_{M}^{*}(p))^{2})_{33}
    ((\mathcal{T}_{M}^{*}(p))^{2})_{12}
    ((\mathcal{T}_{M}^{*}(p))^{2})_{34} - 
    \left(((\mathcal{T}_{M}^{*}(p))^{2})_{13}
    ((\mathcal{T}_{M}^{*}(p))^{2})_{14}\right)^{2}.$$
    We will show that this function in $p$ is not
    identically 0. For  a contradiction, we assume that
    $\xi(((\mathcal{T}_{M}^{*}(p))^{2})) = 0$ for all
    $p \in \mathbb{R}$.
    Thus, we have for    all $p \in \mathbb{R}$
    \begin{equation}\label{formIIIMainEqn}
        ((\mathcal{T}_{M}^{*}(p))^{2})_{11}
        ((\mathcal{T}_{M}^{*}(p))^{2})_{33}
        ((\mathcal{T}_{M}^{*}(p))^{2})_{12}
        ((\mathcal{T}_{M}^{*}(p))^{2})_{34} = 
        \left(((\mathcal{T}_{M}^{*}(p))^{2})_{13}
        ((\mathcal{T}_{M}^{*}(p))^{2})_{14}\right)^{2}
    \end{equation}
    where (using the Form (III))
    \begin{align*}
        ((\mathcal{T}_{M}^{*}(p))^{2})_{11} &= 
        p^{2z_{11}} + p^{2z_{12}} + p^{2z_{13}}+ p^{2z_{14}},\\
        ((\mathcal{T}_{M}^{*}(p))^{2})_{33} &= 
        p^{2z_{13}} + p^{2z_{14}} + p^{2z_{33}}+ p^{2z_{34}},\\
        ((\mathcal{T}_{M}^{*}(p))^{2})_{12} &= 
        2p^{z_{11} + z_{12}} + 2p^{z_{13} + z_{14}},\\
        ((\mathcal{T}_{M}^{*}(p))^{2})_{34} &= 
        2p^{z_{13} + z_{14}} + 2p^{z_{33} + z_{34}},\\
        ((\mathcal{T}_{M}^{*}(p))^{2})_{13} &= 
        p^{z_{11} + z_{13}} + p^{z_{12} + z_{14}} +
        p^{z_{13} + z_{33}}+ p^{z_{14} + z_{34}},\\
        ((\mathcal{T}_{M}^{*}(p))^{2})_{14} &= 
        p^{z_{11} + z_{14}} + p^{z_{12} + z_{13}} +
        p^{z_{13} + z_{34}}+ p^{z_{14} + z_{33}}.
    \end{align*}
\begin{itemize}
    \item 
    Step 1. \emph{The leading degree term of the LHS of
    \cref{formIIIMainEqn} is either 
     $4p^{3z_{11} + 2z_{33} + z_{12} + z_{13} + z_{14}}$
      (in the enumerated Case 1. of \cref{formIIIequality}) or 
    $4p^{2z_{11} + 3z_{33} + z_{34} + z_{13} + z_{14}}$ (in Case 2.) or    has coefficient $16$
      (in Case 0).
    }

    By our choice of $z_{ij}$, we know that $z_{11} > z_{33}
    > z_{ij}$ for all $i \neq j \in [4]$.
    So, $p^{2z_{11}}$ is the
    leading degree term of
    $$((\mathcal{T}_{M}^{*}(p))^{2})_{11} = 
    p^{2z_{11}} + p^{2z_{12}} + p^{2z_{13}}+ p^{2z_{14}},$$
    and $p^{2z_{33}}$ is the leading degree term of
    $$((\mathcal{T}_{M}^{*}(p))^{2})_{33} = 
    p^{2z_{13}} + p^{2z_{14}} + p^{2z_{33}}+ p^{2z_{34}}.$$
    Now, if the two terms of $((\mathcal{T}_{M}^{*}(p))^{2})_{12}$
    collapse into one term, that implies that
    $z_{11} + z_{12} = z_{13} + z_{14}$.
    But then, \cref{formIIIequality} implies 
    that $z_{13} + z_{14} = z_{33} + z_{34}$ as well.
    This would imply that
    $((\mathcal{T}_{M}^{*}(p))^{2})_{12} = 
    4p^{z_{11} + z_{12}} = ((\mathcal{T}_{M}^{*}(p))^{2})_{34}$.
    This means that the leading degree term of the LHS of
    \cref{formIIIMainEqn} will be
    $$16p^{4z_{11} + 2z_{33} + 2z_{12}}.$$
    On the other hand,
    if the two terms of $((\mathcal{T}_{M}^{*}(p))^{2})_{12}$
    do not collapse into each other, that means that
    $z_{11} + z_{12} \neq z_{13} + z_{14}$.
    From \cref{formIIIequality}, this implies that in fact,
    $z_{11} + z_{12} \neq z_{13} + z_{14} \neq z_{33} + z_{34}$
    are pairwise distinct.
    So, it must either be the case that
    $z_{11} + z_{12} > z_{13} + z_{14} > z_{33} + z_{34}$,
    or it must be the case that
    $z_{33} + z_{34} > z_{13} + z_{14} > z_{11} + z_{12}$.
    Depending on which of the two cases occur,
    the leading degree term of the LHS of
    \cref{formIIIMainEqn} will be, respectively,  either
    $$4p^{3z_{11} + 2z_{33} + z_{12} + z_{13} + z_{14}} ~~\text{ or }~~
    4p^{2z_{11} + 3z_{33} + z_{34} + z_{13} + z_{14}}.$$

\item 
Step 2. \emph{
Case 1. or Case 2. hold in \cref{formIIIequality}.
Case 0. does not.
    The  leading degree term of the RHS of
    \cref{formIIIMainEqn} is 
    $4p^{4z_{11} + 2z_{13} + 2z_{14}}.$
Furthermore, we have 
    \begin{description}
    \item{(A)} $z_{11} + z_{14} = z_{12} + z_{13} > z_{13} + z_{34}$, or
    \item{(B)} $z_{11} + z_{14} = z_{13} + z_{34} > z_{12} + z_{13}$.
    \end{description}
    and     the leading degree term of 
    $((\mathcal{T}_{M}^{*}(p))^{2})_{14}$ is obtained by combining two terms
   $p^{z_{11} + z_{14}} + p^{z_{12} + z_{13}}$ in case (A), or
   $p^{z_{11} + z_{14}} + p^{z_{13} + z_{34}}$ in case (B).
    }
    
    We will now analyze the leading degree term of the 
    RHS of \cref{formIIIMainEqn}.
    Let us first focus on the terms of
    $$((\mathcal{T}_{M}^{*}(p))^{2})_{13} = 
    p^{z_{11} + z_{13}} + p^{z_{12} + z_{14}} + p^{z_{13} + z_{33}}
    + p^{z_{14} + z_{34}}.$$
    We note that $z_{11} + z_{13} > z_{12} + z_{14}$,
    since $z_{11} > z_{12}$, and $z_{13} \geq z_{14}$.
    Similarly, $z_{11} + z_{13} > z_{13} + z_{33}$,
    since $z_{11} > z_{33}$.
    Finally, $z_{11} + z_{13} > z_{14} + z_{34}$.
    So, the leading degree term of 
    $((\mathcal{T}_{M}^{*}(p))^{2})_{13}$ is
    $p^{z_{11} + z_{13}}$ with coefficient 1.
    Now, we focus on the terms of
    $$((\mathcal{T}_{M}^{*}(p))^{2})_{14} = 
    p^{z_{11} + z_{14}} + p^{z_{12} + z_{13}} + p^{z_{13} + z_{34}}
    + p^{z_{14} + z_{33}}.$$
    Since $z_{11} > z_{33}$,
    it follows that $z_{11} + z_{14} > z_{14} + z_{33}$.
    So, the leading term coefficient of
    $((\mathcal{T}_{M}^{*}(p))^{2})_{14}$ cannot be 4.
    Since the square of
    this coefficient must be 4 or 16
    to match that of the LHS  of \cref{formIIIMainEqn},
    we see that this coefficient must be 2, and the 
    leading term coefficient of the LHS of \cref{formIIIMainEqn}
    must be 4 (not 16).
    In particular, this also rules out the possibility that
    $z_{11} + z_{12} = z_{13} + z_{14} = z_{33} + z_{34}$,
    i.e., Case 0. in   \cref{formIIIequality} does not hold, and
    Case 1. or Case 2. hold in \cref{formIIIequality}.
    

    Now, in $((\mathcal{T}_{M}^{*}(p))^{2})_{14}$, we have seen that
    the leading degree term cannot be $p^{z_{14} + z_{33}}$.
    Similarly, it is not possible that $z_{12} + z_{13} = 
    z_{13} + z_{34}$, since $z_{12} \neq z_{34}$.
    So, the leading degree term of 
    $((\mathcal{T}_{M}^{*}(p))^{2})_{14}$ must either be
    $p^{z_{11} + z_{14}} + p^{z_{12} + z_{13}}$, or
    $p^{z_{11} + z_{14}} + p^{z_{13} + z_{34}}$.
    Thus, either (A) or (B) hold.
    In either of these cases, we find that the leading
    degree term of the RHS of \cref{formIIIMainEqn} is
    $$4p^{4z_{11} + 2z_{13} + 2z_{14}}.$$

\item
Step 3.
\emph{Case 1.  does not hold in \cref{formIIIequality}, 
which implies that only Case 2 is viable.
 The  leading degree term of the LHS of
    \cref{formIIIMainEqn} is 
       $4p^{2z_{11} + 3z_{33} + z_{34} + z_{13} + z_{14}}$.
}

    By \cref{formIIIequality}, $z_{13} + z_{14} $
    is the average of $z_{11} + z_{12}$ and $z_{33} + z_{34}$.
    Assume  Case 1.  holds, i.e., 
    $z_{11} + z_{12} > z_{13} + z_{14} > z_{33} + z_{34}$. In this case, by Step 1.
    the leading degree term of the first term of $\xi((\mathcal{T}_{M}^{*}(p))^{2})$ is
    $4p^{3z_{11} + 2z_{33} + z_{12} + z_{13} + z_{14}}$.
    Since it has to equal the leading degree term of the second term,
    this means that
    $3z_{11} + 2z_{33} + z_{12} + z_{13} + z_{14} =
    4z_{11} + 2z_{13} + 2z_{14}$.
    Simplifying, we see that
    \begin{equation}\label{formIIISecondEqualityCase1}
        2z_{33} + z_{12} = z_{11} + z_{13} + z_{14}.
    \end{equation}
    But also, for 
    $$((\mathcal{T}_{M}^{*}(p))^{2})_{14} = 
    p^{z_{11} + z_{14}} + p^{z_{12} + z_{13}} + p^{z_{13} + z_{34}}
    + p^{z_{14} + z_{33}}$$
    in case (A), $z_{11} + z_{14} = z_{12} + z_{13}$, the
    leading degree term of 
    $((\mathcal{T}_{M}^{*}(p))^{2})_{14}$ is
    $p^{z_{11} + z_{14}} + p^{z_{12} + z_{13}}$.
    Together with \cref{formIIISecondEqualityCase1},
    this implies that $2z_{13} = 2z_{33}$, which is not possible
    since $z_{33} > z_{13}$.
    The other possibility is  case (B), $z_{11} + z_{14} = z_{13} + z_{34}$,
    the leading degree term of 
    $((\mathcal{T}_{M}^{*}(p))^{2})_{14}$
    is $p^{z_{11} + z_{14}} + p^{z_{13} + z_{34}}$.
    \cref{formIIISecondEqualityCase1} implies that
    $2z_{33} + z_{12} = 2z_{13} + z_{34}$.
    Since $p^{z_{13} + z_{34}}$ has the highest degree in
    $((\mathcal{T}_{M}^{*}(p))^{2})_{14}$, that also implies that  (note that  
    $z_{11} = z_{13} -z_{14} + z_{34}$ from Case (B))
    $$z_{13} + z_{34} = z_{13} - z_{14}  + z_{34} + z_{14} = z_{11} + z_{14}  > z_{14} + z_{33}.$$
    So, $2z_{33} + z_{12} = 2 z_{13}  + z_{34} > z_{13} + (z_{14} + z_{33})$,
    which means that $z_{33} + z_{12} > z_{13} + z_{14}$.
    Since we have assumed that we are in Case 1., $z_{11} + z_{12} >
    z_{13} + z_{14} > z_{33} + z_{34}$,
    we see that $z_{33} + z_{12} >  z_{13} + z_{14} > z_{33} + z_{34}$.
    So, $z_{12} > z_{34}$.
    However, since $p^{z_{13} + z_{34}}$  has the highest degree in
    $((\mathcal{T}_{M}^{*}(p))^{2})_{14}$,
    we also know that
    $z_{13} + z_{34} > z_{12} + z_{13}$.
    But this implies that $z_{34} > z_{12}$,
    which is a contradiction.
    So, it is not possible that
    $z_{11} + z_{12} > z_{13} + z_{14} > z_{33} + z_{34}$.
    We  conclude that in fact, Case 1. is impossible,  only Case 2. remains, namely
    \begin{equation}\label{formIIIInequality}
        z_{33} + z_{34} > z_{13} + z_{14} > z_{11} + z_{12},
    \end{equation}
    and the leading degree term of the LHS of \cref{formIIIMainEqn} is
    $4p^{2z_{11} + 3z_{33} + z_{34} + z_{13} + z_{14}}$.

\item 
Step 4.
\emph{Case (A) is impossible, which implies that only Case (B) is viable.}

    Since the leading degree terms of the LHS and the RHS of
    \cref{formIIIMainEqn} are equal, this means that
    $2z_{11} + 3z_{33} + z_{34} + z_{13} + z_{14} =
    4z_{11} + 2z_{13} + 2z_{14}$.
    From \cref{formIIIequality}, we know that
    $2z_{13} + 2z_{14} = z_{11} + z_{33} + z_{12} + z_{34}$.
    Substituting into the equation above,
    we get that
    $2z_{11} + 3z_{33} + z_{34} + z_{13} + z_{14} =
    5z_{11} + z_{33} + z_{34} + z_{12}$.
    Simplifying, we get
    \begin{equation}\label{formIIISecondEqualityCase2}
        2z_{33} + z_{13} + (z_{14} + z_{11}) = 4z_{11} + z_{12}.
    \end{equation}
    Assume Case (A), $z_{11} + z_{14} = z_{12} + z_{13}$.
    Together with \cref{formIIISecondEqualityCase2},
    this implies that
    $$2z_{33} + 2z_{13} = 4z_{11},$$
    which is not possible since $z_{11} > z_{33}, z_{13}$.
    So, we conclude that  Case (A) is impossible, and only
    Case (B) is viable, i.e., the leading degree term of
    $((\mathcal{T}_{M}^{*}(p))^{2})_{14}$ is
    $p^{z_{11} + z_{14}} + p^{z_{13} + z_{34}}$, and we have
    \begin{equation}\label{formIIIThirdEqualityCase2}
        z_{11} + z_{14} = z_{13} + z_{34} > z_{12} + z_{13}.
    \end{equation}
The only case remaining is Case 2. in combination of Case (B).

\item 
Step 5. 
\emph{Case 2. in combination of Case (B) is impossible.}

    We will now consider the least degree term of the LHS and RHS
    of \cref{formIIIMainEqn}.
    From Case 2. \cref{formIIIInequality},
    we note that the least degree term of
    $$((\mathcal{T}_{M}^{*}(p))^{2})_{12} = 2p^{z_{11} + z_{12}}
    + 2p^{z_{13} + z_{14}}$$
    is $2p^{z_{11} + z_{12}}$,
    and the least degree term of
    $$((\mathcal{T}_{M}^{*}(p))^{2})_{34} = 2p^{z_{13} + z_{14}}
    + 2p^{z_{33} + _{34}}$$
    is $2p^{z_{13} + z_{14}}$.
    Since $z_{11} + z_{12} < z_{33} + z_{34}$, and
    $z_{11} > z_{33}$, we conclude that $z_{12} < z_{34}$.
    Similarly, since $z_{11} + z_{12} < z_{13} + z_{14}$, we also
    conclude that $z_{12} < z_{13}, z_{14}$.
    So, the least degree term of
    $$((\mathcal{T}_{M}^{*}(p))^{2})_{11} = p^{2z_{11}}+  p^{2z_{12}}
    + p^{2z_{13}} + p^{2z_{14}}$$
    is $p^{2z_{12}}$.
    From case (B) \cref{formIIIThirdEqualityCase2}, we see that
    $z_{11} + z_{14} = z_{13} + z_{34}$.
    Since $z_{11} > z_{34}, z_{13}$, we see that
    $z_{14} < z_{13}$, and $z_{14} < z_{34}$.
    So, the least degree term of 
    $$((\mathcal{T}_{M}^{*}(p))^{2})_{33} = p^{2z_{13}} + p^{2z_{14}}
    + p^{2z_{33}} + p^{2z_{34}}$$
    is $p^{2z_{14}}$.
    Summing up, the least degree term of the LHS of
    \cref{formIIIMainEqn} is precisely
    $$4p^{z_{11} + 3z_{12} + z_{13} + 3z_{14}}.$$

    As for the RHS of \cref{formIIIMainEqn},
    we note that since $z_{14} < z_{13}$, and
    $z_{12} < z_{11}, z_{33}, z_{34}$, the least degree term of
    $$((\mathcal{T}_{M}^{*}(p))^{2})_{13} = p^{z_{11} + z_{13}}
    + p^{z_{12} + z_{14}} + p^{z_{13} + z_{33}} + p^{z_{14} + z_{34}}$$
    is $p^{z_{12} + z_{14}}$.
    This means that the least degree term of
    $((\mathcal{T}_{M}^{*}(p))^{2})_{14}$ must have a coefficient of
    exactly $2$ for the coefficients of the LHS and RHS of
    \cref{formIIIMainEqn} to be equal.
    But we are in case (B), and so
    $p^{z_{11} + z_{14}} + p^{z_{13} + z_{34}}$ is the leading 
    degree term of
    $$((\mathcal{T}_{M}^{*}(p))^{2})_{14} = p^{z_{11} + z_{14}}
    + p^{z_{12} + z_{13}} + p^{z_{13} + z_{34}} + p^{z_{14} + z_{33}}.$$
    So, the least degree term must be the combined term from
    $p^{z_{12} + z_{13}} + p^{z_{14} + z_{33}}$.
    In this case, the least degree term of the RHS of 
    \cref{formIIIMainEqn} is
    $$4p^{2z_{12} + 2z_{14} + 2z_{12} + 2z_{13}}.$$
    Since this must be equal to the least degree term of the
    LHS of \cref{formIIIMainEqn}, we find that
    $z_{11} + 3z_{12} + z_{13} + 3z_{14} =
    4z_{12} + 2z_{14} + 2z_{13}$.
    Simplifying, we get
    $$z_{11} + z_{14} = z_{12} + z_{13},$$
    which contradicts case (B) \cref{formIIIThirdEqualityCase2}.    
    So, we see that in fact, 
    Case 2. in combination of case (B)  is impossible.
    This implies that our original assumption that
    $\xi(((\mathcal{T}_{M}^{*}(p))^{2})) = 0$ for all
    $p \in \mathbb{R}$ must be false.
\end{itemize}
    Now, if we let $\xi_{2}: N \mapsto \xi(N^{2})$, we see that
    $\xi_{2}(\mathcal{T}_{M}^{*}(p))$ is not the zero function.
    From the construction of $\mathcal{T}_{M}^{*}$, this implies that
    $\xi_{2}(\mathcal{T}_{M}(\mathbf{p}))$ is
    not the zero function either.
    So, we can use \cref{lemma: thickeningWorks} to find
    $M' \in \mathfrak{R}(M, \mathcal{F} \cup
    \{\rho_{\tt{tensor}}, \zeta, \xi_{2} \}) \cap
    \text{Sym}_{4}^{\tt{pd}}(\mathbb{R}_{> 0})$.
    Since $F(M') \neq 0$ for all $F \in \mathcal{F} \cup
    \{\rho_{\tt{tensor}}, \zeta \}$, and $\xi((M')^{2}) \neq 0$,
    we can use \cref{corollary: stretchingWorksPositive}
    to find some $N \in \mathfrak{R}(M', \mathcal{F} \cup
    \{\rho_{\tt{tensor}}, \zeta, \xi \}) \cap
    \text{Sym}_{4}^{\tt{pd}}(\mathbb{R}_{> 0})$.
    Since $\zeta(N) \neq 0$, if $N$ is not of Form (III), then
    it must be either  diagonal distinct, or isomorphic to 
    a matrix of Form (I) or Form (II), or satisfies
    $N_{11} = N_{22} \neq N_{33} = N_{44}$, but
    is not of Form (III).
    In any case, since $\rho_{\tt{tensor}}(N) \neq 0$,
    we see that $\PlEVAL(N) \leq \PlEVAL(M)$ is
    $\#$P-hard, due to either
    \cref{theorem: diagDistHardness},
    or \cref{lemma: formIIReduction} and
    \cref{lemma: formIHardness}, or
    \cref{lemma: order22Forms}.
    On the other hand, if $N$ is of Form (III), it is the
    required matrix (satisfying $\xi(N) \ne 0$).
\end{proof}

We will now need a few more technical lemmas that show that
since $\rho_{\tt{tensor}}(M) \neq 0$,
$M$ may be assumed to have some more structure.

\begin{lemma}\label{lemma: formIIILogRatioRational}
    Let $M \in \text{Sym}_{4}^{\tt{pd}}(\mathbb{R}_{> 0})$
    be of Form (III) such that $\rho_{\tt{tensor}}(M) \neq 0$.
    Let $\mathcal{F}$ be a countable set of 
    $\text{Sym}_{4}(\mathbb{R})$-polynomials such that
    $F(M) \neq 0$ for all $F \in \mathcal{F}$.
    Then, either $\PlEVAL(M)$ is $\#$P-hard, or
    there exists some $N \in \mathfrak{R}
    (M, \mathcal{F} \cup \mathcal{F}' \cup 
    \{\rho_{\tt{tensor}}\}) \cap
    \text{Sym}_{4}^{\tt{pd}}(\mathbb{R}_{> 0})$
    of Form (III), where
    $\mathcal{F}' = \{\xi_{c_{1}, c_{2}}:
    c_{1}, c_{2} \in (\mathbb{Z}_{> 0})\}$ for
    the $\text{Sym}_{4}(\mathbb{R})$-polynomials
    $$\xi_{c_{1}, c_{2}}: N \mapsto 
    (N_{13}^{c_{1}}N_{12}^{c_{2}} -
    N_{14}^{c_{1}}N_{11}^{c_{2}})^{2} + 
    (N_{33}^{c_{1}}N_{14}^{c_{2}} -
    N_{34}^{c_{1}}N_{13}^{c_{2}})^{2}.$$
\end{lemma}
\begin{proof}
    We let $\zeta: N \mapsto (N_{11} - N_{33})(N_{11} - N_{44})
    (N_{22} - N_{33})(N_{22} -N_{44})$.
    We note that $\zeta(M) \neq 0$.
    Without loss of generality, we may first
    replace $M$ with the matrix $M' \in \mathfrak{R}(M, 
    \mathcal{F} \cup \{\rho_{\tt{tensor}}, \zeta\}) \cap
    \text{Sym}_{4}^{\tt{pd}}(\mathbb{R}_{> 0})$
    whose existence is guaranteed by
    \cref{lemma: DiagLarge}.
    So, we may assume that $M_{ii} > M_{jk}$
    for all $i \in [4]$, $j \neq k \in [4]$.
    Now, we let the entries of $M$ be generated by some
    $\{g_{t}\}_{t \in [d]}$.
    We may replace $M$ with some $c M$, as guaranteed
    by \cref{lemma: MequivalentCM} such that
    $e_{ijt} \geq 0$ for all $i, j \in [4]$, $t \in [d]$.
    We also note that $e_{ij0} = 0$ for all $i, j \in [4]$.
    We will now define the function $\widehat{\mathcal{T}_{M}}:
    \mathbb{R}^{d + 1} \rightarrow \text{Sym}_{4}(\mathbb{R})$
    just as we did in \cref{equation: wideHatM}.
    Then, we similarly define $\mathcal{T}_{M}^{*}: \mathbb{R}
    \rightarrow \text{Sym}_{4}(\mathbb{R})$ as in
    \cref{equation: MStar},
    such that
    $$\mathcal{T}_{M}^{*}(p)_{ij} = p^{z_{ij}},$$
    for some integers $z_{ij}$.
    We may assume that $z_{11} \neq z_{33}$, and that
    $z_{ii} > z_{jk}$ for all $i \in [4]$, 
    and $j \neq k \in [4]$.

    We will now consider $\xi_{c_{1}, c_{2}}((\mathcal{T}_{M}^{*}(p))^{2})$.
    First we want to show that if $\xi_{c_{1}, c_{2}}((\mathcal{T}_{M}^{*}(p))^{2}) = 0$,
    for all $p \in \mathbb{R}$, then $c_1=c_2$.
    We note that 
    \begin{multline*}
        ((\mathcal{T}_{M}^{*}(p))^{2}_{13})^{c_{1}}
        ((\mathcal{T}_{M}^{*}(p))^{2}_{12})^{c_{2}}
        = \left(p^{z_{11} + z_{13}} + p^{z_{12} + z_{14}} + 
        p^{z_{13} + z_{33}} + p^{z_{14} + z_{34}}\right)^{c_{1}}\\
        \cdot \left(2p^{z_{11} + z_{12}} + 2p^{z_{13}
        + z_{14}}\right)^{c_{2}},
        \text{ and }
    \end{multline*}
    \begin{multline*}
        ((\mathcal{T}_{M}^{*}(p))^{2}_{14})^{c_{1}}
        ((\mathcal{T}_{M}^{*}(p))^{2}_{11})^{c_{2}}
        = \left(p^{z_{11} + z_{14}} + p^{z_{12} + z_{13}} + 
        p^{z_{13} + z_{34}} + p^{z_{14} + z_{33}}\right)^{c_{1}}\\
        \cdot \left(p^{2z_{11}} + p^{2z_{12}} + 
        p^{2z_{13}} + p^{2z_{14}}\right)^{c_{2}}.
    \end{multline*}
     The identity $\xi_{c_{1}, c_{2}}((\mathcal{T}_{M}^{*}(p))^{2}) = 0$  for all $p \in \mathbb{R}$ implies that
     both terms in $\xi_{c_{1}, c_{2}}$ are 0.  In particular,
     $((\mathcal{T}_{M}^{*}(p))^{2}_{13})^{c_{1}}
    ((\mathcal{T}_{M}^{*}(p))^{2}_{12})^{c_{2}} = 
    ((\mathcal{T}_{M}^{*}(p))^{2}_{14})^{c_{1}}
    ((\mathcal{T}_{M}^{*}(p))^{2}_{11})^{c_{2}}$,
    for all $p \in \mathbb{R}$,
    we see that the coefficient of the leading term of the LHS here
    is a (possibly non-trivial) multiple of
    $2^{c_{2}}$.
    As for the RHS, by construction, the leading term of
    $$(\mathcal{T}_{M}^{*}(p))^{2}_{11} = p^{2z_{11}} + p^{2z_{12}} 
    + p^{2z_{13}} + p^{2z_{14}}$$
    will be $p^{2z_{11}}$. So, the leading term of 
    $((\mathcal{T}_{M}^{*}(p))^{2}_{11})^{c_{2}}$
    will be
    $p^{2c_{2} z_{11}}$, with a coefficient of $1$.
    So, this means that some terms within 
    $((\mathcal{T}_{M}^{*}(p))^{2}_{14})^{c_{1}}$ must be
    equal to each other.
    In fact, since the leading term coefficient of this term
    must be a multiple of $2^{c_{2}}$, it must either be the case
    that the leading term of $(\mathcal{T}_{M}^{*}(p))^{2}_{14}$
    must have either $2$ or all $4$ of the terms have the
    same degree.
    If all four terms have the same degree, that would imply that
    $z_{11} + z_{14} = z_{33} + z_{14}$, which is not
    true since $z_{11} \neq z_{33}$.
    So, this means that the leading term coefficient of the RHS
    is precisely $2^{c_{1}}$.
    Since this is a multiple of $2^{c_{1}}$, this implies that
    $c_{1} \geq c_{2}$.

    On the other hand, we can apply the same argument on  the second term of
    the identity $\xi_{c_{1}, c_{2}}((\mathcal{T}_{M}^{*}(p))^{2}) = 0$, which is
    $((\mathcal{T}_{M}^{*}(p))^{2}_{33})^{c_{1}}
    ((\mathcal{T}_{M}^{*}(p))^{2}_{14})^{c_{2}} = 
    ((\mathcal{T}_{M}^{*}(p))^{2}_{34})^{c_{1}}
    ((\mathcal{T}_{M}^{*}(p))^{2}_{13})^{c_{2}}$.
    Then we  find that $c_{2} \geq c_{1}$,
    which implies that in fact, $c_{1} = c_{2}$.
    
    Now, if $c_{1} = c_{2}$, let us consider
    the matrix $M$ itself.
    If $(M_{13}M_{12})^{c_{1}} = (M_{11}M_{14})^{c_{1}}$, and
    $(M_{33}M_{14})^{c_{1}} = (M_{13}M_{34})^{c_{1}}$,
    since $M \in \text{Sym}_{4}(\mathbb{R}_{> 0})$, that
    implies that $M_{11}M_{14} = M_{12}M_{13}$, and
    $M_{14}M_{33} = M_{13}M_{34}$.
    We also note that $M_{14} = M_{23}$, and
    $M_{11}M_{44} = M_{22}M_{33}$, since $M$ is
    of Form (III).
    So, this implies that $\rho_{\tt{tensor}}(M) = 0$,
    which contradicts our assumption about $M$.

    So, we see that for all $\xi_{c_{1}, c_{2}}$, either
    $\xi_{c_{1}, c_{2}}(M) \neq 0$, or there
    exists some $p \in \mathbb{R}$ such that
    $\xi_{c_{1}, c_{2}}((\mathcal{T}_{M}^{*}(p))^{2}) \neq 0$.
    We will now let $\mathcal{F}'' = \{F'': N \mapsto
    \xi_{c_{1}, c_{2}}(N^{2})\ |\ c_{1} \neq c_{2} \in
    \mathbb{Z}_{> 0} \}$, and
    $\mathcal{F}''' = \{\xi_{c, c}: c \in \mathbb{Z}_{> 0} \}$.
    We can now use
    \cref{lemma: thickeningWorks} to find $M'' \in \mathfrak{R}
    (M, \mathcal{F} \cup \mathcal{F}'' \cup \mathcal{F}''' \cup
    \{\rho_{\tt{tensor}}, \zeta\}) \cap
    \text{Sym}_{4}^{\tt{pd}}(\mathbb{R}_{> 0})$.
    Since $F(M'') \neq 0$ for all $F \in \mathcal{F} \cup
    \mathcal{F}''' \cup \{\rho_{\tt{tensor}}, \zeta\}$,
 and $\xi_{c_{1}, c_{2}}(M^{2}) \neq 0$ for all  $c_{1} \neq c_{2} \in    \mathbb{Z}_{> 0}$,
    we can then use \cref{corollary: stretchingWorksPositive}
    to find some $N \in \mathfrak{R}
    (M, \mathcal{F} \cup \mathcal{F}' \cup
    \{\rho_{\tt{tensor}}, \zeta\}) \cap
    \text{Sym}_{4}^{\tt{pd}}(\mathbb{R}_{> 0})$.
    Since $\zeta(N) \neq 0$, if $N$ is not of Form (III), then
    it must either be diagonal distinct, or isomorphic to 
    a matrix of Form (I) or Form (II), or satisfies
    $N_{11} = N_{22} \neq N_{33} = N_{44}$, but
    is not of Form (III).
    In any case, since $\rho_{\tt{tensor}}(N) \neq 0$,
    we see that $\PlEVAL(N) \leq \PlEVAL(M)$ is
    $\#$P-hard, due to either
    \cref{theorem: diagDistHardness},
    or \cref{lemma: formIIReduction} and
    \cref{lemma: formIHardness}, or
    \cref{lemma: order22Forms}.
    On the other hand, if $N$ is of Form (III), it is the
    required matrix.
\end{proof}

\begin{lemma}\label{lemma: formIIILogRatio}
    Let $M \in \text{Sym}_{4}^{\tt{pd}}(\mathbb{R}_{> 0})$
    be of Form (III) such that $\rho_{\tt{tensor}}(M) \neq 0$.
    Let $\mathcal{F}$ be a countable set of 
    $\text{Sym}_{4}(\mathbb{R})$-polynomials such that
    $F(M) \neq 0$ for all $F \in \mathcal{F}$.
    Then, either $\PlEVAL(M)$ is $\#$P-hard, or
    there exists some $N \in \mathfrak{R}
    (M, \mathcal{F} \cup \{\rho_{\tt{tensor}}\}) \cap
    \text{Sym}_{4}^{\tt{pd}}(\mathbb{R}_{> 0})$
    of Form (III), such that
    $$\left(\log\left(\frac{N_{13}}{N_{14}}\right)\right)^{2} \neq
    \log\left(\frac{N_{11}}{N_{12}}\right) \cdot
    \log\left(\frac{N_{33}}{N_{34}}\right),$$
    and $N_{ii} > N_{jk}$ for all $i \in [4]$,
    $j \neq k \in [4]$, and
    $N_{11}N_{33}N_{12}N_{34} \neq (N_{13}N_{14})^{2}$.
\end{lemma}
\begin{proof}
    We will let $\zeta: N \mapsto (N_{11} - N_{33})(N_{11} - N_{44})
    (N_{22} - N_{33})(N_{22} -N_{44})$, and
    $\zeta': N \mapsto N_{11}N_{33}N_{12}N_{34} - (N_{13}N_{14})^{2}$.
    We note that $\zeta(M) \neq 0$.
    We will also let $\mathcal{F}' = \{\xi_{c_{1}, c_{2}}: c_{1}, c_{2}
    \in (\mathbb{Z}_{> 0})\}$ as defined in 
    \cref{lemma: formIIILogRatioRational}, such that
    .$$\xi_{c_{1}, c_{2}}: N \mapsto 
    (N_{13}^{c_{1}}N_{12}^{c_{2}} -
     N_{14}^{c_{1}}N_{11}^{c_{2}})^{2} + 
    (N_{33}^{c_{1}}N_{14}^{c_{2}} -
     N_{34}^{c_{1}}N_{13}^{c_{2}})^{2}.$$
    We can first find $M' \in \mathfrak{R}(M, \mathcal{F} \cup
    \mathcal{F}' \cup \{\rho_{\tt{tensor}}, \zeta\}) \cap
    \text{Sym}_{q}^{\tt{pd}}(\mathbb{R}_{> 0})$
    using \cref{lemma: formIIILogRatioRational}.
    We can now use \cref{lemma: formIIICase2Impossible}
    to find $M'' \in \mathfrak{R}(M, \mathcal{F} \cup
    \mathcal{F}' \cup \{\rho_{\tt{tensor}}, \zeta, \zeta'\}) \cap
    \text{Sym}_{q}^{\tt{pd}}(\mathbb{R}_{> 0})$.
    Finally, we can use \cref{lemma: DiagLarge} to find
    $M''' \in \mathfrak{R}(M, \mathcal{F} \cup
    \mathcal{F}' \cup \{\rho_{\tt{tensor}}, \zeta, \zeta'\}) \cap
    \text{Sym}_{q}^{\tt{pd}}(\mathbb{R}_{> 0})$
    such that $(M''')_{ii} > (M''')_{jk}$ for all
    $i \in [4]$, $j \neq k \in [4]$.
    Since $\zeta(M''') \neq 0$, if $M'''$ is not of Form (III), then
    it must either be diagonal distinct, or isomorphic to 
    a matrix of Form (I) or Form (II), or satisfies
    $(M''')_{11} = (M''')_{22} \neq (M''')_{33} = (M''')_{44}$, but
    is not of Form (III).
    In any case, since $\rho_{\tt{tensor}}(M''') \neq 0$,
    we see that $\PlEVAL(M''') \leq \PlEVAL(M)$ is
    $\#$P-hard, due to either
    \cref{theorem: diagDistHardness},
    or \cref{lemma: formIIReduction} and
    \cref{lemma: formIHardness}, or
    \cref{lemma: order22Forms}.
    We may now assume that $M'''$ is of Form (III).

    For convenience, we can rename this $M'''$ as $M$.
    We also note that if $M_{13} = M_{14}$, we are already done,
    since 
    $\nicefrac{M_{11}}{M_{12}} > 1, \nicefrac{M_{33}}{M_{34}} > 1$,
    which implies that
    $$\log\left(\frac{M_{11}}{M_{12}}\right) \cdot
    \log\left(\frac{M_{33}}{M_{34}}\right) > 0
    = \left(\log\left(\frac{M_{13}}{M_{14}}\right) \right)^{2}.$$
    So, we may assume that $M_{13} \neq M_{14}$.
    Now, we can assume without loss of generality
    that $M_{11} > M_{33}$, and that $M_{13} > M_{14}$
    by permuting the rows and columns of $M$
    (first by a possible switch $\{1,2\} \leftrightarrow \{3,4\}$ and then a possible flip $3 \leftrightarrow 4$).
    
    Now, we may let the entries of $M$ be generated by some
    $\{g_{t} \}_{t \in [d]}$, such that
    $e_{ijt} \geq 0$ for all $i, j \in [4]$, $t \in [d]$.    
    So, for all $\mathbf{p} \in (\mathbb{R}_{> 0})^{d}$,
    $$\left(\log\left(\frac{\mathcal{T}_{M}(\mathbf{p})_{13}}
    {\mathcal{T}_{M}(\mathbf{p})_{14}}\right)\right)^{2}
    = \left(\sum_{t \in [d]}(e_{13t} - e_{14t})\log(p_{t})\right)^{2},
    \text{ and }$$
    $$\left(\log\left(\frac{\mathcal{T}_{M}(\mathbf{p})_{33}}
    {\mathcal{T}_{M}(\mathbf{p})_{34}}\right) \right)
    \left(\log\left(\frac{\mathcal{T}_{M}(\mathbf{p})_{11}}
    {\mathcal{T}_{M}(\mathbf{p})_{12}}\right) \right)
    = \left(\sum_{t \in [d]}(e_{33t} - e_{34t})\log(p_{t})\right) 
    \cdot \left(\sum_{t \in [d]}(e_{11t} - e_{12t})\log(p_{t})\right).$$
    We note that none of these three functions are the constant zero function.
    For example, since $M_{11} > M_{12}$, we have 
    $(e_{111}, \dots, e_{11d}) \ne (e_{121}, \dots, e_{12d})$ and thus
    $\sum_{t \in [d]}(e_{11t} - e_{12t})\log(p_{t})$ is not constant zero.
    The other two being not constant zero follow similarly from
    $M_{13} > M_{14}$  and
    $M_{33} > M_{44}$.
    We will now make use of the following claim,
    which we shall prove shortly.
    \begin{claim}\label{claim: formIIIRationalOnly}
        Let $\mathbf{0} \neq (a_{1}, \dots, a_{n}),
        (b_{1}, \dots, b_{n}), 
        (c_{1}, \dots, c_{n}) \in \mathbb{Z}^{n}$ such that
        $$\left(\sum_{i \in [n]}a_{i}x_{i}\right)^{2}
        = \left(\sum_{i \in [n]}b_{i}x_{i} \right)
        \left(\sum_{i \in [n]}c_{i}x_{i}\right),$$
        for all $(x_{1}, \dots, x_{n}) \in
        \mathbb{R}^{n}$, then
        there exists a rational constant $\kappa$ such that
        $$(c_{1}, \dots, c_{n}) = \kappa \cdot (a_{1}, \dots, a_{n}) 
        = \kappa^{2} \cdot (b_{1}, \dots, b_{n}).$$
    \end{claim}

    We will now define the function $\Omega: (\mathbb{R}_{> 0})^{d}
    \rightarrow \mathbb{R}$ such that
    $$\Omega(\mathbf{p}) =
    \left(\log\left(\frac{\mathcal{T}_{M}(\mathbf{p})_{13}}
    {\mathcal{T}_{M}(\mathbf{p})_{14}}\right)\right)^{2} - 
    \left(\log\left(\frac{\mathcal{T}_{M}(\mathbf{p})_{33}}
    {\mathcal{T}_{M}(\mathbf{p})_{34}}\right) \right) \cdot
    \left(\log\left(\frac{\mathcal{T}_{M}(\mathbf{p})_{11}}
    {\mathcal{T}_{M}(\mathbf{p})_{12}}\right) \right).$$
    \cref{claim: formIIIRationalOnly} tells us that
    if $\Omega(\mathbf{p}) = 0$ for all
    $\mathbf{p} \in (\mathbb{R}_{> 0})^{d}$,
    then there exists some rational $\kappa \in \mathbb{Q}$ such that
    $$(e_{111} - e_{121}, \dots, e_{11d} - e_{12d}) =
    \kappa \cdot (e_{131} - e_{141}, \dots, e_{13d} - e_{14d}) 
    = \kappa^{2} \cdot (e_{331} - e_{341}, \dots, e_{33d} - e_{34d}).$$
    
    This implies that
    $$\frac{\left(\log\left(\nicefrac{M_{11}}
    {M_{12}}\right)\right)}
    {\left(\log\left(\nicefrac{M_{13}}{M_{14}}\right)\right)}
    =  \frac{(e_{111} - e_{121})\log(g_{1}) + \cdots +
    (e_{11d} - e_{12d})\log(g_{d})}
    {(e_{131} - e_{141})\log(g_{1}) + \cdots +
    (e_{13d} - e_{14d})\log(g_{d})}
    = \kappa = \frac{c_{1}}{c_{2}}, \text{ and }$$
    $$\frac{\left(\log\left(\nicefrac{M_{13}}
    {M_{14}}\right)\right)}
    {\left(\log\left(\nicefrac{M_{33}}{M_{34}}\right)\right)}
    = \frac{(e_{131} - e_{141})\log(g_{1}) + \cdots +
    (e_{13d} - e_{14d})\log(g_{d})}
    {(e_{331} - e_{341})\log(g_{1}) + \cdots +
    (e_{33d} - e_{34d})\log(g_{d})}
    = \kappa = \frac{c_{1}}{c_{2}},$$
    for some positive integers $c_{1}, c_{2}$.
    But this means that
    $$\left(\frac{M_{11}}{M_{12}}\right)^{c_{2}} = 
    \left(\frac{M_{13}}{M_{14}}\right)^{c_{1}}, \text{ and }
    \left(\frac{M_{13}}
    {M_{14}}\right)^{c_{2}} = 
    \left(\frac{M_{33}}{M_{34}}\right)^{c_{1}}.$$
    In particular, we see that
    $$((M)_{13})^{c_{1}}((M)_{12})^{c_{2}} = 
    ((M)_{11})^{c_{2}}((M)_{14})^{c_{1}}, \text{ and }
    ((M)_{13})^{c_{2}}((M)_{34})^{c_{1}} = 
    ((M)_{33})^{c_{1}}((M)_{14})^{c_{2}}.$$
    
    This contradicts our assumption that
    $\xi_{c_{1}, c_{2}}(M) \ne 0$.
    So, our assumption that
    $\Omega(\mathbf{p}) = 0$ for all $\mathbf{p} \in (\mathbb{R}_{> 0})^{d}$
    must be false.
    Let $\mathbf{p}' = (p'_{1}, \dots, p'_{d})
    \in (\mathbb{R}_{> 0})^{d}$ such that
    $\Omega(\mathbf{p}') \neq 0$.
    We can now define $\mathbf{p}: \mathbb{R} \rightarrow \mathbb{R}^{d}$,
    such that
    $$\mathbf{p}(t) = t \cdot \mathbf{p}' + (1 - t) \cdot
    (g_{1}, \dots, g_{d}).$$
    We note that $\mathbf{p}(0) = (g_{1}, \dots, g_{d}) \in
    (\mathbb{R}_{> 0})^{d}$, and $\mathbf{p}(1) = \mathbf{p}' \in
    (\mathbb{R}_{> 0})^{d}$.
    Therefore, for all $t \in [0, 1]$, and for all $i \in [d]$,
    $\mathbf{p}(t)_{i} = t \cdot p'_{i} + (1 - t)g_{i} > 0$.
    Moreover, since $\mathbf{p}(t)$ is continuous as a function
    of $t$, we see that there exists some $\delta_{1}, \delta_{2} > 0$
    such that $\mathbf{p}(t) \in (\mathbb{R}_{> 0})^{d}$ for all
    $t \in (-\delta_{1}, \delta_{1})$, and
    $t \in (1 - \delta_{2}, 1 + \delta_{2})$.
    We will now let $U = (-\delta_{1}, 1 + \delta_{2}) \subset
    \mathbb{R}$ be an open set, such that
    $\mathbf{p}(t) \in (\mathbb{R}_{> 0})^{d}$ for all
    $t \in U$.
    We now define $\omega: U \rightarrow \mathbb{R}$ such that
    $$\omega(t) = \Omega(\mathbf{p}(t)).$$
    Clearly, $\omega$ is a real valued analytic function on $U$.
    Moreover, by our choice of $\mathbf{p}'$, we know that
    $\omega(1) = \Omega(\mathbf{p}') \neq 0$.
    So, $\omega$ is a non-zero real analytic function on the open set $U$.
    So, we know (\cite{krantz2002primer}, Corollary 1.2.7) that
    the set of zeros of $\omega$ within $[0, 1]$ cannot have an
    accumulation point.
    In particular, there can only be finitely many zeros of $\omega$
    within $[0, 1]$.
    We can denote the set of zeros as $\emptyset_{\omega}
    = \{t \in [0, 1]: \omega(t) = 0\}$.
    
    Since $\mathbf{p}(0) = (g_{1}, \dots, g_{d})$, we see that
    $\mathcal{T}_{M}(\mathbf{p}(0)) = M$.
    So, there exists some $0< \delta <1$ such that
    $\mathcal{T}_{M}(\mathbf{p}(t))_{ii} >
    \mathcal{T}_{M}(\mathbf{p}(t))_{jk} > 0$,
    for all $i \in [4], j \neq k \in [4]$, and $t \in (0, \delta)$.
    We also note that the eigenvalues of $\mathcal{T}_{M}(\mathbf{p}(t))$
    are continuous as functions of $t$.
    Since $\mathcal{T}_{M}(\mathbf{p}(0)) = M$ is positive definite,
    this implies that there exists some $0 < \delta' < 1$ such that
    $\mathcal{T}_{M}(\mathbf{p}(t))$ is positive definite
    for all $t \in (0, \delta')$.
    We will now let $\mathcal{F}'' = \mathcal{F} \cup \mathcal{F}'
    \cup \{\rho_{\tt{tensor}}, \zeta, \zeta'\}$.
    We note that
    $\mathcal{F}''$ is a countable set of
    $\text{Sym}_{4}(\mathbb{R})$-polynomials.
    So, $F(\mathcal{T}_{M}(\mathbf{p}(t)))$ is a polynomial in $t$ for all
    $F \in \mathcal{F}''$.
    Since $F(M) \neq 0$ for all
    $F \in \mathcal{F}''$, by our choice of $M$,
    it follows that $F(\mathcal{T}_{M}(\mathbf{p}(t)))$ is a non-zero
    polynomial in $t$ for all $F \in \mathcal{F}''$.
    We let $\emptyset_{F} = \{t \in (0, 1):
    F(\mathcal{T}_{M}(\mathbf{p}(t))) = 0\}$ for all
    $F \in \mathcal{F}''$.
    Each of these is a finite set, and therefore,
    $\cup_{F \in \mathcal{F}''}\emptyset_{F}$ is a countable set.

    We can therefore find $\mathbf{p}^{*} = \mathbf{p}(t^{*})$ for
    some $t^{*} \in \left((0, \delta)
    \cap (0, \delta')\right) \setminus
    \left(\cup_{F \in \mathcal{F}''}\emptyset_{F} \cup
    \emptyset_{\omega}\right)$.
    By our choice of $t^{*}$, we see that
    $F(\mathcal{T}_{M}(\mathbf{p}^{*})) \neq 0$ for all
    $F \in \mathcal{F} \cup \mathcal{F}' \cup
    \{\rho_{\tt{tensor}}, \zeta, \zeta' \}$, and
    $\mathcal{T}_{M}(\mathbf{p}^{*}) \in 
    \text{Sym}_{4}^{\tt{pd}}(\mathbb{R}_{> 0})$.
    Moreover, we also see that
    $\Omega(\mathbf{p}^{*}) \neq 0$.
    So, $N = \mathcal{T}_{M}(\mathbf{p}^{*}) \in \mathfrak{R}(M,
    \mathcal{F} \cup \mathcal{F}' \cup
    \{\rho_{\tt{tensor}}, \zeta, \zeta' \}) \cap
    \text{Sym}_{4}^{\tt{pd}}(\mathbb{R}_{> 0})$
    such that $\Omega(N) \neq 0$.
    Since $\zeta(N) \neq 0$, if $N$ is not of Form (III), then
    it must either be diagonal distinct, or isomorphic to 
    a matrix of Form (I) or Form (II), or satisfies
    $N_{11} = N_{22} \neq N_{33} = N_{44}$, but
    is not of Form (III).
    In any case, since $\rho_{\tt{tensor}}(N) \neq 0$,
    we see that $\PlEVAL(N) \leq \PlEVAL(M)$ is
    $\#$P-hard, due to either
    \cref{theorem: diagDistHardness},
    or \cref{lemma: formIIReduction} and
    \cref{lemma: formIHardness}, or
    \cref{lemma: order22Forms}.
    On the other hand, if $N$ is of Form (III), it is the
    required matrix.
\end{proof}

We will now prove  \cref{claim: formIIIRationalOnly}.
\begin{claimproof}{\cref{claim: formIIIRationalOnly}}
The zero sets defined by the LHS and the RHS are respectively the
hyperplane $\sum_{i \in [n]}a_{i}x_{i}=0$, and the union of the hyperplanes
      $\sum_{i \in [n]}b_{i}x_{i} =0$ and $\sum_{i \in [n]}c_{i}x_{i} =0$.
      Thus they must all coincide. Their (non-zero) normal vectors must be proportional. Therefore
        there exist  constants $\kappa, \kappa' \in \mathbb{Q}_{\ne 0}$ such that
        $(c_{1}, \dots, c_{n}) = \kappa \cdot (a_{1}, \dots, a_{n})$
        and $(b_{1}, \dots, b_{n})  = \kappa' \cdot (a_{1}, \dots, a_{n})$.
        Then the given equality in the lemma statement implies that $\kappa' = \kappa^{-1}$.
\end{claimproof}

\begin{lemma}\label{lemma: formIIIReduction}
    Let $M \in \text{Sym}_{4}^{\tt{pd}}(\mathbb{R}_{> 0})$
    be of Form (III) such that $\rho_{\tt{tensor}}(M) \neq 0$.
    Let $\mathcal{F}$ be a countable set of 
    $\text{Sym}_{4}(\mathbb{R})$-polynomials such that
    $F(M) \neq 0$ for all $F \in \mathcal{F}$.
    Let $\mathbf{0} \neq \mathbf{x} \in \chi_{4}$ of 
    support size greater than $2$.
    Then, either $\PlEVAL(M)$ is $\#$P-hard, or
    there exists some $N \in \mathfrak{R}
    (M, \mathcal{F} \cup \{\Psi_{\mathbf{x}},
    \rho_{\tt{tensor}}\}) \cap
    \text{Sym}_{4}^{\tt{pd}}(\mathbb{R}_{> 0})$
    of Form (III).
\end{lemma}
\begin{proof}
    We can replace $M$ with the matrix $M'$ that
    is obtained from \cref{lemma: formIIILogRatio}.
    We may also assume without loss of generality that
    $M_{11} > M_{33}$, and that $M_{13} \geq M_{14}$
    by permuting the rows and columns of $M$
    (first by a possible switch $\{1,2\} \leftrightarrow \{3,4\}$ and then a possible flip $3 \leftrightarrow 4$).
    We will now let the entries of $M$
    be generated by some $\{g_{t}\}_{t \in [d]}$.
    Using \cref{lemma: MequivalentCM}, we may
    replace $M$ with some $c M$ such that
    $e_{ijt} \geq 0$ for all $i, j \in [4]$, $t \in [d]$.
    Since $M \in \text{Sym}_{4}(\mathbb{R}_{> 0})$,
    we note that $e_{ij0} = 0$ for all $i, j \in [4]$.
    
    We will now define the function $\widehat{\mathcal{T}_{M}}:
    \mathbb{R}^{d + 1} \rightarrow  \text{Sym}_{4}(\mathbb{R})$,
    just as in \cref{equation: wideHatM},
    such that
    $$\widehat{\mathcal{T}_{M}}(p, z_{1}, \dots, z_{d})_{ij}
    = \mathcal{T}_{M}(p^{z_{1}}, \dots, p^{z_{d}}).$$
    Following the same argument as in the paragraph after
    \cref{equation: wideHatM},
    we can pick some rational $z_{1}^{*}, \dots, z_{d}^{*}$ such that
    $\widehat{\mathcal{T}_{M}}(e, z_{1}^{*}, \dots, 
    z_{d}^{*})$ satisfies the properties of $M$,
    that
    $$M_{11}M_{33}M_{12}M_{34} \neq (M_{13}M_{14})^{2},~~~~
    \left(\log\left(\frac{M_{13}}{M_{14}}\right)\right)^{2}
    \neq \left(\log\left(\frac{M_{11}}{M_{12}}\right)\right) \cdot
    \left(\log\left(\frac{M_{33}}{M_{34}}\right)\right),$$
    $M_{11} > M_{33}$, $M_{13} \geq M_{14}$, and $M_{ii} > M_{jk}$
    for all $i \in [4]$, $j \neq k \in [4]$, by permuting the rows and columns of $M$
    (first by a possible switch $\{1,2\} \leftrightarrow \{3,4\}$ and then a possible flip $3 \leftrightarrow 4$).
    We then let $Z^{*} \in \mathbb{Z}_{> 0}$ such that
    $Z^{*}z_{t}^{*} \in \mathbb{Z}_{> 0}$ for all
    $t \in [d]$.
    We can then define $\mathcal{T}_{M}^{*}: \mathbb{R} \rightarrow
    \text{Sym}_{4}(\mathbb{R})$, as in
    \cref{equation: MStar} such that
    $$\mathcal{T}_{M}^{*}(p)_{ij} =
    \widehat{\mathcal{T}_{M}}(p, Z^{*} \cdot z_{i}^{*}, \dots,
    Z^{*} \cdot z_{d}^{*})_{ij} = p^{z_{ij}},$$
    for some $z_{ij} \in \mathbb{Z}_{\geq 0}$ for all $i, j \in [4]$.
    We may assume from this construction that $z_{ii} > z_{jk}$ for all 
    $i \in [4], j \neq k \in [4]$, $z_{11} > z_{33}$, $z_{13} \geq z_{14}$,
    $z_{11} + z_{33} + z_{12} + z_{34} \neq 2(z_{13} + z_{14})$,
    and that $(z_{13} - z_{14})^{2} \neq
    (z_{11} - z_{12})(z_{33} - z_{34})$.

    We see that
    $$\mathcal{T}_{M}^{*}(p) = \begin{pmatrix}
        p^{z_{11}} & p^{z_{12}} & p^{z_{13}} & p^{z_{14}}\\
        p^{z_{12}} & p^{z_{11}} & p^{z_{14}} & p^{z_{13}}\\
        p^{z_{13}} & p^{z_{14}} & p^{z_{33}} & p^{z_{34}}\\
        p^{z_{14}} & p^{z_{13}} & p^{z_{34}} & p^{z_{33}}\\
    \end{pmatrix}$$
    It can be verified that the eigenvalues of
    $\mathcal{T}_{M}^{*}(p)$ are:
    \begin{equation}\label{eqn: formIIIEigenvalues}
    \begin{aligned}
        \lambda_{1}(\mathcal{T}_{M}^{*}(p))
            &= \frac{1}{2}\left(\mu_{1}(p) + \nu_{1}(p)\right),\\
        \lambda_{2}(\mathcal{T}_{M}^{*}(p))
            &= \frac{1}{2}\left(\mu_{1}(p) - \nu_{1}(p)\right),\\
        \lambda_{3}(\mathcal{T}_{M}^{*}(p))
            &= \frac{1}{2}\left(\mu_{2}(p) + \nu_{2}(p)\right),\\
        \lambda_{4}(\mathcal{T}_{M}^{*}(p))
            &= \frac{1}{2}\left(\mu_{2}(p) - \nu_{2}(p)\right),
    \end{aligned}
    \end{equation}
    where
    \begin{align*}
        \mu_{1}(p) &= p^{z_{11}} + p^{z_{33}} + p^{z_{12}} + p^{z_{34}},\\
        \mu_{2}(p) &= p^{z_{11}} + p^{z_{33}} - p^{z_{12}} - p^{z_{34}},\\
        \nu_{1}(p) &= \sqrt{(p^{z_{11}} - p^{z_{33}}
        + p^{z_{12}} - p^{z_{34}})^{2}
        + 4(p^{z_{13}} + p^{z_{14}})^{2}}, \text{ and }\\
        \nu_{2}(p) &= \sqrt{(p^{z_{11}} - p^{z_{33}}
        - p^{z_{12}} + p^{z_{34}})^{2}
        + 4(p^{z_{13}} - p^{z_{14}})^{2}}.
    \end{align*}

    We will now make use of the following claim, which
    we shall prove shortly.
    \begin{claim}\label{claim: formIIIeigenLimits}
        $$\lim\limits_{p \rightarrow \infty}
        \frac{\lambda_{1}(\mathcal{T}_{M}^{*}(p))}{p^{z_{11}}} = 
        \lim\limits_{p \rightarrow \infty}
        \frac{\lambda_{3}(\mathcal{T}_{M}^{*}(p))}{p^{z_{11}}} = 1,$$
        $$\lim\limits_{p \rightarrow \infty}
        \frac{\lambda_{2}(\mathcal{T}_{M}^{*}(p))}{p^{z_{33}}} = 
        \lim\limits_{p \rightarrow \infty}
        \frac{\lambda_{4}(\mathcal{T}_{M}^{*}(p))}{p^{z_{33}}} = 1.$$
    \end{claim}

    We can see that for large enough values of $p$,
    $\lambda_{i}(\mathcal{T}_{M}^{*}(p)) > 0$
    for all $i \in [4]$.
    Define the function $\varphi_{\mathbf{x}}:
    (\mathbb{R}_{\neq 0})^{4} \rightarrow \mathbb{R}$ as
    $\varphi_{\mathbf{x}}(\alpha_{1}, \dots, \alpha_{4})
    = \prod_{i \in [4]}(\alpha_{i})^{x_{i}}$.
    For large enough values of $p$,
    \cref{claim: formIIIeigenLimits} implies that
    $\varphi_{\mathbf{x}}(\lambda_{1}(\mathcal{T}_{M}^{*}(p)),
    \dots, \lambda_{4}(\mathcal{T}_{M}^{*}(p)))$ is
    well-defined. Moreover, we note that
    $$\phi_{\mathbf{x}}(\lambda_{1}(\mathcal{T}_{M}^{*}(p)),
    \dots, \lambda_{4}(\mathcal{T}_{M}^{*}(p))) = 0 \iff
    \varphi_{\mathbf{x}}(\lambda_{1}(\mathcal{T}_{M}^{*}(p)),
    \dots, \lambda_{4}(\mathcal{T}_{M}^{*}(p))) = 1.$$
    We can analyze
    $\phi_{\mathbf{x}}(\lambda_{1}(\mathcal{T}_{M}^{*}(p)),
    \dots, \lambda_{4}(\mathcal{T}_{M}^{*}(p)))$ by studying
    $\varphi_{\mathbf{x}}(\lambda_{1}(\mathcal{T}_{M}^{*}(p)),
    \dots, \lambda_{4}(\mathcal{T}_{M}^{*}(p)))$.
    We note that
    \begin{multline*}
        \varphi_{\mathbf{x}}(\lambda_{1}(\mathcal{T}_{M}^{*}(p)),
        \dots, \lambda_{4}(\mathcal{T}_{M}^{*}(p))) = 
        (p^{z_{11}})^{x_{1} + x_{3}}(p^{z_{33}})^{x_{2} + x_{4}}
        \left(\frac{\lambda_{1}(\mathcal{T}_{M}^{*}(p))}
        {p^{z_{11}}}\right)^{x_{1}} \cdot
        \left(\frac{\lambda_{2}(\mathcal{T}_{M}^{*}(p))}
        {p^{z_{33}}}\right)^{x_{2}}\\ \cdot
        \left(\frac{\lambda_{3}(\mathcal{T}_{M}^{*}(p))}
        {p^{z_{11}}}\right)^{x_{3}} \cdot
        \left(\frac{\lambda_{4}(\mathcal{T}_{M}^{*}(p))}
        {p^{z_{33}}}\right)^{x_{4}}.
    \end{multline*}

    So,
    $$\lim\limits_{p \rightarrow \infty}
    \varphi_{\mathbf{x}}(\lambda_{1}(\mathcal{T}_{M}^{*}(p)),
    \dots, \lambda_{4}(\mathcal{T}_{M}^{*}(p))) = 
    \lim\limits_{p \rightarrow \infty}p^{(z_{11} - z_{33})(x_{1} + x_{3})}.$$
    Since $z_{11} - z_{33} > 0$ by our assumption, we see that
    $\lim\limits_{p \rightarrow \infty}
    \varphi_{\mathbf{x}}(\lambda_{1}(\mathcal{T}_{M}^{*}(p)),
    \dots, \lambda_{4}(\mathcal{T}_{M}^{*}(p))) \neq 1$ if
    $x_{1} + x_{3} \neq 0$.
    In other words, for large enough $p$,
    $\phi_{\mathbf{x}}(\lambda_{1}(\mathcal{T}_{M}^{*}(p)),
    \dots, \lambda_{4}(\mathcal{T}_{M}^{*}(p))) \neq 0$,
    if $x_{1} + x_{3} \neq 0$.
    So, when $x_{1} + x_{3} \neq 0$, we see from the construction
    of $\mathcal{T}_{M}^{*}$ that there exists some
    $\mathbf{p}^{*} = ((p^{*})^{Z^{*} \cdot z_{1}^{*}}, \dots, 
    (p^{*})^{Z^{*} \cdot z_{d}^{*}})$ such that
    $\phi_{\mathbf{x}}(\mathcal{T}_{M}(\mathbf{p}^{*})) \neq 0$.
    
    We will now assume that $x_{1} + x_{3} = 0$.
    This of course also implies that $x_{2} + x_{4} = 0$, since $\mathbf{x} \in \chi_{4}$.
    If $x_{1} = 0$, or $x_{2} = 0$, then the support
    size of $\mathbf{x}$ would be $\leq 2$.
    We can assume that $x_{1} > 0$, and that
    $x_{2} \neq 0$.
    So, we see that
    $\phi_{\mathbf{x}}(\lambda_{1}(\mathcal{T}_{M}^{*}(p)),
    \dots, \lambda_{4}(\mathcal{T}_{M}^{*}(p))) = 0$
    implies that
    $$\left(\frac{\lambda_{1}(\mathcal{T}_{M}^{*}(p))}
    {\lambda_{3}(\mathcal{T}_{M}^{*}(p))}\right)^{x_{1}}
    = \left(\frac{\lambda_{2}(\mathcal{T}_{M}^{*}(p))}
    {\lambda_{4}(\mathcal{T}_{M}^{*}(p))}\right)^{x_{2}},$$
    for all $p \in \mathbb{R}$ where
    $\lambda_{i}(\mathcal{T}_{M}^{*}(p)) \neq 0$ for all 
    $i \in [4]$.

    We will now make use of the following claim,
    which we shall prove later.
    \begin{claim}\label{claim: formIIIeigenRatios}
        For small $\delta > 0$,
        \begin{align*}
            \lambda_{1}(\mathcal{T}_{M}^{*}(e^{\delta}))
            &= 4 + \frac{\delta}{2}\left(z_{11} + z_{33} + z_{12} + z_{34}
            + 2z_{13} + 2z_{14}\right) + O(\delta^{2}),\\
            \lambda_{2}(\mathcal{T}_{M}^{*}(e^{\delta}))
            &= \frac{\delta}{2}\left(z_{11} + z_{33} + z_{12} + z_{34}
            - 2z_{13} - 2z_{14}\right) + O(\delta^{2}),\\
            \lambda_{3}(\mathcal{T}_{M}^{*}(e^{\delta}))
            &= \frac{\delta}{2}\left((z_{11} + z_{33} - z_{12} - z_{34})
            + \sqrt{(z_{11} - z_{33} - z_{12} + z_{34})^{2} +
            4(z_{13} - z_{14})^{2}}\right) + O(\delta^{2}),\\
            \lambda_{4}(\mathcal{T}_{M}^{*}(e^{\delta}))
            &= \frac{\delta}{2}\left((z_{11} + z_{33} - z_{12} - z_{34})
            - \sqrt{(z_{11} - z_{33} - z_{12} + z_{34})^{2} +
            4(z_{13} - z_{14})^{2}}\right) + O(\delta^{2}).\\
        \end{align*}
    \end{claim}

    We can see that,
    $$\lim\limits_{\delta \rightarrow 0}
    \lambda_{1}(\mathcal{T}_{M}^{*}(e^{\delta})) = 4.$$
    Also $\lim\limits_{\delta \rightarrow 0}
    \lambda_{i}(\mathcal{T}_{M}^{*}(e^{\delta})) = 0$ for $i = 2, 3, 4$.
    From our choice of $z_{ij}$, we know that
    $z_{11}, z_{33} > z_{12}, z_{34}$. So, we know that
    $z_{11} + z_{33} - z_{12} - z_{34} > 0$.
    This implies that
    $$\lim\limits_{\delta \rightarrow 0}
    \frac{\lambda_{3}(\mathcal{T}_{M}^{*}(e^{\delta}))}{\delta}
    = (z_{11} + z_{33} - z_{12} - z_{34}) + 
    \sqrt{(z_{11} - z_{33} - z_{12} + z_{34})^{2} + 4(z_{13} - z_{14})^{2}}
    > 0.$$
    So,
    \begin{equation}\label{limit-la1overla2}
    0 \neq \lim\limits_{\delta \rightarrow 0} \left(\delta^{x_{1}} \cdot 
    \left(\frac{\lambda_{1}(\mathcal{T}_{M}^{*}(e^{\delta}))}
    {\lambda_{3}(\mathcal{T}_{M}^{*}(e^{\delta}))}\right)^{x_{1}}\right)
    = \lim\limits_{\delta \rightarrow 0} \left(\delta^{x_{1}} \cdot 
    \left(\frac{\lambda_{2}(\mathcal{T}_{M}^{*}(e^{\delta}))}
    {\lambda_{4}(\mathcal{T}_{M}^{*}(e^{\delta}))}\right)^{x_{2}}\right).
    \end{equation}

    So, if $x_{2} > 0$, then in the expression in \cref{claim: formIIIeigenRatios},
    if the coefficient of $\delta$ in
    $\lambda_{4}(\mathcal{T}_{M}^{*}(e^{\delta}))$ is non-zero, then
    the ratio $\lambda_{2}/\lambda_{4}$ in \cref{limit-la1overla2} will stay bounded
    as $\delta \rightarrow 0$, leading to a contradiction to 
     \cref{limit-la1overla2} where the limit is non-zero.
    Hence the coefficient of $\delta$ in
    $\lambda_{4}(\mathcal{T}_{M}^{*}(e^{\delta}))$ must be $0$.
    Similarly, if $x_{2} < 0$, the coefficient of $\delta$ in
    $\lambda_{2}(\mathcal{T}_{M}^{*}(e^{\delta}))$ must be $0$.
    Let us first assume that $x_{2} > 0$.
    In that case, we find that
    $$z_{11} + z_{33} - z_{12} - z_{34}
    = \sqrt{(z_{11} - z_{33} - z_{12} + z_{34})^{2} +
    4(z_{13} - z_{14})^{2}}.$$
    On squaring both sides and simplifying, we find that
    \begin{multline*}
        (z_{11} - z_{12})^{2} + (z_{33} - z_{34})^{2} + 
        2(z_{11} - z_{12})(z_{33} - z_{34}) = 
        (z_{11} - z_{12})^{2} + (z_{33} - z_{34})^{2} - 
        2(z_{11} - z_{12})(z_{33} - z_{34})\\ + 
        4(z_{13} - z_{14})^{2}.
    \end{multline*}
    This implies that $(z_{11} - z_{12})(z_{33} - z_{34})
    = (z_{13} - z_{14})^{2}$.
    But by our choice of $z_{ij}$, we ensured
    that $(z_{11} - z_{12})(z_{33} - z_{34})
    \neq (z_{13} - z_{14})^{2}$.
    So, we see that it is not possible that $\phi_{\mathbf{x}}
    (\lambda_{1}(\mathcal{T}_{M}^{*}(e^{\delta})), \dots,
    \lambda_{4}(\mathcal{T}_{M}^{*}(e^{\delta}))) = 0$
    for all $\delta \in \mathbb{R}$, if $x_{2} > 0$.

    On the other hand, if $x_{2} < 0$, then from
    \cref{claim: formIIIeigenRatios},
    we see that the coefficient of $\delta$ in
    $\lambda_{2}(\mathcal{T}_{M}^{*}(e^{\delta}))$ must be $0$.
    This implies that $z_{11} + z_{33} + z_{12} + z_{34}
    = 2z_{13} + 2z_{14}$.
    But once again, from our choice of $z_{ij}$, we ensured
    that $z_{11} + z_{33} + z_{12} + z_{34}
    \neq 2z_{13} + 2z_{14}$.
    So, we see that it is not possible that $\phi_{\mathbf{x}}
    (\lambda_{1}(\mathcal{T}_{M}^{*}(e^{\delta})),  \dots,
    \lambda_{4}(\mathcal{T}_{M}^{*}(e^{\delta}))) = 0$
    for all $\delta \in \mathbb{R}$, if $x_{2} < 0$ either.

    So, we see that for all $\mathbf{x} \in \chi_{4}$ of support
    size greater than $2$, $\phi_{\mathbf{x}}
    (\lambda_{1}(\mathcal{T}_{M}^{*}(e^{\delta})),  \dots,
    \lambda_{4}(\mathcal{T}_{M}^{*}(e^{\delta}))) \neq 0$
    for some $\delta$ that is small enough.
    Moreover, this is true for all
    $\mathbf{y} \in \chi_{q}$ such that
    $y_{\sigma(i)} = x_{i}$, for some $\sigma \in S_{4}$.
    This means that there exists some
    $p^{*} \in \mathbb{R}$ such that
    $$\Psi_{\mathbf{x}}(\mathcal{T}_{M}^{*}(p^{*})) =
    \Phi_{\mathbf{x}}
    (\lambda_{1}(\mathcal{T}_{M}^{*}(p^{*})),  \dots,
    \lambda_{4}(\mathcal{T}_{M}^{*}(p^{*}))) \neq 0.$$
    So, \cref{lemma: thickeningWorks} allows us
    to find the required $N
    \in \mathfrak{R}(M, \mathcal{F} \cup \{\Psi_{\mathbf{x}},
    \rho_{\tt{tensor}}\}) \cap
    \text{Sym}_{4}^{\tt{pd}}(\mathbb{R}_{> 0})$.
\end{proof}

We will now prove \cref{{claim: formIIIeigenLimits}}, and
\cref{claim: formIIIeigenRatios}

\begin{claimproof}{\cref{{claim: formIIIeigenLimits}}}
    We note that
    $$\lim\limits_{p \rightarrow \infty}\frac{\mu_{1}(p)}{p^{z_{11}}}
    = \lim\limits_{p \rightarrow \infty}
    \frac{p^{z_{11}} + p^{z_{33}} + p^{z_{12}} + p^{z_{34}}}
    {p^{z_{11}}} = 1.$$
    Similarly,
    $$
    \lim\limits_{p \rightarrow \infty}\frac{\mu_{2}(p)}{p^{z_{11}}} 
    = \lim\limits_{p \rightarrow \infty}
    \frac{p^{z_{11}} + p^{z_{33}} - p^{z_{12}} - p^{z_{34}}}
    {p^{z_{11}}}= 1.$$

    We can also see that
    $$\lim\limits_{p \rightarrow \infty}\frac{\nu_{1}(p)}{p^{z_{11}}} =
    \lim\limits_{p \rightarrow \infty}\frac{
    \sqrt{(p^{z_{11}} - p^{z_{33}} + p^{z_{12}} - p^{z_{34}})^{2}
    + 4(p^{z_{13}} + p^{z_{14}})^{2}}}{p^{z_{11}}} = 1, \text{ and }$$
    $$\lim\limits_{p \rightarrow \infty}\frac{\nu_{2}(p)}{p^{z_{11}}} = 
    \sqrt{(p^{z_{11}} - p^{z_{33}} - p^{z_{12}} + p^{z_{34}})^{2}
    + 4(p^{z_{13}} - p^{z_{14}})^{2}} = 1.$$
    This immediately proves that
    $$\lim\limits_{p \rightarrow \infty}
    \frac{\lambda_{1}(\mathcal{T}_{M}^{*}(p))}{p^{z_{11}}}
    = \lim\limits_{p \rightarrow \infty}
    \frac{\lambda_{3}(\mathcal{T}_{M}^{*}(p))}{p^{z_{11}}} = 1.$$

    We now note that
    $$\lambda_{2}(\mathcal{T}_{M}^{*}(p)) = \frac{1}{2}
    \left(\frac{(\mu_{1}(p))^{2} - (\nu_{1}(p))^{2}}
    {\mu_{1}(p) + \nu_{1}(p)} \right), ~~\text{ and }~~
    \lambda_{4}(\mathcal{T}_{M}^{*}(p)) = \frac{1}{2}
    \left(\frac{(\mu_{2}(p))^{2} - (\nu_{2}(p))^{2}}
    {\mu_{2}(p) + \nu_{2}(p)} \right).$$
    Now, we can see that
    $$(\mu_{1}(p))^{2} - (\nu_{1}(p))^{2} =
    4(p^{z_{11}} + p^{z_{12}})(p^{z_{33}} + p^{z_{34}})
    - 4(p^{z_{13}} + p^{z_{14}})^{2}, \text{ and }$$
    $$(\mu_{2}(p))^{2} - (\nu_{2}(p))^{2} =
    4(p^{z_{11}} - p^{z_{12}})(p^{z_{33}} - p^{z_{34}})
    - 4(p^{z_{13}} - p^{z_{14}})^{2}.$$
    Since we already know the limiting behaviors of
    $\mu_{1}(p) + \nu(p)$, and $\mu_{2}(p) + \nu_{2}(p)$,
    we can see that
    $$\lim\limits_{p \rightarrow \infty}
    \frac{\lambda_{2}(\mathcal{T}_{M}^{*}(p))}{p^{z_{33}}}
    = \lim\limits_{p \rightarrow \infty}
    \frac{\lambda_{4}(\mathcal{T}_{M}^{*}(p))}{p^{z_{33}}} = 1.$$
\end{claimproof}

\begin{claimproof}{\cref{claim: formIIIeigenRatios}}
    Following the Taylor series expansions of
    $\mu_{1}(e^{\delta}), \mu_{2}(e^{\delta}), (\nu_{1}(e^{\delta}))^{2}$,
    and $(\nu_{2}(e^{\delta}))^{2}$, we see that
    \begin{align*}
        \mu_{1}(e^{\delta}) &= 4 + 
        \delta (z_{11} + z_{33} + z_{12} + z_{34}) + O(\delta^{2}),\\
        \mu_{2}(e^{\delta}) &=
        \delta (z_{11} + z_{33} - z_{12} - z_{34}) + O(\delta^{2}),\\
        (\nu_{1}(e^{\delta}))^{2} &= 16 + 16\delta (z_{13} + z_{14})
        + O(\delta^{2}), \text{ and }\\
        (\nu_{2}(e^{\delta}))^{2} &= \delta^{2}
        ((z_{11} - z_{33} - z_{12} + z_{34})^{2}
        + 4(z_{13} - z_{14})^{2}) + O(\delta^{3}).\\
    \end{align*}
    
    Then, we can use the square root expansion of $\sqrt{1 + x}$ to see that
    \begin{align*}
        \nu_{1}(e^{\delta}) &= 
        4 + 2\delta (z_{13} + z_{14}) + O(\delta^{2}), \text{ and }\\
        \nu_{2}(e^{\delta}) &= 
         \delta \sqrt{(z_{11} - z_{33} - z_{12} + z_{34})^{2}
        + 4(z_{13} - z_{14})^{2}} + O(\delta^{2}).
    \end{align*}
    Putting together these expressions 
    finishes the proof of this claim.
\end{claimproof}

We can now prove that $\PlEVAL(M)$, when $M$ is of Form (III) that is not 
isomorphic to a tensor product,
must also be $\#$P-hard.

\begin{lemma}\label{lemma: formIIIHardness}
    Let $M \in \text{Sym}_{4}^{\tt{pd}}(\mathbb{R}_{> 0})$
    be of Form (III) such that $\rho_{\tt{tensor}}(M) \neq 0$.
    Then, $\PlEVAL(M)$ is $\#$P-hard.
\end{lemma}
\begin{proof}
    The proof is very similar to the proof of
    \cref{lemma: formIVHardness}.
    The only difference is that instead of
    \cref{lemma: formIVReduction}, we will have
    to use \cref{lemma: formIIIReduction} to
    find the matrices $N_{i}$.
    The rest of the proof is identical.
\end{proof}
\subsection{Form (VI)}\label{sec: form_VI}

Finally, there is only one more form of matrix for us to deal with.
We will now prove that $\PlEVAL(M)$ is $\#$P-hard, when
$M$ is of Form (VI)  that is not 
isomorphic to a tensor product.

\begin{figure}[ht]
\[
\begin{pmatrix}
            a & x & y & z\\
            x & a & z & y\\
            y & z & a & x\\
            z & y & x & a
        \end{pmatrix}
        \]
        \caption{\label{Form-VI-in-abc} Form (VI)}
\end{figure}

\begin{lemma}\label{lemma: formVIReduction}
    Let $M \in \text{Sym}_{4}^{\tt{pd}}(\mathbb{R}_{> 0})$
    be of Form (VI) such that $\rho_{\tt{tensor}}(M) \neq 0$.
    Let $\mathcal{F}$ be a countable set of 
    $\text{Sym}_{4}(\mathbb{R})$-polynomials such that
    $F(M) \neq 0$ for all $F \in \mathcal{F}$.
    Let $\mathbf{0} \neq \mathbf{x} \in \chi_{4}$ of 
    support size greater than $2$.
    Then, either $\PlEVAL(M)$ is $\#$P-hard, or
    there exists some $N \in \mathfrak{R}
    (M, \mathcal{F} \cup \{\Psi_{\mathbf{x}},
    \rho_{\tt{tensor}}\}) \cap
    \text{Sym}_{4}^{\tt{pd}}(\mathbb{R}_{> 0})$
    of Form (VI).
\end{lemma}
\begin{proof}
    We can use
    \cref{lemma: DiagLarge} to obtain some
    $M' \in \mathfrak{R}
    (M, \mathcal{F} \cup \{\rho_{\tt{tensor}}\}) \cap
    \text{Sym}_{4}^{\tt{pd}}(\mathbb{R}_{> 0})$,
    such that
    $M'_{ii} > M'_{jk}$ for all $i \in [4]$,
    $j \neq k \in [4]$.
    If $M'$ is diagonal distinct, then \cref{theorem: diagDistHardness}
    implies that $\PlEVAL(M') \leq \PlEVAL(M)$ is
    $\#$P-hard.
    Let us assume that $M'$ is not diagonal distinct.
    If $M'$ is not of Form (VI), then \cref{theorem: nonDiagAllForms}
    implies that either $\PlEVAL(M') \leq \PlEVAL(M)$ is
    $\#$P-hard, or
    there exists some $N \in \mathfrak{R}(M', \mathcal{F}
    \cup \{\rho_{\tt{tensor}} \}) \cap
    \text{Sym}_{4}^{\tt{pd}}(\mathbb{R}_{> 0})$ such that
    $N$ is of Form (I), (III), or (IV), and
    the diagonal entries of $N$ are not all identical.
    But then, we note that since $\rho_{\tt{tensor}}(N) \neq 0$,
    \cref{lemma: formIHardness},
    \cref{lemma: formIVHardness}, and \cref{lemma: formIIIHardness}
    imply that $\PlEVAL(N) \leq \PlEVAL(M)$ is $\#$P-hard.
    So, we may now assume that $M'$ is of Form (VI).
    For convenience let us rename the matrix
    $M'$ as $M$.
    If $M_{11}M_{14} = M_{12}M_{13}$,
     let $\kappa = \nicefrac{M_{13}}{M_{11}} = \nicefrac{M_{14}}{M_{12}}$, then
    $$M = \begin{pmatrix}
            M_{11} & M_{12}\\
            M_{12} & M_{11}
        \end{pmatrix}
        \otimes
        \begin{pmatrix}
            1 & \kappa\\
            \kappa & 1
        \end{pmatrix}
        = \begin{pmatrix}
            a & x\\
            x & a
        \end{pmatrix}
        \otimes
        \begin{pmatrix}
            1 & \kappa\\
            \kappa & 1
        \end{pmatrix}. $$
        Similarly, if  $M_{11}M_{12} = M_{13}M_{14}$, or
    $M_{11}M_{13} = M_{12}M_{14}$, then $M$ is isomorphic to a tensor product
    under the flips
        $2 \leftrightarrow 4$, or $3 \leftrightarrow 4$, respectively.
    So, we have $M_{11}M_{12} \not = M_{13}M_{14}$,  $M_{11}M_{13} \not = M_{12}M_{14}$, and
    $M_{11}M_{14}\not  = M_{12}M_{13}$, since  $\rho_{\tt{tensor}}(M) \neq 0$.

    We will now let the entries of $M$
    be generated by some $\{g_{t}\}_{t \in [d]}$.
    Using \cref{lemma: MequivalentCM}, we may
    replace $M$ with some $c M$ such that
    $e_{ijt} \geq 0$ for all $i, j \in [4]$, $t \in [d]$.
    Since $M \in \text{Sym}_{4}(\mathbb{R}_{> 0})$,
    we note that $e_{ij0} = 0$ for all $i, j \in [4]$.
    
    We will now define the function $\widehat{\mathcal{T}_{M}}:
    \mathbb{R}^{d + 1} \rightarrow \text{Sym}_{q}^{\mathbb{R}}$,
    just as in \cref{equation: wideHatM},
    such that
    $$\widehat{\mathcal{T}_{M}}(p, z_{1}, \dots, z_{d})_{ij}
    = \mathcal{T}_{M}(p^{z_{1}}, \dots, p^{z_{d}}).$$
    We can pick some rational $z_{1}^{*}, \dots, z_{d}^{*}$ such that
    $\widehat{\mathcal{T}_{M}}(p, z_{1}^{*}, \dots, 
    z_{d}^{*})$ satisfies the properties of $M$,
    that
    $$M_{11}M_{12} \neq M_{13}M_{14},~~
    M_{11}M_{13} \neq M_{12}M_{14},~~
    M_{11}M_{14} = M_{12}M_{13}, ~~\text{ and }~~
    M_{ii} > M_{jk}$$
    for all $i \in [4]$, $j \neq k \in [4]$.
    We then let $Z^{*} \in \mathbb{Z}_{> 0}$ such that
    $Z^{*}z_{t}^{*} \in \mathbb{Z}_{> 0}$ for all
    $t \in [d]$.
    We can then define $\mathcal{T}_{M}^{*}: \mathbb{R} \rightarrow
    \text{Sym}_{4}(\mathbb{R})$, as in
    \cref{equation: MStar} such that
    $$\mathcal{T}_{M}^{*}(p)_{ij} =
    \widehat{\mathcal{T}_{M}}(p, Z^{*} \cdot z_{i}^{*}, \dots,
    Z^{*} \cdot z_{d}^{*})_{ij} = p^{z_{ij}},$$
    for some $z_{ij} \in \mathbb{Z}_{\geq 0}$ for all $i, j \in [4]$.
    We may assume from this construction that $z_{ii} > z_{jk}$ for all 
    $i \in [4]$, and $j \neq k \in [4]$,
    $z_{11} + z_{12} \neq z_{13} + z_{14}$,
    $z_{11} + z_{13} \neq z_{12} + z_{14}$, and
    $z_{11} + z_{14} \neq z_{12} + z_{13}$.
    
    We see that
    $$\mathcal{T}_{M}^{*}(p) = \begin{pmatrix}
        p^{z_{11}} & p^{z_{12}} & p^{z_{13}} & p^{z_{14}}\\
        p^{z_{12}} & p^{z_{11}} & p^{z_{14}} & p^{z_{13}}\\
        p^{z_{13}} & p^{z_{14}} & p^{z_{11}} & p^{z_{12}}\\
        p^{z_{14}} & p^{z_{13}} & p^{z_{12}} & p^{z_{11}}\\
    \end{pmatrix}$$
    It can be verified that the eigenvalues of
    $\mathcal{T}_{M}^{*}(p)$ are:
    \begin{equation}\label{eqn: formVIEigenvalues}
    \begin{aligned}
        \lambda_{1}(\mathcal{T}_{M}^{*}(p))
            &= p^{z_{11}} + p^{z_{12}} + p^{z_{13}} + p^{z_{14}},\\
        \lambda_{2}(\mathcal{T}_{M}^{*}(p))
            &= p^{z_{11}} + p^{z_{12}} - p^{z_{13}} - p^{z_{14}},\\
        \lambda_{3}(\mathcal{T}_{M}^{*}(p))
            &= p^{z_{11}} - p^{z_{12}} + p^{z_{13}} - p^{z_{14}},\\
        \lambda_{4}(\mathcal{T}_{M}^{*}(p))
            &= p^{z_{11}} - p^{z_{12}} - p^{z_{13}} + p^{z_{14}}.
    \end{aligned}
    \end{equation}
    
    We will now define the function $\varphi_{\mathbf{x}}:
    (\mathbb{R}_{\neq 0})^{4} \rightarrow \mathbb{R}$ such that
    $\varphi_{\mathbf{x}}(\alpha_{1}, \dots, \alpha_{4})
    = \prod_{i \in [4]}(\alpha_{i})^{x_{i}}$.
    We see that
    $$\varphi_{\mathbf{x}}(\lambda_{1}(\mathcal{T}_{M}^{*}(p),
    \dots, \lambda_{4}(\mathcal{T}_{M}^{*}(p)) = 1 \iff
    \phi_{\mathbf{x}}(\lambda_{1}(\mathcal{T}_{M}^{*}(p),
    \dots, \lambda_{4}(\mathcal{T}_{M}^{*}(p)) = 0.$$
    We will understand the behavior of $\phi_{\mathbf{x}}$
    by studying $\varphi_{\mathbf{x}}$.
    
    From the Taylor series expansions of
    $\lambda_{i}(\mathcal{T}_{M}^{*}(p))$ for each $i \in [4]$,
    we can see that for small $\delta > 0$,
    \begin{align*}
        \lambda_{1}(\mathcal{T}_{M}^{*}(e^{\delta}))
        &= 4 + \delta
        \left(z_{11} + z_{12} + z_{13} + z_{14}\right)
        + O(\delta^{2}),\\
        \lambda_{2}(\mathcal{T}_{M}^{*}(e^{\delta}))
        &= \delta\left(z_{11} + z_{12}
        - z_{13} - z_{14}\right) + O(\delta^{2}),\\
        \lambda_{3}(\mathcal{T}_{M}^{*}(e^{\delta}))
        &= \delta\left(z_{11} - z_{12}
        + z_{13} - z_{14}\right) + O(\delta^{2}),\\
        \lambda_{4}(\mathcal{T}_{M}^{*}(e^{\delta}))
        &= \delta\left(z_{11} - z_{12}
        - z_{13} + z_{14}\right) + O(\delta^{2}).\\
    \end{align*}
    
    From our choice of $z_{ij}$, we can see that the
    coefficients of $\delta$ in each of the eigenvalues is
    non-zero.
    So, the following limits all exist and are non-zero constants,
    $$\lim\limits_{\delta \rightarrow 0}
    \frac{\lambda_{i}(\mathcal{T}_{M}^{*}(e^{\delta}))}{\delta} \neq 0,$$
    for all $i \in \{2, 3, 4\}$.
    For the given $\mathbf{x} \in \chi_{4}$,
    \begin{multline*}
        \varphi_{\mathbf{x}}\left(\lambda_{1}(\mathcal{T}_{M}^{*}(e^{\delta})),
        \dots, \lambda_{4}(\mathcal{T}_{M}^{*}(e^{\delta}))\right) = 
        \delta^{x_{2} + x_{3} + x_{4}} \cdot
        \left(\lambda_{1}(\mathcal{T}_{M}^{*}(e^{\delta}))\right)^{x_{1}} \cdot
        \left(\frac{\lambda_{2}(\mathcal{T}_{M}^{*}(e^{\delta}))}
        {\delta} \right)^{x_{2}} \cdot\\
        \left(\frac{\lambda_{3}(\mathcal{T}_{M}^{*}(e^{\delta}))}
        {\delta} \right)^{x_{3}} \cdot
        \left(\frac{\lambda_{4}(\mathcal{T}_{M}^{*}(e^{\delta}))}
        {\delta} \right)^{x_{4}}.
    \end{multline*}
    So, we see that
    $$\lim\limits_{\delta \rightarrow 0}
    \varphi_{\mathbf{x}}\left(\lambda_{1}(\mathcal{T}_{M}^{*}(e^{\delta})),
    \dots, \lambda_{4}(\mathcal{T}_{M}^{*}(e^{\delta}))\right) =
    4^{x_1}
    \lim\limits_{\delta \rightarrow 0}
    \left[
    (\delta^{x_{1} + x_{2} + x_{3}})\prod_{i = 2, 3, 4}\left(
    \frac{\lambda_{i}(\mathcal{T}_{M}^{*}(e^{\delta}))}{\delta}\right)^{x_{i}}\right ].$$
    This implies that unless $x_{2} + x_{3} + x_{4} = (-x_{1}) = 0$,
    $$\lim\limits_{\delta \rightarrow 0}
    \varphi_{\mathbf{x}}\left(\lambda_{1}(\mathcal{T}_{M}^{*}(e^{\delta})),
    \dots, \lambda_{4}(\mathcal{T}_{M}^{*}(e^{\delta}))\right)
    \neq 1,$$
    which in turn implies that
    $$\lim\limits_{\delta \rightarrow 0}
    \phi_{\mathbf{x}}(\lambda_{1}(\mathcal{T}_{M}^{*}(e^{\delta})), 
    \dots, \lambda_{4}(\mathcal{T}_{M}^{*}(e^{\delta}))) \neq 0.$$
    
    So, we may assume that $x_{1} = 0$.
    Since $\mathbf{x}$ has a support size $> 2$, and $x_2 + x_3 + x_4 =0$,
    we now may assume that none of $x_2, x_3, x_4 =0$.
    
    Suppose $z_{12} = z_{13}$, then from
    \cref{eqn: formVIEigenvalues}, we see that
    $$\lambda_{2}(\mathcal{T}_{M}^{*}(p)) 
    = p^{z_{11}} - p^{z_{14}}
    = \lambda_{3}(\mathcal{T}_{M}^{*}(p)).$$
    So,
    $\varphi_{\mathbf{x}}\left(\lambda_{1}(\mathcal{T}_{M}^{*}(p)),
    \dots, \lambda_{4}(\mathcal{T}_{M}^{*}(p))\right) = 1$
    for all $p \in \mathbb{R}$
     implies that
    $\lambda_{2}(\mathcal{T}_{M}^{*}(p))^{x_{2} + x_{3}} = 
    \lambda_{4}(\mathcal{T}_{M}^{*}(p))^{x_{2} + x_{3}}$
    for all $p \in \mathbb{R}$.
    If $x_{2} + x_{3} = 0$, then $x_4 = -(x_{2} + x_{3}) = 0$
    as well, and $\mathbf{x}$ would have
    support size $\leq 2$, so we may assume otherwise.
    Without loss generality, we may assume that
    $x_{2} + x_{3} > 0$.
    So, we see that $\lambda_{2}(\mathcal{T}_{M}^{*}(p)) = 
    \lambda_{4}(\mathcal{T}_{M}^{*}(p))$ for all
    $p \in \mathbb{R}$.
    By \cref{eqn: formVIEigenvalues}, this implies that
    $$p^{z_{11}} - p^{z_{14}} = p^{z_{11}} - 2p^{z_{12}} + p^{z_{14}}.$$
    Simplifying this, we find that $z_{12} = z_{13} = z_{14}$
    as well.

     We also know that $z_{11} \neq z_{12}$.
     So, we see that $\mathcal{T}_{M}^{*}(p)$ is of the form
     $$\mathcal{T}_{M^{*}}(p) = \begin{pmatrix}
         p^{z_{11}} & p^{z_{12}} & p^{z_{12}} & p^{z_{12}}\\
         p^{z_{12}} & p^{z_{11}} & p^{z_{12}} & p^{z_{12}}\\
         p^{z_{12}} & p^{z_{12}} & p^{z_{11}} & p^{z_{12}}\\
         p^{z_{12}} & p^{z_{12}} & p^{z_{12}} & p^{z_{11}}\\
    \end{pmatrix}.$$
    But then, it is already known from 
    \cref{corollary: tutteHardness} that when $p > 1$,
    $\PlEVAL(\mathcal{T}_{M^{*}}(p)) \leq \PlEVAL(M)$
    is $\#$P-hard.
    
    From the above proof we may now assume we 
    are in the case where $z_{12} \neq z_{13}$. By the same reasoning, 
    we may also assume that in fact,
    $z_{12} \neq z_{13} \neq z_{14}$ are pairwise distinct.
    Without loss of generality, we may assume that $x_{2} > 0$.
    We can also assume by symmetry, that
    $x_{3} \geq x_{4}$.
    So, the two cases we have to consider are:
    $x_{3} > 0$ (which would imply that $x_{4} = -(x_2 + x_3)< 0$), and
    $x_{3} < 0$ (which would also imply that $x_{4} \le x_3 < 0$).
    Let us first consider the case where $x_{3} > 0$.
    We see that if
    $\phi_{\mathbf{x}}\left(\lambda_{1}(\mathcal{T}_{M}^{*}(p)),
    \dots, \lambda_{4}(\mathcal{T}_{M}^{*}(p))\right) = 0$, then
    $$\left(p^{z_{11}} + p^{z_{12}} - p^{z_{13}} - p^{z_{14}}\right)^{x_{2}}
    \left(p^{z_{11}} - p^{z_{12}} + p^{z_{13}} - p^{z_{14}}\right)^{x_{3}}
    = \left(p^{z_{11}} - p^{z_{12}} - p^{z_{13}} + p^{z_{14}}
    \right)^{x_{2} + x_{3}}.$$
    The highest degree term of both the LHS and the RHS are both
    equal to $p^{z_{11}(x_{2} + x_{3})}$.
    The second highest degree term will depend on
    which of $z_{12}, z_{13}, z_{14}$ is the largest.
    If $z_{12} > z_{13}, z_{14}$, then the second highest
    degree term of the LHS is
    $(x_{2} - x_{3})p^{z_{11}(x_{2} + x_{3} - 1) + z_{12}}$,
    while the second highest degree term of the RHS is
    $(- x_{2} - x_{3})p^{z_{11}(x_{2} + x_{3} - 1) + z_{12}}$.
    These terms are clearly not equal to each other.
    On the other hand, if $z_{13} > z_{12}, z_{14}$, then
    the second highest degree term of the LHS is
    $(- x_{2} + x_{3})p^{z_{11}(x_{2} + x_{3} - 1) + z_{13}}$,
    while the second highest degree term of the RHS is
    $(- x_{2} - x_{3})p^{z_{11}(x_{2} + x_{3} - 1) + z_{13}}$,
    which are also not equal to each other.
    Finally, if $z_{14} > z_{12}, z_{13}$, the second highest
    degree term of the LHS is
    $(- x_{2} - x_{3})p^{z_{11}(x_{2} + x_{3} - 1) + z_{14}}$,
    while the second highest degree term of the RHS is
    $(x_{2} + x_{3})p^{z_{11}(x_{2} + x_{3} - 1) + z_{14}}$,
    which are not equal to each other either.
    This proves that it is not possible that
    $\phi_{\mathbf{x}}\left(\lambda_{1}(\mathcal{T}_{M}^{*}(p)),
    \dots, \lambda_{4}(\mathcal{T}_{M}^{*}(p))\right) = 0$.
    
    We will now consider the case where $x_{3} < 0$.
    In this case, if
    $\phi_{\mathbf{x}}\left(\lambda_{1}(\mathcal{T}_{M}^{*}(p)),
    \dots, \lambda_{4}(\mathcal{T}_{M}^{*}(p))\right) = 0$, then
    $$\left(p^{z_{11}} + p^{z_{12}} - p^{z_{13}} - p^{z_{14}}
    \right)^{-x_{3} - x_{4}}
    = \left(p^{z_{11}} - p^{z_{12}} + p^{z_{13}} - p^{z_{14}}\right)^{-x_{3}}
    \left(p^{z_{11}} - p^{z_{12}} - p^{z_{13}} + p^{z_{14}}\right)^{-x_{4}}.$$
    The highest degree term of both the LHS and the RHS are both
    equal to $p^{z_{11}(-x_{3} - x_{4})}$.
    The second highest degree term will again depend on
    which of $z_{12}, z_{13}, z_{14}$ is the largest.
    If $z_{12} > z_{13}, z_{14}$, then the second highest
    degree term of the LHS is
    $(-x_{3} - x_{4})p^{z_{11}(-x_{3} - x_{4} - 1) + z_{12}}$,
    while the second highest degree term of the RHS is
    $(x_{3} + x_{4})p^{z_{11}(-x_{3} - x_{4} - 1) + z_{12}}$.
    These terms are clearly not equal to each other.
    On the other hand, if $z_{13} > z_{12}, z_{14}$, then
    the second highest degree term of the LHS is
    $(x_{3} + x_{4})p^{z_{11}(-x_{3} - x_{4} - 1) + z_{13}}$,
    while the second highest degree term of the RHS is
    $(-x_{3} + x_{4})p^{z_{11}(-x_{3} - x_{4} - 1) + z_{13}}$,
    which are also not equal to each other.
    Finally, if $z_{14} > z_{12}, z_{13}$, the second highest
    degree term of the LHS is
    $(x_{3} + x_{4})p^{z_{11}(-x_{3} - x_{4} - 1) + z_{14}}$,
    while the second highest degree term of the RHS is
    $(x_{3} - x_{4})p^{z_{11}(-x_{3} - x_{4} - 1) + z_{14}}$,
    which are not equal to each other either.
    This proves that it is not possible that
    $\phi_{\mathbf{x}}\left(\lambda_{1}(\mathcal{T}_{M}^{*}(p)),
    \dots, \lambda_{4}(\mathcal{T}_{M}^{*}(p))\right) = 0$,
    even in this case.

    So, we see that for all $\mathbf{x} \in \chi_{4}$ of support
    size greater than $2$, $\phi_{\mathbf{x}}
    (\lambda_{1}(\mathcal{T}_{M}^{*}(p),  \dots,
    \lambda_{4}(\mathcal{T}_{M}^{*}(p)))$
    is a non-zero polynomial.
    Moreover, this is true for all
    $\mathbf{y} \in \chi_{q}$ such that
    $y_{\sigma(i)} = x_{i}$, for some $\sigma \in S_{4}$.
    In particular, this means that there exists some
    $p^{*} \in \mathbb{R}$ such that
    $$\Psi_{\mathbf{x}}(\mathcal{T}_{M}^{*}(p^{*})) =
    \Phi_{\mathbf{x}}
    (\lambda_{1}(\mathcal{T}_{M}^{*}(p^{*})),  \dots,
    \lambda_{4}(\mathcal{T}_{M}^{*}(p^{*}))) \neq 0.$$
    So, \cref{lemma: thickeningWorks} allows us
    to find some $N = \mathcal{T}_{M}(\mathbf{p}^{*})
    \in \mathfrak{R}(M, \mathcal{F} \cup \{\Psi_{\mathbf{x}},
    \rho_{\tt{tensor}}\}) \cap
    \text{Sym}_{4}^{\tt{pd}}(\mathbb{R}_{> 0})$.
    If $N$ is diagonal distinct, then \cref{theorem: diagDistHardness}
    implies that $\PlEVAL(N) \leq \PlEVAL(M)$ is
    $\#$P-hard.
    So, we may assume that $N$ is not diagonal distinct.
    
    Let us now assume that $N$ is not of Form (VI).
    Now, from \cref{theorem: nonDiagAllForms},
    we see that either $\PlEVAL(N) \leq \PlEVAL(M)$ is
    $\#$P-had, or there exists
    some $N' \in \mathfrak{R}(M, \mathcal{F}
    \cup \{\Psi_{\mathbf{x}}, \rho_{\tt{tensor}}\}) \cap
    \text{Sym}_{4}^{\tt{pd}}(\mathbb{R}_{> 0})$, that is
    isomorphic to a matrix of Form (I), (III), or (IV),
    and does not have diagonal entries that are all identical
    to each other.
    But in that case,
    we note that since $\rho_{\tt{tensor}}(N') \neq 0$,
    \cref{lemma: formIHardness},
    \cref{lemma: formIVHardness}, and \cref{lemma: formIIIHardness}
    imply that $\PlEVAL(N') \leq \PlEVAL(M)$ is $\#$P-hard.
    Otherwise, $N$ is the required matrix of Form (VI), gaining the crucial property that
    $\Psi_{\mathbf{x}}(N) \ne 0$.
\end{proof}

We can now prove that $\PlEVAL(M)$, when $M$ is of Form (VI) that is not 
isomorphic to a tensor product, 
must also be $\#$P-hard.

\begin{lemma}\label{lemma: formVIHardness}
    Let $M \in \text{Sym}_{4}^{\tt{pd}}(\mathbb{R}_{> 0})$
    be of Form (VI) such that $\rho_{\tt{tensor}}(M) \neq 0$.
    Then, $\PlEVAL(M)$ is $\#$P-hard.
\end{lemma}
\begin{proof}
    The proof is very similar to the proof of
    \cref{lemma: formIVHardness}.
    The only difference is that instead of
    \cref{lemma: formIVReduction}, we will have
    to use \cref{lemma: formVIReduction} to
    find the matrices $N_{i}$.
    The rest of the proof is identical.
\end{proof}

We can finally prove the following theorem.

\begin{theorem}\label{theorem: TensorOrHard}
    Let $M \in \text{Sym}_{4}^{\tt{pd}}(\mathbb{R}_{> 0})$
    such that $\rho_{\tt{tensor}}(M) \neq 0$.
    Then, $\PlEVAL(M)$ is $\#$P-hard.    
\end{theorem}
\begin{proof}
    If $M$ is diagonal distinct, then
    \cref{theorem: diagDistHardness} implies that
    $\PlEVAL(M)$ is $\#$P-hard.
    Otherwise, from \cref{theorem: nonDiagAllForms},
    we know that either $\PlEVAL(M)$ is $\#$P-hard, or
    there exists some $M' \in \mathfrak{R}(M, \mathcal{F}_{M}
    \cup \{\rho_{\tt{tensor}} \})
    \cap \text{Sym}_{4}^{\tt{pd}}(\mathbb{R}_{> 0})$
    that is of Form (I), (III), (IV), or (VI).
    In any case,
    \cref{lemma: formIHardness}. \cref{lemma: formIIIHardness},
    \cref{lemma: formIVHardness}, or \cref{lemma: formVIHardness}
    allow us to prove that $\PlEVAL(M') \neq \PlEVAL(M)$ is
    $\#$P-hard.
\end{proof}

\section[Dichotomy of 4 x 4 matrices]
{Dichotomy for $4 \times 4$ matrices}\label{sec: 4x4Dichotomy}

We will now first deal with matrices $M \in
\text{Sym}_{4}^{\tt{pd}}(\mathbb{R}_{> 0})$
such that $M$ is isomorphic to $A \otimes B$
for some $A, B \in \text{Sym}_{2}(\mathbb{R})$.
We can assume without loss of generality, that in fact,
$M = A \otimes B$. In this case, we see that
given any graph $G = (V, E)$,
$$Z_{A \otimes B}(G) =
\sum_{\sigma: V \rightarrow [4]}
\prod_{(u, v) \in E}(A \otimes B)_{\sigma(u)\sigma(v)} =
\sum_{\sigma_{1}, \sigma_{2}: V \rightarrow [2]}
\prod_{(u, v) \in E}
\left((A \otimes B)_{\tau\left(
{\scriptscriptstyle{\sigma_{1}(u), \sigma_{2}(u)}}\right)
\tau\left(
{\scriptscriptstyle{\sigma_{1}(v), \sigma_{2}(v)}}\right)}\right),$$
where
$\tau(i, j) = 2(i-1) + j$ for all $i, j \in [2]$.

But then, we see that
\begin{align*}
    Z_{A \otimes B}(G)
    &= \sum_{\sigma_{1}, \sigma_{2}: V \rightarrow [2]}
    \prod_{(u, v) \in E}A_{\sigma_{1}(u)\sigma_{1}(v)}
    B_{\sigma_{2}(u)\sigma_{2}(v)}\\
    &= \left(\sum_{\sigma_{1}: V \rightarrow [2]}
    \prod_{(u, v) \in E}A_{\sigma_{1}(u)\sigma_{1}(v)}\right) \cdot
    \left(\sum_{\sigma_{2}: V \rightarrow [2]}
    \prod_{(u, v) \in E}B_{\sigma_{2}(u)\sigma_{2}(v)}\right)\\
    &= \left(Z_{A}(G) \right) \cdot \left(Z_{B}(G) \right).
\end{align*}

So, we see that when $M = A \otimes B$, it is in some sense,
equivalent to problems of a smaller size.
We shall formalize this notion shortly.
In the meantime, we can now immediately prove the following lemma.

\begin{lemma}\label{lemma: tensorSimple}
    Let $M \in \text{Sym}_{4}^{\tt{pd}}(\mathbb{R}_{> 0})$,
    such that $M = A \otimes B$ for some $A, B \in
    \text{Sym}_{2}(\mathbb{R})$.
    If $\PlEVAL(A)$ is polynomial time tractable, then
    $\PlEVAL(M) \equiv \PlEVAL(B)$.
\end{lemma}
\begin{proof}
    Let $M = A \otimes B$.
    If $A$ or $B$ has any zero entries entries, then
    $M$ would have zero entries, so we see that $A$ and $B$
    cannot have any zero entries.
    Now, if we assume that $A$ and $B$ have entries $a$ and $b$ of
    the opposite signs, then $ab < 0$ would be an entry of $M$.
    So, all the entries of $A$ and $B$ must have the same sign.
    If this sign is $-$, then we may
    replace $A$ and $B$ with $-A$ and $-B$.
    So, we may assume that $A, B \in \text{Sym}_{2}(\mathbb{R}_{> 0})$.
    This implies that for any graph $G = (V, E)$,
    $Z_{A}(G) > 0$, and $Z_{B}(G) > 0$.
    If we now had oracle access to $\PlEVAL(B)$, we can use it
    to compute $Z_{M}(G) = Z_{A}(G) \cdot Z_{B}(G)$ for all
    planar $G = (V, E)$.
    Similarly, if we had oracle access to $\PlEVAL(M)$, we could use
    it to compute $Z_{B}(G) = \frac{Z_{M}(G)}{Z_{A}(G)}$
    for any planar $G = (V, E)$.
\end{proof}

\begin{remark*}
    Let $M \in \text{Sym}_{4}^{\tt{pd}}(\mathbb{R}_{> 0})$,
    such that $M = A \otimes B$ for some $A, B \in
    \text{Sym}_{2}(\mathbb{R})$.
    Then, as we saw in the proof of \cref{lemma: tensorSimple},
    we may assume that in fact, $A, B \in \text{Sym}_{2}(\mathbb{R}_{> 0})$.
    Now, we may let $(\lambda_{1}, \lambda_{2})$ be the eigenvalues
    of $A$, and $(\mu_{1}, \mu_{2})$ be the eigenvalues of $B$.
    This implies that the eigenvalues of $M = A \otimes B$ are:
    $$\lambda_{1}\mu_{1}, ~~~~ \lambda_{1}\mu_{2}, ~~~~
    \lambda_{2}\mu_{1}, ~~~~ \lambda_{2}\mu_{2}.$$
    We may assume without loss of generality that
    $\lambda_{1} \geq \lambda_{2}$, and $\mu_{1} \geq \mu_{2}$.
    Since $A \in \text{Sym}_{2}(\mathbb{R}_{> 0})$, we know
    from the Perron-Frobenius Theorem, that in fact,
    $\lambda_{1} > |\lambda_{2}|$.
    Similarly, since $B \in \text{Sym}_{2}(\mathbb{R}_{> 0})$,
    we see that $\mu_{1} > |\mu_{2}|$.
    Now, if $\lambda_{2} \leq 0$, we note that
    $\lambda_{2}\mu_{1} \leq 0$ would be an eigenvalue of
    $A \otimes B = M$, which contradicts our assumption that
    $M \in \text{Sym}_{4}^{\tt{pd}}(\mathbb{R}_{> 0})$.
    So, in fact, $\lambda_{2} > 0$.
    Similarly, $\mu_{2} > 0$ as well.
    This means that if
    $M \in \text{Sym}_{4}^{\tt{pd}}(\mathbb{R}_{> 0})$ is such that
    $M = A \otimes B$ for some $A, B \in \text{Sym}_{2}(\mathbb{R})$,
    we may assume that
    $A, B \in  \text{Sym}_{2}^{\tt{pd}}(\mathbb{R}_{> 0})$.
    Similarly, if $M \in \text{Sym}_{4}^{\tt{F}}(\mathbb{R}_{> 0})$,
    and $M = A \otimes B$, we may assume $A, B \in \text{Sym}_{2}(\mathbb{R}_{> 0})$.
    If $A$ or $B$ has an eigenvalue 0, so would 
    $M = A \otimes B$.
    So, if $M \in \text{Sym}_{4}^{\tt{F}}(\mathbb{R}_{> 0})$ is such that
    $M = A \otimes B$ for some $A, B \in \text{Sym}_{2}(\mathbb{R})$,
    we may assume that
    $A, B \in  \text{Sym}_{2}^{\tt{F}}(\mathbb{R}_{> 0})$.
\end{remark*}

We will now state the following theorem~\cite{guo2020complexity}:

\begin{theorem}\label{theorem: fullDomain2Hardness}
	The problem $\PlEVAL(M)$ is $\#$P-hard, for
	$M = \left( \begin{smallmatrix}
	x & y\\
	y & z\\
	\end{smallmatrix} \right) \in 
    \text{Sym}_{2}(\mathbb{R}_{\geq 0})$,
    unless $xz = y^{2}$, $y = 0$, or $x = z$, in which case $\PlEVAL(M)$ is in
    polynomial time.
\end{theorem}

This theorem now allows us to prove the following
lemma immediately.

\begin{lemma}\label{lemma: tensorTractable}
    Let $M \in \text{Sym}_{4}^{\tt{pd}}(\mathbb{R}_{> 0})$,
    such that $M = A \otimes B$ for some $A, B \in
    \text{Sym}_{2}^{\tt{pd}}(\mathbb{R}_{> 0})$.
    Then $\PlEVAL(M)$ is polynomial time tractable if
    $A_{11} = A_{22}$, and $B_{11} = B_{22}$.
    If $A_{11} = A_{22}$, but $B_{11} \neq B_{22}$,
    or if $B_{11} = B_{22}$, but $A_{11} \neq A_{22}$,
    then $\PlEVAL(M)$ is $\#$P-hard.
\end{lemma}
\begin{proof}
The tractability part follows from \cref{theorem: fullDomain2Hardness}.
    For the \#P-hardness part, by symmetry,
    let us assume that $A_{11} = A_{22}$, but $B_{11} \neq B_{22}$.
    From \cref{theorem: fullDomain2Hardness}, we see that
    since $A_{11} = A_{22}$, $\PlEVAL(A)$ is polynomial time tractable.
    On the other hand, we apply  \cref{theorem: fullDomain2Hardness} to
    $\PlEVAL(B)$, and see that (1)  $B_{11}B_{22} \not = (B_{12})^{2}$
    since $B \in
    \text{Sym}_{2}^{\tt{pd}}(\mathbb{R}_{> 0})$ has full rank,
    (2) $B_{12} \not = 0$
    as $B \in
    \text{Sym}_{2}(\mathbb{R}_{> 0})$,  and (3) $B_{11} \neq B_{22}$ as given.
    Hence,
    $\PlEVAL(B)$ is $\#$P-hard.
    Now, since
    $\PlEVAL(A)$ is polynomially tractable, \cref{lemma: tensorSimple} implies that 
    $\PlEVAL(M)$ is $\#$P-hard.
\end{proof}

To deal with the case where $A_{11} \neq A_{22}$ and
$B_{11} \neq B_{22}$ will require just a little bit more work.

\begin{lemma}\label{lemma: TensorDiagDist}
    Let $M \in \text{Sym}_{4}^{\tt{pd}}(\mathbb{R}_{> 0})$,
    such that $M = A \otimes B$ for some $A, B \in
    \text{Sym}_{2}^{\tt{pd}}(\mathbb{R}_{> 0})$, that satisfy
    $A_{11} \neq A_{22}$ and $B_{11} \neq B_{22}$.
    Let $\mathcal{F}$ be a countable set of
    $\text{Sym}_{4}(\mathbb{R})$-polynomials such that
    $F(M) \neq 0$ for all $F \in \mathcal{F}$.
    Then, either $\PlEVAL(M)$ is $\#$P-hard, or there
    exists some diagonal distinct
    $N \in \mathfrak{R}(M, \mathcal{F}) \cap
    \text{Sym}_{4}^{\tt{pd}}(\mathbb{R}_{> 0})$, such that
    $\rho_{\tt{tensor}}(N) = 0$.
\end{lemma}
\begin{proof}
    We may assume without loss of generality that $A_{11} > A_{22}$, and
    that $B_{11} > B_{22}$, by permuting the rows and columns
    of $A$ and $B$ (and correspondingly, $M$).
    We can see that $M_{11} = A_{11}B_{11} > M_{ij} > M_{44}$ for any $(i,j) \ne (1,1), (4,4)$.
   So, rows $1$ and $i$ are
    not order identical for any $i \ne 1$. 
    Similarly, rows $4$ and $i$ are
    not order identical for any $i \ne 4$.  
    We have $M_{11} > M_{22} > M_{44}$ and  $M_{11} >  M_{33} > M_{44}$. 
    If $M_{22} \neq M_{33}$, then $M$ is already diagonal distinct,
    and we are done.
    So, we may assume that $M_{22} = M_{33}$.
    Let us first assume that rows $2$ and $3$ are not order identical.
    Then $M$ is p.o.\,distinct, and by \cref{corollary: ordDistDiagDistReduct}
    there exists some diagonal distinct
    $N \in \mathfrak{R}(M, \mathcal{F}) \cap
    \text{Sym}_{4}^{\tt{pd}}(\mathbb{R}_{> 0})$.
    If $\rho_{\tt{tensor}}(N) \neq 0$, then
    \cref{theorem: TensorOrHard}
    immediately implies that $\PlEVAL(N) \leq \PlEVAL(M)$
    is $\#$P-hard.
On the other hand, if $\rho_{\tt{tensor}}(N) = 0$ then $N$ is our required matrix.

    Finally, we consider the case where rows $2$ and $3$
    are order identical.
    We see that the multi-set of elements in row $2$ of $M$
    is $\{A_{11}B_{12}, A_{11}B_{22}, A_{12}B_{12}, A_{12}B_{22}\}$,
    and the multi-set of elements in row $3$ of $M$ is
    $\{A_{12}B_{11}, A_{12}B_{12}, A_{22}B_{11}, A_{22}B_{12}\}$.
    We can see that $A_{12}B_{12}$ and $A_{11}B_{22} = A_{22}B_{11}$ (which are 
     $M_{22} = M_{33}$) 
    appear in both row $2$ and row $3$.
    So, it must be the case that either
    $A_{11}B_{12} = A_{22}B_{12}$ and $A_{12}B_{22} = A_{12}B_{11}$, 
    or
    $A_{11}B_{12} = A_{12}B_{11}$ and $A_{12}B_{22} = A_{22}B_{12}$.
    Since we have assumed that $A_{11} > A_{22}$, and that
    $B_{11} > B_{22}$, we can see the first option here
    is not possible.
    So, we see that in fact,
    $$\frac{A_{11}}{A_{12}} = \frac{B_{11}}{B_{12}}, ~~\text{ and }~~
    \frac{A_{22}}{A_{12}} = \frac{B_{22}}{B_{12}}.$$
    But this means that $B = \kappa \cdot A$
    for some constant $\kappa > 0$, and in this case, we see that
    for any graph $G = (V, E)$,
    $$Z_{M}(G) = Z_{A \otimes B}(G) = Z_{A}(G) \cdot Z_{\kappa \cdot A}(G) =
    \kappa^{|E|} \cdot (Z_{A}(G))^{2}.$$
    Since $A \in \text{Sym}_{2}(\mathbb{R}_{> 0})$, 
    we see that for any graph $G = (V, E)$,
    $Z_{A}(G) > 0$.
    So, given oracle access to $\PlEVAL(M)$, we can compute
    $Z_{M}(G)$ for any planar graph $G = (V, E)$, and use it to compute
    $$Z_{A}(G) = \sqrt{\frac{Z_{M}(G)}{\kappa^{|E|}}}.$$
    So, we see that $\PlEVAL(A) \leq \PlEVAL(M)$.
    Since $A \in \text{Sym}_{2}^{\tt{pd}}(\mathbb{R}_{> 0})$, and
    $A_{11} > A_{22}$, \cref{theorem: fullDomain2Hardness}
    implies that $\PlEVAL(A) \leq \PlEVAL(M)$ is
    $\#$P-hard.
\end{proof}

\begin{lemma}\label{lemma: TensorDiagDistHard}
    Let $M \in \text{Sym}_{4}^{\tt{pd}}(\mathbb{R}_{> 0})$ be a diagonal
    distinct matrix,
    such that $M = A \otimes B$ for some $A, B \in
    \text{Sym}_{2}^{\tt{pd}}(\mathbb{R})$.
    Then, $\PlEVAL(M)$ is $\#$P-hard.
\end{lemma}
\begin{proof}
    Since $M$ is diagonal distinct, it must be the
    case that $A_{11} \neq A_{22}$, and
    $B_{11} \neq B_{22}$.
    By permuting the rows and columns of $A$ and $B$ (and
    correspondingly, the rows and columns of $M$), we may assume that
    $A_{11} > A_{22}$, and
    that $B_{11} > B_{22}$.

    Let $\mathcal{L}_{M}$ be the lattice of the
    eigenvalues of $M$, and let $\mathcal{B}$ be a lattice
    basis of $\mathcal{L}_{M}$.
    If there is any $\mathbf{x} \in \mathcal{B}$  that  is confluent,
    we can use the combination of
    \cref{lemma: diagDistinctImpliesOrderDist},
    \cref{corollary: pairwiseDistImpliesIndep}, and
    \cref{theorem: confluentReduction}, to find some
    diagonal distinct $N_{1} \in 
    \mathfrak{R}(M, \mathcal{F}_{M} \cup \{\Psi_{\mathbf{x}},
    \rho_{\tt{diag}}\})
    \cap \text{Sym}_{4}^{\tt{pd}}(\mathbb{R}_{> 0})$.
    
    We know from \cref{lemma: dimReduction} that after repeating
    this step at most 4 times, we will have
    some diagonal distinct $N \in \text{Sym}_{4}^{\tt{pd}}(\mathbb{R}_{> 0})$,
    such that $\PlEVAL(N) \leq \PlEVAL(M)$, and
    $$\mathbf{0} \ne \mathbf{x} \in \overline{\mathcal{L}_{N}} \implies 
    \mathbf{x} \text{ is not confluent.}$$
    Once again, we see that if $\rho_{\tt{tensor}}(N) \neq 0$,
    \cref{theorem: TensorOrHard}
    immediately implies that $\PlEVAL(N) \leq \PlEVAL(M)$
    is $\#$P-hard, so let us assume $\rho_{\tt{tensor}}(N) = 0$.
    Now, 
    since any non-zero $\mathbf{x} \in \overline{\mathcal{L}_{N}}$
    is not confluent, \cref{lemma: almostAllConfluent} implies that
    \begin{equation}\label{tensorHardnessEqn}
        \overline{\mathcal{L}_{N}} \subseteq \left\{(x_{1}, x_{2}, x_{3}, x_{4})
        \in \chi_{4}
        \hspace{0.08cm}\Big|\hspace{0.1cm}
        |x_{1}| = |x_{2}| = |x_{3}| = |x_{4}|\right\}.
    \end{equation}
    We can now assume (after some permutation of the rows and columns
    of $N$) that $N = A' \otimes B'$ for
    some $A', B' \in \text{Sym}_{2}^{\tt{pd}}(\mathbb{R}_{> 0})$.
    Since $N$ is diagonal distinct, we see that
    $A'$ and $B'$ are diagonal distinct as well.
    We may let $A' = H_{1}D_{1}H_{1}^{\tt{T}}$, and
    $B' = H_{2}D_{2}H_{2}^{\tt{T}}$, for some orthogonal
    matrices $H_{1}, H_{2}$, and diagonal matrices $D_{1}, D_{2}$.
    Since $N = A' \otimes B'$, this implies that
    $N = (H_{1} \otimes H_{2})(D_{1} \otimes D_{2})
    (H_{1} \otimes H_{2})^{\tt{T}}$.
    We may further assume that the eigenvalues of $A'$ and $B'$
    are $(\lambda_{1}, \lambda_{2})$, and $(\mu_{1}, \mu_{2})$
    respectively.
    This implies that the eigenvalues of $N = A' \otimes B'$ are:
    $$\lambda_{1}\mu_{1}, ~~~~ \lambda_{1}\mu_{2}, ~~~~
    \lambda_{2}\mu_{1}, ~~~~ \lambda_{2}\mu_{2}.$$
    We can clearly see that $(\lambda_{1}\mu_{1})(\lambda_{2}\mu_{2})
    = (\lambda_{1}\mu_{2})(\lambda_{2}\mu_{1})$.
    So, $(1, -1, -1, 1) \in \mathcal{L}_{N}$.
    Let us now consider any $\mathbf{0} \neq \mathbf{x} \in 
    \mathcal{L}_{N}$.
    From \cref{tensorHardnessEqn}, we know that it must be the case that
    $|x_{1}| = |x_{2}| = |x_{3}| = |x_{4}|$.
    We may assume without loss of generality, that
    $x_{1} > 0$.
    If $x_{2} = x_{1}$,  then $\mathbf{x}$ being in $\chi_{4}$  implies that $x_3 = x_4 = -x_1$.
    Then
    $x_1(1, 1, -1, -1) \in \mathcal{L}_{N}$.
    This implies that $x_1(1, -1, -1, 1) + x_1(1, 1, -1, -1) = x_1(2, 0, -2, 0)
    \in \mathcal{L}_{N} \subseteq \overline{\mathcal{L}_{N}}$
    which contradicts \cref{tensorHardnessEqn}.
    So, we conclude that $x_{2} = - x_{1}$. By a symmetric argument $x_{3} = - x_{1}$.
    Since $\mathbf{x} \in \chi_{4}$, this forces
    $x_{4} = x_{1}$.
    Hence $\mathbf{x} = x_{1}(1, -1, -1, 1)$.
    Therefore, we conclude that if
    $\mathbf{x} \in \mathcal{L}_{N}$, then $\mathbf{x}$ must be
    a multiple of $(1, -1, -1, 1)$.
    In other words,
    $$\mathcal{L}_{N} = \left\{ c \cdot (1, -1, -1, 1) \ |\ 
    c \in \mathbb{Z} \right\}.$$
    
    We now define  $N' = A' \otimes I$, where $I$ is the 2 by 2 identity matrix.
    Recall that $N = A' \otimes B' = 
    (H_{1} \otimes H_{2})(D_{1} \otimes D_{2})
    (H_{1} \otimes H_{2})^{\tt{T}}$, and we can now see that
    $N' = A' \otimes I = 
    (H_{1} \otimes H_{2})(D_{1} \otimes I)
    (H_{1} \otimes H_{2})^{\tt{T}}$.
    So, if we let $H = H_{1} \otimes H_{2}$, $D = D_{1} \otimes D_{2}$,
    and $D' = D_{1} \otimes I$, we see that
    $N = HDH^{\tt{T}}$, and $N' = HD'H^{\tt{T}}$.
    Moreover, by construction, we note that the eigenvalues of
    $N'$ are: $(\lambda_{1}, \lambda_{2}, \lambda_{1}, \lambda_{2})$.
    So, $(1, -1, -1, 1) \in \mathcal{L}_{N'}$,
    which implies that
    $$\mathcal{L}_{N} \subseteq \mathcal{L}_{N'}.$$
    But then, \cref{lemma: stretchingBasic} implies that
    $\PlEVAL(N') \leq \PlEVAL(N)$.
    But by construction of $N'$, we see that for any graph
    $G = (V, E)$,
    $$Z_{N'}(G) = Z_{A' \otimes I}(G) = Z_{A'}(G) \cdot Z_{I}(G).$$
    We note that for any connected graph $G = (V, E)$, $Z_{I}(G) = 2$,
    as $I$ is the 2 by 2 identity matrix.
    This means that if we had oracle access to $\PlEVAL(N')$, we could
    use it to compute $Z_{A'}(G)$
    for any connected (and thus all) planar graphs $G = (V, E)$.
    So, we see that
    $\PlEVAL(A') \leq \PlEVAL(N')$.
    Since $(A')_{11} \neq (A')_{22}$ by construction,
    \cref{theorem: fullDomain2Hardness} implies that
    $\PlEVAL(A')$ is $\#$P-hard.
    This proves that $\PlEVAL(A') \leq \PlEVAL(N') \leq 
    \PlEVAL(N) \leq \PlEVAL(M)$ is $\#$P-hard as well.
\end{proof}

\begin{remark*}
    Let $M = A \otimes B$ for any $M \in \text{Sym}_{4}(\mathbb{R}_{> 0})$.
    In the proof above, we made use of the argument that
    if the matrices $N$ we obtained using the
    theorems and lemmas from the previous section satisfy
    the equation $\rho_{\tt{tensor}}(N) \neq 0$,
    then we can use \cref{theorem: TensorOrHard} to prove
    the $\#$P-hardness of $\PlEVAL(N) \leq \PlEVAL(M)$.
There is an alternative route in this proof. This alternative route starts by observing that,
    in fact, if $\rho_{\tt{tensor}}(M) = 0$, then
    any matrix $N$ that we obtain using the techniques
    from the rest of this paper will also satisfy the
    equation $\rho_{\tt{tensor}}(N) = 0$.
    This is because of the following set of properties:
    $T_{n}(M) = T_{n}(A) \otimes T_{n}(B)$,
    $S_{n}(M) = S_{n}(A) \otimes S_{n}(B)$, and
    $R_{n}(M) = R_{n}(A) \otimes R_{n}(B)$,
    for all $n \geq 1$.
    Moreover, if we let the entries of $M$ be generated by
    some $\{g_{t}\}_{t \in [d]}$, this implies that
    $$\rho_{\tt{tensor}}(M) = 0 \implies \rho_{\tt{tensor}}(\mathcal{T}_{M}(\mathbf{p})) = 0$$
    for all $\mathbf{p} \in \mathbb{R}^{d}$.
    Similarly,
    $$\rho_{\tt{tensor}}(M) = 0 \implies 
    \rho_{\tt{tensor}}(\mathcal{S}_{M}(\theta))) = 0, \text{ and }
    \rho_{\tt{tensor}}(\mathcal{R}_{M}(\theta)) = 0,$$
    for all $\theta \in \mathbb{R}$.
\end{remark*}

\cref{lemma: TensorDiagDist} and
\cref{lemma: TensorDiagDistHard} form the last piece of the puzzle for
matricces in $\text{Sym}_{4}^{\tt{pd}}(\mathbb{R}_{> 0})$,
and  now we can prove the following theorem.

\begin{theorem}\label{theorem: posDefPositiveDichotomy}
    Let $M \in \text{Sym}_{4}^{\tt{pd}}(\mathbb{R}_{> 0})$.
    Then $\PlEVAL(M)$ is $\#$P-hard, unless
    $M$ is isomorphic to $A \otimes B$ for some
    $A, B \in \text{Sym}_{2}^{\tt{pd}}(\mathbb{R}_{> 0})$
    such that $A_{11} = A_{22}$, and $B_{11} = B_{22}$,
    in which case, $\PlEVAL(M)$ is polynomial time tractable.
\end{theorem}
\begin{proof}
    We know from \cref{theorem: TensorOrHard} that
    $\PlEVAL(M)$ is $\#$P-hard, unless
    $\rho_{\tt{tensor}}(M) = 0$.
    But, if $\rho_{\tt{tensor}}(M) = 0$
    then (by the defining property of $\rho_{\tt{tensor}}$) $M$ is isomorphic to
    $A \otimes B$ for some $A, B \in
    \text{Sym}_{2}^{\tt{pd}}(\mathbb{R}_{> 0})$.
    By permuting the rows and columns of $M$, we may
    assume that in fact, $M = A \otimes B$.
    
    We will now consider the diagonal entries of $A$ and $B$.
    If $A_{11} = A_{22}$, but $B_{11} \neq B_{22}$, or
    if $B_{11} = B_{22}$, but $A_{11} \neq A_{22}$, 
    \cref{lemma: tensorTractable} implies that
    $\PlEVAL(M)$ is $\#$P-hard.
    If $A_{11} \neq A_{22}$, and $B_{11} \neq B_{22}$,
    \cref{lemma: TensorDiagDist} implies that either
    $\PlEVAL(M)$ is $\#$P-hard, or there
    exists some diagonal distinct $N \in \mathfrak{R}(M,
    \mathcal{F}_{M}) \cap
    \text{Sym}_{4}^{\tt{pd}}(\mathbb{R}_{> 0})$ such that
    $\rho_{\tt{tensor}}(N) = 0$.
    But from \cref{lemma: TensorDiagDistHard}, we see that
    $\PlEVAL(N) \leq \PlEVAL(M)$ is $\#$P-hard.
    So, we see that even if $M$ is isomorphic to some
    $A \otimes B$, unless
    $A_{11} = A_{22}$ and $B_{11} = B_{22}$,
    $\PlEVAL(M)$ is $\#$P-hard.
    Moreover, if $A_{11} = A_{22}$, and $B_{11} = B_{22}$,
    we see from \cref{lemma: tensorTractable} that
    $\PlEVAL(M)$ is polynomially tractable.
\end{proof}

We can extend this dichotomy to all full rank matrices
with some additional effort.
When we consider $M \in \text{Sym}_{4}^{\tt{F}}(\mathbb{R}_{> 0})$,
we come across a family of matrices for which
the gadgets we have constructed so far are insufficient
for proving the $\#$P-hardness of $\PlEVAL(M)$.
So, we will need to introduce a new edge gadget.
Given any graph $G = (V, E)$, we will construct the
graph $BG$ by replacing each  edge of $G$ with the
gadget in \cref{fig: bridgeGadget}.
Clearly, this gadget preserves planarity.
Moreover, we note that for any $M \in \text{Sym}_{q}(\mathbb{R})$,
$$Z_{M}(BG) = \sum_{\sigma: V \rightarrow [q]}\prod_{(u, v) \in E}
\sum_{a, b \in [q]}
M_{\sigma(u)a}M_{\sigma(u)b}M_{\sigma(v)a}M_{\sigma(v)b}M_{ab}
= \sum_{\sigma: V \rightarrow [q]}\prod_{(u, v) \in E}
(BM)_{\sigma(u)\sigma(v)},$$
where $BM \in \text{Sym}_{q}(\mathbb{R})$ such that
for all $i, j \in [q]$,
$$(BM)_{ij} = \sum_{a, b \in [q]}M_{ia}M_{ib}M_{ja}M_{jb}M_{ab}.$$
Therefore, we see that
$\PlEVAL(BM) \leq \PlEVAL(M)$
for all $M \in \text{Sym}_{q}(\mathbb{R})$.
We will first use this gadget to prove the
$\#$P-hardness of one special family of
matrices in $\text{Sym}_{4}^{\tt{F}}(\mathbb{R}_{> 0})$.

\setcounter{figure}{2}
\begin{figure}[ht]
    \centering
    \begin{subfigure}{0.45\textwidth}
        \centering
        \raisebox{0.8\height}{
            \scalebox{0.8}{\begin{tikzpicture}[line join=miter, draw opacity=1]

\node[circle, draw=black, fill=black] (c) at (-2, 0){};
\node[circle, draw=black, fill=black] (c) at (-1, 0){};
\node[circle, draw=black, fill=black] (c) at (1, 0){};
\node[circle, draw=black, fill=black] (c) at (2, 0){};
\node[circle, draw=black, fill=black] (c) at (0, 1){};
\node[circle, draw=black, fill=black] (c) at (0, -1){};

\draw[line width=0.5mm, black] (-2, 0) -- (-1, 0);
\draw[line width=0.5mm, black] (1, 0) -- (2, 0);

\draw[line width=0.5mm, black] (-1, 0) -- (0, 1);
\draw[line width=0.5mm, black] (-1, 0) -- (0, -1);
\draw[line width=0.5mm, black] (1, 0) -- (0, 1);
\draw[line width=0.5mm, black] (1, 0) -- (0, -1);
\draw[line width=0.5mm, black] (0, 1) -- (0, -1);

\end{tikzpicture}}
        }
	\caption{The edge gadget $B$.}
	\label{fig: bridgeGadget}
    \end{subfigure}
    \begin{subfigure}{0.45\textwidth}
        \centering
        \scalebox{0.8}{\begin{tikzpicture}[line join=miter, draw opacity=1]

\node[circle, draw=black, fill=black] (c) at (-1, -1){};
\node[circle, draw=black, fill=black] (c) at (-1, 1){};
\node[circle, draw=black, fill=black] (c) at (1, 1){};
\node[circle, draw=black, fill=black] (c) at (1, -1){};

\draw[line width=0.5mm, black] (-1, -1) -- (-1, 1);
\draw[line width=0.5mm, black] (-1, 1) -- (1, 1);
\draw[line width=0.5mm, black] (1, 1) -- (1, -1);
\draw[line width=0.5mm, black] (1, -1) -- (-1, -1);

\begin{scope}[xshift = 5cm]

\node[circle, draw=black, fill=black] (c) at (-2, -2){};
\node[circle, draw=black, fill=black] (c) at (-2, 2){};
\node[circle, draw=black, fill=black] (c) at (2, 2){};
\node[circle, draw=black, fill=black] (c) at (2, -2){};

\node[circle, draw=black, fill=black] (c) at (-2, -1){};
\node[circle, draw=black, fill=black] (c) at (-2, 1){};
\node[circle, draw=black, fill=black] (c) at (-3, 0){};
\node[circle, draw=black, fill=black] (c) at (-1, 0){};

\node[circle, draw=black, fill=black] (c) at (-1, 2){};
\node[circle, draw=black, fill=black] (c) at (1, 2){};
\node[circle, draw=black, fill=black] (c) at (0, 3){};
\node[circle, draw=black, fill=black] (c) at (0, 1){};

\node[circle, draw=black, fill=black] (c) at (2, -1){};
\node[circle, draw=black, fill=black] (c) at (2, 1){};
\node[circle, draw=black, fill=black] (c) at (3, 0){};
\node[circle, draw=black, fill=black] (c) at (1, 0){};

\node[circle, draw=black, fill=black] (c) at (-1, -2){};
\node[circle, draw=black, fill=black] (c) at (1, -2){};
\node[circle, draw=black, fill=black] (c) at (0, -3){};
\node[circle, draw=black, fill=black] (c) at (0, -1){};

\draw[line width=0.5mm, black] (-2, -2) -- (-2, -1);
\draw[line width=0.5mm, black] (-2, 1) -- (-2, 2);
\draw[line width=0.5mm, black] (-2, -1) -- (-3, 0);
\draw[line width=0.5mm, black] (-2, -1) -- (-1, 0);
\draw[line width=0.5mm, black] (-2, 1) -- (-3, 0);
\draw[line width=0.5mm, black] (-2, 1) -- (-1, 0);
\draw[line width=0.5mm, black] (-3, 0) -- (-1, 0);

\draw[line width=0.5mm, black] (-2, 2) -- (-1, 2);
\draw[line width=0.5mm, black] (1, 2) -- (2, 2);
\draw[line width=0.5mm, black] (-1, 2) -- (0, 3);
\draw[line width=0.5mm, black] (-1, 2) -- (0, 1);
\draw[line width=0.5mm, black] (1, 2) -- (0, 3);
\draw[line width=0.5mm, black] (1, 2) -- (0, 1);
\draw[line width=0.5mm, black] (0, 3) -- (0, 1);

\draw[line width=0.5mm, black] (2, 2) -- (2, 1);
\draw[line width=0.5mm, black] (2, -1) -- (2, -2);
\draw[line width=0.5mm, black] (2, 1) -- (3, 0);
\draw[line width=0.5mm, black] (2, 1) -- (1, 0);
\draw[line width=0.5mm, black] (2, -1) -- (3, 0);
\draw[line width=0.5mm, black] (2, -1) -- (1, 0);
\draw[line width=0.5mm, black] (3, 0) -- (1, 0);

\draw[line width=0.5mm, black] (2, -2) -- (1, -2);
\draw[line width=0.5mm, black] (-1, -2) -- (-2, -2);
\draw[line width=0.5mm, black] (-1, -2) -- (0, -3);
\draw[line width=0.5mm, black] (-1, -2) -- (0, -1);
\draw[line width=0.5mm, black] (1, -2) -- (0, -3);
\draw[line width=0.5mm, black] (1, -2) -- (0, -1);
\draw[line width=0.5mm, black] (0, -3) -- (0, -1);

\end{scope}

\end{tikzpicture}}
	\caption{An example graph $G$,
                    and the graph $BG$.}
	\label{fig: bridgeGadgetExample}
    \end{subfigure}
    \caption{The edge gadget $B$}
    \label{fig: bridgeGadgetAll}
\end{figure}
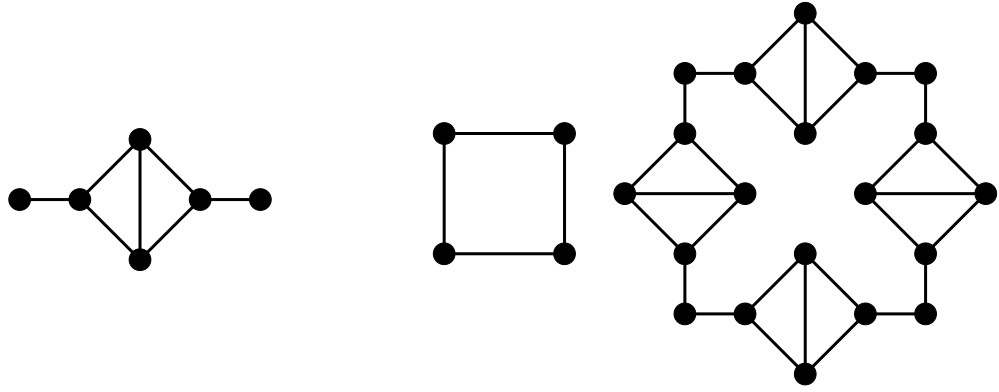

\begin{lemma}\label{lemma: fullRankSpecialHardness}
    Let $M \in \text{Sym}_{4}^{\tt{F}}(\mathbb{R}_{> 0})$
    such that
    $$M = \begin{pmatrix}
        p^{2} & pq & pq & q^{2}\\
        pq & q^{2} & p^{2} & pq\\
        pq & p^{2} & q^{2} & pq\\
        q^{2} & pq & pq & p^{2}\\
    \end{pmatrix}$$
    where $p \neq q \in \mathbb{R}_{> 0}$.
    Then, $\PlEVAL(M)$ is $\#$P-hard.
\end{lemma}
\begin{proof}
    We will consider $B(T_{n}(M))$ for $n \geq 1$.
    From the definitions, we can verify that
    \begin{align*}
        (B(T_{n}M))_{11} = (B(T_{n}M))_{44}
        &= p^{10n} + 4p^{7n}q^{3n} + 2p^{6n}q^{4n} + 4p^{4n}q^{6n} 
        + 4p^{3n}q^{7n} + p^{2n}q^{8n},\\
        (B(T_{n}M))_{22} = (B(T_{n}M))_{33}
        &= p^{8n}q^{2n} + 4p^{7n}q^{3n} + 4p^{6n}q^{4n} + 
        2p^{4n}q^{6n} + 4p^{3n}q^{7n} + q^{10n},\\
        (B(T_{n}M))_{12} = (B(T_{n}M))_{13}
        &= p^{8n}q^{2n} + 2p^{7n}q^{3n} + 3p^{6n}q^{4n} + 4p^{5n}q^{5n} + 
        3p^{4n}q^{6n} + 2p^{3n}q^{7n} + p^{2n}q^{8n},\\
        (B(T_{n}M))_{24} = (B(T_{n}M))_{34}
        &= p^{8n}q^{2n} + 2p^{7n}q^{3n} + 3p^{6n}q^{4n} + 4p^{5n}q^{5n} + 
        3p^{4n}q^{6n} + 2p^{3n}q^{7n} + p^{2n}q^{8n},\\
        (B(T_{n}M))_{14} = (B(T_{n}M))_{23}
        &= 4p^{6n}q^{4n} + 8p^{5n}q^{5n} + 4p^{4n}q^{6n}.
    \end{align*}

    So, it is seen that
    \begin{align*}
        \det(B(T_{n}M)) = \quad &
        p^{36n}q^{4n} + 8p^{35n}q^{5n}  + 20p^{34n}q^{6n} 
        + 8p^{33n}q^{7n} - 32p^{32n}q^{8n} \\
        &- 40p^{31n}q^{9n} - 100p^{30n}q^{10n} - 296p^{29n}q^{11n}
        - 84p^{28n}q^{12n}\\ 
        &+ 840p^{27n}q^{13n} + 1204p^{26n}q^{14n} + 72p^{25n}q^{15n} 
        - 1440p^{24n}q^{16n}\\
        &- 2088p^{23n}q^{17n} - 1124p^{22n}q^{18n} + 1496p^{21n}q^{19n} 
        + 3110p^{20n}q^{20n}\\
        &+ 1496p^{19n}q^{21n} - 1124p^{18n}q^{22n} 
        - 2088p^{17n}q^{23n} - 1440p^{16n}q^{24n} \\
        &+ 72p^{15n}q^{25n} + 1204p^{14n}q^{26n}
        + 840p^{13n}q^{27n} - 84p^{12n}q^{28n}\\
        &- 296p^{11n}q^{29n} - 100p^{10n}q^{30n} 
        - 40p^{9n}q^{31n} - 32p^{8n}q^{32n}\\
        &+ 8p^{7n}q^{33n} + 20p^{6n}q^{34n} 
        + 8p^{5n}q^{35n} + p^{4n}q^{36n}.
    \end{align*}

    We will define the polynomial
    $f: \mathbb{R} \rightarrow \mathbb{R}$ such that
   \begin{align*}
        f(x) = \quad & x^{36} + 8x^{35} + 20x^{34} + 8x^{33} - 32x^{32}
        - 40x^{31} - 100x^{30} - 296x^{29} - 84x^{28}\\
        &+ 840x^{27} + 1204x^{26} + 72x^{25} - 1440x^{24}
        - 2088x^{23} + 1124x^{22} + 1496x^{21} + 3110x^{20}\\
        &+ 1496x^{19} + 1124x^{18} - 2088x^{17} - 1440x^{16}
        + 72x^{15} + 1204x^{14} + 840x^{13}\\
        &- 84x^{12} - 296x^{11} - 100x^{10} - 40x^{9}
        - 32x^{8} + 8x^{7} + 20x^{6} + 8x^{5} + x^{4}.
    \end{align*}
    By construction, $\det(B(T_{n}M)) = q^{40n}f((\nicefrac{p}{q})^{n})$.
    Since $f$ is a non-zero polynomial,
    we know that it has finitely many real roots.
    Since $p \neq q \in \mathbb{R}_{> 0}$,
    we know that if $p > q$, then $(\nicefrac{p}{q})^{n + 1}
    > (\nicefrac{p}{q})^{n} > 1$ for all $n \geq 1$, and
    if $p < q$, then $0 < (\nicefrac{p}{q})^{n + 1}
    < (\nicefrac{p}{q})^{n}$ for all $n \geq 1$.
    In particular, this implies that for large enough
    $n \geq 1$,
    $\det(B(T_{n}M)) = q^{40n}f((\nicefrac{p}{q})^{n}) \neq 0$.

    Therefore, we see that for large enough $n \geq 1$,
    $B(T_{n}M) \in \text{Sym}_{4}^{\tt{F}}(\mathbb{R}_{> 0})$.
    We will now let $N_{n} = (B(T_{n}M))^{2}$ for all 
    $n \geq 1$.
    We see that $N_{n} \in \text{Sym}_{4}^{\tt{pd}}(\mathbb{R}_{> 0})$
    for large enough $n$.
    We will now define the $\text{Sym}_{4}(\mathbb{R})$-polynomials
    $\xi_{1}, \xi_{2}, \xi_{3}$ such that
    $\xi_{1}: N \mapsto (N_{11}N_{14} - N_{12}N_{13})$,
    $\xi_{2}: N \mapsto (N_{11}N_{13} - N_{12}N_{14})$, and
    $\xi_{3}: N \mapsto (N_{11}N_{12} - N_{13}N_{14})$.
    From the construction of $\varrho_{\tt{tensor}}$
    in \cref{lemma: tensorPolynomial},
    we can see that if $\xi_{1}(N) \neq 0$, then
    $\varrho_{\tt{tensor}}(N) \neq 0$.
    We note that the matrix $N^{\sigma}$ is defined such that
    $(N^{\sigma})_{ij} = N_{\sigma(i)\sigma(j)}$
    for any $\sigma \in S_{4}$.
    In fact, given any $\sigma \in S_{4}$ such that
    $\{\sigma(1), \sigma(4)\} = \{1, 4\}$, or
    $\{2, 3\}$, we see that $\xi_{1}(N) \neq 0$ implies
    that $\varrho_{\tt{tensor}}(N^{\sigma}) \neq 0$.
    Similarly, given any $\sigma \in S_{4}$
    such that $\{\sigma(1), \sigma(3)\} = \{1, 4\}$, or
    $\{2, 3\}$, we see that
    $\xi_{2}(N) \neq 0$ implies that 
    $\varrho_{\tt{tensor}}(N^{\sigma}) \neq 0$, and
    given any $\sigma \in S_{4}$
    such that $\{\sigma(1), \sigma(2)\} = \{1, 4\}$, or
    $\{2, 3\}$, we see that
    $\xi_{3}(N) \neq 0$ implies that 
    $\varrho_{\tt{tensor}}(N^{\sigma}) \neq 0$.
    Summing up, we find that if
    $\xi_{i}(N) \neq 0$ for all $i \in [3]$, then
    $\rho_{\tt{tensor}}(N) \neq 0$.

    We will now show that for large enough $n \geq 1$,
    $\xi_{i}(N_{n}) \neq 0$ for all $i \in [3]$.
    To do this, we will compute $\xi_{i}(N_{n})$
    for all $n \geq 1$.
    It can be verified that
    \begin{align*}
        \xi_{1}(N_{n}) = \xi_{2}(N_{n}) = \quad &
        p^{38n}q^{2n} + 2p^{37n}q^{3n} + 4p^{36n}q^{4n} 
        + 22p^{35n}q^{5n} + 48p^{34n}q^{6n}\\
        &+ 82p^{33n}q^{7n} + 204p^{32n}q^{8n} + 358p^{31n}q^{9n}
        + 396p^{30n}q^{10n}\\
        &+ 394p^{29n}q^{11n} + 164p^{28n}q^{12n} - 722p^{27n}q^{13n} 
        - 1968p^{26n}q^{14n}\\
        &- 2886p^{25n}q^{15n} - 2740p^{24n}q^{16n} - 610p^{23n}q^{17n}
        + 2918p^{22n}q^{18n}\\
        &+ 5382p^{21n}q^{19n} + 5100p^{20n}q^{20n} + 
        2242p^{19n}q^{21n} - 1648p^{18n}q^{22n}\\
        &- 4042p^{17n}q^{23n} - 3644p^{16n}q^{24n} 
        - 1678p^{15n}q^{25n} + 140p^{14n}q^{26n}\\
        &+ 1006p^{13n}q^{27n} + 876p^{12n}q^{28n} 
        + 378p^{11n}q^{29n} + 112p^{10n}q^{30n} \\
        &+ 62p^{9n}q^{31n} + 36p^{8n}q^{32n} + 
        10p^{7n}q^{33n} + p^{6n}q^{34n}, \text{ and }
        \end{align*}
        \begin{align*}
         \xi_{3}(N_{n}) \quad  = \quad \quad  \quad  \quad \quad  &
        9p^{36n}q^{4n} + 20p^{35n}q^{5n} + 
        16p^{34n}q^{6n} + 108p^{33n}q^{7n} + 250p^{32n}q^{8n}\\
        &+ 188p^{31n}q^{9n} + 484p^{30n}q^{10n} + 
        1236p^{29n}q^{11n} + 966p^{28n}q^{12n}\\
        &+ 724p^{27n}q^{13n} + 1736p^{26n}q^{14n}
        - 244p^{25n}q^{15n} - 4638p^{24n}q^{16n}\\
        &- 4196p^{23n}q^{17n} - 2388p^{22n}q^{18n} 
        - 5356p^{21n}q^{19n} - 4616p^{20n}q^{20n}\\
        &+ 4060p^{19n}q^{21n} + 9088p^{18n}q^{22n} 
        + 5604p^{17n}q^{23n} + 1454p^{16n}q^{24n}\\
        &+ 468p^{15n}q^{25n} - 4p^{14n}q^{26n} 
        - 996p^{13n}q^{27n} - 1478p^{12n}q^{28n}\\
        &- 1220p^{11n}q^{29n} - 728p^{10n}q^{30n}
        - 348p^{9n}q^{31n} - 138p^{8n}q^{32n}\\
        &- 44p^{7n}q^{33n} - 12p^{6n}q^{34n} - 
        4p^{5n}q^{35n} - p^{4n}q^{36n}.
    \end{align*}

    We can now define the polynomials $f_{1}, f_{3}: \mathbb{R}
    \rightarrow \mathbb{R}$ such that
    \begin{align*}
        f_{1}(x) = &\
        x^{38} + 2x^{37} + 4x^{36} + 22x^{35} + 48x^{34} 
        + 82x^{33} + 204x^{32} + 358x^{31}\\
        &+ 396x^{30} + 394x^{29} + 164x^{28} - 722x^{27} - 1968x^{26}
        - 2886x^{25}\\
        &- 2740x^{24} - 610x^{23}
        + 2918x^{22} + 5382x^{21} + 5100x^{20} + 2242x^{19}\\
        &- 1648x^{18} - 4042x^{17}
        - 3644x^{16} - 1678x^{15} + 140x^{14} + 1006x^{13}\\
        &+ 876x^{12} + 378x^{11}
        + 112x^{10} + 62x^{9} + 36x^{8} + 10x^{7} + x^{6},
    \end{align*}
    \begin{align*}
        f_{3}(x) = &\ 
        9x^{36} + 20x^{35} + 16x^{34} + 108x^{33} 
        + 250x^{32} + 188x^{31} + 484x^{30}\\
        &+ 1236x^{29} + 966x^{28} + 724x^{27}
        + 1736x^{26} - 244x^{25}\\
        &- 4638x^{24} - 4196x^{23} - 2388x^{22} 
        - 5356x^{21} - 4616x^{20}\\
        &+ 4060x^{19} + 9088x^{18} + 5604x^{17}
        + 1454x^{16} + 468x^{15}\\
        &- 4x^{14} - 996x^{13} - 1478x^{12}
        - 1220x^{11} - 728x^{10}\\
        &- 348x^{9} - 138x^{8} - 44x^{7}
        - 12x^{6} - 4x^{5} - x^{4}.
    \end{align*}
    By construction, we see that
    $\xi_{1}(N_{n}) = \xi_{2}(N_{n}) = q^{40n}f_{1}((\nicefrac{p}{q})^{n})$,
    and $\xi_{3}(N_{n}) = q^{40n}f_{3}((\nicefrac{p}{q})^{n})$.
    Since $f_{1}$ and $f_{3}$ are non-zero polynomials,
    we know that they have finitely many real roots.
    We are given $p \neq q \in \mathbb{R}_{> 0}$.
    So, $(\nicefrac{p}{q})^{n + 1} > (\nicefrac{p}{q})^{n}$
    for all $n \geq 1$, or
    $(\nicefrac{p}{q})^{n + 1} < (\nicefrac{p}{q})^{n}$
    for all $n \geq 1$.
    In either case, we see that for large enough $n \geq 1$,
    $\xi_{1}(N_{n}) = \xi_{2}(N_{n}) \neq 0$, and
    $\xi_{3}(N_{n}) \neq 0$.

    So, there exists some $n^{*} \geq 1$
    such that $N_{n^{*}} \in \text{Sym}_{4}^{\tt{pd}}(\mathbb{R}_{> 0})$,
    and $\rho_{\tt{tensor}}(N_{n^{*}}) \neq 0$.
    From \cref{theorem: TensorOrHard}, we know that
    $\PlEVAL(N_{n^{*}}) \leq \PlEVAL(B(T_{n^{*}}M))
    \leq \PlEVAL(T_{n^{*}M}) \leq \PlEVAL(M)$ is
    $\#$P-hard.
    This proves that $\PlEVAL(M)$ is $\#$P-hard.
\end{proof}

With the special case taken care of, we return to the more general case
where $M$ is ``close'' to being a tensor product.
We will need a few more lemmas.

\begin{lemma}\label{lemma: fullRankExtension}
    Let $M \in \text{Sym}_{4}^{\tt{F}}(\mathbb{R}_{> 0})$,
    such that $\rho_{\tt{tensor}}(M) \neq 0$, but
    $\varrho_{\tt{tensor}}(M^{2}) = 0$.
    Then, either $\PlEVAL(M)$ is $\#$P-hard, or
    $M = KDK^{\tt{T}}$, where
    $$K = \frac{1}{2}
    \begin{pmatrix}
        1 & u & v & 1\\
        1 & -v & u & -1\\
        1 & v & -u & -1\\
        1 & -u & -v & 1\\
    \end{pmatrix}$$
    for some $u, v \in \mathbb{R}$ with $u^2+v^2 =2$, and
    $D = \kappa \cdot \text{\rm diag}(1, x, y, z)$
    for some $\kappa \in \mathbb{R}_{> 0}$, and $x, y, z \in \mathbb{R}$
    such that $0 < |x|, |y|, |z| < 1$.
    Moreover,
    \begin{enumerate}
        \item If $|x| \neq |y|$, then $u = v = 1$, and $z = -xy$.
        \item If $x = y$, then $z = -x^{2}$.
        \item If $x = -y$, then either $z = x^{2}$, 
        or $z = -x^{2}$ and $|u| \neq |v|$.
    \end{enumerate}
\end{lemma}
\begin{remark*}
Given $u^2+v^2 =2$, we have 
$uv=1$ iff $u=v$. In this case, $K$ is a tensor product. In particular,
when $u=v=1$, $K = H^{\otimes 2}$, where $H$ is the 2 by 2 Hadamard matrix. 
\end{remark*}
\begin{proof}
    Since $M^{2} \in \text{Sym}_{4}^{\tt{pd}}(\mathbb{R}_{> 0})$,
    such that $\varrho_{\tt{tensor}}(M^{2}) = 0$,
    we may let $M^{2} = A \otimes B$ for some 
    $A, B \in \text{Sym}_{2}^{\tt{pd}}(\mathbb{R}_{> 0})$.
    From \cref{theorem: posDefPositiveDichotomy},
    we know that $\PlEVAL(M^{2}) \leq \PlEVAL(M)$ is
    $\#$P-hard unless $A_{11} = A_{22}$,
    $B_{11} = B_{22}$.
    So, let us now assume that $A_{11} = A_{22}$,
    $B_{11} = B_{22}$.
    We note that in this case,
    $A = HD_{1}H^{\tt{T}}$, and $B = HD_{2}H^{\tt{T}}$,
    where
    $$H = \frac{1}{\sqrt{2}} \begin{pmatrix}
        1 & 1\\
        1 & -1\\
    \end{pmatrix},~~ D_{1} = \begin{pmatrix}
        A_{11} + A_{12} & 0\\
        0 & A_{11} - A_{12}\\
    \end{pmatrix}, ~~\text{ and }~~ D_{2} = \begin{pmatrix}
        B_{11} + B_{12} & 0\\
        0 & B_{11} - B_{12}\\
    \end{pmatrix}.$$
    So, we see that $M^{2} = (H \otimes H) (D_{1} \otimes D_{2})
    (H \otimes H)^{\tt{T}}$.
    We may now let $\lambda_{1} = A_{11} + A_{12} >
    A_{11} - A_{12} = \lambda_{2}$ be the eigenvalues of $A$, and
    $\mu_{1} = B_{11} + B_{12} > B_{11}- B_{12} = \mu_{2}$
    be the eigenvalues of $B$.
    So, the eigenvalues of $M^{2}$ are $\lambda_{1}\mu_{1}$,
    $\lambda_{1}\mu_{2}$, $\lambda_{2}\mu_{1}$, and
    $\lambda_{2}\mu_{2}$.

    We will now consider the eigenvalues of $M$.
    We know that they must be such that their squares are the
    eigenvalues of $M^{2}$.
    Moreover, since $M \in \text{Sym}_{4}(\mathbb{R}_{> 0})$,
    we know from the Perron-Frobenius Theorem, that $M$
    has a unique positive eigenvalue with the largest absolute
    value.
    Since $\lambda_{1}\mu_{1}$ is the unique eigenvalue of $M^{2}$ with the
    largest absolute value, this implies that
    $\kappa = \sqrt{\lambda_{1}\mu_{1}}$ is the
    unique positive eigenvalue of $M$ with the largest
    absolute value.
    We may now define $x, y, z \in \mathbb{R}_{\neq 0}$ to be such that
    $\kappa  x$, $\kappa  y$, and
    $\kappa  z$ are the eigenvalues
    of $M$ that satisfy the equations:
    $(\kappa  x)^{2} = \lambda_{1}\mu_{2}$,
    $(\kappa  y)^{2} = \lambda_{2}\mu_{1}$, and
    $(\kappa  z)^{2} = \lambda_{2}\mu_{2}$.
    Note that since $\kappa$ is positive and is the unique eigenvalue with the
    largest absolute eigenvalue, it follows that
    $|x|, |y|, |z| < 1$.
    Moreover, since $(\lambda_{1}\mu_{1})(\lambda_{2}\mu_{2})
    = (\lambda_{1}\mu_{2})(\lambda_{2}\mu_{1})$,
    we see that
    $z^{2} = (xy)^{2}$.
    So, either $z = xy$, or $z = -xy$.
    In either case, since $|x|, |y| < 1$, it follows that
    $\kappa > \kappa  |x|, \kappa  |y| > \kappa |z|$.
    So, $1 > |x|, |y| > |z|$.

    We will first consider the case where $|x| \neq |y|$.
    In that case, we see that the eigenvalues of $M^{2}$
    are all distinct. Then their corresponding eigenspaces are all 
    one dimensional, thus the corresponding unit eigenvectors are unique up to a $\pm$ sign.
    This means that $H \otimes H$ is the
    unique orthogonal matrix (up to a $\pm 1$ multiplier per each column) such that
    $M^{2} = (H \otimes H)(D_{1} \otimes D_{2})(H \otimes H)^{\tt{T}}$.
    Since $M \cdot M^{2} = M^{3} = M^{2} \cdot M$,
    we know that $M$ and $M^{2}$ must be simultaneously
    diagonalizable.
    But this then implies that $M$ must also be diagonalized by
    $H \otimes H$, which precisely takes the
    form of the matrix $K$ in the statement of this lemma 
    (with $u = v = 1$).
    Moreover, if $z = xy$, that implies that
    $M = (H \otimes H)(D_{1}' \otimes D_{2}')(H \otimes H)^{\tt{T}}
    = (HD_{1}'H^{\tt{T}}) \otimes (HD_{2}'H^{\tt{T}})$,
    where $D_{1}' = \sqrt{\kappa} \cdot \text{diag}(1, y)$, and
    $D_{2}' = \sqrt{\kappa} \cdot \text{diag}(1, x)$.
    But this implies that $\varrho_{\tt{tensor}}(M) = 0$,
    which contradicts our assumption about $M$.
    So, our assumption that $z = xy$ must be false.
    This means that $z = -xy$, in which case,
    $M = KDK^{\tt{T}}$ with $u = v = 1$, and $D = \kappa \cdot 
    \text{diag}(1, x, y, -xy)$.
    So, the statement of the
    lemma is proved.

    Now, let us consider the case where $|x| = |y|$.
    We will characterize the set of all matrices $H'$
    such that $M^{2} = (H')(D_{1} \otimes D_{2})(H')^{\tt{T}}$.
    We note that the eigenvalues $(\lambda_{1}\mu_{1})$
    and $(\lambda_{2}\mu_{2})$ of $M^{2}$ have multiplicity $1$.
    So, their corresponding eigenspaces are one dimensional.
    So, $\nicefrac{1}{2}(1, 1, 1, 1)^{\tt{T}}$ and
    $\nicefrac{1}{2}(1, -1, -1, 1)^{\tt{T}}$ must be column
    1 and column 4 of $H'$ respectively (upto a factor of $\pm 1$).
    But since $|x| = |y|$, it follows that
    the duplicate eigenvalue $(\lambda_{1}\mu_{2}) = \kappa^{2}x^{2} = 
    \kappa^{2}y^{2} = (\lambda_{2}\mu_{1})$ has multiplicity
     $2$.
    So, the corresponding eigenspace has dimension $2$.
    So, all unit vectors that are orthogonal to both
    $\nicefrac{1}{2}(1, 1, 1, 1)^{\tt{T}}$, and 
    $\nicefrac{1}{2}(1, -1, -1, 1)^{\tt{T}}$
    are unit eigenvectors of $M^{2}$ with the eigenvalue
    $(\lambda_{1}\mu_{2}) = (\lambda_{2}\mu_{1})$.
    Therefore, any unit vectors $\mathbf{v}_{1},
    \mathbf{v}_{2}$ that are orthogonal to each other,
    and also to both $(1, 1, 1, 1)^{\tt{T}}$ and
    $(1, -1, -1, 1)^{\tt{T}}$ could be columns of $H'$.
    We may now assume that $\mathbf{v}_{1} = 
    (v_{11}, v_{12}, v_{13}, v_{14})^{\tt{T}}$
    and $\mathbf{v}_{2} = (v_{21}, v_{22}, v_{23}, v_{24})^{\tt{T}}$.
    Since $\mathbf{v}_{1}$ and $\mathbf{v}_{2}$ are
    both orthogonal to $(1, 1, 1, 1)^{\tt{T}}$ and
    $(1, -1, -1, 1)^{\tt{T}}$,
    it must be the case that
    $$v_{11} + v_{12} + v_{13} + v_{14} = 0, \quad
    v_{11} - v_{12} - v_{13} + v_{14} = 0 \implies 
    (v_{11} + v_{14}) = (v_{12} + v_{13}) = 0,$$
    $$v_{21} + v_{22} + v_{23} + v_{24} = 0, \quad
    v_{21} - v_{22} - v_{23} + v_{24} = 0 \implies 
    (v_{21} + v_{24}) = (v_{22} + v_{23}) = 0.$$
    We may let $u = v_{11}$, and $v = v_{13}$.
    This implies
    that $\mathbf{v}_{1} = (u, -v, v, -u)^{\tt{T}}$ is one of the
    eigenvectors of $M^{2}$ with the eigenvalue
    $(\lambda_{1}\mu_{2})$.
    Since eigenvectors are preserved under scalar multiplication,
    we may assume that $u^{2} + v^{2} = 2$,
    and then $\mathbf{v}_{1} = \frac{1}{2}(u, -v, v, -u)^{\tt{T}}$
    is a unit vector.
    Since $\mathbf{v}_{1}$ and $\mathbf{v}_{2}$ are also
    orthogonal to each other, this implies that
    $u(v_{21} - v_{24}) + v(v_{22} - v_{23}) = 0$.
    Since $v_{24} = -v_{21}$, and $v_{23} = -v_{22}$,
    this means that
    $u(v_{21}) - v(v_{22}) = 0$.
    It follows that we have
    $\mathbf{v}_{2} = \frac{1}{2}(v, u, -u, -v)^{\tt{T}}$, 
    as the other unit eigenvector of $(\lambda_{2}\mu_{1})
    = (\lambda_{1}\mu_{2})$.

    Summing up, if $|x| = |y|$, the eigenvectors of $M^{2}$
    are $\frac{1}{2}(1, 1, 1, 1)^{\tt{T}}$,
    $\frac{1}{2}(u, -v, v, u)^{\tt{T}}$,
    $\frac{1}{2}(v, u, -u, -v)^{\tt{T}}$, and
    $\frac{1}{2}(1, -1, -1, 1)^{\tt{T}}$, for all 
    $u, v \in \mathbb{R}$ such that $u^2 + v^2 =2$.
    In other words, if $M^{2} = (H')(D_{1} \otimes D_{2})(H')^{\tt{T}}$,
    we see that $H'$ must be of the form
    of the matrix $K$ in the statement of the lemma.
    Since $M \cdot M^{2} = M^{3} = M^{2} \cdot M$,
    we know that $M$ and $M^{2}$ must be simultaneously
    diagonalizable.
    So, this implies that
    when $|x| = |y|$,
    $M = KDK^{\tt{T}}$ for some
    orthogonal matrix $K$ of the stated form
    in the statement of this lemma, and
    $D = \kappa \cdot \text{diag}(1, x, y, z)$, where
    $z = \pm xy$.
    
    If $x = y$, then this means that the
    multiplicity of the eigenvalues  $\kappa  x
    = \kappa  y$ is once again $2$.
    This implies that any nonzero linear combination of
    $\mathbf{v}_{1}$ and $\mathbf{v}_{2}$ is
    also  an eigenvector of $M$. In particular,
    that implies that $\left(\frac{u + v}{2}\right) \mathbf{v}_{1}
    + \left(\frac{v - u}{2}\right) \mathbf{v}_{2}
    =  \frac{1}{2} (1, -1, 1, -1)^{\tt{T}}$, 
    and $\left(\frac{u - v}{2}\right) \mathbf{v}_{1}
    + \left(\frac{u + v}{2}\right) \mathbf{v}_{2}
    = \frac{1}{2} (1, 1, -1, -1)^{\tt{T}}$
    are eigenvectors of $M$ as well, with the eigenvalue
    $\kappa  x = \kappa y$.
    This means that $M = (H \otimes H) D (H \otimes H)^{\tt{T}}$.
    If $z = xy$, that would again imply 
    that $\varrho_{\tt{tensor}}(M) = 0$.
    So, we conclude that $z = -xy = -x^{2}$, in which
    case, $M$ is of the form in the statement of the lemma.

    Finally, we consider the case where $x = -y$. 
    We have $z = \pm x^{2}$.
    If $z = x^{2}$, then $M$ is already of the form in the
    statement of the lemma and we are done.
    We will now assume that $z = -x^{2}$. For a contradiction assume $|u| = |v|$.
    As $u^2 + v^2 =2$, we have $|u|=|v|=1$.
    We may assume that $u = 1$,
    since unit eigenvectors are preserved when multiplied
    by $\pm 1$.
    If $v = 1$ as well, we find that
    $K = H \otimes H$.
    So, $M = (HD'_{1} H^{\tt{T}}) \otimes
    (H D'_{2}H^{\tt{T}})$, where
    $D'_{1} = \sqrt{\kappa} \cdot \text{diag}(1, -x)$,
    and $D'_{2} = \sqrt{\kappa} \cdot \text{diag}(1, x)$,
    which implies once again that
    $\varrho_{\tt{tensor}}(M) = 0$,
    contradicting our assumption about $M$.
    Finally, if $v = -1$, we will consider the matrix
    $M' \in \text{Sym}_{4}^{\tt{F}}(\mathbb{R}_{> 0})$
    such that $(M')_{ij} = M_{\sigma(i)\sigma(j)}$,
    where $\sigma \in S_{4}$ is such that
    $\sigma(1) = 1, \sigma(2) = 3, \sigma(3) = 2$,
    and $\sigma(4) = 4$.
    We note that $M$ and $M'$ are isomorphic
    to each other.
    Since $M = KDK^{\tt{T}}$, we see that
    $M' = (K')D'(K')^{\tt{T}}$, where $K'$ is obtained by
    applying $\sigma$ to the rows of $K$, and multiplying the 3rd column
    by $-1$, and $D'$ is obtained by
    multiplying both the 3rd row and column of $D$ by $-1$.
    But then, we see that $K' = H  \otimes H$, and
    $D' = D = \kappa \cdot 
    \text{diag}(1, x, -x, -x^2)$. Hence,
    $M' = (HD'_{1} H^{\tt{T}}) \otimes
    (H D'_{2}H^{\tt{T}})$, where
    $D'_{1} = \sqrt{\kappa} \cdot \text{diag}(1, -x)$,
    and $D'_{2} = \sqrt{\kappa} \cdot \text{diag}(1, x)$.
    This implies that $\varrho_{\tt{tensor}}(M') = 0$,
    which contradicts our assumption that
    $\rho_{\tt{tensor}}(M) \neq 0$.
    So, we conclude that when $x = -y$ and $z = -x^{2}$ we have
    $|u| \neq |v|$.
    
    In all cases, we have proved 
    that either $\PlEVAL(M)$ is $\#$P-hard, or
    $M$, the diagonalizing $K$, and the diagonal matrix $D$, 
    take the form  in the statement of the lemma.
\end{proof}

\begin{lemma}\label{lemma: fullRankVarRho}
    Let $M \in \text{Sym}_{4}^{\tt{pd}}(\mathbb{R}_{> 0})$,
    such that $M = A \otimes B$ for some
    $A, B \in \text{Sym}_{2}^{\tt{pd}}(\mathbb{R}_{> 0})$.
    For any $\sigma \in S_{4}$, if 
    $\{\sigma(1), \sigma(4)\} = \{1, 2\}$, or
    $\{3, 4\}$, or $\{1, 3\}$, or $\{2, 4\}$, then
    $\varrho_{\tt{tensor}}(M^{\sigma}) \neq 0$,
    where $(M^{\sigma})_{ij} = M_{\sigma(i)\sigma(j)}$
    for all $i, j \in [4]$.
\end{lemma}
\begin{proof}
    Since $M = A \otimes B$, we have  $\varrho_{\tt{tensor}}(M) = 0$, where
    $$M = \begin{pmatrix}
        A_{11}B_{11} & A_{11}B_{12} & A_{12}B_{11} & A_{12}B_{12}\\
        A_{11}B_{12} & A_{11}B_{22} & A_{12}B_{12} & A_{12}B_{22}\\
        A_{12}B_{11} & A_{12}B_{12} & A_{22}B_{11} & A_{22}B_{12}\\
        A_{12}B_{12} & A_{12}B_{22} & A_{22}B_{12} & A_{22}B_{22}\\
    \end{pmatrix}.$$
    We note from the definition of $\varrho_{\tt{tensor}}$
    in \cref{lemma: tensorPolynomial} that
    $\varrho_{\tt{tensor}}(N) = 0$ implies that
    $N_{14} = N_{23}$, and $N_{i1}N_{i4} = N_{i2}N_{i3}$
    for all $i \in [4]$.
    
    We will first consider $\sigma \in S_{4}$ such that
    $\{\sigma(1), \sigma(4)\} = \{1, 2\}$, or
    $\{3, 4\}$.
    So, $\varrho_{\tt{tensor}}(M^{\sigma}) = 0$ implies
    that $M_{\sigma(1)\sigma(4)} = M_{\sigma(2)\sigma(3)}$.
    This implies that $M_{12} = (A_{11}B_{12}) = 
    (A_{22}B_{12}) = M_{34}$,
    which is only possible if $A_{11} = A_{22}$, since
    $B_{12} > 0$.
    We also note that $\varrho_{\tt{tensor}}(M^{\sigma}) = 0$ implies
    that $M_{1\sigma(1)}M_{1\sigma(4)} 
    = M_{1\sigma(2)}M_{1\sigma(3)}$.
    This implies that
    $M_{11}M_{12} = M_{13}M_{14}$.
    In other words, it must be the case that
    $(A_{11}B_{11})(A_{11}B_{12}) = 
    (A_{12}B_{11})(A_{12}B_{12})$.
    Since $A_{11}, A_{12}, B_{11}, B_{12} > 0$,
    this is only possible if $A_{11} = A_{12}$.
    But if $A_{11} = A_{12} = A_{22}$, that
    means $A$ is a rank $1$ matrix, which
    contradicts the fact that
    $A \in \text{Sym}_{2}^{\tt{pd}}(\mathbb{R}_{> 0})$.
    So, we conclude that if $\{\sigma(1), \sigma(4)\} = \{1, 2\}$,
    or if $\{\sigma(1), \sigma(4)\} = \{3, 4\}$
    then $\varrho_{\tt{tensor}}(M^{\sigma}) \neq 0$.

    For $\sigma \in S_{4}$ such that
    $\{\sigma(1), \sigma(4)\} = \{1, 3\}$, or
    $\{2, 4\}$, let $\tau \in S_{4}$ be the transposition $(23)$.
    Then $M^{\tau} = (A \otimes B)^{\tau} = B \otimes A$.
    Now  $M_{\sigma(1)\sigma(4)} = M_{\sigma(2)\sigma(3)}$ becomes
    $M_{13}  = M_{24}$, i.e., $(M^{\tau})_{12} = (M^{\tau})_{34}$.
    Similarly, $M_{1\sigma(1)}M_{1\sigma(4)} 
    = M_{1\sigma(2)}M_{1\sigma(3)}$ is
    $M_{11}M_{13} = M_{12}M_{14}$, i.e., 
    $(M^{\tau})_{11}(M^{\tau})_{12} = (M^{\tau})_{13}(M^{\tau})_{14}$.
    Since $M^{\tau}$
    is also a tensor product, only switching the roles of $A, B  
    \in \text{Sym}_{2}^{\tt{pd}}(\mathbb{R}_{> 0})$ in $M$,
    $\varrho_{\tt{tensor}}(M^{\sigma}) = 0$ leads to the same
    contradiction as in the previous paragraph
    (with $M^{\tau}$ in place of $M$).
    So, $\varrho_{\tt{tensor}}(M^{\sigma}) \neq 0$
    when $\{\sigma(1), \sigma(4)\} = \{1, 3\}$,
    or 
    $\{2, 4\}$.

  
\end{proof}

\begin{lemma}\label{lemma: fullRankHardnessCase1}
    Let $M \in \text{Sym}_{4}^{\tt{F}}(\mathbb{R}_{> 0})$
    be a matrix of the form $M = KDK^{\tt{T}}$, where the matrix $K$
    and $D$ are as in
    \cref{lemma: fullRankExtension}, such that
    $|x| \neq |y|$, $u = v = 1$, and $z = -xy$.
    Then, $\PlEVAL(M)$ is $\#$P-hard.
\end{lemma}
\begin{proof}
    We will first define
    the $\text{Sym}_{4}(\mathbb{R})$-polynomial $\xi$ such that
    $\xi: N \mapsto (N_{11}N_{14} - N_{12}N_{13})$.
    We note from our choice of $M$, that
    $M = KDK^{\tt{T}}$, where
    $$K = \frac{1}{2}
    \begin{pmatrix}
        1 & 1 & 1 & 1\\
        1 & -1 & 1 & -1\\
        1 & 1 & -1 & -1\\
        1 & -1 & -1 & 1\\
    \end{pmatrix}$$
    and
    $D = \kappa \cdot \text{diag}(1, x, y, -xy)$
    for some $\kappa \in \mathbb{R}_{> 0}$, and $x, y \in \mathbb{R}$
    such that $0 < |x|, |y| < 1$, and $|x| \neq |y|$.
    We note that for any odd $n \geq 1$,
    $M^{n} = KD^{n}K^{\tt{T}}$, where
    $D^{n} = (\kappa)^{n} \cdot \text{diag}(1, x^{n}, y^{n}, -(xy)^{n})$.
    We note that
    \begin{align*}
        (M^{n})_{11} = (M^{n})_{22} = (M^{n})_{33} = (M^{n})_{44}
        &=  \kappa^{n} \cdot \left(
        \frac{1 + x^{n} + y^{n} - (xy)^{n}}{4}\right),\\
        (M^{n})_{12} = (M^{n})_{34}
        &=  \kappa^{n} \cdot \left(
        \frac{1 - x^{n} + y^{n} + (xy)^{n}}{4}\right),\\
        (M^{n})_{13} = (M^{n})_{24}
        &=  \kappa^{n} \cdot \left(
        \frac{1 + x^{n} - y^{n} + (xy)^{n}}{4}\right),\\
        (M^{n})_{14} = (M^{n})_{23}
        &=  \kappa^{n} \cdot \left(
        \frac{1 - x^{n} - y^{n} - (xy)^{n}}{4}\right).
    \end{align*}

    We will now let $N_{n} = (T_{2}((\nicefrac{4}{\kappa^{n}})M^{n}))^{2}$
    for any odd $n \geq 1$.
    We note that
    $$(N_{n})_{11} = (1 + x^{n} + y^{n} - (xy)^{n})^{4} 
    + (1 - x^{n} + y^{n} + (xy)^{n})^{4}
    + (1 + x^{n} - y^{n} + (xy)^{n})^{4}
    + (1 - x^{n} - y^{n} - (xy)^{n})^{4},$$
    $$(N_{n})_{12} = 2((1 + x^{n} + y^{n} - (xy)^{n})
    (1 - x^{n} + y^{n} + (xy)^{n}))^{2} 
    + 2((1 + x^{n} - y^{n} + (xy)^{n})
    (1 - x^{n} - y^{n} - (xy)^{n}))^{2},$$
    $$(N_{n})_{13} = 2((1 + x^{n} + y^{n} - (xy)^{n})
    (1 + x^{n} - y^{n} + (xy)^{n}))^{2} 
    + 2((1 - x^{n} + y^{n} + (xy)^{n})
    (1 - x^{n} - y^{n} - (xy)^{n}))^{2},$$
    $$(N_{n})_{14} = 2((1 + x^{n} + y^{n} - (xy)^{n})
    (1 - x^{n} - y^{n} - (xy)^{n}))^{2} 
    + 2((1 - x^{n} + y^{n} + (xy)^{n})
    (1 + x^{n} - y^{n} + (xy)^{n}))^{2}.$$
    After simplifying these equations, and computing
    $\xi(N_{n})$, we find that
    $$\xi(N_{n}) = -1024x^{2n}y^{2n}\left(
    1 - 2y^{2n} - 2x^{2n} + y^{4n} + x^{4n}
    + 4x^{2n}y^{2n} - 2x^{2n}y^{4n} - 2x^{4n}y^{2n}
    + x^{4n}y^{4n}\right).$$
    Since $|x|, |y| < 1$, we know that for large enough $n$,
    $$1 > |2y^{2n}| + |2x^{2n}| + |y^{4n}| + |x^{4n}|
    + |4x^{2n}y^{2n}| + |x^{2n}y^{4n}| + |2x^{4n}y^{2n}|
    + |x^{4n}y^{4n}|.$$
    Moreover, we know that $|x|, |y| \neq 0$.
    So, we see that there exists some odd $n^{*}$, such that
    $\xi(N_{n^{*}}) \neq 0$.
    Since $\xi(N_{n^{*}}) = \left(\nicefrac{4}{\kappa^{n^{*}}}\right)^{4}
    \xi((T_{2}(M^{n^{*}}))^{2})$, this implies that
    $\xi((T_{2}(M^{n^{*}}))^{2}) \neq 0$.

    We will now let $M' = M^{n^{*}}$.
    Since $M \in \text{Sym}_{4}^{\tt{F}}(\mathbb{R}_{> 0})$,
    it follows that $M' \in \text{Sym}_{4}^{\tt{F}}(\mathbb{R}_{> 0})$
    as well.
    From the construction of $\varrho_{\tt{tensor}}$ in
    \cref{lemma: tensorPolynomial}, we see that
    $\xi(N) \neq 0$ implies that $\varrho_{\tt{tensor}}(N^{\sigma}) \neq 0$
    for all $\sigma \in S_{4}$ such that 
    $\{\sigma(1), \sigma(4)\} = \{1, 4\}$, or $\{2, 3\}$.
    We will now define
    $$T_{1} = \Big\{\sigma \in S_{4}: 
    \{\sigma(1), \sigma(4)\} \in 
    \big\{ \{1, 4\}, \{2, 3\} \big\} \Big\},$$
    and define the $\text{Sym}_{4}(\mathbb{R})$-polynomial $\zeta_{1}$
    such that $\zeta_{1}: N \mapsto \prod_{\sigma \in T_{1}}
    \varrho_{\tt{tensor}}((N^{2})^{\sigma})$.
    By construction, we see that
    $\zeta_{1}(T_{2}(M')) \neq 0$.
    We also note that $(M')^{2} = M^{2n^{*}} = (M^{2})^{n^{*}}$.
    By our choice of $M$,
    we know that $\varrho_{\tt{tensor}}(M^{2}) = 0$.
    So, there exist $A, B \in \text{Sym}_{2}^{\tt{pd}}(\mathbb{R}_{> 0})$
    such that $M^{2} = A \otimes B$.
    But this means that $M^{2n^{*}} = A^{n^{*}} \otimes B^{n^{*}}$,
    which means that $\varrho_{\tt{tensor}}((M')^{2}) = 0$
    as well.
    So, from \cref{lemma: fullRankVarRho}, we
    note that for any $\sigma \in S_{4}$,
    such that $\{\sigma(1), \sigma(4)\} = 
    \{1, 2\}$, or $\{3, 4\}$, or $\{1, 3\}$, or
    $\{2, 4\}$, $\varrho_{\tt{tensor}}(((M')^{2})^{\sigma}) \neq 0$.
    We will now let
    $$T_{2} = S_{4} \setminus T_{1} = \Big\{\sigma \in S_{4}: 
    \{\sigma(1), \sigma(4)\} \in 
    \big\{ \{1, 2\}, \{3, 4\}, \{1, 3\}, \{2, 4\} \big\} \Big\},$$
    and define the $\text{Sym}_{4}(\mathbb{R})$-polynomial
    $\zeta_{2}$ such that $\zeta_{2}: 
    N \mapsto \prod_{\sigma \in T_{2}}
    \varrho_{\tt{tensor}}((N^{2})^{\sigma})$.
    We have just seen that $\zeta_{2}(M') \neq 0$.
    
    Since $M' \in \text{Sym}_{4}^{\tt{F}}(\mathbb{R}_{> 0})$,
    \cref{lemma: thickeningWorks} lets us find
    $M'' \in \mathfrak{R}(M', \{\zeta_{1}, \zeta_{2}\}) \cap
    \text{Sym}_{4}^{\tt{F}}(\mathbb{R}_{> 0})$.
    Since $\zeta_{1}(M'') \neq 0$, and $\zeta_{2}(M'') \neq 0$,
    it follows that $\rho_{\tt{tensor}}((M'')^{2}) \neq 0$.
    Since $(M'')^{2} \in \text{Sym}_{4}^{\tt{pd}}(\mathbb{R}_{> 0})$,
    \cref{theorem: TensorOrHard} implies that 
    $\PlEVAL((M'')^{2}) \leq \PlEVAL(M'') \leq \PlEVAL(M')$ 
    is $\#$P-hard.
    Since $M' = M^{n^{*}}$ for some integer $n^{*} \geq 1$,
    we also see that $\PlEVAL(M') \leq \PlEVAL(M)$.
    This proves that $\PlEVAL(M)$ is $\#$P-hard.
\end{proof}

\begin{remark*}
    We should note that in
    \cref{lemma: fullRankHardnessCase1}, we do not
    require that
    $\rho_{\tt{tensor}}(M) \neq 0$.
    We only require that $M$ be of the form $M = KDK^{\tt{T}}$,
    where $K$ and $D$ are as in \cref{lemma: fullRankExtension}
    (with the stipulation on $x, y, z, u, v$ in the statement of \cref{lemma: fullRankHardnessCase1}).
    In this case, it trivially follows that
    $\varrho_{\tt{tensor}}(M^{2}) = 0$, which
    is the only property of $M$ that is used in the proof.
    This will also be true for
    \cref{lemma: fullRankHardnessCase2}, 
    \cref{lemma: fullRankHardnessCase3}, and
    \cref{lemma: fullRankHardnessCase4} below.
\end{remark*}

\begin{lemma}\label{lemma: fullRankHardnessCase2}
    Let $M \in \text{Sym}_{4}^{\tt{F}}(\mathbb{R}_{> 0})$
    be a matrix of the form $M = KDK^{\tt{T}}$, where the matrix $K$
    and $D$ are as in
    \cref{lemma: fullRankExtension}, such that
    $x = y$, and $z = -x^{2}$.
    Then, $\PlEVAL(M)$ is $\#$P-hard.
\end{lemma}
\begin{proof}
    We will define $\xi$ such that
    $\xi: N \mapsto (N_{11}N_{14} - N_{12}N_{13})$, as in
    \cref{lemma: fullRankHardnessCase1}.
    We note from \cref{lemma: fullRankExtension}, that
    $M = KDK^{\tt{T}}$, where
    $$K = \frac{1}{2}
    \begin{pmatrix}
        1 & u & v & 1\\
        1 & -v & u & -1\\
        1 & v & -u & -1\\
        1 & -u & -v & 1\\
    \end{pmatrix}$$
    for some $u, v \in \mathbb{R}$ with $u^2+v^2 =2$, and
    $D = \kappa \cdot \text{diag}(1, x, x, -x^{2})$
    for some $\kappa \in \mathbb{R}_{> 0}$, and $x \in \mathbb{R}$
    such that $|x| < 1$.
    For odd $n \geq 1$, we note that $M^{n} = KD^{n}K^{\tt{T}}$,
    where $D^{n} = (\kappa)^{n} \cdot (1, x^{n}, x^{n}, - x^{2n})$.
    So, we note that (since $u^{2} + v^{2} = 2$),
    \begin{align*}
        (M^{n})_{11} = (M^{n})_{22} = (M^{n})_{33} = (M^{n})_{44}
        &=  \kappa^{n} \cdot \left(
        \frac{1 + (u^{2} + v^{2})x^{n} - x^{2n}}{4}\right) 
        = \kappa^{n} \cdot \left(\frac{1 + 2x^{n} - x^{2n}}{4} \right),\\
        (M^{n})_{12} = (M^{n})_{34} = (M^{n})_{13} = (M^{n})_{24}
        &=  \kappa^{n} \cdot \left(
        \frac{1 + (uv)(-x + x) + x^{2n}}{4}\right)
        = \kappa^{n} \cdot \left(\frac{1 + x^{2n}}{4}\right),\\
        (M^{n})_{14} = (M^{n})_{23}
        &=  \kappa^{n} \cdot \left(
        \frac{1 - (u^{2} + v^{2})x - x^{2n}}{4}\right)
        = \kappa^{n} \cdot \left(\frac{1 - 2x^{n} - x^{2n}}{4} \right).
    \end{align*}
    We will let 
    $N_{n} = (T_{2}((\nicefrac{4}{\kappa^{n}})M^{n}))^{2}$
    for any odd $n \geq 1$.
    We note that
    \begin{align*}
        (N_{n})_{11}
        &= (1 + 2x^{n} - x^{2n})^{4} + 2(1 + x^{2n})^{4} 
        + (1 - 2x^{n} - x^{2n})^{4},\\
        (N_{n})_{12} = (N_{n})_{13}
        &= 2((1 + 2x^{n} - x^{2n})(1 + x^{2n}))^{2} + 
        2((1 - 2x^{n} - x^{2n})(1 + x^{2n}))^{2},\\
        (N_{n})_{14}
        &= 2((1 + 2x^{n} - x^{2n})(1 - 2x^{n} - x^{2n}))^{2}
        + 2(1 + x^{2n})^{4}.
    \end{align*}
    After simplifying and computing $\xi(N_{n})$, we find that
    $$\xi(N_{n}) = (-1024)x^{4n}\left(1 - x^{2n}\right)^{4}.$$
    Since $0 < |x| < 1$, we see that
    there exists some odd $n^{*} \geq 1$ such that
    $\xi(N_{n^{*}}) \neq 0$.
    Since $\xi(N_{n^{*}}) = \left(\nicefrac{4}{\kappa^{n^{*}}}\right)^{4}
    \xi((T_{2}(M^{n^{*}}))^{2})$, this implies that
    $\xi((T_{2}(M^{n^{*}}))^{2}) \neq 0$.

    We will now let $M' = M^{n^{*}}$.
    Now, by the same argument as in the proof of
    \cref{lemma: fullRankHardnessCase1}, we can now find some
    $M'' \in \text{Sym}_{4}^{\tt{F}}(\mathbb{R}_{> 0})$, such that
    $\PlEVAL(M'') \leq \PlEVAL(M')$, and
    $\rho_{\tt{tensor}}((M'')^{2}) \neq 0$.
    So, \cref{theorem: TensorOrHard} once again lets us prove
    that $\PlEVAL((M'')^{2}) \leq \PlEVAL(M') \leq \PlEVAL(M)$
    is $\#$P-hard.
\end{proof}

\begin{lemma}\label{lemma: fullRankHardnessCase3}
    Let $M \in \text{Sym}_{4}^{\tt{F}}(\mathbb{R}_{> 0})$
    be a matrix of the form $M = KDK^{\tt{T}}$, where the matrix $K$
    and $D$ are as in
    \cref{lemma: fullRankExtension}, such that
    $x = -y$, $z = -x^{2}$, and $|u| \neq |v|$.
    Then, $\PlEVAL(M)$ is $\#$P-hard.
\end{lemma}
\begin{proof}
    We will once again define $\xi$ such that
    $\xi: N \mapsto (N_{11}N_{14} - N_{12}N_{13})$.
    We note from \cref{lemma: fullRankExtension}, that
    $M = KDK^{\tt{T}}$, where
    $$K = \frac{1}{2}
    \begin{pmatrix}
        1 & u & v & 1\\
        1 & -v & u & -1\\
        1 & v & -u & -1\\
        1 & -u & -v & 1\\
    \end{pmatrix}$$
    for some $u, v \in \mathbb{R}$ with $u^2+v^2 =2$, $|u| \neq |v|$, 
    and $D = \kappa \cdot \text{diag}(1, x, -x, -x^{2})$
    for some $\kappa \in \mathbb{R}_{> 0}$, and $x, y, z \in \mathbb{R}$
    such that $|x|, |y|, |z| < 1$.

    We note that for odd $n \geq 1$,
    $M^{n} = KD^{n}K^{\tt{T}}$, where
    $D^{n} = (\kappa)^{n} \cdot \text{diag}(1, x^{n} -x^{n}, -x^{2n})$.
    So, we note that
    \begin{align*}
        (M^{n})_{11} = (M^{n})_{22} = (M^{n})_{33} = (M^{n})_{44}
        &=  \kappa^{n} \cdot \left(
        \frac{1 + (u^{2} - v^{2})x^{n} - x^{2n}}{4} \right),\\
        (M^{n})_{12} = (M^{n})_{34}
        &=  \kappa^{n} \cdot \left(
        \frac{1 -2uvx^{n} + x^{2n}}{4}\right),\\
        (M^{n})_{13} = (M^{n})_{24}
        &=  \kappa^{n} \cdot \left(
        \frac{1 + 2uvx^{n} + x^{2n}}{4} \right),\\
        (M^{n})_{14} = (M^{n})_{23}
        &=  \kappa^{n} \cdot \left(
        \frac{1 + (v^{2} - u^{2})x^{n} - x^{2n}}{4} \right).
    \end{align*}
    We will now let 
    $N_{n} = (T_{2}((\nicefrac{4}{\kappa^{n}})M^{n}))^{2}$
    for any odd $n \geq 1$.
    We note that
    $$(N_{n})_{11}
    = (1 + (u^{2} - v^{2})x^{n} - x^{2n})^{4}
    + (1 - 2uvx^{n} + x^{2n})^{4} + (1 + 2uvx^{n} + x^{2n})^{4}
    + (1 + (v^{2} - u^{2})x^{n} - x^{2n})^{4},$$
    $$(N_{n})_{12}
    = 2((1 + (u^{2} - v^{2})x^{n} - x^{2n})(1 - 2uvx^{n} + x^{2n}))^{2}
    + 2((1 + (v^{2} - u^{2})x^{n} - x^{2n})(1 + 2uvx^{n} + x^{2n}))^{2},$$
    $$(N_{n})_{13}
    = 2((1 + (u^{2} - v^{2})x^{n} - x^{2n})(1 + 2uvx^{n} + x^{2n}))^{2}
    + 2((1 + (v^{2} - u^{2})x^{n} - x^{2n})(1 - 2uvx^{n} + x^{2n}))^{2},$$
    $$(N_{n})_{14}
    = 2((1 + (u^{2} - v^{2})x^{n} - x^{2n})
    (1 + (v^{2} - u^{2})x^{n} - x^{2n}))^{2}
    + 2((1 - 2uvx^{n} + x^{2n})(1 + 2uvx^{n} + x^{2n}))^{2}.$$
    If we compute $\xi(N_{n})$, it can be verified that
    $$\xi(N_{n}) = f_{4}(u, v)x^{4n} + f_{6}(u, v)x^{6n} + 
    f_{8}(u, v)x^{8n} + f_{10}(u, v)x^{10n} + f_{12}(u, v)x^{12n},$$
    where
    $$f_{4}(u, v) = f_{12}(u, v) = 
    -48u^8 - 48v^8 + 576u^6v^2 + 576u^2v^6 -
    128u^4 - 128v^4 + 768u^2v^2 - 1824u^4v^4 + 256,$$
    \begin{multline*}
        f_{6}(u, v) = f_{10}(u, v) = 
        16u^{12} - 160u^{10}v^2 + 240u^8v^4 + 832u^6v^6 + 
        240u^4v^8 - 160u^2v^{10} + 16v^{12} \\
        + 128u^8 -  512u^6v^2 - 1280u^4v^4 - 512u^2v^6 + 
        128v^8 + 256u^4 + 512u^2v^2 + 256v^4,
    \end{multline*}
    \begin{multline*}
        f_{8}(u, v) = 4u^{16} - 32u^{14}v^2 - 16u^{12}v^4
        + 288u^{10}v^6 + 536u^8v^8 + 288u^6v^{10} 
        - 16u^4v^{12} - 32u^2v^{14} + 4v^{16}\\
        - 32u^{12} + 64u^{10}v^2 + 544u^8v^4 
        + 896u^6v^6 + 544u^4v^8 + 64u^2v^{10}
        - 32v^{12} - 288u^8\\ 
        + 384u^6v^2 - 4800u^4v^4 + 384u^2v^6 - 288v^8
        - 256u^4 + 1536u^2v^2 - 256v^4 + 512.
    \end{multline*}
    While these equations appear a bit intimidating,
    we will focus on $f_{4}(u, v)$. 
    We note that $u^{2} + v^{2} = 2$.
    So, $f_{4}(u, v) = 0$ if and only if
    \begin{multline*}
        -16\Big(3u^{8} + 3(2 - u^{2})^{4} - 36u^{6}(2 - u^{2}) 
        - 36u^{2}(2 - u^{2})^{3} + 8u^{4} + 8(2 - u^{2})^{2}\\
        - 48u^{2}(2 - u^{2}) + 114u^{2}(2 - u^{2})^{2}
        - 16\Big) = 0.
    \end{multline*}
    After simplifying, we find that $f_{4}(u, v) = 0$ if and only if
    $$-1024(3u^8 - 12u^6 + 16u^4 - 8u^2 + 1) = 0.$$
    This is a degree $4$ polynomial in $u^{2}$,
    and it can be verified that
    $$3u^8 - 12u^6 + 16u^4 - 8u^2 + 1 = 3(u^{2} - 1)^{2}
    \left(u^{2} - 1 - \frac{\sqrt{6}}{3}\right)
    \left(u^{2} - 1 + \frac{\sqrt{6}}{3}\right).$$
    Since we know that $|u| \neq |v|$, we can 
    see that it is not possible that $u^{2} = 1$.
    So, $f_{4}(u, v) = 0$ if and only if
    $u^{2} = (1 - \nicefrac{\sqrt{6}}{3})$ and
    $v^{2} = (1 + \nicefrac{\sqrt{6}}{3})$, or if
    $u^{2} = (1 + \nicefrac{\sqrt{6}}{3})$ and
    $v^{2} = (1 - \nicefrac{\sqrt{6}}{3})$.
    But in either case, for these values of $u^2$ and $v^2$, it is seen that
    $f_{6}(u, v) = (\nicefrac{16384}{9}) \neq 0$.
    Summing up, we find that
    if $f_{4}(u, v) = 0$, then
    $f_{6}(u, v) \neq 0$.
    If $f_{4}(u, v) \neq 0$, then
    we note that since $|x| < 1$, 
    for large enough $n \geq 1$,
    $|f_{4}(u, v)| > |f_{6}(u, v)|x^{2n} + |f_{8}(u, v)|x^{4n}
    + |f_{10}(u, v)|x^{6n} + |f_{12}(u, v)|x^{8n}.$
    So, there exists some odd $n^{*} \geq 1$
    such that $\xi(N_{n^{*}}) \neq 0$.
    On the other hand, if $f_{4}(u, v) = f_{12}(u, v) = 0$,
    we note that for  large enough $n \geq 1$,
    $|f_{6}(u, v)| > |f_{8}(u, v)|x^{2n} + |f_{10}(u, v)|x^{4n}$.
    So, once again there exists some odd $n^{*} \geq 1$
    such that $\xi(N_{n^{*}}) \neq 0$.
    Since $\xi(N_{n^{*}}) = \left(\nicefrac{4}{\kappa^{n^{*}}}\right)^{4}
    \xi((T_{2}(M^{n^{*}}))^{2})$, this implies that
    $\xi((T_{2}(M^{n^{*}}))^{2}) \neq 0$.

    We will now let $M' = M^{n^{*}}$.
    Now, by the same argument as in the proof of
    \cref{lemma: fullRankHardnessCase1}, we can now find some
    $M'' \in \text{Sym}_{4}^{\tt{F}}(\mathbb{R}_{> 0})$, such that
    $\PlEVAL(M'') \leq \PlEVAL(M')$, and
    $\rho_{\tt{tensor}}((M'')^{2}) \neq 0$.
    So, \cref{theorem: TensorOrHard} once again lets us prove
    that $\PlEVAL((M'')^{2}) \leq \PlEVAL(M') \leq \PlEVAL(M)$
    is $\#$P-hard.
\end{proof}

\begin{lemma}\label{lemma: fullRankHardnessCase4}
    Let $M \in \text{Sym}_{4}^{\tt{F}}(\mathbb{R}_{> 0})$
    be a matrix of the form $M = KDK^{\tt{T}}$, where the matrix $K$
    and $D$ are as in
    \cref{lemma: fullRankExtension}, such that
    $x = -y$, and $z = x^{2}$.
    Then, $\PlEVAL(M)$ is $\#$P-hard.
\end{lemma}
\begin{proof}
    We note from \cref{lemma: fullRankExtension}, that
    $M = KDK^{\tt{T}}$, where
    $$K = \frac{1}{2}
    \begin{pmatrix}
        1 & u & v & 1\\
        1 & -v & u & -1\\
        1 & v & -u & -1\\
        1 & -u & -v & 1\\
    \end{pmatrix}$$
    for some $u, v \in \mathbb{R}$ with $u^2+v^2 =2$, and
    $D = \kappa \cdot \text{diag}(1, x, -x, x^{2})$
    for some $\kappa \in \mathbb{R}_{> 0}$, and $x \in \mathbb{R}$
    such that $|x| < 1$.
    
    We will first consider the case where $u = 0$.
    Since $u^{2} + v^{2} = 2$,
    this implies that $|v| = \sqrt{2}$.
    Since eigenvectors are equivalent upto scaling,
    we may assume that in fact, $v = \sqrt{2}$.
    In this case, we find that since
    $M = KDK^{\tt{T}}$,
    \begin{align*}
        M
        &= \frac{1}{2}\begin{pmatrix}
            (1 - 2x + x^{2}) & (1 - x^{2}) & (1 - x^{2}) & (1 + 2x + x^{2})\\
            (1 - x^{2}) & (1 + 2x + x^{2}) & (1 - 2x + x^{2}) & (1 - x^{2})\\
            (1 - x^{2}) & (1 - 2x + x^{2}) & (1 + 2x + x^{2}) & (1 - x^{2})\\
            (1 + 2x + x^{2}) & (1 - x^{2}) & (1 - x^{2}) & (1 - 2x + x^{2})\\
        \end{pmatrix}\\
        &= \frac{1}{2}\begin{pmatrix}
            (1 - x)^{2} & (1 - x)(1 + x) & (1 - x)(1 + x) & (1 + x)^{2}\\
            (1 - x)(1 + x) & (1 + x)^{2} & (1 - x)^{2} & (1 - x)(1 + x)\\
            (1 - x)(1 + x) & (1 - x)^{2} & (1 + x)^{2} & (1 - x)(1 + x)\\
            (1 + x)^{2} & (1 - x)(1 + x) & (1 - x)(1 + x) & (1 - x)^{2}\\
        \end{pmatrix}
    \end{align*}
    Since $x \neq 0$, we see that this matrix $M$
    has the exact form as in \cref{lemma: fullRankSpecialHardness}
    with $p = (\nicefrac{1}{\sqrt{2}})(1 - x)$, and 
    $q = (\nicefrac{1}{\sqrt{2}})(1 + x)$.
    So, we already know that $\PlEVAL(M)$ is $\#$P-hard.
    Similarly, when $v = 0$, it is seen that
    $M$ has the exact  form as in \cref{lemma: fullRankSpecialHardness}
    with $p = (\nicefrac{1}{\sqrt{2}})(1 + x)$, and 
    $q = (\nicefrac{1}{\sqrt{2}})(1 - x)$,
    and therefore, $\PlEVAL(M)$ is $\#$P-hard.

    We can therefore now assume that $u \neq 0$ and $|u| \neq \sqrt{2}$.    
    We will again define $\xi$ such that
    $\xi: N \mapsto (N_{11}N_{14} - N_{12}N_{13})$.
    We note that for odd $n \geq 1$,
    $M^{n} = KD^{n}K^{\tt{T}}$, where
    $D^{n} = (\kappa)^{n} \cdot \text{diag}(1, x^{n} -x^{n}, x^{2n})$.
    So, we note that
    \begin{align*}
        (M^{n})_{11} = (M^{n})_{22} = (M^{n})_{33} = (M^{n})_{44}
        &=  \kappa^{n} \cdot \left(
        \frac{1 + (u^{2} - v^{2})x^{n} + x^{2n}}{4} \right),\\
        (M^{n})_{12} = (M^{n})_{34}
        &=  \kappa^{n} \cdot \left(
        \frac{1 -2uvx^{n} - x^{2n}}{4}\right),\\
        (M^{n})_{13} = (M^{n})_{24}
        &=  \kappa^{n} \cdot \left(
        \frac{1 + 2uvx^{n} - x^{2n}}{4} \right),\\
        (M^{n})_{14} = (M^{n})_{23}
        &=  \kappa^{n} \cdot \left(
        \frac{1 + (v^{2} - u^{2})x^{n} + x^{2n}}{4} \right).
    \end{align*}
    We will now let 
    $N_{n} = (T_{2}((\nicefrac{4}{\kappa^{n}})M^{n}))^{2}$
    for any odd $n \geq 1$.
    We note that
    $$(N_{n})_{11}
    = (1 + (u^{2} - v^{2})x^{n} + x^{2n})^{4}
    + (1 - 2uvx^{n} - x^{2n})^{4} + (1 + 2uvx^{n} - x^{2n})^{4}
    + (1 + (v^{2} - u^{2})x^{n} + x^{2n})^{4},$$
    $$(N_{n})_{12}
    = 2((1 + (u^{2} - v^{2})x^{n} + x^{2n})(1 - 2uvx^{n} - x^{2n}))^{2}
    + 2((1 + (v^{2} - u^{2})x^{n} + x^{2n})(1 + 2uvx^{n} - x^{2n}))^{2},$$
    $$(N_{n})_{13}
    = 2((1 + (u^{2} - v^{2})x^{n} + x^{2n})(1 + 2uvx^{n} - x^{2n}))^{2}
    + 2((1 + (v^{2} - u^{2})x^{n} + x^{2n})(1 - 2uvx^{n} - x^{2n}))^{2},$$
    $$(N_{n})_{14}
    = 2((1 + (u^{2} - v^{2})x^{n} + x^{2n})
    (1 + (v^{2} - u^{2})x^{n} + x^{2n}))^{2}
    + 2((1 - 2uvx^{n} - x^{2n})(1 + 2uvx^{n} - x^{2n}))^{2}.$$
    If we compute $\xi(N_{n})$, it can be verified that
    $$\xi(N_{n}) = f_{4}(u, v)x^{4n} + f_{6}(u, v)x^{6n} + 
    f_{8}(u, v)x^{8n} + f_{10}(u, v)x^{10n} + f_{12}(u, v)x^{12n},$$
    where
    $$f_{4}(u, v) = f_{12}(u, v) = 
    -48u^8 - 48v^8 + 576u^6v^2 + 576u^2v^6 +
    128u^4 + 128v^4 - 768u^2v^2 - 1824u^4v^4 + 256,$$
    \begin{multline*}
        f_{6}(u, v) = f_{10}(u, v) = 
        16u^{12} - 160u^{10}v^2 + 240u^8v^4 + 832u^6v^6 + 
        240u^4v^8 - 160u^2v^{10} + 16v^{12} \\
        - 128u^8 + 512u^6v^2 + 1280u^4v^4 + 512u^2v^6 - 
        128v^8 + 256u^4 + 512u^2v^2 + 256v^4,
    \end{multline*}
    \begin{multline*}
        f_{8}(u, v) = 4u^{16} - 32u^{14}v^2 - 16u^{12}v^4
        + 288u^{10}v^6 + 536u^8v^8 + 288u^6v^{10} 
        - 16u^4v^{12} - 32u^2v^{14} + 4v^{16}\\
        + 32u^{12} - 64u^{10}v^2 - 544u^8v^4 
        - 896u^6v^6 - 544u^4v^8 - 64u^2v^{10}
        + 32v^{12} - 288u^8\\ 
        + 384u^6v^2 - 4800u^4v^4 + 384u^2v^6 - 288v^8
        + 256u^4 - 1536u^2v^2 + 256v^4 + 512.
    \end{multline*}
    We will focus on $f_{4}(u, v)$. 
    We note that $u^{2} + v^{2} = 2$.
    So, $f_{4}(u, v) = 0$ if and only if
    \begin{multline*}
        -16\Big(3u^{8} + 3(2 - u^{2})^{4} - 36u^{6}(2 - u^{2}) 
        - 36u^{2}(2 - u^{2})^{3} - 8u^{4} - 8(2 - u^{2})^{2}\\
        + 48u^{2}(2 - u^{2}) + 114u^{2}(2 - u^{2})^{2}
        - 16\Big) = 0.
    \end{multline*}
    After simplifying, we find that $f_{4}(u, v) = 0$ if and only if
    $$-1024(3u^8 - 12u^6 + 14u^4 - 4u^2) = 0.$$
    It can be verified that
    $$3u^8 - 12u^6 + 14u^4 - 4u^2 = 3u^{2}(u^{2} - 2)
    \left(u^{2} - 1 - \frac{1}{\sqrt{3}} \right)
    \left(u^{2} - 1 + \frac{1}{\sqrt{3}} \right).$$
    Since we have already assumed that that $u \neq 0$ and 
    $|u| \neq \sqrt{2}$,
    we see that $f_{4}(u, v) = 0$ if and only if
    $u^{2} = (1 - \nicefrac{1}{\sqrt{3}})$
    and $v^{2} = (1 + \nicefrac{1}{\sqrt{3}})$, or
    if $u^{2} = (1 + \nicefrac{1}{\sqrt{3}})$ and
    $v^{2} = (1 - \nicefrac{1}{\sqrt{3}})$.
    In either case, it is seen that
    $f_{6}(u, v) = (\nicefrac{16384}{9}) \neq 0$.
    So, we find that
    if $f_{4}(u, v) = 0$, then
    $f_{6}(u, v) \neq 0$.
    Following the same argument as in the proof
    \cref{lemma: fullRankHardnessCase3} lets us find
    some odd $n^{*} \geq 1$
    such that $\xi(N_{n^{*}}) \neq 0$.
    Since $\xi(N_{n^{*}}) = \left(\nicefrac{4}{\kappa^{n^{*}}}\right)^{4}
    \xi((T_{2}(M^{n^{*}}))^{2})$, this implies that
    $\xi((T_{2}(M^{n^{*}}))^{2}) \neq 0$.

    We will now let $M' = M^{n^{*}}$.
    Now, by the same argument as in the proof of
    \cref{lemma: fullRankHardnessCase1}, we can now find some
    $M'' \in \text{Sym}_{4}^{\tt{F}}(\mathbb{R}_{> 0})$, such that
    $\PlEVAL(M'') \leq \PlEVAL(M')$, and
    $\rho_{\tt{tensor}}((M'')^{2}) \neq 0$.
    So, \cref{theorem: TensorOrHard} once again lets us prove
    that $\PlEVAL((M'')^{2}) \leq \PlEVAL(M') \leq \PlEVAL(M)$
    is $\#$P-hard.
\end{proof}

We are finally ready to prove the following theorem.

\begin{theorem}\label{theorem: fullRankPositiveDichotomy}
    Let $M \in \text{Sym}_{4}^{\tt{F}}(\mathbb{R}_{> 0})$.
    Then $\PlEVAL(M)$ is $\#$P-hard, unless
    $M$ is isomorphic to $A \otimes B$ for some
    $A, B \in \text{Sym}_{2}^{\tt{F}}(\mathbb{R}_{> 0})$
    such that $A_{11} = A_{22}$, and $B_{11} = B_{22}$,
    in which case, $\PlEVAL(M)$ is polynomial time tractable.
\end{theorem}
\begin{proof}
    We will let $N = M^{2}$. Since $M \in \text{Sym}_{4}^{\tt{F}}
    (\mathbb{R}_{> 0})$, we see that
    $N \in \text{Sym}_{4}^{\tt{pd}}(\mathbb{R}_{> 0}).$
    We know from \cref{theorem: posDefPositiveDichotomy} that unless
    $N$ is isomorphic to $A' \otimes B'$ for some
    $A', B' \in \text{Sym}_{2}^{\tt{pd}}(\mathbb{R}_{> 0})$
    that satisfy the conditions that 
    $(A')_{11} = (A')_{22}$, and $(B')_{11} = (B')_{22}$,
    $\PlEVAL(N) \leq \PlEVAL(M)$ is $\#$P-hard.

    We will now consider $M$ such that $\rho_{\tt{tensor}}(M) \neq 0$,
    but $\rho_{\tt{tensor}}(N) = 0$.
    We may assume that after permutation of rows and columns of
    $M$ (and correspondingly $N$), that
    $N = A' \otimes B'$ for some $A'$, and $B'$ as above.
    In that case, we see that
    $\rho_{\tt{tensor}}(M) \neq 0$,
    but $\varrho_{\tt{tensor}}(N) = 0$.
    So, from \cref{lemma: fullRankExtension}, 
    \cref{lemma: fullRankHardnessCase1},
    \cref{lemma: fullRankHardnessCase2}, 
    \cref{lemma: fullRankHardnessCase3},
    and \cref{lemma: fullRankHardnessCase4},
    we see that $\PlEVAL(M)$ is $\#$P-hard.

    Finally, we consider the case where $\rho_{\tt{tensor}}(M) = 0$.
    We may assume that after permutation of rows and columns of
    $M$ (and correspondingly $N$), that
    $N = A' \otimes B'$ for some $A'$, and $B'$ as above.
    Since $(A')_{11} = (A')_{22}$, we note that
    $A' = HD'_{1}H^{\tt{T}}$, where
    $$H = \frac{1}{\sqrt{2}} \begin{pmatrix}
    1 & 1\\
    1 & -1 \\
    \end{pmatrix}, ~~\text{ and }~~ D'_{1} = \begin{pmatrix}
        A'_{11} + A'_{12} & 0\\
        0 & A'_{11} - A'_{12}\\
    \end{pmatrix}$$
    Similarly, since $(B')_{11} = (B')_{22}$, we can verify that
    $B' = HD'_{2}H^{\tt{T}}$, where $H$ is the same matrix
    as above, and
    $$D'_{2} = \begin{pmatrix}
        B'_{11} + B'_{12} & 0\\
        0 & B'_{11} - B'_{12}\\
    \end{pmatrix}$$
    We may let $\lambda'_{1} = A'_{11} + A'_{12}$, and $\lambda'_{2}
    = A'_{11} - A'_{12}$.
    Similarly, we may let $\mu'_{1} = B'_{11} + B'_{12}$, and
    $\mu'_{2} = B'_{11} - B'_{12}$.
    We note that since $A'_{12}, B'_{12} > 0$, it follows that
    $\lambda'_{1} > \lambda'_{2}$, and
    $\mu_{1}' > \mu'_{2}$.
    Moreover, since $A', B' \in 
    \text{Sym}_{2}^{\tt{pd}}(\mathbb{R}_{> 0})$,
    we also see that $\lambda'_{2}, \mu'_{2} > 0$.
    Now, it follows that
    $\lambda'_{1}\mu'_{1} > \lambda'_{1}\mu'_{2} > \lambda'_{2}\mu'_{2}$,
    and that
    $\lambda'_{1}\mu'_{1} > \lambda'_{2}\mu'_{1} > \lambda'_{2}\mu'_{2}$.    
    Now, if we let $H' = H \otimes H$, and
    $D' = D'_{1} \otimes D'_{2}$,
    we see that $N = A' \otimes B' = (H')D'(H')^{\tt{T}}$, where
    $$H' = \frac{1}{2} \begin{pmatrix}
    1 & 1 & 1 & 1\\
    1 & -1 & 1 & -1\\
    1 & 1 & -1 & -1 \\
    1 & -1 & -1&  1\\
    \end{pmatrix}, ~~\text{ and }~~ D' = \begin{pmatrix}
        \lambda'_{1}\mu'_{1} & 0 & 0 & 0\\
        0 & \lambda'_{1}\mu'_{2} & 0 & 0\\
        0 & 0 & \lambda'_{2}\mu'_{1} & 0\\
        0 & 0 & 0 & \lambda'_{2}\mu'_{2}\\
    \end{pmatrix}.$$

    We will now consider the eigenvalues of the matrix $M$.
    Since $M^{2} = N$, we know that the squares of the eigenvalues
    of $M$ must precisely be the eigenvalues of $N$.
    Moreover, since $M \in \text{Sym}_{4}(\mathbb{R}_{> 0})$,
    from the Perron-Frobenius Theorem, we know that
    $M$ has a unique positive eigenvalue  with the largest
    absolute value. Since $(\lambda_{1}'\mu_{1}')$ is the
    unique eigenvalue of $N$ with the largest absolute value,
    it follows that $\kappa_{1} = \sqrt{\lambda_{1}'\mu_{1}'}$ is one
    of the eigenvalues of $M$.
    If we now let $\lambda_{1} = \sqrt{\lambda_{1}'}$, and
    $\mu_{1} = \sqrt{\mu_{1}'}$, this means that
    $\kappa_{1}  = \lambda_{1}\mu_{1}$ is an eigenvalue of $M$.
    We also know that there exists some eigenvalue
    $\kappa_{2}$ of $M$ such that $(\kappa_{2})^{2} = \lambda_{2}'\mu_{1}'$.
    We will let $\lambda_{2} = (\nicefrac{\kappa_{2}}{\mu_{1}})
    = \pm \sqrt{\lambda_{2}'}$.
    Therefore, $\lambda_{2}\mu_{1}$ is an eigenvalue of $M$.
    Similarly, we know that there exists an eigenvalue $\kappa_{3}$ of $M$
    such that $(\kappa_{3})^{2} = \lambda_{1}'\mu_{2}'$.
    We can then let $\mu_{2} = (\nicefrac{\kappa_{3}}{\lambda_{1}})
    = \pm \sqrt{\mu_{2}'}$, such that
    $\lambda_{1}\mu_{2}$ is an eigenvalue of $M$.
    Finally, we know that there exists some eigenvalue
    $\kappa_{4}$ of $M$ such that $(\kappa_{4})^{2} = \lambda_{2}'\mu_{2}'$.
    Therefore, we see that $\kappa_{4} = \pm \sqrt{\lambda_{2}'\mu_{2}'}
    = \pm \lambda_{2}\mu_{2}$.
    As $\rho_{\tt{tensor}}(M) = 0$, $M$ is isomorphic to a tensor product.
    Then its 4 eigenvalues must satisfy an equation of the form
    $\kappa_{i} \kappa_{j} = \kappa_{k} \kappa_{\ell}$, for some
    $i, j, k, \ell$ with  $\{i, j, k, \ell \} = \{1,2,3,4\}$.
    Since $\kappa_{1}$
    and $\kappa_{4}$ have the maximum and minimum absolute value
    respectively
    among the 4 eigenvalues, we must have
    $\kappa_{1} \kappa_{4} = \kappa_{2} \kappa_{3}$.
    Since $\kappa_{1} = \lambda_{1}\mu_{1}$, 
    $\kappa_{2} = \lambda_{2}\mu_{1}$,
    $\kappa_{3} = \lambda_{1}\mu_{2}$,
    we must have  $\kappa_{4} = \lambda_{2}\mu_{2}$.
    %
    
    We now note that since $M^{2} = N$, $N \cdot M = M^{3} = M \cdot N$.
    Since these matrices commute, we see that they can
    both be diagonalized by some orthogonal matrix $K$.
    If we let $D = \text{diag}(\lambda_{1}\mu_{1}, \lambda_{1}\mu_{2},
    \lambda_{2}\mu_{1}, \lambda_{2}\mu_{2}) = D_{1} \otimes D_{2}$,
    where $D_{1} = \text{diag}(\lambda_{1}, \lambda_{2})$, and
    $D_{2} = \text{diag}(\mu_{1}, \mu_{2})$,
    We see that there exists some $K$
    such that $N = KD'K^{\tt{T}}$, and $M = KDK^{\tt{T}}$.
    
    If $\lambda_{1}'\mu_{2}' \neq \lambda_{2}'\mu_{1}'$,
    then $N$ has distinct eigenvalues, which implies that
    $H' = H \otimes H$ is the only matrix (upto scaling each
    column by $-1$) which can diagonalize $N$.
    This would then imply that $H'$ can also diagonalize $M$.
    This means that we can take $K = H' = (H \otimes H)$.
    In other words, $M = (HD_{1}H^{\tt{T}}) \otimes
    (HD_{2}H^{\tt{T}})$.
    
    If $\lambda_{1}'\mu_{2}' = \lambda_{2}'\mu_{1}'$,
    then $N$ can be diagonalized by any $K$ such that
    $$K = \frac{1}{2}\begin{pmatrix}
        1 & u & v & 1\\
        1 & -v & u & -1\\
        1 & v & -u & -1\\
        1 & -u & -v & 1\\
    \end{pmatrix}$$
    for $u, v \in \mathbb{R}$ such that $u^{2} + v^{2} = 2$.
    Moreover, $\lambda_{1}'\mu_{2}' = \lambda_{2}'\mu_{1}'$ implies
    that $\lambda_{1}\mu_{2} = \pm \lambda_{2}\mu_{1}$.
    If in fact, $\lambda_{1}\mu_{2} = \lambda_{2}\mu_{1}$,
    we can replace the two middle columns of $K$ with any
    two orthogonal vectors that lie in their span.
    So, we see that $M$ can also be diagonalized
    by $H' = H \otimes H$.
    So, once again, we find that $M = (HD_{1}H^{\tt{T}}) \otimes
    (HD_{2}H^{\tt{T}})$.
    Finally, if $\lambda_{1}\mu_{2} = - \lambda_{2}\mu_{1}$,
    we see that $M = KDK^{\tt{T}}$,
    such that $D = (\lambda_{1}\mu_{1}) \cdot \text{diag}
    (1, x, -x, -x^{2})$, where $x = (\nicefrac{\mu_{2}}{\mu_{1}})
    = -(\nicefrac{\lambda_{2}}{\lambda_{1}})$.
    So, from \cref{lemma: fullRankHardnessCase3}
    we see that unless $|u| = |v| = 1$, $\PlEVAL(M)$ is $\#$P-hard.

    Now assume $|u| = |v| = 1$.
    Since columns of $K$ can be scaled by $-1$,
    we may assume that $u = 1$.
    If $v = 1$, then we see that once again,
    $M$ is diagonalized by $K = H' = H \otimes H$.
    So, $M = (HD_{1}H^{\tt{T}}) \otimes
    (HD_{2}H^{\tt{T}})$.
    If $v = -1$, then we claim that
    $M = (HD_{2}H^{\tt{T}}) \otimes
    (HD_{1}H^{\tt{T}})$. Indeed, if we let
    \[K' = K \begin{pmatrix}
        1 & 0 & 0 & 0\\
        0 & 0 & 1 & 0\\
        0 & 1 & 0 & 0\\
        0 & 0 & 0 & 1\\
    \end{pmatrix}
    \begin{pmatrix}
        1 & 0 & 0 & 0\\
        0 & -1 &0 & 0\\
        0 & 0 & 1 & 0\\
        0 & 0 & 0 & 1\\
    \end{pmatrix}
    = \frac{1}{2}\begin{pmatrix}
        1 & 1 & 1 & 1\\
        1 & -1 & 1 & -1\\
        1 & 1 & -1 & -1\\
        1 & -1 & -1 & 1\\
    \end{pmatrix}
    = H \otimes H,
    \]
    then $M =  K D K^{\tt{T}}= K (D_1 \otimes D_2 )K^{\tt{T}} = K' (D_2 \otimes D_1 )K^{\tt{T}} = (HD_{2}H^{\tt{T}}) \otimes
    (HD_{1}H^{\tt{T}}) $.

    So, in any case, we see that $M = A \otimes B$ for some
    $A, B \in \text{Sym}_{2}^{\tt{F}}(\mathbb{R}_{> 0})$.
    Moreover, since $(A')_{11} = (A')_{22}$, and
    $(B')_{11} = (B')_{22}$, we see that
    $N_{11} = N_{22} = N_{33} = N_{44}$.
    Now, $N_{11} = N_{22}$ implies that
    \begin{align*}
        &((A_{11})^{2} + (A_{12})^{2}) \cdot ((B_{11})^{2} + (B_{12})^{2})
        = (M_{11})^{2} + (M_{12})^{2} + (M_{13})^{2} + (M_{14})^{2}\\
        &= (M_{12})^{2} + (M_{22})^{2} + (M_{23})^{2} + (M_{24})^{2}
        = ((A_{11})^{2} + (A_{12})^{2}) \cdot ((B_{12})^{2} + (B_{22})^{2}),
    \end{align*}
    which implies that $(B_{11})^{2} = (B_{22})^{2}$.
    Since $B \in \text{Sym}_{2}(\mathbb{R}_{> 0})$, in fact,
    we see that $B_{11} = B_{22}$.
    Similarly, $N_{11} = N_{33}$ would imply that
    $A_{11} = A_{22}$ as well.
    So, if $(A')_{11} = (A')_{22}$, and $(B')_{11} = (B')_{22}$,
    we see that $A_{11} = A_{22}$, and $B_{11} = B_{22}$ as well,
    in which case, 
    \cref{lemma: tensorTractable} implies that
    $\PlEVAL(M)$ is polynomial time tractable.

    So, we see that $\PlEVAL(M)$ is $\#$P-hard,
    unless $M$ is isomorphic to $A \otimes B$ for some
    $A, B \in \text{Sym}_{2}^{\tt{F}}(\mathbb{R}_{> 0})$ such that
    $A_{11} = A_{22}$, and $B_{11} = B_{22}$, in which case,
    $\PlEVAL(M)$ is polynomial time tractable.
\end{proof}

We can extend this dichotomy to all non-negative real valued
full rank matrices with minimal effort.

\begin{definition}\label{definition: domainSeparable}
    Let $M \in \text{Sym}_{q}(\mathbb{R}_{\geq 0})$.
    We say that $M$ is domain separable if it is
    isomorphic to some $A \oplus B = \left(
    \begin{smallmatrix}
        A & \mathbf{0}\\
        \mathbf{0} & B
    \end{smallmatrix}\right)$ for some non-empty matrices
    $A \in \text{Sym}_{q_{1}}(\mathbb{R}_{\geq 0})$, and
    $B \in \text{Sym}_{q_{2}}(\mathbb{R}_{\geq 0})$ where
    $q_{1} + q_{2} = q$.
\end{definition}

It is known from \cite{cai2013graph} (Lemma 4.6, p.~940, the proof of which uses Lemma 4.1, p.~937,
called the  first pinning lemma)
that $\EVAL(M)$ is $\#$P-hard iff
at least one of  $\EVAL(A)$ or $\EVAL(B)$ is $\#$P-hard,
and that $\EVAL(M)$ is polynomial time tractable iff both 
$\EVAL(A)$ and $\EVAL(B)$ are polynomial time tractable.
Furthermore, the proof of Lemma 4.1 in \cite{cai2013graph} uses only planar gadgets
(and in that proof  we can place each identifying vertex, called $w$ and $w^*$'s in the paper,
to be on the outer face).
Thus, this proof works for planar graphs,
i.e., $\PlEVAL(M)$ is $\#$P-hard
iff at least one of $\PlEVAL(A)$ or $\PlEVAL(B)$ is 
$\#$P-hard, and that $\PlEVAL(M)$ is polynomial time
tractable iff both 
$\PlEVAL(A)$ and $\PlEVAL(B)$ are polynomial time tractable.
Now, we already have a dichotomy
from ~\cite{guo2020complexity, cai2023complexity}
for $\PlEVAL(M)$ when $M \in \text{Sym}_{q}(\mathbb{R}_{\geq 0})$
for $q < 4$.
So, that allows us to handle domain separable matrices with ease.

\begin{lemma}\label{lemma: domainSeparableDichotomy}
    Let $M \in \text{Sym}_{4}^{\tt{F}}(\mathbb{R}_{\geq 0})$
    such that $M = A \oplus B$ is domain separable.
    Then $\PlEVAL(M)$ is $\#$P-hard unless $\PlEVAL(A)$ 
    and $\PlEVAL(B)$ are both polynomial time tractable,
    in which case, $\PlEVAL(M)$ is also polynomial time tractable.
\end{lemma}

\begin{definition}\label{definition: biPartite}
    Let $M \in \text{Sym}_{q}(\mathbb{R}_{\geq 0})$.
    We say that $M$ is bipartite if 
    there exists some $A \in \mathbb{R}^{q_{1} \times q_{2}}$
    for some $q_{1} + q_{2} = q$, such that
    $M = \left(
    \begin{smallmatrix}
        \mathbf{0} & A\\
        A^{\tt{T}} & \mathbf{0}\\
    \end{smallmatrix}\right)$.
\end{definition}

\begin{lemma}\label{lemma: bipartiteDichotomy}
    Let $M \in \text{Sym}_{4}^{\tt{F}}(\mathbb{R}_{\geq 0})$
    be bipartite.
    Then $\PlEVAL(M)$ is $\#$P-hard unless
    $M$ is isomorphic to $A \otimes B$ for some
    $A, B \in \text{Sym}_{2}^{\tt{F}}(\mathbb{R}_{\geq 0})$
    such that $A_{11} = A_{22}$, and $B_{11} = B_{22}$,
    in which case, $\PlEVAL(M)$ is also polynomial time tractable.
\end{lemma}
\begin{proof}
    We may assume that $M = \left(
    \begin{smallmatrix}
        \mathbf{0} & A\\
        A^{\tt{T}} & \mathbf{0}\\
    \end{smallmatrix}\right)$,
    for some $A \in (\mathbb{R}_{\geq 0})^{q_{1} \times q_{2}}$,
    where $q_{1} + q_{2} = 4$.
    If $q_{1} = 1$ (or $q_{2} = 1$),
    we can see that $M$ has rank at most $2$.
    So, it is not possible that $M \in \text{Sym}_{4}^{\tt{F}}(\mathbb{R}_{\geq 0})$.

    So, we may assume that $q_{1} = q_{2} = 2$.
    Then, $M$ must be (upto some isomorphism) of the form
    $$M = \begin{pmatrix}
        0 & 0 & M_{13} & M_{14}\\
        0 & 0 & M_{23} & M_{24}\\
        M_{13} & M_{23} & 0 & 0\\
        M_{14} & M_{24} & 0 & 0\\
    \end{pmatrix}$$
    
    But then, we see that $(T_{n}M)^{2} = A_{n} \oplus A_{n}$, where
    $$A_n = \begin{pmatrix}
        (M_{13})^{2n} + (M_{14})^{2n} & (M_{13}M_{23})^{n} + (M_{14}M_{24})^{n}\\
         (M_{13}M_{23})^{n} + (M_{14}M_{24})^{n} & (M_{23})^{2n} + (M_{24})^{2n}\\
    \end{pmatrix}$$

    If there exists some $n \geq 1$ such that
    $(A_{n})_{11} \neq (A_{n})_{22}$, then
    from \cref{theorem: fullDomain2Hardness} and
    \cref{lemma: domainSeparableDichotomy}, it follows that
    $\PlEVAL((T_{n}M)^{2}) \leq \PlEVAL(M)$ is $\#$P-hard.

    If $(A_{n})_{11} = (A_{n})_{22}$ for all $n \geq 1$, that
    implies that $\{M_{13}, M_{14}\} = \{M_{23}, M_{24}\}$ as multi-sets.
    But if $M_{13} = M_{23}$, and $M_{14} = M_{24}$, then
    $M$ can once again, not be full rank.
    This implies that $M_{13} = M_{24}$, and $M_{14} = M_{23}$.
    But in this case, we see that
    $M = B \otimes A$, where $B = \left(\begin{smallmatrix}
        0 & 1\\
        1 & 0\\
    \end{smallmatrix}\right)$.
    In this case, we see from \cref{theorem: fullDomain2Hardness}
    that $\PlEVAL(A)$ and $\PlEVAL(B)$ are both polynomial
    time tractable.
    Since $Z_{M}(G) = Z_{A}(G) \cdot Z_{B}(G)$ for all
    graphs $G = (V, E)$, this implies that
    $\PlEVAL(M)$ is also polynomial time tractable.
\end{proof}

This lets us prove the following dichotomy theorem.

\begin{theorem}\label{theorem: fullRankNonNegativeDichotomy}
    Let $M \in \text{Sym}_{4}^{\tt{F}}(\mathbb{R}_{\geq 0})$.
    Then $\PlEVAL(M)$ is $\#$P-hard, unless
    one of the following conditions is true, in which case,
    $\PlEVAL(M)$ is polynomial tractable.
     \begin{align*}
         &\mbox{(1)~~} (\text{direct sum})~ M \cong A \oplus B \text{ for polynomially tractable }
            \PlEVAL(A), \PlEVAL(B),\\
	&\mbox{(2)~~} (\text{tensor product})~ M \cong A \otimes B \text{ for polynomially tractable }
            \PlEVAL(A), \PlEVAL(B).
     \end{align*}
\end{theorem}
\begin{proof}
    If $M$ is domain separable, i.e., $M \cong A \oplus B$ for any $A, B$, we know from
    \cref{lemma: domainSeparableDichotomy} that
    $\PlEVAL(M)$ is $\#$P-hard, unless
    $\PlEVAL(A)$, and $\PlEVAL(B)$ are both polynomially tractable,
    in which case, $\PlEVAL(M)$ is also polynomially
    tractable.
    If $M$ is bipartite, we know from
    \cref{lemma: bipartiteDichotomy} that $\PlEVAL(M)$ is 
    $\#$P-hard, unless $M \cong A \otimes B$, where
    $\PlEVAL(A)$, and $\PlEVAL(B)$ are both polynomially tractable,
    in which case, $\PlEVAL(M)$ is also polynomially tractable.

    Now, we consider any $M$ that is neither
    domain separable, nor bipartite.
    The underlying graph of $M$ where an edge $(i,j)$ exists iff  $M_{ij} >0$
    is connected and non-bipartite.
    In this case, since
    $M \in \text{Sym}_{4}(\mathbb{R}_{\geq 0})$,
    there exists some $n^{*} \geq 1$ such that $M^{n} \in
    \text{Sym}_{4}^{\tt{F}}(\mathbb{R}_{> 0})$ for all 
    $n \geq n^{*}$. 
    We let $n \geq n^{*}$ be some odd integer.
    From \cref{theorem: fullRankPositiveDichotomy},
    we know that $\PlEVAL(M^{n}) \leq \PlEVAL(M)$ is
    $\#$P-hard, unless $M^{n}$ is isomorphic to 
    $A' \otimes B'$,  for some
    polynomial time tractable $\PlEVAL(A'), \PlEVAL(B')$, 
    where $A', B' \in \text{Sym}_{2}^{\tt{F}}(\mathbb{R}_{> 0})$, such that
    $(A')_{11} = (A')_{22}$, and $(B')_{11} = (B')_{22}$.

    Without loss of generality, we will now consider the case
    where $M^{n} = A' \otimes B'$ for some
    $A', B' \in \text{Sym}_{2}^{\tt{F}}(\mathbb{R}_{> 0})$.
    Since $(A')_{11} = (A')_{22}$, and
    $(B')_{11} = (B')_{22}$, we know that
    $A' = H D'_{1} H^{\tt{T}}$, and
    $B' = H D'_{2} H^{\tt{T}}$, where
    $$H = \frac{1}{\sqrt{2}}\begin{pmatrix}
        1 & 1\\
        1 & -1\\
    \end{pmatrix}, D'_{1} = \begin{pmatrix}
        A'_{11} + A'_{12} & 0\\
        0 & A'_{11} - A'_{12}\\
    \end{pmatrix}, D'_{2} = \begin{pmatrix}
        B'_{11} + B'_{12} & 0\\
        0 & B'_{11} - B'_{12}\\
    \end{pmatrix}$$

    So, if we let $\lambda_{1}' = A'_{11} + A'_{12}$,
    $\lambda_{2}' = A'_{11} - A'_{12}$, 
    $\mu_{1}' = B'_{11} + B'_{12}$, and
    $\mu_{2}' = B'_{11} - B'_{12}$, we see that
    $M^{n} = A' \otimes B' = (H \otimes H) D' 
    (H \otimes H)^{\tt{T}}$, where
    $D' = \text{diag}(\lambda_{1}'\mu_{1}', \lambda_{1}'\mu_{2}',
    \lambda_{2}'\mu_{1}', \lambda_{2}'\mu_{2}')$.
    We also note that since $A', B' \in 
    \text{Sym}_{2}(\mathbb{R}_{> 0})$,
    $\lambda_{1}' > |\lambda_{2}'|$, and
    $\mu_{1}' > |\mu_{2}'|$.
    So, $|\lambda_{1}'\mu_{1}'| > |\lambda_{2}'\mu_{1}'|
    > |\lambda_{2}'\mu_{2}'|$, and
    $|\lambda_{1}'\mu_{1}'| > |\lambda_{1}'\mu_{2}'|
    > |\lambda_{2}'\mu_{2}'|$.
    If $\lambda_{2}'\mu_{1}' \neq \lambda_{1}'\mu_{2}'$,
    then $M^{n}$ would have four distinct eigenvalues.
    Otherwise, it will have three distinct eigenvalues,
    with one of them being repeated.
    
    Since $A' \otimes B' = M^{n}$,  we may assume that the
    eigenvalues of $M$ are: $\nu_{1}, \nu_{2}, \nu_{3}, \nu_{4}$,
    such that
    $(\nu_{1})^{n} = \lambda_{1}'\mu_{1}'$,
    $(\nu_{2})^{n} = \lambda_{1}'\mu_{2}'$,
    $(\nu_{3})^{n} = \lambda_{2}'\mu_{1}'$,
    $(\nu_{4})^{n} = \lambda_{2}'\mu_{2}'$.
    Since $n$ is an odd integer, and 
    $\nu_{i}, \lambda_{i}', \mu_{i}'$ are all real,
    this implies that
    $\nu_{1} = (\lambda_{1}'\mu_{1}')^{\frac{1}{n}}$,
    $\nu_{2} = (\lambda_{1}'\mu_{2}')^{\frac{1}{n}}$,
    $\nu_{3} = (\lambda_{2}'\mu_{1}')^{\frac{1}{n}}$,
    $\nu_{4} = (\lambda_{2}'\mu_{2}')^{\frac{1}{n}}$.
    Now, if we let $\lambda_{i} = (\lambda_{i}')^{\frac{1}{n}}$,
    and $\mu_{i} = (\mu_{i}')^{\frac{1}{n}}$
    for $i \in [2]$, we find that
    $\nu_{1} = \lambda_{1}\mu_{1}$,
    $\nu_{2} = \lambda_{1}\mu_{2}$,
    $\nu_{3} = \lambda_{2}\mu_{1}$,
    $\nu_{4} = \lambda_{2}\mu_{2}$.

    Since $M^{n} \cdot M = M \cdot M^{n}$, we know that 
    $M$ and $M^{n}$ can be simultaneously diagonalized by the same
    orthogonal matrix.
    Now, if $\lambda_{2}'\mu_{1}' \neq \lambda_{1}'\mu_{2}'$,
    we know that the diagonalizing orthogonal matrix for $M^{n}$ is essentially
    unique, whose columns are the unit column eigenvectors of $M^{n}$
    corresponding to the respective eigenvalues. This is
    the matrix $(H \otimes H) K$, where
    $K = \left( \begin{smallmatrix}
        \pm 1 & 0 & 0 & 0\\
        0 & \pm 1 & 0 & 0\\
        0 & 0 & \pm 1 & 0\\
        0 & 0 & 0 & \pm 1\\
    \end{smallmatrix} \right )$.
    Note that for any diagonal matrix $D'$, we have $KD' = D'K$.
    This implies that $M = (H D_{1} H^{\tt{T}}) \otimes
    (H D_{2} H^{\tt{T}})$, where
    $D_{1} = \text{diag}(\lambda_{1}, \lambda_{2})$, and
    $D_{2} = \text{diag}(\mu_{1}, \mu_{2})$.
    On the other hand, if 
    $\lambda_{2}'\mu_{1}' = \lambda_{1}'\mu_{2}'$, we see that
    $(\lambda_{2}'\mu_{1}')^{\frac{1}{n}} =
    (\lambda_{1}'\mu_{2}')^{\frac{1}{n}}$.
    So, $\nu_{3} = \lambda_{2}\mu_{1} 
    = \lambda_{1}\mu_{2} = \nu_{2}$ as well.
    Therefore, we have the following orthogonal  decomposition of
    $\mathbb{R}^4 = V_1 \oplus V_{23}  \oplus V_4$ as a  direct sum: 
    $V_1$ (respectively, $V_4$)  is a one-dimensional
    eigenspace  corresponding to the eigenvalue $\lambda_{1}'\mu_{1}'$ 
    (respectively, $\lambda_{2}'\mu_{2}'$) of
    $M^n$ which is also a one-dimensional
    eigenspace corresponding to the eigenvalue $\nu_1 = \lambda_{1}\mu_{1}$ (respectively, 
    $\nu_4 = \lambda_{2}\mu_{2}$) of
    $M$; $V_{23}$  is a two-dimensional
    eigenspace  where both $M^n$ and $M$  act as scalar matrices.
    Thus any orthogonal matrix that diagonalizes $M^n$ also diagonalizes $M$.
    So, once again, we see that the matrix $M$ can be
    diagonalized by $H \otimes H$.
    So, in either case, we find that
    $M = (H D_{1} H^{\tt{T}}) \otimes
    (H D_{2} H^{\tt{T}})$, where
    $D_{1} = \text{diag}(\lambda_{1}, \lambda_{2})$, and
    $D_{2} = \text{diag}(\mu_{1}, \mu_{2})$.

    If we let $A = H D_{1} H^{\tt{T}}$, and
    $B = H D_{2} H^{\tt{T}}$, we see that
    $M = A \otimes B$.
    Moreover, we find that
    $A_{11} = \frac{1}{2}(\lambda_{1} + \lambda_{2}) = A_{22}$,
    and $B_{11} = \frac{1}{2}(\mu_{1} + \mu_{2}) = B_{22}$.
    So, we see that when
    $M^{n} = A' \otimes B'$, where
    $(A')_{11} = (A')_{22}$, and $(B')_{11} = (B')_{22}$, then
    $M = A \otimes B$ for some $A, B$ such that
    $A_{11} = A_{22}$, and $B_{11} = B_{22}$.
    Moreover, in this case, it is easily seen that
    $\PlEVAL(M)$ is polynomially tractable.
\end{proof}

\printbibliography

\end{document}